\documentclass[10pt]{article}
\usepackage{amsthm,amsfonts,amssymb,euscript,amsmath,comment,slashed,color,enumerate,upgreek,mathrsfs}

\usepackage[affil-it]{authblk}
\usepackage{bbm}
\usepackage{graphicx}
\usepackage{enumerate}
\usepackage[margin=1in,dvips]{geometry}

\usepackage[utf8]{inputenc}
\usepackage[LGR,T1]{fontenc}

\newcommand{\f}{\frac}
\newcommand{\rd}{\partial}
\newcommand{\nab}{\nabla}
\newcommand{\alp}{\alpha}

\newcommand{\bt}{\beta}

\newcommand{\gi}{(g^{-1})}
\newcommand{\mfg}{\mathfrak g}

\newcommand{\ls}{\lesssim}
\newcommand{\de}{\delta}
 \usepackage{bm}
 \usepackage[hidelinks]{hyperref} 
 	\usepackage[dvipsnames]{xcolor}
 \def\f {\frac}
 \def\rd {\partial}
 \def\ls {\lesssim}
 \def\de {\delta}
 \def\ep {\epsilon}
 \def\i {\infty}
 \def\alp {\alpha}
 \def\bt {\beta}
 \def\Db {\langle D_x \rangle}
  
 \def\la {\langle}
 \def\ra {\rangle}
 \def\th {\theta}
 \def\nab {\nabla}
 \def\wo2 {\la x\ra^{-\f r2}}

    \newcommand{\Bes}{  B^{u_k,u_{k'}}_{\infty,1}(\Sigma_t)}

  \newcommand{\uphi}{\underline{\phi}}
  \newcommand{\uphil}{\underline{\phi}^{\bm \lambda}}
  \newcommand{\srd}{\slashed{\rd}}
  \newcommand{\partialuk}{  \srd_{{u}_k}}
   \newcommand{\partialukp}{  \srd_{{u}_{k'}}}

  \def\ls {\lesssim}
  \def\om {\omega}
  \def\th {\theta}
  \def\rd {\partial}
  
  \def\ep {\epsilon}

  \def\nab {\nabla}
  \def\f {\frac}

  \def\i {\infty}
  \def\de {\delta}
  \newcommand{\ud}{\mathrm{d}}
  \def\alp {\alpha}
  \def\bt {\beta}

  \def\Hb {\underline{H}}
  
  \def\mfg {\mathfrak{g}}
  \newcommand{\tphi}{ \widetilde{\phi}_{k}}
  \def\Omg {\Omega} 
  \newcommand{\Sd}{S_{\delta}}
    \newcommand{\sdelta}{\delta^{\frac{1}{2}}}
  
  \def\Db {\langle D_x \rangle}
   
  \def\la {\langle}
  \def\ra {\rangle}
  \def\th {\theta}
  \def\nab {\nabla}
  \def\wo2 {\la x\ra^{-\f r2}}
   \def\gi {(g^{-1})}

 \theoremstyle{plain}
 \newtheorem{theorem}{Theorem}[section]
\newtheorem{proposition}[theorem]{Proposition}
\newtheorem{lemma}[theorem]{Lemma}
\newtheorem{remark}[theorem]{Remark}
\newtheorem{rmk}[theorem]{Remark}
\newtheorem{definition}[theorem]{Definition}

 \newtheorem{thm}[theorem]{Theorem}

 \newtheorem{lem}[theorem]{Lemma}
 \newtheorem{prop}[theorem]{Proposition}
 \newtheorem{cor}[theorem]{Corollary}
 
 \newtheorem{defn}[theorem]{Definition}

 \newcommand{\RR}{\mathbb{R}}

 \newcommand{\Lgeo}{L_k^{geo}}

 \newcommand{\barL}{X_k}

 \newcommand{\n}{ \vec{n}}
 \newcommand{\rphi}{ \phi_{reg}}

 \newcommand{\pfstep}[1]{\vspace{.5em} {\it \noindent #1.}}

 \newcommand{\blue}[1]{{\color{blue} #1}}

 \def\f {\frac}
 \def\rd {\partial}
 \def\ls {\lesssim}
 \def\de {\delta}
 \def\ep {\epsilon}
 \def\i {\infty}
 \def\alp {\alpha}
 \def\bt {\beta}
 \def\Db {\langle D_x \rangle}
  
 \def\la {\langle}
 \def\ra {\rangle}
 \def\th {\theta}
 \def\nab {\nabla}
 \def\wo2 {\la x\ra^{-\f r2}}

 \DeclareMathOperator*{\esssup}{ess\,sup}

 \usepackage{color}

 \numberwithin{equation}{section}

\title{Nonlinear interaction of three impulsive gravitational waves I:\\
 main result and the geometric estimates}

\author{Jonathan Luk\thanks{jluk@stanford.edu}}
\affil{\small  Department of Mathematics, Stanford University, 450~Serra~Mall~Building~380,~Stanford~CA~94305-2125,~United~States~of~America \ }

\author{Maxime Van de Moortel\thanks{mmoortel@princeton.edu}}
\affil{\small  Department of Mathematics, Princeton University, Fine~Hall,~Washington~Road,~Princeton~NJ~08544,~United~States~of~America \ }

\begin{document}

\maketitle

\begin{abstract}
Impulsive gravitational waves are (weak) solutions to the Einstein vacuum equations such that the Riemann curvature tensor admits a delta singularity along a null hypersurface. The interaction of impulsive gravitational waves is then represented by the transversal intersection of these singular null hypersurfaces. 

This is the first of a series of two papers in which we prove that for all suitable $\mathbb U(1)$-symmetric initial data representing three ``small amplitude'' impulsive gravitational waves propagating towards each other transversally, there exists a local solution to the Einstein vacuum equations featuring the interaction of these waves. Moreover, we show that the solution remains Lipschitz everywhere and is $H^2_{loc} \cap C_{loc}^{1, \f 14-}$ away from the impulsive gravitational waves. This is the first construction of solutions to the Einstein vacuum equations featuring the interaction of three impulsive gravitational waves.

In this paper, we focus on the geometric estimates, i.e.~we control the metric and the null hypersurfaces assuming the wave estimates. The geometric estimates rely crucially on the features of the spacetime with three interacting impulsive gravitational waves, particularly that each wave is highly localized and that the waves are transversal to each other. In the second paper of the series, we will prove the wave estimates and complete the proof.
\end{abstract}

 	\tableofcontents

\section{Introduction}

It is well-known that  the Einstein vacuum equations
\begin{equation}\label{EE}
Ric(g) = 0
\end{equation}
admit (weak) solutions $(\mathcal M,g)$ in $(3+1)$ dimensions for which the Riemann curvature tensor admits delta singularities on a null hypersurface; see for instance \cite{PenroseIGW}. These are interpreted as \emph{impulsive gravitational waves}. 

Remarkably, an \emph{explicit} solution has been discovered by Khan--Penrose \cite{KhanPenrose} (see also \cite{Szekeres1}), in which there are two transversally intersecting null hypersurfaces on which (different) components of the Riemann curvature tensor admit delta singularities. This was interpreted as representing the interaction of two impulsive gravitational waves. The Khan--Penrose solution also exhibits interesting global properties in that the spacetime remains smooth locally beyond the interaction of the two impulsive gravitational waves but eventually a stronger Kasner-like spacelike singularity develops in the future \cite{Yurtsever88}. This is thought of as an idealized situation representing two very strong gravitational waves coming together from (infinitely) far-away strongly gravitating objects such that the interaction of the gravitational waves gives rise to a focusing effect and ultimately leads to a (more severe) singularity. After Khan--Penrose, there are many other explicit constructions of solutions featuring the interaction of impulsive gravitational waves, all of which rely on introducing a high degree of symmetry; see Section~\ref{sec:history.IGW}.

In \cite{LR1,LR2}, Luk--Rodnianski initiated the study of the propagation and interaction of impulsive gravitational waves without any symmetry assumptions. Even though the impulsive gravitational waves have lower regularity than that required for general local existence results \cite{sKiR2003, sKiRjS2015, hSdT2005}, Luk--Rodnianski developed a general local theory for solutions to the Einstein vacuum equations incorporating not only the propagation of one impulsive gravitational wave, but also the interactions of two impulsive gravitational waves. Their results can be summarized as follows:
\begin{theorem}[Luk--Rodnianski \cite{LR1,LR2}]\label{thm:LR}
Consider the characteristic initial value problem for the Einstein vacuum equations with characteristic initial data posed on two null hypersurfaces $H_0$ and $\Hb_0$ transversally intersecting at a spacelike $2$-sphere. Suppose on the initial hypersurface $H_0$ (respectively $\Hb_0$), the null second fundamental form has a jump discontinuity across the $2$-sphere $S_*$ (respectively $\underline{S}_*$) but smooth otherwise.

Then, assuming $S_*$ and $\underline{S}_*$ are sufficiently close to each other, there exists a unique local solution to the Einstein vacuum equations with two singular hypersurfaces emanating from the initial singularities $S_*$ and $\underline{S}_*$, and intersecting in the future. Moreover, the spacetime metric is everywhere Lipschitz and is smooth away from the union of null hypersurface emanating from $S_*$ and the null hypersurface emanating from $\underline{S}_{*}$.
\end{theorem}

Theorem~\ref{thm:LR} shows that at least \emph{locally} near the interaction, the structure of the spacetime (in terms of smoothness) is similar to that of the Khan--Penrose solution. It moreover provides the setup to understand more generally the \emph{global} structure of spacetimes.

However, all the existing examples and results cover only the interaction of two impulsive gravitational waves (despite the fact that the Luk--Rodnianski theory applies without any symmetries and allows for very general wave fronts). The question remains as to what is the structure of the spacetime singularities --- even locally! --- associated with the interaction of three impulsive gravitational waves coming in from different directions. In fact, there is not even a single example of a solution to \eqref{EE} featuring the transversal interaction of three impulsive gravitational waves. In particular, a symmetry assumption such as $\mathbb T^2$-symmetry is too restrictive to allow for the construction of such examples.

\textbf{The purpose of this work is to go beyond the interaction of \emph{two} impulsive gravitational waves and to consider the interaction of \emph{three} impulsive gravitational waves. Our main result is a local theory for vacuum spacetime solutions under polarized $\mathbb U(1)$ symmetry which feature the transversal interaction of three small amplitude impulsive gravitational waves.}

To see the difference between two and three impulsive gravitational waves, first recall that in the proof of Theorem~\ref{thm:LR}, one fundamental insight is that even though the spacetime metric necessarily very singular in the directions transversal to each of the impulsive gravitational waves, one can find \emph{two vector fields which are linearly independent at every spacetime point}, such that the spacetime metric is more regular when (Lie-)differentiated in the direction of these two vector fields. These vector fields are constructed with the use of a so-called double null foliation. As a result of the strong reliance of the double null foliation, the methods of Theorem~\ref{thm:LR} cannot be extended to in the case of three impulsive gravitational waves, which necessarily requires new techniques. Moreover, known results on much weaker singularities for semi-linear problems suggest that the local singularity structure after the interaction of three impulsive gravitational waves may even be qualitatively different from that for two impulsive gravitational waves \cite{jRmR1982}; see further discussions in Remark~\ref{rmk:higher.reg} and Section~\ref{sec:semilinear}.

To make the problem slightly more tractable, we impose the simplifying assumption that the spacetime is polarized $\mathbb U(1)$ symmetric, i.e.~we consider an ambient manifold $(I\times \mathbb R^{3}, ^{(4)}g)$, where $I\subseteq \mathbb R$ is an interval, and stipulate that the metric takes the following ansatz
$$^{(4)}g= e^{-2\phi}g+ e^{2\phi}(dx^3)^2,$$
where $\phi:I\times \mathbb R^{2}\to \mathbb R$ is a scalar function and $g$ is a Lorentzian metric on $I\times \mathbb R^{2}$. The Einstein vacuum equations then reduce to the $(2+1)$-dimensional Einstein--scalar field problem
\begin{equation}\label{eq:Einstein.scalar.field}
\begin{cases}
Ric(g) = 2 \ud \phi \otimes \ud \phi \\
\Box_g \phi = 0
\end{cases},
\end{equation}
which simplifies the analysis. Notice that unlike $\mathbb T^2$-symmetry, in our setting the symmetry group is one-dimensional and therefore the transversal interaction of three impulsive gravitational waves is still allowed.

The setup of the problem (see the precise statements in Section~\ref{dataroughsection}) is the following. Consider initial data which are compactly supported such that $\rd_i \phi$ and $\rd_t\phi$ are small in $L^\i$ and are smooth except along three lines $\{ \ell_k \}_{k=1}^3$ where they have a (small) jump discontinuity. Moreover, prescribe the jump discontinuity in a manner such that locally they propagate towards each other, and, assuming that the metric remains $C^1$-close to Minkowski, arrange them to interact before time $t=1$. Prescribe $g$ by solving the constraint equations and imposing suitable gauge conditions (using modifications of methods in \cite{Huneau.constraints}).

The following is an informal version of our main theorem (see Section~\ref{sec:precise.statements} for a more precise statement):
\begin{theorem}\label{thm:intro}
Given a polarized $\mathbb U(1)$ symmetric initial data set corresponding to three (non-degenerate) small-amplitude impulsive gravitational waves propagating towards each other, there exists a weak solution to the Einstein vacuum equations corresponding to the given data up to and beyond the transversal interaction of these waves. In particular, in the solution, the metric is everywhere Lipschitz and is $H^2_{loc}\cap C^{1,\th}_{loc}$ for some $\th \in (0, \f 14)$ away from the three null hypersurfaces corresponding to the impulsive gravitational waves. 
\end{theorem}

\begin{figure}
\centering
\includegraphics[width=60mm]{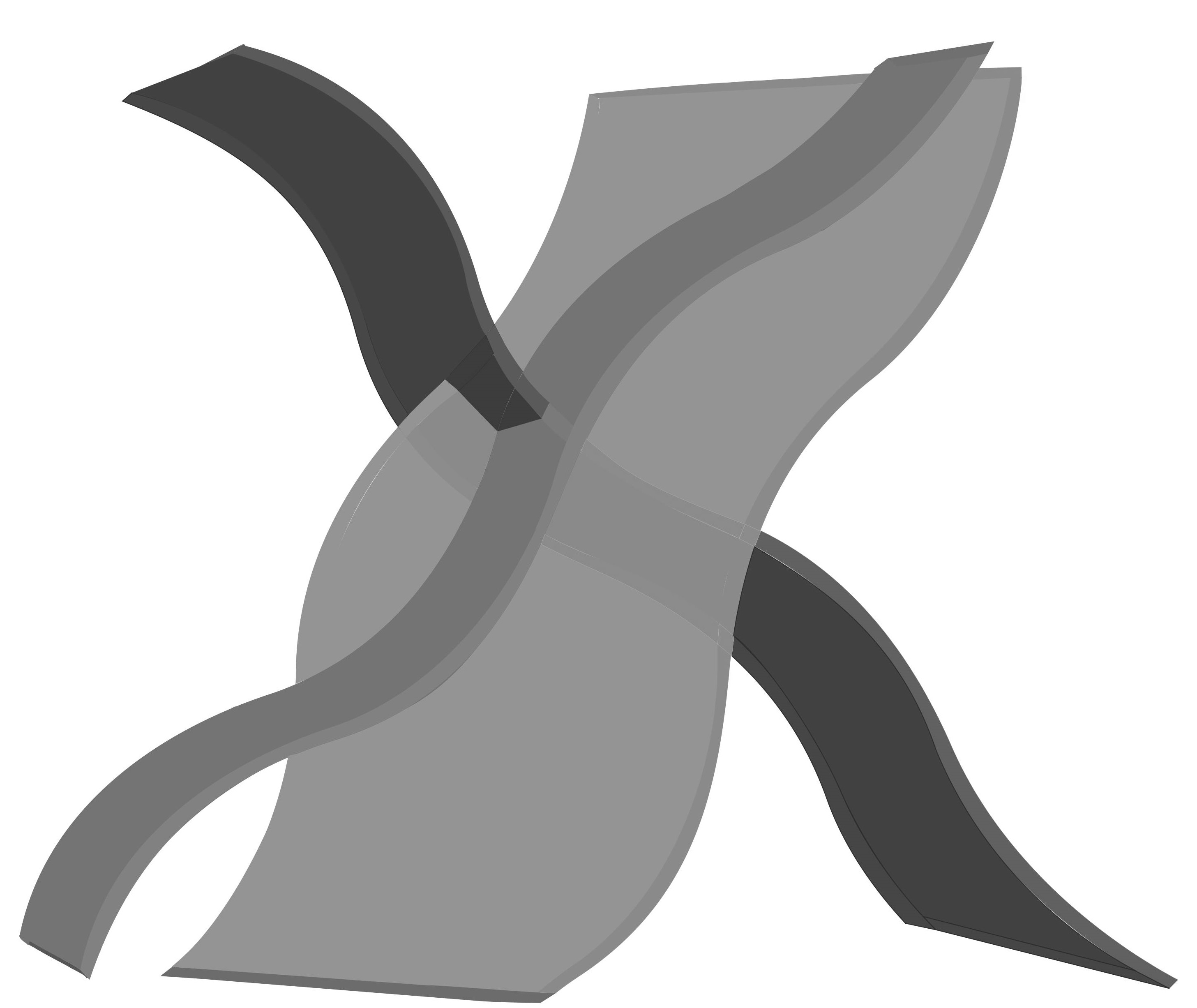}
\caption{Nonlinear interaction of three impulsive gravitational waves}
\label{fig:only.figure}
\end{figure}

A few remarks regarding Theorem~\ref{thm:intro} are in order:

\begin{remark}[More than three waves]\label{rmk:>3}
While Theorem~\ref{thm:intro} only explicitly treats the case of three transversally interacting impulsive gravitational waves, the techniques that are introduced can handle initial data featuring any finite number of impulsive gravitational waves propagating in different directions.

When there are more than three waves in the initial data, generically at each interaction point only three waves interact. Moreover, even when four waves are arranged to interact at the same point in the reduced $2+1$-dimensional spacetime, in the original $3+1$-dimensional spacetime, the waves interact at a one-dimension curve, which should be considered as a non-generic case. Put differently, to understand genuine interaction of four impulsive gravitational waves, it seems necessary to relax the symmetry assumption.
\end{remark}

\begin{remark}[Anisotropic estimates in $L^2$ spaces]
Even though Theorem~\ref{thm:intro} is most conveniently stated in terms of isotropic spaces (Lipschitz, $C^{1,\alp}$ and $H^2$), when we prove the Lipschitz and $C^{1,\alp}$ estimates, we need to first obtain higher regularity estimates in $L^2$ based spaces with respect to some geometrically defined vector fields; see Section~\ref{sec:aprioriestimates2}.
\end{remark}

\begin{remark}[$\de$-impulsive gravitational waves]\label{rmk:smooth.approx}
Impulsive gravitational waves should be viewed as an idealized description of very strong and localized gravitational waves. It may be argued that instead of having a solution whose curvature has a delta singularity, a more physically relevant description would be a \underline{smooth} solution whose curvature scales like an \underline{approximate delta singularity}. To capture this more general class of solutions, we introduce the notion of $\de$-impulsive gravitational waves (for $\de >0$), which roughly speaking corresponds to solutions to the Einstein vacuum equations whose Riemann curvature tensor is of amplitude $O(\de^{-1})$ in a $\de$-neighborhood of a null hypersurface (and of $O(1)$ otherwise). The class of $\de$-impulsive gravitational waves also includes the particular cases where the curvature profile has zero average (which is not possible with a delta function) as discussed \cite{KhanPenrose, Szekeres2}.

In this paper, we will also prove a version of Theorem~\ref{thm:intro} for the nonlinear interaction of three $\de$-impulsive gravitational waves (for all small $\de>0$); see Theorem~\ref{smooththeorem}. In fact, our approach in proving Theorem~\ref{thm:intro} for the impulsive gravitational waves proceeds by first approximating the impulsive wave data by those of $\de$-impulsive waves, and then passing to the $\de \to 0$ limit; see Section~\ref{sec:intro.delta}.

We remark also that the non-degeneracy condition in Theorem~\ref{thm:intro} (see 7 in Definition~\ref{roughdef}) will only be used for solving the constraint equations, in order to show that the impulsive wave data can indeed be approximated by $\de$-impulsive wave data.
\end{remark}

\begin{remark}[Uniqueness]\label{rmk:uniqueness}
We note explicitly that our proof does not give a uniqueness statement as we rely on a compactness argument. 
\end{remark}

\begin{remark}[Relation to low regularity problem]
The main difficulty of Theorem~\ref{thm:intro} (and of understanding impulsive gravitational waves in general) is the low regularity of the initial data. Without any symmetry assumptions, the best-known general local result is the celebrated bounded $L^2$ curvature theorem, which requires the initial data to have curvature in $L^2$ \cite{sKiRjS2015}, while impulsive gravitational waves have much lower regularity.

In this paper, we impose polarized $\mathbb U(1)$ symmetry, and under such a symmetry assumption, local well-posedness holds in a lower regularity than in the general $(3+1)$-dimensional case without symmetry. For $\mathbb U(1)$ symmetry (even without polarization), the results of \cite{hSdT2005} imply that local well-posedness can be obtained for in $H^{\f 74+\ep}$. While the optimal regularity in polarized $\mathbb U(1)$ symmetry is not explicitly discussed in the literature, some interesting progress has been made on a related quasilinear model problem \cite{hBjyC2001, Huneau.mastersthesis}.

Note that the initial data that we consider barely fail to be in the $H^{\f 32}$ space, and are below the threshold for any standard theorem.
More importantly, our main focus is not just to obtain a local existence result. Instead, we also obtain control of the Lipschitz norm and a finer description of the singularity structure. 
\end{remark}

\begin{remark}[Higher regularity]\label{rmk:higher.reg}
In the case of the interaction of two impulsive gravitational waves (recall Theorem~\ref{thm:LR}), the spacetime metric is smooth away from the impulsive gravitational waves. In our setting, the improved regularity we obtain away from the impulsive gravitational waves is only in class $C^{1,\alp}$. However, in view of some examples for even much weaker singularities for some simpler model semilinear problems (see~Section~\ref{sec:semilinear}), one may conjecture that in our setting the spacetime metric is \underline{not} smooth away from the union of impulsive gravitational waves, and that there is a weaker singularity that emanated from the intersection point of the three impulsive gravitational waves. It would be interesting to understand what is the optimal regularity that can be obtained.
\end{remark}

Our proof of Theorem~\ref{thm:intro} has three main components, which are highly coupled to each other. 
\begin{enumerate}
\item Control the geometric quantities, including the metric components, the null hypersurfaces and the commuting vector fields, assuming suitable bounds on the scalar wave. 
\item Show that the $L^2$-based wave energy estimates for the scalar wave imply via an anisotropic Sobolev embedding theorem that the scalar field is everywhere Lipschitz with improved H\"older regularity away from the singular hypersurfaces.
\item Prove $L^2$-based wave energy estimates for the scalar wave with appropriately chosen commutators.
\end{enumerate}
The three parts are of somewhat different nature, and are further discussed in Section~\ref{sec:aprioriestimates1}--\ref{sec:aprioriestimates3} respectively. In this paper, we will discuss the relation between the three steps, and carry out Step 1. Steps 2 and 3 will be performed in the companion paper \cite{LVdM2}.

The remainder of the introduction is structured as follows: In \textbf{Section~\ref{sec:method}} we will give a brief indication of the ideas used in the proof of Theorem~\ref{thm:intro}, emphasizing the ideas for the geometric estimates. In \textbf{Section~\ref{sec:discussions}}, we will discuss some related works. In \textbf{Section~\ref{sec:open.problems}}, we then give a list of some related problems. Finally, in \textbf{Section~\ref{sec:outline}}, we will outline the remainder of the paper.

\subsection{Ideas of the proof}\label{sec:method}

\subsubsection{$\de$-impulsive gravitational waves}\label{sec:intro.delta}

Rather than directly constructing a solution with the impulsive wave data, our strategy will be to consider a $\de$-approximate problem and the pass to the $\de \to 0$ limit. Recall that in the reduction \eqref{eq:Einstein.scalar.field}, the original Einstein vacuum equations reduce to a lower dimensional Einstein--scalar field system. This naturally separates the estimates into the scalar field part and the geometry part. For the impulsive gravitational wave problem, the scalar field $\phi$ (in the reduced system) is only Lipschitz, and $\rd \phi$ has jump discontinuities along three different hypersurfaces. We will instead first regularize the initial data, so that the data are smooth, and while $\rd\phi$ remains $O(\ep)$, we only have $\rd^2\phi = O(\ep \de^{-1})$ in a $\de$-neighborhood of three curves (for $\de \ll \ep \ll 1$). We will call these regularized waves the $\de$-impulsive gravitational waves; see precise conditions in Section~\ref{datasmoothsection}. 

The advantage of first considering the $\de$-impulsive waves before passing to the limit is that we can bound some norms which blow up in a controlled manner in terms of $\de^{-1}$. This is useful in the analysis because after introducing a suitable decomposition (see \eqref{eq:intro.decomposition} below), some quantities are \emph{small} in terms of $\de$, and can compensate for the large $\de^{-1}$ powers in the estimates. This is reminiscent of Christodoulou's short pulse method; see \cite{dC2009}.

The challenge will now be to show that for all $\delta>0$ sufficiently small, (1) there is a uniform time of existence of the solutions,  that (2) we can prove some estimates that are independent of $\delta$, and that (3) the estimates are sufficiently strong for us to pass to the $\de\to 0$ limit to obtain a solution.

We rely on a compactness argument to extract the $\de \to 0$; for this reason, we do not prove uniqueness (see Remark~\ref{rmk:uniqueness}). Note the importance to have \emph{strong} convergence of $(\phi,g)$ in $H^1$ since general weak $H^1$ limits of solutions to the Einstein vacuum equations are not necessarily (weak) vacuum solutions (see \cite{Burnett,HLHF}).

\subsubsection{Geometric constructions and the choice of gauge}\label{sec:intro.gauge}

The choice of gauge plays a fundamental role for low-regularity problems in general relativity. In the present work, we will in fact need to choose multiple gauges: one global system of coordinates determined by an elliptic gauge and three sets of null coordinates. The global elliptic gauge is chosen to maximize the regularity of the (reduced $(2+1)$-dimensional) metric coefficients given the low regularity setting, and each set of null coordinates is adapted to each propagating impulsive gravitational wave. Because we use multiple sets of coordinates, it is also important to control the transformation between any two sets of coordinates.

\textbf{Elliptic gauge.} Since the $\mathbb U(1)$-reduced problem \eqref{eq:Einstein.scalar.field} is effectively $(2+1)$-dimensional, the scalar field, which determines the Ricci curvature tensor, completely determines the Riemann curvature tensor.  In order to maximize the gain in regularity when reconstructing the metric from the curvature tensor, we use an elliptic gauge. More precisely, we foliate the spacetime by maximal hypersurfaces $\{\Sigma_t\}_{t\in [0,T]}$ and we choose spatial coordinates such that the induced metric is conformal to the flat metric on each $t\in [0,T]$. In doing so, each metric component $\mathfrak g$ obeys a spatial elliptic equation schematically of the form
\begin{equation}\label{eq:intro.g.elliptic}
\Delta \mathfrak g = (\rd \phi)^2 + (\rd_x \mathfrak g)^2,
\end{equation}
where $\Delta$ is the flat Laplacian and we use the convention that $\rd_x$ denotes a spatial derivative, while $\rd$ denotes either a spatial or a time derivative.

\textbf{Eikonal functions, null frames, and geometric coordinates.} To understand the propagation of the $\de$-impulsive waves, a crucial role is played by the \emph{eikonal functions} $u_k$, for $k=1,2,3$. Each $u_k$ is defined so that its level sets correspond to the null hypersurfaces along which one of the $\de$-impulsive wave propagates.

Associated with each eikonal function $u_k$, we will introduce
\begin{itemize}
\item a null frame $(L_k, E_k, X_k)$ such that $L_k$ and $E_k$ are tangential to constant-$u_k$ (null) hypersurfaces, and
\item a system of geometric coordinates $(u_k, t_k,\th_k)$ with $u_k$ as above, $t_k = t$ and $\th_k$ transported by $L_k \th_k = 0$.
\end{itemize}

The significance of the eikonal functions and the null frames lie in that
\begin{itemize}
\item $u_k$ captures the location of each $\de$-impulsive wave. In particular, for the $k$-th wave\footnote{To make precise the notion of the three different propagating $\de$-impulsive waves require a decomposition of $\phi$; see \eqref{eq:intro.decomposition} below.}, the most singular behavior is only expected in $u_k \in [-\de, \de]$.
\item The $L_k$ and $E_k$ vector fields corresponds to regular directions for the $k$-th wave. In other words, the $L_k$ and $E_k$ derivatives are better behaved than a generic derivative.
\end{itemize}

While the geometric constructions associated with $u_k$ are important for capturing the propagation of the $\de$-impulsive waves, in order for them to be useful, we need to obtain the relevant geometric control, including estimating the connection coefficients such as $\nab_{L_k} E_k$ etc. All the connection coefficients can algebraically determined from $\chi_k := g(\nab_{E_k} L_k, E_k)$, $\eta_k:=g(\nab_{X_k}L_k, E_k)$ and spatial derivatives of the metric coefficients $\mfg$ in the elliptic gauge coordinates. A bound for $\chi_k$ is particular means that the constant-$u_k$ null hypersurface are regular without conjugate points. 

The Einstein equations imply that $\chi_k$ and $\eta_k$ satisfy nonlinear transport equations (see \eqref{Leta}, \eqref{Lchi})
\begin{equation}\label{eq:Ray.intro}
L_k \chi_k = -2(L_k \phi)^2 + \cdots,\quad L_k \eta_k = -2(L_k\phi)(E_k \phi) + \cdots,
\end{equation}
where $\cdots$ are lower order terms.

\subsubsection{Estimates for the wave part}\label{sec:intro.waves}

The reduced equations \eqref{eq:Einstein.scalar.field} naturally divide the estimates into those for the wave part and for the geometric part. This paper is focused on the geometric part; we refer the reader to the introduction to \cite{LVdM2} for the discussion of the proof of the wave estimates. Nevertheless, since the geometric estimates are highly coupled with the wave estimates, we first point out the main wave estimates, before we explain how they dictate the geometric estimates that we prove.

The wave estimates we state here are natural from the point of view of propagation of singularities for linear wave equations. The much less obvious part, which will be addressed in \cite{LVdM2}, is that these estimates continue to hold in the quasilinear setting, particularly under the low regularity of the metric that we establish in this paper.

In order to capture the three propagating singularities,  we show that $\phi$ admits a decomposition 
\begin{equation}\label{eq:intro.decomposition}
\phi = \rphi + \sum_{k=1}^3 \tphi,
\end{equation}
where $\rphi$ is a ``regular'' part, and $\tphi$ are the ``singular'' parts, each corresponding to one of the impulsive waves. Each of these parts are defined to satisfy the wave equation, i.e.~$\Box_g \rphi = 0$ and $\Box_g \tphi = 0$.

The following are the most important features of $\rphi$ and $\tphi$.
\begin{enumerate}
\item On constant-$t$ hypersurfaces $\Sigma_t$, $\phi$ obeys the following isotropic bounds:
\begin{enumerate}
\item For fixed $s' \in (0, \f 12)$, $\|\phi \|_{H^{1+s'}(\Sigma_t)} + \|\rd \phi \|_{H^{s'}(\Sigma_t)} \ls \ep$.
\item $\|\rd \phi \|_{L^\i(\Sigma_t)} \ls \ep$. Importantly, this also holds with a Besov improvement\footnote{Notice that the Besov estimate implies that $\rd\phi$ is continuous, and thus fails for the (rough) impulsive gravitational waves. Importantly for our argument, the Besov estimate nonetheless holds for the $\de$-impulsive waves, uniformly for all sufficiently small $\de>0$.}. 
\end{enumerate}
\item The second derivatives of $\phi$ is not better than $\|\rd^2 \phi \|_{L^2(\Sigma_t)} \ls \ep \de^{-\f 12}$, but have the following features:
\begin{enumerate}
\item The regular part is better: $\| \rd \rphi \|_{H^{1+s'}(\Sigma_t)} \ls \ep$.
\item The bad part for $\rd^2 \tphi$ is only localized to $S^k_\de = \{ u_k \in [-\de,\de]\}$: in fact $\|\rd^2\tphi \|_{L^2(\Sigma_t \setminus S^k_\de)} \ls \ep$.
\item $L_k$ and $E_k$ are \emph{better} than general derivatives on $\tphi$: $\| \rd L_k \tphi \|_{L^2(\Sigma_t)} + \|\rd E_k \tphi \|_{L^2(\Sigma_t)} \ls \ep$. 
\end{enumerate}
\item The following flux estimates on constant-$u_{k'}$ null hypersurfaces $C^{k'}_{u_{k'}}$:
\begin{equation}\label{eq:intro.flux}
\sum_{Z_{k'} \in \{L_{k'}, E_{k'} \}} \|Z_{k'} \rd_x \tphi\|_{L^2(C^{k'}_{u_{k'}})} \ls \ep \de^{-\f 12}, \quad \sum_{Z_{k'} \in \{L_{k'}, E_{k'} \}} \|Z_{k'} \rd_x \rphi\|_{L^2(C^{k'}_{u_{k'}})} \ls \ep.
\end{equation}
There are two additional improvements to \eqref{eq:intro.flux}:
\begin{enumerate}
\item The bounds for $\tphi$ improve away from $S^k_\de = \{ u_k \in [-\de,\de]\}$:
\begin{equation}\label{eq:intro.flux.imp.1}
\sum_{Z_{k'} \in \{L_{k'}, E_{k'} \}} \|Z_{k'} \rd_x \tphi\|_{L^2(C^{k'}_{u_{k'}}\setminus S^k_\de)} \ls \ep.
\end{equation}
\item When $k=k'$, there is a $\de$-independent bound for all $u_k$ for a more restricted choice of derivatives:
\begin{equation}\label{eq:intro.flux.imp.2}
\|L_k \rd_x \tphi\|_{L^2(C^k_{u_k})} + \|E_k E_k \tphi\|_{L^2(C^k_{u_k})} \ls \ep.
\end{equation}
\end{enumerate}
\item Singular parts of \emph{different} impulsive waves $\widetilde{\phi}_k$ and $\widetilde{\phi}_{k'}$ ($k \neq k'$) are \emph{transversal} in a quantitative manner.
\end{enumerate}

For brevity, we have suppressed some additional wave estimates for $\phi$ which are proven and used in our arguments. Moreover, in order to obtain the Lipschitz bound for $\phi$, we need further bound $\|\rd E_k \rphi \|_{H^{s''}(\Sigma_t)} \ls \ep$ (for fixed $s'' \in (0, s')$). We will defer all these discussions to \cite{LVdM2}.

\subsubsection{Estimates for the metric component in the elliptic gauge}\label{sec:intro.elliptic} Terms related to the elliptic gauge are in principle the most regular due to the ellipticity of the equations \eqref{eq:intro.g.elliptic}. There are two main technical issues:
\begin{itemize}
\item In order to close our estimates, we need to bound $\rd^2_{ij} \mfg$ in $L^\i$. Since $(\rd \phi)^2 \in L^\infty$, this corresponds exactly to an end-point elliptic estimate that fails.
\item A priori, ellipticity only gains in \underline{spatial}, but not \underline{temporal} regularity.
\end{itemize}

\textbf{The easy estimates.} Since $\phi\in W^{1,\infty}\cap H^{1+s'}$ (1 in Section~\ref{sec:intro.waves}) and is compactly supported, standard elliptic estimates immediately imply that $\mathfrak{g} \in H^{2+s'}$ and $\mathfrak{g} \in W^{2,p}$ (with weights) for any $p\in [1,+\infty)$. The only subtlety concerns the weights at infinity: $\rd_x \mfg$ are no better than $\la x \ra^{-1}$ at $\infty$, and thus to handle the $(\rd_x \mfg)^2$ term on the RHS of \eqref{eq:intro.g.elliptic} requires using the precise structure of the nonlinear terms. These weight issues can be handled in a similar manner as \cite{Huneau.constraints, HL.elliptic}.

\textbf{The endpoint Besov space elliptic estimates.} To close our estimates we need further that $\rd^2_x \mfg \in L^\i$. This corresponds to the $p=+\infty$ case for the $L^p$ Calderon--Zygmund elliptic theory, which does \underline{not} hold. Hence to obtain the $W^{2,\infty}$ estimate we need a slightly stronger \underline{Besov} type estimate for $\rd\phi$ (recall 1(a) in Section~\ref{sec:intro.waves}). Our anisotropic Sobolev embedding theorem (see Theorem~\ref{thm:bootstrap.Li}), which we use to obtain Lipschitz bounds for $\phi$, naturally gives a Besov strengthening. However, the Besov estimates we get are with respect to $(u_k, u_{k'})$ coordinates ($k\neq k'$), and for this reason we introduce an extra physical space argument to obtain a good endpoint elliptic estimates for the operator $\Delta$ defined in the elliptic gauge coordinates.

\textbf{Estimates for $\rd_t \mfg$.} It is slightly more delicate to control second derivatives of the metric coefficients with one spatial and one $\rd_t$ derivative. Differentiating \eqref{eq:intro.g.elliptic} by $\rd_t$, we obtain an equation with the follow main term:
\begin{equation}\label{eq:dtN.intro}
\Delta \rd_t \mathfrak g = (\rd^2 \phi) (\rd \phi) + \dots
\end{equation}
In terms of $W^{s,p}$ spaces, the RHS is no better than being in $L^1$. This by itself would only give an estimate for $\rd_t \mathfrak g$ which is much worse than that for $\rd_i \mathfrak g$.

Here is the key idea: we decompose $\rd_t$ as linear combination of a good null derivative and spatial derivatives. Take for instance a contribution from two parallel waves, i.e.~a term $\rd_t  [(\rd \widetilde{\phi}_k)^2]$ on the RHS of \eqref{eq:dtN.intro}. With some (well-controlled) coefficients $\alp$ and $\sigma^i$, we have
$$\rd_t [(\rd \widetilde{\phi}_k)^2]  = \alp L_k [(\rd \widetilde{\phi}_k)^2] + \sigma^i \rd_i [(\rd \widetilde{\phi}_k)^2] =  \alp L_k [(\rd \widetilde{\phi}_k)^2] + \rd_i [\sigma^i  (\rd \widetilde{\phi}_k)^2] + \dots.$$
The $L_k$ derivative is a good derivative for $\tphi$ which then gives better estimates. The other term is essentially a total \emph{spatial} derivative so that we gain with the ellipticity of the equation \eqref{eq:dtN.intro}.

The above argument becomes more subtle when there is an interaction of two waves propagating in different directions. Nevertheless, for a term such as
$\rd_t [(\rd \widetilde{\phi}_j)(\rd \widetilde{\phi}_k)],$ (with $j\neq k$)
we exploit precisely that the waves are \emph{transversal} (point 4 in Section~\ref{sec:intro.waves}) and decompose $\rd_t$ is a \underline{spatially-dependent} manner. The resulting decomposition of $\rd_t$ is not regular so that we are not able to control $\rd_x \rd_t\mfg$ in $L^\i$, yet the error generated is sufficiently lower order that we still bound $\rd_x \rd_t \mfg$ in sufficiently high $L^p$ spaces, as well as in $H^{s'}$.

\subsubsection{Bounds for the eikonal functions and Ricci coefficients}\label{sec:intro.eikonal}
We now turn to the bounds for geometric quantities related to the eikonal functions $u_k$, particularly the Ricci coefficients $\chi_k$ and $\eta_k$, which obey the transport equations \eqref{eq:Ray.intro}. 

Using the transport equation \eqref{eq:Ray.intro} and the Lipschitz bound for $\phi$, it immediately follows that $\chi_k$ and $\eta_k$ are bounded in (weighted) $L^\i$. 
 
The first derivative estimates for $\chi_k$ and $\eta_k$ are more subtle. Clearly, $L_k \chi_k$ and $L_k \eta_k$ are in $L^\i$ by \eqref{eq:Ray.intro} and the above discussions. For the other first derivative, consider first $\chi_k$. We differentiate \eqref{eq:Ray.intro}  by $\rd_q$ to get the following schematic equation:
\begin{equation}\label{eq:LEchi.intro}
L_k \rd_q \chi_k = - 2(L_k \rd_q \phi)(L_k \phi) + \cdots,
\end{equation}
where $\cdots$ are lower order terms as before.

To control the second derivative term $L_k \rd_q \phi$ (recall that the first derivative $L_k\phi$ is bounded), we use the flux estimates (3 in Section~\ref{sec:intro.waves}). Decompose the top derivative $L_k \rd_q \phi= L_k \rd_q \phi_{reg} + \sum_{j=1,2,3} L_k \rd_q \widetilde{\phi}_j$. The regular part is well under control by \eqref{eq:intro.flux}, but for the singular parts, \eqref{eq:intro.flux} itself (which has a $\de^{-\f 12}$ weight) is too weak, and we need to separately consider the cases $k=j$ or $k\neq j$. 
\begin{itemize}
\item First, when $k=j$, we use the fact $L_k$ is a good derivative for $\phi_k$ and the term $L_k \rd_q \widetilde{\phi}_k$ is controlled by \eqref{eq:intro.flux.imp.2} without $\de^{-\f 12}$ weights.
\item Second, when $k \neq j$, the energy flux does \underline{not} give good estimates for  $L_k \rd_q \widetilde{\phi}_j$. Here, we use crucially the \emph{transversality} of the $\de$-impulsive waves (4 in Section~\ref{sec:intro.waves}). While the $L^2$ norm of $L_k \rd_q \widetilde{\phi}_k$ is large it is concentrated in $S^k_\de$, which is a small region of length scale $\de$. On the other hand, integrating \eqref{eq:LEchi.intro} requires an $L^1$ --- instead of $L^2$ --- estimate along the integral curve of $L_k$. We can thus gain a power of $\de^{\f 12}$ (by the Cauchy--Schwarz inequality) using the smallness of the length scale.
\end{itemize}

This gives a good estimate for $\rd_q \chi_k$ in a mixed $L^\i_{u_k} L^2_{\th_k}$ type space. A similar argument controls $E_k \eta_k$. However, for a general spatial derivative $\rd_x \eta_k$, when running the above argument for the $k = j$ case, there is a second derivative term which is \underline{not} controlled by \eqref{eq:intro.flux.imp.2}, which results in the $L^\i_{u_k} L^2_{\th_k}$ norm of $\rd_x\eta_k$ blowing up as $\de^{-\f 12}$. Instead, we can only control $\rd_x \eta_k$ in the $L^2_{u_k}L^2_{\th_k}$ space.

Finally, we have some bounds for special combinations of second derivatives such as $L_k^2 \chi_k$, $L_k \rd_q \chi_k$, etc., by virtue of the equations \eqref{eq:Ray.intro} and the already established bounds.

\subsubsection{Final remarks}

Ultimately, the geometric estimates are important because they are needed to close the wave estimates. In \cite{LVdM2}, we will show indeed that the geometric estimates we obtain are sufficient. 

Since we will need some anisotropic bounds for the wave variables up to $2+s''$ derivatives, as well as some $\de^{-\f12}$-dependent bounds up to $3$ derivatives, all the geometric estimates that we mentioned above (for instance the $L^\i$ bound for $\rd^2_x \mfg$, the $L^p\cap H^{s'}$ ($p\in [4,\infty)$) bound for $\rd_x \rd_t \mfg$ and the $L^\i_{u_k} L^2_{\th_k}$ bound for $\chi_k$, etc.) are necessary to carry out the energy estimates for the scalar wave.

On the other hand, it is quite remarkable that even though for some geometric quantities we only have weaker estimates (for instance we do \underline{not} put $\rd_x \rd_t \mfg$ in $L^\i$ or obtain a $\de^{-\f 12}$-independent $L^\i_{u_k} L^2_{\th_k}$ bound for $\rd_x \eta_k$ or control general second derivatives of $\chi_k$ and $\eta_k$), theses bounds are sufficient in the commutator estimates that we need to bound the wave part in \cite{LVdM2}. This is for instance because certain potentially dangerous terms do not appear due to the structure of the commutators.

Finally, in order to close the argument, we need to control the change between the elliptic gauge coordinates and the geometric coordinates, as well as the commutators for various vector fields. It will turn out that the control we establish for the geometric quantities will just be sufficient to justify that the eikonal function $u_k$ is a $W^{2,\infty}$ function in terms of the elliptic gauge coordinates, and that the second derivatives of the commutation vector fields $(L_k, E_k, X_k)$ with respect to the elliptic gauge coordinate derivatives are in $L^2$. Both of these statements are used in order to close the geometric and wave estimates.

\subsection{Related works}\label{sec:discussions}

\subsubsection{Impulsive gravitational waves}\label{sec:history.IGW}

Beyond \cite{KhanPenrose,Szekeres1}, there are further examples of interactions of two impulsive gravitational waves, see for instance \cite{sCvF1984,sCbcX1986,vFjI1987,HE1,HE2,HE3,NutkuHalil}. All these constructions rely heavily on symmetry assumptions. The singularity structures in these examples and their stability were further discussed in \cite{Tipler,Yurtsever,Yurtsever88}. We refer the readers to the books \cite{cBpaH2003,jbG1991} for further details and related examples (including those where matter fields are present).

In terms of mathematical results, priori to the works \cite{LR1, LR2}, there were low-regularity existence results in $\mathbb T^2$-symmetry \cite{pgLjS2010, pgLjmS2011} which in particular included impulsive gravitational waves and their interactions. Relatedly, Christodoulou \cite{dC1993} constructed solutions in the BV class to the Einstein--scalar field system in spherical symmetry. This can be thought of as including as a particular case a scalar field analogue of impulsive gravitational waves. Finally, very recently, a class of spacetimes featuring the interactions of two impulsive gravitational waves without any symmetry but still possessing a piece of future null infinity has been constructed in \cite{Yannis}.

\subsubsection{Low-regularity problems in general relativity and beyond}

Our problem can be viewed in the larger context of low-regularity problems in general relativity. In the Sobolev $H^s$ spaces, this has attracted much interest \cite{bHjyC99b,bHjyC99a,sKiR2003,sKiR2005d,sKiR2010,hSdT2005,dT2002}, culminating in the seminal proof of the bounded $L^2$ curvature theorem \cite{sKiRjS2015}, which requires the initial data to only be in $H^2$. 

Low-regularity problems are interesting for quasilinear wave equations beyond the Einstein equations, see for example \cite{mDcLgMjS2019,hSdT2005,qW2019}. We highlight particularly the work \cite{hBjyC2001} of Bahouri--Chemin on the high-dimensional low-regularity well-posedness of a coupled wave-elliptic system similar to the structure of the polarized $\mathbb U(1)$-reduced Einstein vacuum equations in an elliptic gauge.

\subsubsection{High-frequency waves and high-frequency limits}\label{sec:HF}

We compare our result with the work of Huneau--Luk \cite{HLHF} on high-frequency limits in polarized $\mathbb U(1)$ symmetry. In both \cite{HLHF} and this paper, a local existence result is proved where $\phi$ is Lipschitz but not better. On the one hand, the use of the elliptic gauge and (approximate) eikonal functions plays an important role in both papers. On the other hand, however, the analysis is quite different as one needs to rely on precise features of the problems (either that the waves are of high frequency in \cite{HLHF} or are highly localized in our setting).  

\subsubsection{Semi-linear model problems and propagation of weak singularities}\label{sec:semilinear}

The present work can be viewed in the larger context of interaction of conormal singularities for hyperbolic equations. There is a large literature for \emph{weak} conormal singularities, beginning with the pioneering works of Bony \cite{jmB1982, jmB1981, jmBICM}. In particular, these singularities are sufficiently weak so that \emph{classical well-posedness results hold}. For the interaction of two conormal singularities in the quasilinear case, see \cite{sA1988, sjK1995}. 

The interaction of three conormal singularities --- even for very weak singularities --- has only been studied for semilinear model problems. It has been shown \cite{jmB1984, jmB1986, rbM1984, rMnR1985} that in this case the only possible new singularity after the triple interaction must be weaker and lies in the cone emanating from the intersection. Moreover, it has been demonstrated that in general a new singularity could indeed arise in various different models \cite{jRmR1982, aSByW2018, aSB2020}.

For a sample of further related works, see \cite{mB1983,mB1988,hyC1987,pG1988,gL1989,rbMaSB1995,jRmcR1980,aSB1990,mZ1994} and the references therein. See also \cite{xCmLlOgP2019, yKmLgU2018, mLgUyW2017, mLgUyW2018, gUyW2018, yWtZ2019} for related more recent works concerning inverse problems.

\subsection{Open problems and discussions}\label{sec:open.problems}

We discuss some open problems and possible future directions related to our work.

\begin{enumerate}
\item (\textbf{Non-compactly supported initial data}) Our main theorem assumes that $\phi$ is initially compactly supported. It could be expected that the compact support can be replaced by fast decay of the initial data, but this creates a few technical issues in view of the fact that the metric coefficients grow logarithmically as $|x|\to \infty$.

\item (\textbf{Large data}) The theory of \cite{LR2} allows also for the interaction of impulsive gravitational waves of \emph{large} amplitude. Among other things, our theory is limited to the small amplitude region due to the global elliptic gauge.\footnote{We remark that a smallness assumption is already needed in the smooth theory in such a gauge \cite{HL.elliptic}.}

\item (\textbf{Beyond polarized $\mathbb U(1)$ symmetry}) The present work restricts to polarized $\mathbb U(1)$ symmetry. While the completely general case seems out of reach at the moment, the natural next step would be to study the interaction still under the $\mathbb U(1)$ symmetry assumption but without polarization. In this case, the equations reduce to an Einstein--wave map system (as opposed to simply the Einstein--scalar field system) in $(2+1)$ dimensions. It seems plausible that the extra ``wave map'' part can be controlled after choosing an elliptic (e.g.~Coulomb) gauge, so that one can use some of the techniques in this paper. We hope to return to this problem in the future. 

\item (\textbf{More singular initial data}) In \cite{LR2}, a more general theorem was proven, which allows the interaction of not only impulsive gravitational waves, but also of more singular data where the worst Christoffel symbol is only $L^2$ (instead of being in $L^\infty\cap BV$). This stronger result has various other applications \cite{mDjL2017,jL2013,LR3}, and is in particular related to the interaction of null dust shells \cite{LR3}. It is therefore natural to ask whether we can extend our results in the present paper on the interaction of three impulsive gravitational waves to more singular initial data.

\item (\textbf{Uniqueness}) As already mentioned in Remark~\ref{rmk:uniqueness}, our main theorem does not give uniqueness. It is of interest to appropriately formulate and prove a uniqueness result for these solutions.\footnote{Note that on the other hand, uniqueness for the $\de$-impulsive waves of course follows from standard theory.}

\item (\textbf{Higher regularity})
Ideally one would like to prove stronger regularity statements away from the union of the impulsive waves, or better yet to understand the optimal regularity.

\item (\textbf{Lower bounds and creation of new singularities})
Related to the last point, it would be interesting to show that the jump in the data persists along null characteristics, or even to derive a transport equation for the jump. More ambitiously, one can study whether new (but weaker) singularities appear in the cone emanating from the intersection point of the three impulsive gravitational waves as in the semilinear model problems \cite{jRmR1982}.

\item (\textbf{Interaction of four impulsive gravitational waves}) While our work allows for the transversal interaction of any number of impulsive gravitational waves \emph{under the polarized $\mathbb U(1)$ symmetry assumption}, it would be of interest to study the \emph{generic} transversal interaction of four impulsive gravitational waves, where four waves interact at a point in $(3+1)$ dimensions. See Remark~\ref{rmk:>3}.

\end{enumerate}

\subsection{Outline of the paper}\label{sec:outline}
The remainder of the paper is structured as follows. 

We begin with definitions for the geometric setup and the norms in \textbf{Section~\ref{sec:setup} and \ref{sec:norms}} respectively.

In \textbf{Section~\ref{data}}, we then define the class of data corresponding to both impulsive gravitational waves and $\de$-impulsive gravitational waves (recall Remark~\ref{rmk:smooth.approx}). The precise statements of the main theorem for impulsive gravitational waves (Theorem~\ref{roughtheorem}) and for $\de$-impulsive gravitational waves (Theorem~\ref{smooththeorem}) are then given in \textbf{Section~\ref{sec:precise.statements}}. In \textbf{Section~\ref{sec:approximation}}, we prove Theorem~\ref{roughtheorem} assuming Theorem~\ref{smooththeorem}. In \textbf{Section~\ref{sec:proof.smooth.theorem}}, we prove Theorem~\ref{smooththeorem} by reducing it to three theorems on a priori estimates (Theorems~\ref{thm:bootstrap.metric}, \ref{thm:bootstrap.Li} and \ref{thm:energyest}). 

The remainder of this paper is devoted to the proof of Theorem~\ref{thm:bootstrap.metric} (Theorems~\ref{thm:bootstrap.Li} and \ref{thm:energyest} will be proven in \cite{LVdM2}). After proving preliminary estimates in \textbf{Section~\ref{preliminary.estimates.section}}, we obtain geometric estimates associated to the elliptic gauge in \textbf{Section~\ref{metricsection}} and geometric estimates associated to the eikonal functions in \textbf{Section~\ref{sec:Ricci.coeff}}. In \textbf{Section~\ref{sec:thm.bootstrap.metric.conclusion}}, we then conclude the proof Theorem~\ref{thm:bootstrap.metric}.

Finally, in \textbf{Appendix~\ref{sec:appendix}}, we handle all issues regarding initial data and constraint equations.

\subsection*{Acknowledgements} We are grateful to Igor Rodnianski for inspiring conversations. We particularly thank him for some ideas that greatly simplified our original arguments. Those ideas now appear as part of Section~\ref{sec:metric.inho.1}. We also thank Hayd\'ee Pacheco for Figure~\ref{fig:only.figure}.

Part of this work was carried out when M.~Van de Moortel~was a visiting student at Stanford University. During the time that this work was pursued, J.~Luk~has been supported by a Terman fellowship and the NSF grants DMS-1709458 and DMS-2005435.

 	 	\section{Basic geometric setup and the Einstein equations}\label{sec:setup}
 	
 	In this section, we introduce the basic geometric setup. This plays a fundamental role for the whole series of papers.
 	
 	In \textbf{Section~\ref{ellipticgaugedef}}, we introduce the polarized $\mathbb U(1)$ symmetry and our elliptic gauge condition. In \textbf{Section~\ref{Einsteineqsection}}, we discuss the Einstein vacuum equations under these symmetry and gauge conditions.
 	
 	In \textbf{Section~\ref{XELexpressionsection}}, we introduce the eikonal functions and the related null frames $(L_k,\barL,E_k)$, which are important to capture the propagating impulsive waves. In \textbf{Section~\ref{relationXELgeocoordinatesection}}, we introduce a system of geometric coordinates associated to the eikonal functions. In \textbf{Section~\ref{coordinatetransform}}, we derive transformation formulas between the null frames, and the coordinate vector fields in various different coordinate system. 
	
	In \textbf{Section~\ref{riccinullframesection}}, we compute all the connection coefficients with respect to the null frames $(L_k,\barL,E_k)$. In \textbf{Section~\ref{sec:XEL.derivative}}, we compute the derivatives of the coefficients of $(L_k,\barL,E_k)$ in the $(\rd_t, \rd_1, \rd_2)$ basis. 
	
	In \textbf{Section~\ref{torsionnal}} and \textbf{Section~\ref{Lderivativericcisection}}, we compute respectively the transport equations for the frame coefficients and the connection coefficients.
	
	Finally, in \textbf{Section~\ref{sec:eikonal.initial}}, we compute the initial values of all the eikonal quantities.

 	\subsection{Elliptic gauge and conformally flat spatial coordinates} \label{ellipticgaugedef}

 	\begin{defn}[Polarized $\mathbb U(1)$ symmetry]\label{def:U1}
 		We say that a (3+1) Lorentzian manifold $(\mathcal{M} = I \times\RR^2 \times \mathbb{S}^1, ^{(4)}g)$, where $I\subseteq \RR$ is an interval, has \textbf{polarized $\mathbb U(1)$ symmetry} if the metric $^{(4)}g$ can be expressed as: 
 		
 		\begin{equation}\label{eq:U1}
 		^{(4)}g = e^{-2\phi} g + e^{2\phi} (dx^3)^2,
 		\end{equation}
 		where $\phi$ is a scalar function on $I\times \RR^2$ and $g$ is a $(2+1)$ Lorentzian metric\footnote{Note that since $\phi$ and $g$ are defined on $I\times \RR^2 $, they do not depend on $x^3$, the coordinate on $\mathbb{S}^1$.} on  $I \times \RR^2$.
 	\end{defn}
 	
 	\begin{defn}[The foliation $\Sigma_t$]\label{def:Sigmat}
 		Given a spacetime as in Definition~\ref{def:U1}, we foliate the $2+1$ spacetime $(I\times \RR^2, g)$ with hypersurfaces $\{ \Sigma_t\}_{t\in I}$, where each $\Sigma_t$ is spacelike. We will later make a particular choice of $t$; see Definition~\ref{def:gauge}. 
 		
 		The metric can then be written as 
 		
 		\begin{equation} \label{metric2+1}
 		g = -N^2 dt^2 +\bar{g}_{i j} (dx^i + \beta^i dt) (dx^j + \beta^j dt),
 		\end{equation}
 	for some function $N>0$ and Riemannian metric $\bar{g}_{ij}$.
	
	Here, and the remainder of the paper, we use the convention the \textbf{lower case Latin indices refer to the spatial coordinates $(x^1, x^2)$}, and repeated indices are summed over. In contrast, we use \textbf{lower case Greek indices to refer to spacetime coordinates $(x^0,x^1,x^2) := (t,x^1,x^2)$.} 
	\end{defn}
	
	\begin{defn}[Coordinate derivatives]\label{def:coord.der}
	From now on, we use $\rd_i = \rd_{x^i}$ ($i=1,2$) to denote the spatial coordinate partial derivatives, 
	and $\rd_\alp$ ($\alp = 0,1,2$) to denote spacetime coordinate partial derivatives with respect to the $(x^0,x^1,x^2) := (t,x^1,x^2)$ coordinate system in \eqref{metric2+1}. 
 	\end{defn}

 	\begin{defn}\label{def:miscellaneous}
 		Given $(I \times \RR^2, g)$ and $\{ \Sigma_t\}_{t\in I}$ in Definition~\ref{def:Sigmat}.
 		\begin{enumerate}
 			\item (Spacetime connection) Denote by $\nabla$ the Levi--Civita connection of the spacetime metric $g$.
 			\item (Induced metric) Denote by $\bar{g}$ the induced metric on the two-dimensional hypersurface $\Sigma_t$.
 			\item (Normal to $\Sigma_t$) Denote by $\n$ the future-directed unit normal to $\Sigma_t$;
 				\begin{equation} \label{defnormal}
 				\n =  \frac{  \partial_t - \beta^i \partial_{i} }{N}.
 				\end{equation}
 				satisfying $g(\n,\n)=-1$. 
 				 Define also $e_0$ to be the vector field
 			\begin{equation}\label{def:e0}
 			e_0=  \partial_t - \beta^i \partial_{i} = N\cdot \n.
 			\end{equation}
 			\item (Second fundamental form) Define\footnote{We remark that the definition of $K$ differs by a factor of $-1$ from that in \cite{HL.elliptic}.} $K$ to be the second fundamental form on $\Sigma_t$:
 			\begin{equation} \label{Kdef}
 			K(Y,Z) = g( \nabla_Y \n , Z),
 			\end{equation}
 			for every $Y,\,Z \in T\Sigma_t$. 
			 \end{enumerate} 
 	\end{defn}
 	
 	\begin{defn}[Gauge conditions]\label{def:gauge}
 		We define our gauge conditions (assuming already \eqref{eq:U1}) as follows:
 		\begin{enumerate}
 			\item For every $t\in I$, $\Sigma_t$ is required to be maximal, i.e.
 			\begin{equation} \label{maximality}
			(\bar{g}^{-1})^{ i j}  K_{i j} = 0 .
 			\end{equation}
 			Note that \eqref{maximality} defines the coordinate $t$.
 			\item We choose the coordinate system on $\Sigma_t$ so that $\bar{g}_{i j}$ is conformally flat, i.e.
 			\begin{equation} \label{gauge}
 			\bar{g}_{ i j} = e^{2\gamma} \delta_{i j},
 			\end{equation}
 			where, from now on, $\de_{ij}$ (or $\de^{ij}$) denotes the Kronecker delta.
 		\end{enumerate}
 	\end{defn}

 	We collect some simple computations:
 	
 	\begin{lem}
 		The following holds for $g$ of the form \eqref{metric2+1} satisfying Definition~\ref{def:gauge}:
 		\begin{enumerate}
 			\item The inverse metric $g^{-1}$ is given by
 			\begin{equation} \label{inversegelliptic}
 			g^{-1}=\frac{1}{N^2}\left(\begin{array}{ccc}-1 & \beta^1 & \beta^2\\
 			\beta^1 & N^2e^{-2\gamma}-\beta^1\beta^1 & -\beta^1 \beta^2\\
 			\beta^2 & -\beta^1 \beta^2 & N^2e^{-2\gamma}-\beta^2\beta^2
 			\end{array}
 			\right).
 			\end{equation} 	
			\item The following commutation formula holds:
			\begin{equation} \label{nspatial-spatialn}
\left[\n, \partial_q \right]= \partial_q \log(N) \cdot \n+ \f 1N (\rd_q \bt^i) \cdot \partial_i.
 				\end{equation}
 			\item The spacetime volume form associated to $g$ is given by 
 			\begin{equation} \label{volelliptic}dvol= Ne^{2\gamma} dx^1 dx^2 dt.\end{equation}
 			The induced volume form on the spacelike hypersurface $\Sigma_t$ associated to $g$ is given by
 			\begin{equation} \label{volelliptict=0} dvol_{\Sigma_t}= e^{2\gamma} dx^1 dx^2.\end{equation}
 			\item The wave operator $\Box_g$ is defined to be the Laplace--Beltrami operator associated to $g$, which is given by
 			\begin{equation} \label{Box2+1}
 			\Box_g f = \frac{-e_0^2 f}{N^2} + e^{-2 \gamma} \delta^{i j} \partial^{2}_{i j} f + \frac{e_0 N}{N^3} e_0 f + \frac{ e^{-2 \gamma}}{N}  \delta^{i j} \partial_{i} N \partial_{j} f = - \n^2 f + e^{-2 \gamma} \delta^{i j} \partial^{2}_{i j} f  + \frac{ e^{-2 \gamma}}{N}  \delta^{i j} \partial_{i} N  \partial_j f.
 			\end{equation}
 			\item The condition \eqref{maximality} can be rephrased as
 			\begin{equation} \label{maximality2}
 			\partial_q \beta^q = 2 e_0 (\gamma),
 			\end{equation} 
 			\item The second fundamental form is given by
 			\begin{equation} \label{maximality3}
 			K_{i j } = \frac{  e^{2\gamma}}{2  N} \cdot \left(  \partial_q \beta^q  \cdot  \delta_{i j} - \partial_i \beta^q \cdot  \delta_{q j}  -  \partial_j \beta^q \cdot \delta_{i q}\right) =: - \f{e^{2\gamma}}{2 N}(\mathfrak L\bt)_{ij},
 			\end{equation} 
			where $\mathfrak L$ is the conformal Killing operator $(\mathfrak L \bt)_{ij}:= - \partial_q \beta^q  \cdot  \delta_{i j} + \partial_i \beta^q \cdot  \delta_{q j}  +  \partial_j \beta^q \cdot \delta_{i q}$.
 		\end{enumerate}
 	\end{lem}

 	Finally, we compute the connection coefficients with respect to $\{e_0,\rd_1,\rd_2\}$:
 	
 	\begin{lem} \label{Christoffel}
 		Given $g$ of the form \eqref{metric2+1} satisfying Definition~\ref{def:gauge},
 		\begin{equation} \label{e0e0e0}
 		g(\nabla_{e_0} e_0, e_0) = -N \cdot e_0 N,
 		\end{equation}
 		\begin{equation} \label{e0e0ei}
 		g(\nabla_{e_0} e_0, \rd_i) = -	g(\nabla_{\rd_i} e_0, e_0)=g(\nabla_{e_0} \rd_i, e_0)=N \cdot \partial_i N,
 		\end{equation}		
 		\begin{equation} \label{eje0ei} \begin{split}
 		g(\nabla_{\rd_j} e_0, \rd_i) = &\:	g(\nabla_{e_0} \rd_j, \rd_i)-e^{2\gamma} \cdot \partial_j \beta^l \delta_{i l}=-	g(\nabla_{\rd_j} \rd_i, e_0) \\ = &\: \frac{e^{2\gamma} }{2}\cdot \left( 2e_0 \gamma \cdot \delta_{i j}-\partial_i \beta^q \cdot \delta_{ j q}  -\partial_j \beta^q  \cdot \delta_{i q}                          \right) = \frac{e^{2\gamma} }{2}\cdot \left( \partial_q \beta^q  \cdot \delta_{i j}-\partial_i \beta^q \cdot \delta_{ j q}  -\partial_j \beta^q \cdot \delta_{i q}                          \right). \end{split} \end{equation}
 		
 		Moreover,		
 		\begin{equation} \label{Diej}
 		\nabla_{\rd_i} \rd_j= \frac{e^{2\gamma} }{2N}\cdot \left( \partial_q \beta^q \cdot \delta_{i j}-\partial_i \beta^q \cdot \delta_{ j q}  -\partial_j \beta^q  \cdot \delta_{i q}                          \right) \n + \left(\delta^q_i \partial_j \gamma + \delta^q_j \partial_i \gamma - \delta_{i j} \delta^{q l} \partial_l \gamma  \right) \rd_q, \end{equation}
 		\begin{equation} \label{D0e0}
 		\nabla_{e_0} e_0= \frac{e_0 N }{N}\cdot e_0+  e^{-2\gamma} \delta^{i j} N \partial_i N \rd_j,\end{equation}
 		\begin{equation} \label{D0ei}
 		\nabla_ {e_0}\rd_i= \nabla_{\rd_i} e_0+\partial_i \beta^j \blue{\rd_j}= \frac{\partial_i N }{N}\cdot e_0+\frac{1 }{2}\cdot \left( \partial_q \beta^q \cdot \delta_{i}^{ j}+\partial_i \beta^j   -\delta_{ i q} \delta^{j l}\partial_l \beta^q            \right) \blue{\rd_j}.  \end{equation}
 		
 	\end{lem}
 	
 	\subsection{Einstein equations} \label{Einsteineqsection}
 	
 	With the polarized $\mathbb U(1)$ symmetry \eqref{eq:U1}, the Einstein equation $Ric_{\mu \nu}( ^{(4)}g)=0$ can be re-written in terms of the $(2+1)$-dimensional metric $g$ and the scalar field $\phi$ as:
 	
 	\begin{equation} \label{Einsteinricci}
 	Ric_{\alpha \beta}(g) = 2 \partial_{\alpha}  \phi \partial_{\beta} \phi,
 	\end{equation}
 	\begin{equation} \label{Einsteinwave}
 	\Box_g \phi= 0 .
 	\end{equation}
 	
 	Additionally, given the form of the metric \eqref{metric2+1} and the gauge conditions in Definition~\ref{def:gauge}, \eqref{Einsteinricci} implies the following elliptic equations (see \cite[(4.26)--(4.28)]{HLHF}, but note the sign difference in definitions of $K$): 
 	
	\begin{align}
	\label{Kellipticequation}
 	\delta^{i k } \partial_k K_{i j} = &\: 2e^{2\gamma} \cdot \n \phi \cdot \partial_j \phi, \\
	\label{Nellipticequation}
 	\Delta N = &\: \f{e^{2\gamma}}{4N} |\mathfrak L \bt|^2 + 2 N e^{2\gamma} \cdot  (\n \phi)^2, \\
 	\label{gammaellipticequation}
 	\Delta \gamma = &\: -\de^{il} (\rd_i \phi)(\rd_l\phi) - \f{e^{2\gamma}}{8 N^2} |\mathfrak L \bt|^2 -  e^{2\gamma} \cdot (\n \phi)^2, \\
 	\label{betaellipticequation}
 	\Delta \beta^j = &\:   \delta^{i k } \delta^{j l } (\f{\rd_k N}N - 2\rd_k\gamma) (\mathfrak L \bt)_{i l } -4  N   \delta^{j l } \cdot \n \phi \cdot \partial_l \phi,
 	\end{align}
 	where $\Delta$ denotes the \emph{Euclidean} Laplacian $\Delta =\sum_{i=1}^2 \rd^2_{ii}$.
 	
 	We remark that the above equations are not independent, as \eqref{betaellipticequation} can be derived from taking the divergence of \eqref{maximality3} and using \eqref{Kellipticequation}.

 	The equations \eqref{metric2+1}, Definition~\ref{def:gauge} and \eqref{Einsteinricci} also imply
 	\begin{equation} \label{Khypequation} 
	\begin{split}
 	&\: \n (K_{i j}) - N^{-1} \partial_i \partial_j N + \frac{1}{2}  \cdot \delta_{i j} \cdot N^{-1} \cdot \Delta N \\
	= &\: 2e^{-2\gamma} K_i^l K_{j l } + N^{-1} \cdot ( \partial_j \beta^k K_{k i }+\partial_i \beta^k K_{k j })   \\ 
 	&\: - N^{-1} \cdot ( \delta_{i}^k \partial_j \gamma +\delta_{j}^k \partial_i \gamma- \delta_{i j }  \delta^{l k} \partial_l \gamma) \cdot \partial_k N + 2 \partial_i \phi \cdot \partial_j \phi - \delta_{i j } \de^{kl}\rd_k \phi\cdot \rd_l\phi.
	\end{split}
 	\end{equation}

 	\subsection{Eikonal functions and null frames}\label{XELexpressionsection}
 	
 	We will define three eikonal functions together with null hypersurfaces and null frames. Each of these will later be chosen to adapted to one propagating wave.
 	
 	\begin{defn}[Eikonal functions]\label{def:eikonal}
 		Given a spacetime $(I\times \mathbb R^2, g)$ of the form \eqref{metric2+1} satisfying Definition~\ref{def:gauge}, define three eikonal functions $u_k$, $k=1,2,3$, corresponding to the three impulsive waves, as the unique solutions to\footnote{For simplicity, we stipulate here that the initial wavefront $(u_k)_{|\Sigma_0}$ are exact lines in the $(x^1,x^2)$ coordinates. This can easily be relaxed so that they are only approximate lines.}
 		
 		\begin{equation} \label{eikonal1}
 		\gi^{\alpha \beta} \partial_{\alpha} u_k \partial_{\beta}  u_k =0,
 		\end{equation}
 		\begin{equation} \label{eikonalinit}
 		(u_k)_{|\Sigma_0} = a_k + c_{k j }x^j,
 		\end{equation}
 		which satisfies $e_0 u_k >0$. Here, $a_k$, $c_{k j } \in \RR$ are constants obeying the following conditions\footnote{The identity $| -c_{k 2} \cdot c_{k' 1} + c_{k 1} \cdot c_{k' 2}|^2 =1- | c_{k 1} \cdot c_{k' 1} + c_{k 2} \cdot c_{k' 2}|^2$ is an immediate consequence of \eqref{cnormalization}.}: 
		\begin{align} 
		\label{cnormalization}
 		\sqrt{c_{k1}^2+c_{k2}^2}= &\: 1, \\
		\label{cangle2} 
		|-c_{k 2} \cdot c_{k' 1} + c_{k 1} \cdot c_{k' 2}| = \sqrt{1-|c_{k 1} \cdot c_{k' 1}+ c_{k 2} \cdot c_{k' 2}|^2} \geq &\: \upkappa_0,
 		\end{align} 
		for some fixed constant $\upkappa_0 \in (0,\frac{\pi}{2})$, and for every  $k \neq k' \in \{1,2,3\}$. 
 	\end{defn} 
	 	
	\begin{defn}[Sets associated with the eikonal functions]\label{def:sets.eikonal}
 		Let $u_k$ ($k=1,2,3$) satisfying \eqref{eikonal1} and \eqref{eikonalinit} in $(I\times\mathbb R^2, g)$ be given.
 		\begin{enumerate}
 			\item For all $w \in \mathbb{R}$, define
 			\begin{equation} \label{defCu}
 			C_{w}^k := \{(t,x): u_k(t,x)=w\}, \quad C_{\leq w}^k := \bigcup\limits_{u_k \leq w}  	C_{u_k}^k, \quad C_{\geq w}^k := \bigcup\limits_{u_k \geq w}  	C_{u_k}^k,
 			\end{equation}
			and, for every $T\in I$, define
			\begin{equation} \label{defCuCutoff}
 			C_{w}^k([0,T)) := C_{w}^k \cap (\cup_{t\in [0,T)} \Sigma_t).
 			\end{equation}
 			\item For all $w_1,\,w_2 \in \mathbb{R}$,  
			define \begin{equation} \label{defStwosided}
S^k(w_1,w_2):= \bigcup\limits_{ w_1 \leq u_k \leq w_2}  	C_{u_k}^k.
 			\end{equation}
			For $\de_0>0$, define also 
			\begin{equation} \label{defS}
 			S_{\delta_0}^k := S^k(-\delta_0,\delta_0).
 			\end{equation}
 			We will later understand $\Sd^k$ as ``the singular zone'' for $\tphi$.
 		\end{enumerate}
 	\end{defn}

 	\begin{defn}[Definition of the null frame]\label{def:null.frame}
 		\begin{enumerate}
 			\item Define the null vector $\Lgeo$ associated to the eikonal function $u_k$ by
 			\begin{equation} \label{Lgeodefinition}
 			\Lgeo= - \gi^{\alpha \beta} \partial_{\beta} u_k \cdot   \partial_{\alpha}.
 			\end{equation}
 			\item Define $L_k$ to be the vector field parallel to $\Lgeo$ which satisfies $L_k t = N^{-1}$, i.e.
 			\begin{equation} \label{Ldefinition}
 			L_k= \mu_k \cdot \Lgeo,\quad \mu_k =  (N \cdot \Lgeo t)^{-1}.
 			\end{equation}
 			\item  Define the vector field $X_k$ to be the unique vector field tangential to $\Sigma_{t}$  which is everywhere orthogonal (with respect to $\bar{g}$) to $C_{u_k}^k \cap \Sigma_t$ and such that $g(X_k,L_k)=-1$.
 			\item Define $E_k$ to be the unique vector field which is tangent to $C_{u_k}^k \cap \Sigma_t$, satisfies $ g(E_k,E_k)=1$, and such that $(X_k, E_k)$ has the same orientation as $(\rd_1, \rd_2)$. 
 		\end{enumerate}
 	\end{defn}
 	
 	\begin{lem}
 		\begin{enumerate}
 			\item $\Lgeo$ is null and geodesic, i.e.
 			\begin{equation} \label{eikonal2}
 			g(\Lgeo, \Lgeo) = 0,\quad \nabla_{\Lgeo} \Lgeo=0.
 			\end{equation}
 			\item The following holds:
 			\begin{equation}\label{eq:silly.tangential}
 			L_k u_k = E_k u_k = 0,\quad E_k t = X_k t = 0,\quad L_k t = N^{-1},\quad X_k u_k = \mu_k^{-1}.
 			\end{equation}
 			\item The normal $\n$ can be expressed in terms of $X_k$ and $L_k$ as:	
 			\begin{equation} \label{nXEL}
 			\n = L_k +  X_k .
 			\end{equation}
 			\item The triplet $(X_k,E_k,L_k)$ forms a null frame, i.e.~it satisfies 
 			\begin{equation} \label{XELframecondition}
 			g(L_k,X_k) = -1, \quad g(E_k,L_k)=g(E_k,X_k)=g(L_k,L_k)=0, \quad g(E_k,E_k)=g(X_k,X_k)=1 .
 			\end{equation}
 			\item $g^{-1}$ can be given in terms of the $(X_k, E_k, L_k)$ frame by
 			\begin{equation} \label{inversegXEL}
 			g^{-1} = -  L_k \otimes L_k - L_k \otimes X_k  -  X_k \otimes L_k  + E_k \otimes E_k.
 			\end{equation}
 		\end{enumerate}
 	\end{lem}
 	\begin{proof}
 		\eqref{eikonal2} is an immediate consequence of \eqref{eikonal1}.
 		
 		For \eqref{eq:silly.tangential}, the first two chains of equalities simply follow from tangential properties of the vector fields. That $L_k t = N^{-1}$ follows from Definition~\ref{def:null.frame}.1. Finally, using Definition~\ref{def:null.frame}, $X_k u_k = g_{\sigma \rho}  \gi^{\alpha \rho} \partial_{\alpha} u_k X_k^{\sigma}=-g(\Lgeo,X_k)=\mu_k^{-1}$.
 		
 		To establish \eqref{nXEL}, we need to show that $-L_k +\n$ satisfies all the defining properties of $X_k$ in Definition~\ref{def:null.frame}. First, \eqref{defnormal} and \eqref{Ldefinition} imply $(-L_k + \n) t = 0$, i.e.~$L_k - \n$ is tangent to $\Sigma_t$. Moreover, $g(-L_k + \n, E_k) = -g(L_k, E_k) + g(\n, E_k) = 0$, and also $g(-L_k + \n, X_k) = -g(L_k, X_k) = 1$. Hence $-L_k + \n = X_k$, i.e.~\eqref{nXEL} holds.
 		
 		Turning to \eqref{XELframecondition}, first note that $g(L_k,X_k) = -1$, $g(E_k,E_k) = 1$ and $g(E_k, X_k) = 0$ by Definition~\ref{def:null.frame}, and $g(L_k,L_k)=0$ can be derived using additionally \eqref{eikonal2}.
 		
 		Next, $g(E_k,\Lgeo)=0$ (and hence $g(E_k,L_k)=0$) follows from $E_k u_k=0$ and \eqref{Lgeodefinition}. Finally, note that using $g(L_k,L_k)=0$, \eqref{nXEL} and the fact that $\n$ is the unit normal to $\Sigma_t$, we have 
 			$$0=g(L_k,L_k) = g(\n-X_k, \n -X_k) = g(\n,\n) -2g(\n,X_k) + g(X_k,X_k) = -1 +g(X_k,X_k), $$
 			which gives $g(X_k,X_k)=1$.

 		Now that we have established \eqref{XELframecondition}, the equation \eqref{inversegXEL} follows as an immediate consequence.
 		\qedhere

 	\end{proof}

 	\subsection{Geometric coordinate system $(u_k,\theta_k,t_k)$} \label{relationXELgeocoordinatesection}
 	
 	We now introduce the coordinate $\theta_k$ such that $(u_k,\theta_k,t_k)$ is a regular coordinate system on $I \times \mathbb R^2$.
 	\begin{defn}
 		\begin{enumerate}
 			\item Given $u_k$ satisfying \eqref{eikonal1}--\eqref{eikonalinit}, and fixing some constants $b_k$, define $\th_k$ by
 			\begin{equation} \label{thetadef}
 			L_k\theta_k =0,
 			\end{equation}
			\begin{equation} \label{thetainit}
			(\theta_k)_{|\Sigma_0}=  b_k + c_{k j}^{\perp} x^j,
			\end{equation} 
			where $c_{k 1}^{\perp} = -c_{k 2}$ and $c_{k 2}^{\perp} = c_{k 1}$, and $c_{k i}$ are the constants in \eqref{eikonalinit}.
 			\item Let $t_k = t$. 
			\item Denote by $(\partial_{u_k},\partial_{\theta_k}, \partial_{t_k})$ the coordinate vector fields in the $(u_k,\theta_k,t_k)$ coordinate system. (Note that we continue to use $\rd_t$ to denote the coordinate derivative in the $(x^1,x^2,t)$ coordinate system of Section~\ref{ellipticgaugedef}.)
 		\end{enumerate}
 	\end{defn}

 	\begin{lem}\label{lem:XEL.in.rd}
 		Defining $\varTheta_k = (E_k\theta_k)^{-1}$ and $\Xi_k = X_k\th_k$, we have
 		\begin{equation} \label{thetaE}
 		L_k = \f 1N \cdot \rd_{t_k},\quad E_k =  \varTheta_k^{-1} \cdot \rd_{\th_k}, \quad X_k = \mu_k^{-1}\cdot \partial_{u_k}+ \Xi_k \cdot  \partial_{\theta_k}.
 		\end{equation}
 	\end{lem}
 	\begin{proof}
 		This follows from combining \eqref{eq:silly.tangential}, \eqref{thetadef} with the definitions of $\varTheta_k$ and $\Xi_k$. \qedhere
 	\end{proof}

 	\begin{lem}
 		\begin{enumerate}
 			\item The metric $g$ in the $(t_k,u_k,\theta_k)$ coordinate system is given by
 			\begin{equation} \label{gthetau}
 			g= \varTheta_k^2 \,d\theta^2_k - 2 \mu_k N \,dt_k\, du_k - 2 \mu_k \varXi_k  \varTheta_k^2 \,du_k\, d\theta_k + \mu_k^2(1+  \varXi_k^2  \varTheta_k^2 ) \,du_k^2.
 			\end{equation} 
 			\item The volume form induced by $g$ and $\bar{g}$ in the $(t_k,u_k,\theta_k)$ coordinate system are given by
 			\begin{equation} \label{volthetau}
 			dvol = \mu_k   \cdot N \cdot \varTheta_k \, dt_k \, du_k \, d\theta_k,\quad dvol_{\Sigma_t} = \mu_k^2 \varTheta_k^2 \, du_k\, d\th_k.
 			\end{equation}
 			\item Letting $dvol_{C_{u_k}}$ be the volume form on $C_{u_k}$ such that $du_k \wedge dvol_{C_{u_k}} =dvol$. Then
 			\begin{equation} \label{volCuk}
 			dvol_{C_{u_k}}= \mu_k \cdot N \cdot \varTheta_k \, dt_k\, d\theta_k.
 			\end{equation}
 		\end{enumerate}
 	\end{lem}
 	\begin{proof}
 		\eqref{gthetau} follows from \eqref{thetaE} and \eqref{XELframecondition}; \eqref{volthetau} and \eqref{volCuk} follow from \eqref{gthetau} directly.
 	\end{proof}

 	We establish using Proposition~\ref{lem:XEL.in.rd} a first set of relations between the frame coefficients. 	
 	
 	\begin{lem} We have the following relations between the commutator and the frame coefficients:
 		\begin{equation} \label{LE-EL-torsionnal}
 		\left[E_k,L_k\right]= L_k \log(\varTheta_k) \cdot E_k - E_k \log(N) \cdot L_k .
 		\end{equation}
 		\begin{equation} \label{LX-XL-torsionnal}
 		\left[L_k,X_k\right]= -L_k \log(\mu_k) \cdot X_k+ (  L_k \log(\mu_k) \cdot  \Xi_k+  L_k \Xi_k ) \cdot \varTheta_k \cdot E_k+ X_k \log(N) \cdot L_k,
 		\end{equation}
 		\begin{equation} \label{XE-EX-torsionnal}
 		\left[E_k,X_k\right]= -E_k \log(\mu_k) \cdot X_k+ \left( E_k \Xi_k -X_k \varTheta_k^{-1} + \Xi_k \cdot E_k \log(\mu_k)\right) \cdot \varTheta_k \cdot E_k.
 		\end{equation}
 	\end{lem}
 	
 	\begin{proof} \eqref{LE-EL-torsionnal} follows directly from  \eqref{thetaE}.
 		
 		For \eqref{LX-XL-torsionnal}, we first compute 
		$$ \left[L_k,X_k\right]= \left[N^{-1} \partial_{t_k},\mu_k^{-1}\partial_{u_k}+ \Xi_k \cdot  \partial_{\theta_k}\right]= X_k \log(N) \cdot L_k-L_k \log(\mu_k) \mu_k^{-1}\partial_{u_k}+L_k \Xi_k  \cdot  \partial_{\theta_k} ,$$
 		and then use $$ L_k \log(\mu_k) \cdot \mu_k^{-1}\partial_{u_k} = L_k \log(\mu_k) \cdot  X_k - \Xi_k \cdot L_k \log(\mu_k) \cdot\partial_{\theta_k}, $$
 		which gives \eqref{LX-XL-torsionnal}. 
 		
 		Finally \eqref{XE-EX-torsionnal} is an easy consequence of 
 			$$[E_k, X_k] = \varTheta_k^{-1} (\rd_{\th_k} \mu_k^{-1}) \rd_{u_k} + \varTheta_k^{-1} (\rd_{\th_k} \Xi_k) \rd_{\th_k} - (X_k \varTheta_k^{-1}) \rd_{\th_k}, $$
 			and 
 		$$E_k \log \mu_k \cdot \mu_k^{-1} \rd_{u_k} = E_k \log \mu_k \cdot X_k - \Xi_k \cdot E_k \log \mu_k \cdot \rd_{\th_k}. \qedhere$$

 	\end{proof}
 	
	\subsection{Transformations between different coordinate systems} \label{coordinatetransform}

 	\subsubsection{Relations on $\Sigma_t$ between $(X_k,E_k)$ and the elliptic coordinate vector fields $(\partial_1,\partial_2)$ }

	\begin{lem}
 		The following identities between $E^i_k$ and $X^i_k$ hold: 
 		\begin{equation} \label{EXinellipticcoord}  
		\begin{split}
 		E^1_k=-X_k^2, \quad
 		E^2_k=X_k^1. 
		\end{split}
 		\end{equation}
 		Moreover, the coordinate vector fields $(\partial_1,\partial_2)$ can be expressed in terms of $(E_k,X_k)$ as follows: 
 		\begin{equation} \label{partial12EX}
 		\partial_1 = e^{2\gamma} \cdot \left( -X^2_k \cdot E_k +  E^2_k \cdot X_k \right),\quad \partial_2 = e^{2\gamma} \cdot \left( X^1_k \cdot E_k -  E^1_k \cdot X_k \right).
 		\end{equation}

 	\end{lem}
 	\begin{proof}

 		Since $E_k$ and $X_k$ are orthonormal for $\bar{g}$ (see \eqref{XELframecondition}), it follows that $\delta_{i j} X_k^i X_k^j=\delta_{i j} E_k^i E_k^j=e^{-2\gamma}$.
 		
 		Now set $\tilde{E}_k=e^{\gamma} \cdot E_k$ and $\tilde{X}_k=e^{\gamma} \cdot X_k$. Then, using also point 4 of Definition~\ref{def:null.frame}, $(\tilde{X}_k, \tilde{E}_k)$ is an orthonormal basis for the Euclidean metric $\delta_{i j}$ on $\RR^2$ with the same orientation as $(\rd_1, \rd_2)$. Hence there exists $\varphi \in \RR$ such that 
 		\begin{equation*}   \begin{split}
 		\tilde{E}_k= \cos(\varphi) \cdot \partial_1+\sin(\varphi) \cdot \partial_2, \quad
 		\tilde{X}_k= \sin(\varphi) \cdot \partial_1-\cos(\varphi) \cdot \partial_2, \end{split}
 		\end{equation*}
 		
 		This, in particular, gives \eqref{EXinellipticcoord}. Inverting the rotation matrix, we obtain
 		\begin{equation*}   \begin{split}
 		\partial_1= \cos(\varphi) \cdot 	\tilde{E}_k-\sin(\varphi) \cdot 	\tilde{X}_k, \quad
 		\partial_2= -\sin(\varphi) \cdot \tilde{E}_k-\cos(\varphi) \cdot \tilde{X}_k, \end{split}
 		\end{equation*}
 		which gives directly \eqref{partial12EX}. \qedhere

 	\end{proof}

 	\subsubsection{Elliptic coordinate derivatives of $u_k$}

 	Now we compute the derivatives of $u_k$ and $\th_k$ with respect to $(\rd_t, \rd_i)$.
 	\begin{lem} \label{duprop}

 		The following identities hold: 
 		
 		\begin{equation} \label{ucartderivative}
 		\partial_i u_k = e^{2\gamma} \cdot \mu_k^{-1} \cdot  \delta_{i j}X^j_k,
 		\end{equation}
		\begin{equation} \label{thetacartderivative}
 		\partial_i \theta_k = e^{2\gamma}  \delta_{i j} \cdot \left(\varTheta_k^{-1} \cdot E^j_k + \Xi_k  \cdot X^j_k\right),
 		\end{equation}
 		\begin{equation}  \label{dtui}  	
		\partial_{t} u_k =  \beta^q \partial_q u_k+ N \cdot \mu_k^{-1},
		\end{equation}
		\begin{equation} \label{dttheta}
 		\partial_t \theta_k = \beta^i \partial_i \theta_k+ N \cdot   \Xi_k = e^{2\gamma} \cdot \beta_j \cdot \left(\varTheta_k^{-1} \cdot E^j_k + \Xi_k  \cdot X^j_k\right)+ N \cdot  \Xi_k.
 		\end{equation}
 		Moreover, for all vector field $Y$ in the tangent space of $\Sigma_t$, we have 
		\begin{equation} \label{Yu}
Yu_k = \mu_k^{-1}  \cdot g(Y,X_k).
 		\end{equation}

 	\end{lem}
 	\begin{proof}
 		We use \eqref{partial12EX} to compute $\rd_i u_k$ and $\rd_i \th_k$, and apply also \eqref{thetaE} and \eqref{EXinellipticcoord} to obtain  \eqref{ucartderivative} and \eqref{thetacartderivative}. 
		
		Both \eqref{dtui} and \eqref{dttheta} can be derived by $L_k u_k = L_k \th_k = 0$ (by \eqref{thetaE}), the identity $ L_k = N^{-1} \cdot (\partial_t -\beta^q \partial_q) -X_k$ (by  \eqref{nXEL} and \eqref{defnormal}), and \eqref{thetaE}.
				
		Finally, the identity \eqref{Yu} follows from \eqref{ucartderivative} and \eqref{gauge}. \qedhere

 	\end{proof}
 	
 	\subsubsection{Spatial coordinate system $(u_k,u_{k'})$ on $\Sigma_t$} \label{ukuk'coordinatesection}
 	
 	Fix $k,\,k'\in \{1,2,3\}$ with $k \neq k'$. Introduce the spatial coordinate system $(u_k,u_{k'})$. So as to distinguish it from other coordinate derivatives, we define the coordinate vector fields on $\Sigma_t$ in the $(u_k,u_{k'})$ coordinate system by $(\partialuk, \partialukp)$. 
	
	We now express $(\partialuk, \partialukp)$ in terms of $(X_k,E_k)$ in the following lemma: 
	\begin{lem} \label{ukuk'toXE.lemma}
 		The vector fields $X_k$ and $E_k$ can be expressed in the $(u_k, u_{k'})$ coordinate system as follows:
 		\begin{equation} \label{Xukukprime}
 	X_k =\mu_{k}^{-1} \cdot  \partialuk+ \mu_{k'}^{-1} \cdot g(X_k,X_{k'}) \cdot 	\partialukp	,
 		\end{equation}
 		\begin{equation} \label{Eukukprime}
 	E_k= \mu_{k'}^{-1} \cdot g(E_k,X_{k'}) \cdot  	\partialukp.
 		\end{equation}
 		The above transformation can be inverted to give
		\begin{equation} \label{partialukpEX}
	\partialukp=  \mu_{k'} \cdot g(E_k,X_{k'}) ^{-1} E_k,
 		\end{equation}
 		\begin{equation}
 \partialuk = \mu_k X_k- \frac{ \mu_k\cdot g(X_k,X_{k'})}{ g(E_k,X_{k'})} \cdot E_k.
 		\end{equation}

 	\end{lem}
 	\begin{proof}
We start to define $a$, $b$, $c$, $d$ as $$ X_k = a  \cdot  \partialuk+  b \cdot \partialukp,$$  $$ E_k = c  \cdot \partialuk+d	\cdot \partialukp.$$ Since $E_k u_k=0$ we know that $c=0$. To determine $d$ we compute $d= E_k(u_{k'})=\mu_{k'}^{-1} \cdot g(E_k,X_{k'}) $ by \eqref{Yu}.

We also know by \eqref{thetaE} that $a=X_ku_k = \mu_k^{-1}$. To determine $b$ we compute $b= X_k(u_{k'})=\mu_{k'}^{-1} \cdot g(X_k,X_{k'}) $ by \eqref{Yu} again. \qedhere
 	\end{proof}

 	\subsection{Ricci coefficients, covariant derivatives and commutators in the XEL frame} \label{riccinullframesection}
 	
 	We now define some Ricci coefficients in terms of the frame $(X_k, E_k, L_k)$:
 	
 	\begin{equation} \label{chi}
 	\chi_k= g(\nabla_{E_k} L_k, E_k)=-g(\nabla_{E_k} E_k,L_k),
 	\end{equation}
 	\begin{equation} \label{eta}
 	\eta_k=   g(\nabla_{X_k} L_k, E_k)=- g(\nabla_{X_k} E_k, L_k).
 	\end{equation}
 	
 	All the other Ricci coefficients can, in fact, be determined from $\chi_k$, $\eta_k$, $N$, $\mu_k$ and the contractions of $K$.
 	
 	\begin{lem} \label{riccibarXEL}
 		The following identities hold:
 		\begin{equation} \label{barchiXEL}
 		g( \nabla_{E_k}X_k,E_k) = -g( X_k, \nabla_{E_k}E_k)=  K(E_k,E_k) - \chi_k ,
 		\end{equation}	\begin{equation} \label{zetaXEL}
 		g( \nabla_{E_k}L_k,X_k) =-	g( \nabla_{E_k}X_k,L_k) = K(E_k,X_k),
 		\end{equation}
 		\begin{equation} \label{baretaXEL}
 		g( \nabla_{L_k}E_k,X_k)=-	g( \nabla_{L_k}X_k,E_k)=K(E_k,X_k)-E_k \log(N),
 		\end{equation}
 		\begin{equation} \label{baromegaXEL}
 		g( \nabla_{X_k}L_k,X_k)=-	g( \nabla_{X_k}X_k,L_k) 	=K(X_k,X_k), 
 		\end{equation}
 		\begin{equation} \label{barnuXEL}
 		g( \nabla_{X_k}X_k,E_k) =	-	g( \nabla_{X_k}E_k,X_k)= K(E_k,X_k) -\eta_k,
 		\end{equation}	
 		\begin{equation} \label{muXEL}
 		g( \nabla_{L_k}X_k,L_k)		=	-g( \nabla_{L_k}L_k,X_k)= L_k \log(\mu_k)=K( X_k,X_k) -X_k\log(N).
 		\end{equation}	
 		All the other Ricci coefficients that have not been mentioned in \eqref{chi}--\eqref{muXEL} are zero.
 		
 		As a consequence, we have the following covariant derivatives and commutators: 
 		
 		\begin{equation} \label{EL}
 		\nabla_{E_k} L_k = \chi_k \cdot E_k -K(E_k,X_k) L_k ,
 		\end{equation}
 		\begin{equation} \label{LE}
 		\nabla_{L_k} E_k = (E_k \log(N)-K(E_k,X_k) )\cdot L_k ,
 		\end{equation}
 		\begin{equation} \label{EL-LE}
 		\left[E_k,L_k \right] = \chi_k \cdot E_k - E_k \log(N) \cdot L_k.
 		\end{equation}	\begin{equation} \label{EX}
 		\nabla_{E_k} X_k = K(E_k,X_k) X_k +  (K(E_k,E_k)- \chi_k) \cdot E_k +  K(E_k,X_k) L_k,
 		\end{equation}
 		\begin{equation} \label{XE}
 		\nabla_{X_k} E_k =  \eta_k X_k + K(E_k,X_k) L_k, \end{equation}
 		\begin{equation} \label{EX-XE}
 		\left[E_k,X_k\right]  =   (K(E_k,X_k)-\eta_k) \cdot X_k +  (K(E_k,E_k)- \chi_k) \cdot E_k, 
 		\end{equation} 			\begin{equation} \label{LX}
 		\nabla_{L_k} X_k = (-K(E_k,X_k) + E_k\log N) \cdot E_k  -(K( X_k,X_k) -X_k\log(N)) \cdot  X_k-( K( X_k,X_k) -X_k\log(N) )\cdot L_k,
 		\end{equation}
 		\begin{equation} \label{XL}
 		\nabla_{X_k} L_k = \eta_k \cdot E_k -K(X_k,X_k) \cdot L_k,
 		\end{equation}
 		\begin{equation} \label{LX-XL}
 		[L_k,X_k] =-(K(E_k,X_k) -E_k\log N  +\eta_k) \cdot E_k  - (K( X_k,X_k) -X_k\log(N)) \cdot  X_k+  X_k\log(N) \cdot L_k,
 		\end{equation}	
 		\begin{equation} \label{EE2}
 		\nabla_{E_k} E_k=  \chi_k \cdot X_k+  K(E_k,E_k) \cdot L_k,\end{equation} 
 		\begin{equation} \label{XX}
 		\nabla_{X_k} X_k = K(X_k,X_k) \cdot X_k+  (K(E_k,X_k)-\eta_k) \cdot E_k + K(X_k,X_k) \cdot  L_k,
 		\end{equation}	
 		\begin{equation} \label{LL}
 		\nabla_{L_k}L_k = (K( X_k,X_k) -X_k\log(N)) \cdot L_k. 
 		\end{equation}

 	\end{lem}
 	
 	\begin{proof}
 		
 		Recall from \eqref{nXEL} that $L_k = \n - X_k$, an identity we will use throughout the proof.
 		In particular \eqref{barchiXEL}	follows immediately from this identity, \eqref{chi} and the definition of $K$; so does \eqref{zetaXEL}, after noticing that $g(\nabla_{E_k}X_k,X_k)=\frac{1}{2} \cdot E_k( g(X_k,X_k))=0$ as $X_k$ is $g$-unitary.
 		
 		\eqref{baretaXEL} follows from \eqref{zetaXEL} and the observation that, by \eqref{LE-EL-torsionnal},  $$ g(\left[E_k,L_k\right],X_k)= E_k \log(N).$$
 		
 		\eqref{baromegaXEL} follows from $L_k = \n - X_k$ and the fact that $X_k$ is $g$-unit; so does \eqref{barnuXEL}, using also the definition \eqref{eta}. 
 		
 		For \eqref{muXEL}, first note using \eqref{LX-XL-torsionnal} and the fact that $L_k$ is null, we have $g( \nabla_{L_k}X_k,L_k) = g([L_k, X_k], L_k)= L_k \log(\mu_k)$. Then, by \eqref{XL}, \eqref{XELframecondition} and \eqref{LX-XL-torsionnal},
 			$$K(X_k,X_k) = g(\nab_{X_k}L_k, X_k) =  g([X_k, L_k], X_k) = L_k \log(\mu_k) +X_k \log(N),$$
 			which implies the last equation in \eqref{muXEL} after rearranging.
 		
 		The fact that all the other Ricci coefficients vanish is mostly trivial, except for $g(\nabla_{L_k} E_K, L_k) = 0$, which holds by \eqref{LE-EL-torsionnal}.
 		
 		Finally, the covariant derivatives and commutators follow straightforwardly from the Ricci coefficients and the frame conditions \eqref{XELframecondition}. Details are left to the reader. \qedhere

 	\end{proof}

 	\subsection{Derivatives of the components of the $XEL$ vector fields in the elliptic gauge}\label{sec:XEL.derivative}
 	
 	The goal of this section is to compute $\partial (Y_k^{\alpha})$ for $Y_k \in \{ X_k, E_k, L_k\}$ and $Y_k=Y_k^{\alpha} \partial_{\alpha}$ in the coordinate system $(t,x^1,x^2)$ of section \ref{ellipticgaugedef}.
 	
 	\begin{prop} \label{dXELellipticprop}
	The derivatives of $E^i_k$ with respect to $(L_k, X_k, E_k)$ can be expressed as follows:
 		\begin{subequations} 
		\begin{align}
		\notag
 		L_k(E^i_k) = &\: -(E_k \log(N)-K(E_k,X_k)) \cdot  X^i_k-\frac{ 1 }{2 N}\cdot \left( E_k^i \partial_q \beta^q + E_k^j\partial_j \beta^i   - \delta_{j q}  E_k^j        \cdot \delta^{i l} \partial_l \beta^q                   \right) \\
		&\: - \left( X^i_k  E^j_k  \partial_j \gamma + E^i_k  X^j_k \partial_j \gamma \right), \label{LEi}\\
 		\label{XEi}
 		X_k(E^i_k) = &\: - (X^i_k \cdot E_k \gamma  + E^i_k \cdot X_k \gamma    ) + (\eta_k -K(E_k,X_k)) \cdot X^i_k, \\
 		\label{EEi}
 		E_k(E^i_k) = &\: -2 E^i_k \cdot E_k \gamma  + e^{-2\gamma}\cdot \de^{il} \cdot \rd_l\gamma + (-K(E_k,E_k)+\chi_k) \cdot X_k^i.
		\end{align}
 		\end{subequations}
	The derivatives of $X^i_k$ with respect to $Y \in \{L_k, X_k, E_k\}$ can be expressed as follows: 
		\begin{align}
		\label{YXj}
 		Y (X^1_k) = Y (E^2_k),\quad Y (X^2_k) = -Y (E^1_k).
 		\end{align}
 	The derivatives of $L^t_k$ with respect to $Y \in \{L_k, X_k, E_k\}$ can be expressed as follows:
		\begin{equation} \label{YLt}
 		Y(L^{t}_k) = -\frac{Y \log(N)}{N}.
 		\end{equation}
	The derivatives of $L^i_k$ with respect to $Y \in \{L_k, X_k, E_k\}$ can be expressed as follows:
 		\begin{equation} \label{YLi}
 		Y(L^i_k)= -Y (X^i_k) - Y (\f{\bt^i}{N}).
 		\end{equation}

 	\end{prop}
 	\begin{proof}

 		\pfstep{Step~1: Proof of \eqref{LEi}} We start with the elementary  
		\begin{equation}\label{eq:nabLE.easy.Liebniz}
		\nabla_{L_k} E_k = \nabla_{L_k} ( E^j_k \partial_j ) = L_k(E^j_k)\partial_j + E_k^j \nabla_{L_k} ( \partial_j).
		\end{equation}
		Writing $L_k = \f 1N e_0 - X_k^l\rd_l$ (by \eqref{def:e0} and \eqref{nXEL}). Hence, using \eqref{eje0ei},  \eqref{Diej}, we have 
 		\begin{equation} \label{nablaLij}
 		g( \nabla_{L_k} (\partial_j), \partial_i ) =  \frac{e^{2\gamma} \cdot   }{2 N}\cdot \left( \partial_q \beta^q  \cdot \delta_{i j}  + \partial_j \beta^q \cdot \delta_{ q i }  -\partial_i \beta^q  \cdot \delta_{j q}                          \right) -  e^{2\gamma} \cdot X_k^l \cdot  \left(\delta_{l i} \partial_j \gamma - \delta_{l j} \partial_i \gamma + \delta_{i j} \partial_l \gamma  \right) .
 		\end{equation} 
 		Hence, combining \eqref{eq:nabLE.easy.Liebniz} and \eqref{nablaLij}, and using \eqref{gauge}, we obtain
		\begin{equation}
		\begin{split}
		g(\nabla_{L_k} E_k , \partial_i) = &\: e^{2\gamma} \cdot \delta_{i l} \cdot L_k (E^l) + \frac{e^{2\gamma} \cdot   E_k^j }{2 N}\cdot \left( \partial_q \beta^q  \cdot \delta_{i j} + \partial_j \beta^q \cdot \delta_{ q i }  -\partial_i \beta^q  \cdot \delta_{j q}                          \right) \\
		&\: -  e^{2\gamma} \cdot X^l_k E_k^j  \cdot  \left(\delta_{l i} \partial_j \gamma - \delta_{l j} \partial_i \gamma + \delta_{i j} \partial_l \gamma  \right).
		\end{split}
		\end{equation}
 		
 		Now, by \eqref{LE} (and \eqref{def:e0} and \eqref{nXEL}), we also know that $g(\nabla_{L_k} E_k , \partial_i) =  - e^{2\gamma} \cdot (E_k \log(N)-K(E_k,X_k)) \cdot  \delta_{i l} X^l_k$. Hence, we get
 		\begin{equation*}
		\begin{split}
		L_k(E^i_k) = &\: -(E_k \log(N)-K(E_k,X_k)) \cdot  X^i_k-\frac{ 1 }{2 N}\cdot \left( E_k^i \partial_q \beta^q  + E_k^j\partial_j \beta^i   - \delta_{j q}  E_k^j        \cdot \delta^{i l} \partial_l \beta^q                   \right) \\
		&\: - \left( X^i_k  E^j_k  \partial_j \gamma + E^i_k  X^j_k \partial_j \gamma \right),
		\end{split}
		\end{equation*}
 		using the fact that $\delta_{l j} X^l_k  E^j_k  =0 $ (by \eqref{XELframecondition}). This gives \eqref{LEi}.
 		 		
 		\pfstep{Step~2: Proof of \eqref{XEi} and \eqref{EEi}} We use \eqref{Diej} and \eqref{gauge} to deduce that for $Z_k \in \{ X_k,E_k\}$: 		\begin{equation} \label{nablaZij}
 		g( \nabla_{Z_k} (\partial_j), \partial_i ) =  e^{2\gamma} \cdot Z_k^l \cdot  \left(\delta_{l i} \partial_j \gamma +\delta_{i j} \partial_l \gamma - \delta_{l j} \partial_i \gamma  \right).
 		\end{equation} 
		
		We now combine \eqref{nablaZij} with $X_k (E^j_k) \rd_j = \nab_{X_k} E_k - E_k^j \nab_{X_k}\rd_j$ and \eqref{XE}, and use additionally \eqref{gauge}, \eqref{def:e0} and \eqref{nXEL}, to obtain
		\begin{equation*}
		\begin{split}
		&\:(\eta_k  - K(E_k, X_k)) \cdot X_k^l \cdot e^{2\gamma} \cdot \de_{il} = g(\nab_{X_k} E_k, \rd_i) \\
		= &\: X_k(E_k^j)\cdot e^{2\gamma} \cdot \de_{ij} + E_k^j \cdot X_k^l \cdot e^{2\gamma} \cdot  \left(\delta_{l i} \partial_j \gamma +\delta_{i j} \partial_l \gamma - \delta_{l j} \partial_i \gamma  \right) \\
		= &\: X_k(E_k^j)\cdot e^{2\gamma} \cdot \de_{ij} +  e^{2\gamma} \cdot  \left(X_k^l \cdot \delta_{l i} \cdot E_k \gamma + E_k^j \cdot \delta_{i j} \cdot X_k \gamma  \right),
		\end{split}
		\end{equation*}
		where in the last line we used $\delta_{l j} X^l_k  E^j_k  =0 $ (by \eqref{XELframecondition}). We obtain \eqref{XEi} after rearranging.
		
		To obtain \eqref{EEi}, we argue similarly. Combining \eqref{nablaZij} with $E_k (E^j_k) \rd_j = \nab_{E_k} E_k - E_k^j \nab_{E_k}\rd_j$ and \eqref{EE2}, and using also \eqref{gauge}, \eqref{def:e0} and \eqref{nXEL}, we obtain
		\begin{equation*}
		\begin{split}
		&\:(\chi_k  - K(E_k, E_k)) \cdot X_k^l \cdot e^{2\gamma} \cdot \de_{il} = g(\nab_{E_k} E_k, \rd_i) \\
		= &\: E_k(E_k^j)\cdot e^{2\gamma} \cdot \de_{ij} + E_k^j \cdot E_k^l \cdot e^{2\gamma} \cdot  \left(\delta_{l i} \partial_j \gamma +\delta_{i j} \partial_l \gamma - \delta_{l j} \partial_i \gamma  \right) \\
		= &\: E_k(E_k^j)\cdot e^{2\gamma} \cdot \de_{ij} +  2 e^{2\gamma} \cdot E_k^l \cdot \delta_{l i} \cdot E_k \gamma - \rd_i\gamma,
		\end{split}
		\end{equation*}
		where in the last line we used $\de_{lj} E_k^l E_k^j = e^{-2\gamma}$ (by \eqref{gauge} and \eqref{XELframecondition}).

 		\pfstep{Step~3: Proof of \eqref{YXj}} This is an immediate consequence of \eqref{EXinellipticcoord}.
		
 		\pfstep{Step~4: Proof of \eqref{YLt} and \eqref{YLi}} Finally, we get \eqref{YLt} and \eqref{YLi} from the formulas 
		$$L^t_k= N^{-1},\quad L^i_k=-X^i_k - N^{-1} \cdot \beta ^i,$$ 
		which in turn follow from \eqref{nXEL} and \eqref{defnormal}. \qedhere
 	\end{proof}
 	
 	\subsection{Transport equations for the frame coefficients} \label{torsionnal}
 	We now derive transport equations for $\mu_k$ and $\varTheta_k$.
 	
 	\begin{lem} \label{torsionnallemma}
 		The frame coefficients $\mu_k$ and $\varTheta_k$ satisfy the following transport equations:
 		\begin{equation}\label{Lmu} 	
 		L_k  \log(\mu_k) =  K( X_k,X_k) -X_k \log(N),
 		\end{equation}
 		\begin{equation} \label{Lvartheta}
 		L_k( \log( \varTheta_k)) = \chi_k.
 		\end{equation}

 	\end{lem}

 	\begin{proof}  \eqref{Lmu} has already been proven in \eqref{muXEL}. To obtain \eqref{Lvartheta}, it suffices to compare the expressions in \eqref{EL-LE} with \eqref{LE-EL-torsionnal}. \qedhere 
 	\end{proof}

 	\subsection{Null structure equations for the Ricci coefficients} \label{Lderivativericcisection}

 	We now derive transport equations for $\chi_k$ and $\eta_k$. These equations will involve the Ricci curvature, which can then be expressed in terms of derivatives of the scalar field using the Einstein equations \eqref{Einsteinricci}.

 	\begin{lem} \label{Lricci}
 		Given $(I \times \mathbb R^2, g)$ in the gauge of Definition~\ref{def:gauge} and solving the Einstein equations \eqref{Einsteinricci}, it holds that for $k=1,2,3$,
 		\begin{equation} \label{Leta}
 		L_k \eta_k = -2 L_k \phi \cdot E_k \phi-\chi_k  \cdot ( K(E_k,X_k) - E_k\log N +  \eta_k),
 		\end{equation}
 		\begin{equation} \label{Lchi}
 		L_k\chi_k = -2 (L_k \phi )^2 -\chi^2_k +( K( X_k,X_k) -X_k \log(N)) \cdot \chi_k.
 		\end{equation}
 		
 	\end{lem}
 	
 	\begin{proof}\pfstep{Step~1: Proof of \eqref{Leta}} By \eqref{eta},
 		$$ L_k \eta_k =   g( \nabla_{X_k}L_k,\nabla_{L_k} E_k)+  g( \nabla_{L_k}(\nabla_{X_k}L_k),E_k)=g( \nabla_{L_k}(\nabla_{X_k}L_k),E_k),$$
 		where for the second equality we used $g( \nabla_{X_k}L_k,\nabla_{L_k} E_k)=0$ coming directly from \eqref{XL} and \eqref{LE}. 
 		
 		Then, by the definition of the Riemann curvature tensor, we obtain 
 		\begin{equation}\label{eq:eta.eqn.Riemann}
		g( \nabla_{L_k}(\nabla_{X_k}L_k),E_k) =g( \nabla_{X_k}(\nabla_{L_k}L_k),E_k)+ R(L_k,\barL,L_k,E_k)+ g( \nabla_{\left[L_k, \barL \right]} L_k,E_k). 
		\end{equation}
 		
 		Using \eqref{LL} and $g(E_k,L_k)=0$, we rewrite the first term in \eqref{eq:eta.eqn.Riemann} as 
		$$ g( \nabla_{X_k}(\nabla_{L_k}L_k),E_k)=   (K( X_k,X_k) -X_k \log(N)) \cdot g(\nabla_{X_k} L_k,E_k)= (K( X_k,X_k) -X_k \log(N)) \cdot  \eta_k.$$
 		For the second term in \eqref{eq:eta.eqn.Riemann}, we use \eqref{inversegXEL} to deduce that
 		\begin{equation*}
 		\begin{split}
 		-R(L_k,\barL,L_k,E_k) = &\: -\overbrace{R(L_k,L_k,L_k,E_k)}^{=0} -R(L_k,\barL,L_k,E_k) -  \overbrace{R(L_k,L_k,\barL,E_k)}^{=0} +\overbrace{R(L_k,E_k,E_k,E_k)}^{=0} \\
 		=&\: Ric(L_k,E_k).
 		\end{split}
 		\end{equation*}
 		The third term in \eqref{eq:eta.eqn.Riemann} can be computed using \eqref{LX-XL} and also the definition \eqref{chi} and \eqref{eta} as 
 		$$ g(  \nabla_{ \left[ L_k , \barL\right]} L_k, E_k) = -\chi_k \cdot ( K(E_k,X_k)  - E_k\log N  + \eta_k)+ \eta_k \cdot ( X_k \log(N)-K(X_k,X_k)).$$
 		
 		Combining the three terms and using \eqref{Einsteinricci} give \eqref{Leta}.

 		\pfstep{Step~2: Proof of \eqref{Lchi}} A similar computation as in Step~1, but using \eqref{chi} instead, gives
 		\begin{equation}\label{eq:chi.eqn.Riemann}
 		L_k \chi_k = g( \nabla_{E_k}( \nabla_{L_k}L_k),E_k)+R(L_k,E_k,L_k,E_k)  + g(  \nabla_{ \left[ L_k , E_k\right]} L_k, E_k) . 
 		\end{equation}
		
 		The first term in \eqref{eq:chi.eqn.Riemann} can be written using \eqref{LL} and \eqref{LE} as 
		$$ g( \nabla_{E_k}( \nabla_{L_k}L_k),E_k)=( K( X_k,X_k) -X_k \log(N)) \cdot \chi_k,$$
 		and for the second term in \eqref{eq:chi.eqn.Riemann}, we get from \eqref{inversegXEL} that
 		\begin{equation*}
 		\begin{split}
 		R(L_k,E_k,L_k,E_k)= &\: -R(L_k,E_k,E_k,L_k)  -\overbrace{R(L_k,\barL,L_k,L_k)}^{=0}- \overbrace{R(L_k,L_k,\barL,L_k)}^{=0}-\overbrace{R(L_k,L_k,L_k,L_k)}^{=0}  \\
 		= &\: -Ric(L_k,L_k) .
 		\end{split}
 		\end{equation*}
 		
 		Since $g(  \nabla_{ \left[ L_k , E_k\right]} L_k, E_k) = -\chi_k^2$ by \eqref{EL-LE}, we get \eqref{Lchi} using \eqref{Einsteinricci}. \qedhere

 	\end{proof}

 	\subsection{Initial values of eikonal quantities on $\protect\Sigma_0$}\label{sec:eikonal.initial}

 	\begin{lem} \label{riccisigma0expression}
 		Let $(u_k)_{|\Sigma_0}= a_k + c_{k j } x^j$ and $(\theta_k)_{|\Sigma_0}= b_k+ c_{k j}^{\perp} x^j $, where $a_k$, $b_k$ and $c_{ki}$ are as in \eqref{eikonalinit} and \eqref{thetainit}, $c_{ki}$ satisfies \eqref{cnormalization}--\eqref{cangle2}, and $c_{k 1}^{\perp} = -c_{k 2}$ and $c_{k 2}^{\perp} = c_{k 1}$. 
		
		Then the following identities hold on $\Sigma_0$:  	
		\begin{equation} \label{Xi(0)}
 		(\Xi_k)_{|\Sigma_0} =0,
 		\end{equation}	
		\begin{equation} \label{mu(0)formula}
 		(\mu_k)_{|\Sigma_0}=  e^{\gamma},
 		\end{equation} 	\begin{equation} \label{initialvarTheta}
 	(\varTheta_k)_{|\Sigma_0}=  e^{\gamma},
 		\end{equation} 	 	 \begin{equation} \label{X^i(0)formula}
 		(X_k^i)_{|\Sigma_0}= e^{-\gamma} \cdot  \delta^{i q} \cdot c_{k q},
 		\end{equation}		\begin{equation} \label{initialEj}
 		(E^i_k)_{|\Sigma_0} = e^{-\gamma} \cdot   \delta^{i q} \cdot  c_{k q}^{\perp}.
 		\end{equation}	 	 	
 		
 		The two Ricci coefficients $\chi_k$ and $\eta_k$ are given initially by:	
 		\begin{equation} \label{chisigma0}
 		(\chi_k)_{|\Sigma_0}= e^{-2\gamma} \cdot \de^{ii'} \de^{jj'} K_{i' j'} c_{k i}^{\perp} c_{k j}^{\perp}- X_k \gamma=e^{-2\gamma} \cdot \de^{ii'} \de^{jj'} K_{i' j'} c_{k i}^{\perp} c_{k j}^{\perp}- e^{-\gamma} c_{k q} \delta^{i q}  \partial_i \gamma,
 		\end{equation}	\begin{equation} \label{etasigma0}
 		(\eta_k)_{|\Sigma_0}=e^{-2\gamma} \cdot \de^{ii'} \de^{jj'} K_{i' j'} c_{k i} c_{k j}^{\perp}+E_k \gamma= e^{-2\gamma} \cdot \de^{ii'} \de^{jj'} K_{i' j'} c_{k i} c_{k j}^{\perp}+e^{-\gamma}\cdot  c_{k q}^{\perp}  \cdot   \delta^{j q}\partial_j \gamma.
 		\end{equation}
 	\end{lem}
 	\begin{proof}  
	\pfstep{Step~1: Proof of \eqref{Xi(0)}--\eqref{initialEj}} First, we notice that, for our choice of $(u_k)_{|\Sigma_0}$ and $(\theta_k)_{|\Sigma_0}$, $\partial_{u_k}$ and $\partial_{\theta_k}$ are $g$-orthogonal. Since $E_k$ is proportional to $\partial_{\theta_k}$ and $g$-orthogonal to $X_k$, it means that $X_k$ is proportional to $\partial_{u_k}$, hence \eqref{Xi(0)}.

 		We expand \eqref{eikonal1} to find the following equation: 
		$$ (g^{-1})^{0 0} ((\partial_t u_k)_{|\Sigma_0})^2 + 2 (g^{-1})^{0 i} (\partial_{i}u_k)_{|\Sigma_0} \cdot (\partial_t u_k)_{|\Sigma_0}+  (g^{-1})^{i j} (\partial_{i}u_k)_{|\Sigma_0} (\partial_{j}u_k)_{|\Sigma_0}=0.$$
 		Solving the quadratic equation in $\rd_t u_k$, and using also \eqref{inversegelliptic}, we obtain that on $\Sigma_0$,
 		$$ \partial_t u_k = \beta^i \partial_{i}u_k \pm N \cdot e^{-\gamma}.$$
 		Since $e_0 u_k = \rd_t u_k - \bt^i \rd_i u_k >0$ and $N>0$ (by Definitions~\ref{def:eikonal} and \ref{def:Sigmat}), we have $\rd_t u_k = \bt^i \rd_i u_k + N e^{-\gamma}$. Comparing this with \eqref{dtui}, we obtain $\partial_t u_k = \beta^i c_{k i} + N e^{-\gamma}= \beta^i c_{k i} + N \mu_k^{-1}$, which gives \eqref{mu(0)formula}.
 		
 		To prove \eqref{X^i(0)formula}, we combine $(\partial_{i}u_k)_{|\Sigma_0} = c_{k i}$ with \eqref{ucartderivative} and \eqref{mu(0)formula}. By \eqref{EXinellipticcoord}, this also gives \eqref{initialEj}. 
		
		Then we use \eqref{thetaE} and \eqref{initialEj}, \eqref{thetacartderivative} to get $$ E_k \theta_k = \varTheta_k^{-1}= e^{-\gamma} \cdot   \delta^{i j} \cdot  c_{k i}^{\perp} \cdot  c_{k j}^{\perp}=e^{-\gamma},$$ which is \eqref{initialvarTheta}.
 		
 		\pfstep{Step~2: Proof of \eqref{chisigma0}--\eqref{etasigma0}} By \eqref{barchiXEL}, \eqref{barnuXEL}, \eqref{X^i(0)formula} and \eqref{initialEj}, we find 
		\begin{equation}\label{eq:chi.eta.data.0}
		 (\chi_k)_{|\Sigma_0} =  e^{-2\gamma} \cdot \de^{ii'} \de^{jj'} K_{i' j'} c_{k i}^{\perp} c_{k j}^{\perp}-g(\nabla_{E_k} X_k, E_k), \quad \eta_k =  e^{-2\gamma} \cdot \de^{ii'} \de^{jj'} K_{i' j'} c_{k i} c_{k j}^{\perp}-g(\nabla_{X_k} X_k, E_k).
 		\end{equation}
 		
 		Since $g(E_k,X_k)=0$, we use \eqref{X^i(0)formula}--\eqref{initialEj} to obtain that for any vector field $Y$,
 		$$  g( \nabla_Y X_k, E_k)_{|\Sigma_0} = e^{-\gamma} \delta^{i q} c_{k q} Y^j E_k^l  \cdot g( \nabla_{\partial_j} \partial_i, \partial_l). $$
 		
 		By \eqref{Diej}, we find the following formula: 
		\begin{equation}\label{eq:YE.Gamma.ijl}
		 Y^j E_k^l  \cdot g( \nabla_{\partial_j} \partial_i, \partial_l)= e^{2\gamma}  \cdot Y^j E_k^l  \cdot ( \delta_{i l} \partial_j \gamma+ \delta_{j l} \partial_i \gamma-\delta_{i j}  \partial_l \gamma). 
		 \end{equation}
 		
 		Using \eqref{eq:YE.Gamma.ijl} for $Y= E_k$ and $Y=X_k$, and applying \eqref{X^i(0)formula} and \eqref{initialEj}, we then find that  
		$$ E^j_k E_k^l  \cdot g( \nabla_{\partial_j} \partial_i, \partial_l)=  \partial_i \gamma, \quad X^j_k E_k^l  \cdot g( \nabla_{\partial_j} \partial_i, \partial_l)= (  c_{k i}^{\perp} c_{k q}-  c_{k i} c_{k q}^{\perp} ) \cdot   \delta^{j q}\partial_j \gamma,$$
 		which ultimately implies, using the orthonormality of $c$: 
		\begin{equation}\label{eq:chi.eta.data.1}
		e^{-\gamma}\cdot  \delta^{i q} c_{k q} \cdot E^j_k E_k^l  \cdot g( \nabla_{\partial_j} \partial_i, \partial_l)=  e^{-\gamma}\cdot  c_{k q} \cdot   \delta^{j q}\partial_j \gamma,\quad 
 		e^{-\gamma} \delta^{i q} c_{k q} \cdot X^j_k E_k^l  \cdot g( \nabla_{\partial_j} \partial_i, \partial_l)= -  e^{-\gamma}\cdot  c_{k q}^{\perp}  \cdot   \delta^{j q}\partial_j \gamma.
		\end{equation}
		
 		Combining \eqref{eq:chi.eta.data.0} with \eqref{eq:chi.eta.data.1} gives \eqref{chisigma0} and \eqref{etasigma0}.  \qedhere

 	\end{proof}

	\section{Function spaces and norms}\label{sec:norms}
	
	This section is devoted to the definition of all the function spaces and norms that are used throughout the remainder of the paper.
	
	\subsection{Pointwise norms}
	
	\begin{defn}\label{def:pointwise.norm}
	Define the following pointwise norms in the coordinate system $(t,x^1,x^2)$ associated to the elliptic gauge (see \ref{ellipticgaugedef}):
	\begin{enumerate}
	\item Given a scalar function $f$, define
	$$|\rd_x f|^2 := \sum_{i=1}^2 (\partial_{i}f)^2, \quad 
	|\partial f|^2 :=  \sum_{\alp =0}^2 (\partial_{\alp} f)^2.$$
	\item Given a higher order tensor field, define its norm and the norms of its derivatives componentwise, e.g.
	$$|\bt|^2:= \sum_{i=1}^2 |\bt^i|,\quad |\rd_x \bt|^2:= \sum_{i,j =1}^2|\rd_i \bt^j|^2,\quad |K|^2:= \sum_{i,j=1}^2 |K_{ij}|^2, \quad |\rd K|^2:= \sum_{\alp=0}^2 \sum_{i,j=1}^2 |\rd_\alp K_{ij}|^2\quad etc.$$
	\item Higher derivatives are defined analogously, e.g.
	$$|\rd^2 f|^2 := \sum_{\alp,\sigma = 0}^2 (\rd_{\alp\sigma}^2 f)^2, \quad |\rd \rd_x K|^2:= \sum_{\substack{\alp =0,1,2 \\ i,j,l = 1,2}} |\rd_\alp\rd_i K_{jl}|^2, \quad etc.$$
	\end{enumerate}
 	\end{defn}
 	
	\subsection{Lebesgue and Sobolev spaces on $\Sigma_t$}
	
	Unless otherwise stated, all Lebesgue spaces are defined with respect to the measure $dx^1\, dx^2$ (which is in general different from the volume form induced by $\bar{g}$).
	
	Before we define the norms, we define the following weight function.
	\begin{defn}[Japanese brackets]\label{def:JapBra}
	Define $ \la x \ra := \sqrt{1+|x|^2}$ for $x\in \RR^2$ and $\la s \ra := \sqrt{1 + s^2}$ for $s\in \RR$.
	\end{defn}

	\begin{defn}[$C^k$ and H\"older norms]\label{def:Holder}
	For $k\in \mathbb N \cup \{0\}$ and $s \in (0,1)$, define $C^k(\Sigma_t)$ to be the space of continuously $k$-differentiable functions with respect to elliptic gauge coordinate vector fields $\rd_x$ with norm $\|f \|_{C^k(\Sigma_t)} := \sum_{|\alp|\leq k} \sup_{\Sigma_t} |\rd_x^\alp f|$, and define $C^{k,s}(\Sigma_t) \subseteq C^k(\Sigma_t)$ with H\"older norm defined with respect to the elliptic gauge coordinates as $\|f\|_{C^{k,s}(\Sigma_t)} := \|f \|_{C^k(\Sigma_t)} + \sup_{\substack{ x,y\in \Sigma \\ x\neq y}} \sum_{|\alp| = k} \f{|\rd_x^\alp f(x) - \rd_x^\alp f(y)|}{|x-y|^s}$.
	\end{defn}
	
	\begin{defn}[Standard Lebesgue and Sobolev norms]\label{def:Sobolev.norm}
	\begin{enumerate}
	\item For $k\in \mathbb N \cup \{0\}$ and $p \in [1,+\infty)$, define the (unweighted) Sobolev norms $$ \| f \|_{W^{k,p}(\Sigma_t)} := \sum_{|\alpha| \leq k } \left(\int_{\Sigma_t}  |\partial_x^{\alpha} f|^p(t,x^1,x^2)\, dx^1 dx^2 \right)^{\frac{1}{p}}.$$
	For $k\in \mathbb N \cup \{0\}$, define
	$$\| f \|_{W^{k,\infty}(\Sigma_t)} := \sum_{|\alp|\leq k} \esssup_{(x^1,x^2)\in \Sigma_t} |\partial_x^{\alpha} f|^p(t,x^1,x^2).$$
	\item Define $L^p(\Sigma_t) :=W^{0,p}(\Sigma_t)$ and $H^k(\Sigma_t):=W^{k,2}(\Sigma_t)$.
	\end{enumerate}
	\end{defn}
	
	\begin{defn}[Fractional Sobolev norms]\label{def:fractional.Sobolev.norm}
	For $s\in \mathbb R\setminus (\mathbb N \cup \{0\})$, define $H^{s}(\Sigma_t)$ by
	$$\|f\|_{H^{s}(\Sigma_t)} := \|\Db^s f\|_{L^2(\Sigma_t)}.$$
	where $\Db^s$ is defined via the (spatial) Fourier transform $\mathcal F$ (in the $x$ coordinates) by $\mathcal F(\Db^s f):= \la \xi\ra^s \mathcal F$.
	\end{defn}
	
	\begin{defn}[Weighted norms]\label{def:weighted.Sobolev.norm}
	\begin{enumerate}
	\item For $k\in \mathbb N \cup \{0\}$, $p \in [1,+\infty)$ and $r\in \RR$, define the weighted Sobolev norms by
	$$ \| f \|_{W^{k,p}_r(\Sigma_t)}= \sum_{|\alpha| \leq k } \left(\int_{\Sigma_t} \la x \ra^{p \cdot (r+|\alpha|)} |\partial_x^{\alpha} f|^p(t,x^1,x^2) \,dx^1 \,dx^2 \right)^{\f 1p},$$
	with obvious modifications for $p = \infty$. 
	\item Define also $L^p_r(\Sigma_t) := W^{0,p}_r(\Sigma_t)$ and $H^k_r(\Sigma_t) := W^{k,2}_r(\Sigma_t)$. Moreover, define $C^k_r(\Sigma_t)$ as the closure of Schwartz functions under the $L^\i_r(\Sigma_t)$ norm.
	\end{enumerate}
	\end{defn}
	
	\begin{defn}[Mixed norms]\label{def:mixed.Sobolev.norm}
	We will use mixed Sobolev norms, mostly in the $(u_k,\th_k,t_k)$ coordinates in spacetime or the $(u_k,u_{k'})$ coordinates on $\Sigma_t$. Our convention is that the norm on the right is taken first. For instance,
	$$ \| f \|_{L^2_{u_{k'}} L^{\infty}_{u_k}(\Sigma_t)}= (\int_{u_{k'} \in \RR} (\sup_{u_k \in \RR} f(t,u_k,u_{k'}))^2 d u_{k'})^{\frac{1}{2}},$$
	and analogously for other combinations.
	\end{defn}
	
	\begin{defn}[Norms for derivatives]
	We combine the notations in Definition~\ref{def:pointwise.norm} with those in Definitions~\ref{def:Sobolev.norm}--\ref{def:mixed.Sobolev.norm}. For instance, given a scalar function $f$,
	$$\|\rd f\|_{L^2(\Sigma_t)} := (\int_{\Sigma_t} \sum_{\alp = 0}^2 |\rd_\alp f|^2 \, dx^1\, dx^2)^{\f 12},$$
	and similarly for $\|\rd_x f \|_{L^2(\Sigma_t)}$, $\|\rd\rd_x f\|_{L^2(\Sigma_t)}$, etc.
	\end{defn}
	
	\subsection{The Littlewood--Paley projection and Besov spaces in $(u_k,u_{k'})$ coordinates}\label{sec:LPukukp}
	
	Assume for this subsection that $k \neq k'$, so that $(u_k, u_{k'})$ forms a coordinate system on $\Sigma_t$. 
	
	\begin{defn}[Littlewood--Paley projection]
	Define the Fourier transform in the $(u_k, u_{k'})$ coordinates by
	$$(\mathcal F^{u_k, u_{k'}} f)(\xi_k, \xi_{k'}) =  \iint_{\mathbb R^2} f(u_k, u_{k'}) e^{-2\pi i(u_k \xi_k + u_{k'} \xi_{k'})}\, du_k \, du_{k'}.$$
	Let $\varphi:\mathbb R^2 \to [0,1]$ be radial, smooth such that $\varphi(\xi) = \begin{cases} 1\quad \mbox{for $|\xi|\leq 1$} \\
	 0\quad \mbox{for $|\xi|\geq 2$}
	 \end{cases}$, where $|\xi| = \sqrt{|\xi_k|^2 + |\xi_{k'}|^2}$. 
	
	Define $P_0^{u_k,u_{k'}}$ by 
	$$P_0^{u_k,u_{k'}} f := (\mathcal F^{u_k, u_{k'}})^{-1} (\varphi(\xi) \mathcal F^{u_k, u_{k'}} f),$$
	and for $q\geq 1$, define $P_q^{u_k,u_{k'}} f$ by
	$$P_q^{u_k,u_{k'}} f := ( \mathcal F^{u_k, u_{k'}})^{-1} ((\varphi(2^{-q} \xi) - (\varphi(2^{-q+1} \xi)) \mathcal F^{u_k, u_{k'}} f(\xi)).$$
	\end{defn}
	
	\begin{defn}[The Besov space $B^{u_k,u_{k'}}_{\infty,1}$]\label{def:Besov}
	Define the Besov norm $\Bes$ by
	$$ \|f \|_{\Bes}:= \sum_{q \geq 0} \|P^{u_k, u_{k'}}_q f \|_{L^{\infty}(\Sigma_t)}.$$
	\end{defn}
	
	\subsection{Lebesgue norms on $C^k_{u_k}$ and $\Sigma_t\cap C^k_{u_k}$}
	Recall the definition of $C^k_{u_k}$ from Definition~\ref{def:sets.eikonal}. The $L^2$ norm on $C^k_{u_k}$ is defined with respect to the measure $d\th_k\, dt_k$.
	\begin{defn}[$L^2$ norm on $C^k_{u_k}$]
	For every fixed $u_k$, define the $L^2(C_{u_k}^k([0,T)))$ norm by
	$$\| f\|_{L^2(C_{u_k}^k([0,T)))} :=  (\int_0^T \int_{\mathbb R} |f|^2(u_k,\th_k,t_k) \, d\th_k\, dt_k)^{\f 12}.$$
	\end{defn}
	
	The $L^2$ norm $\Sigma_t\cap C^k_{u_k}$ is defined with respect to the measure $d\th_k$.
	\begin{defn}[$L^2$ norm on $\Sigma_t \cap C_{u_k}^k$]\label{def:L2.in.th_k}
	For every fixed $t$ and $u_k$ (and recall $t=t_k$), define the $L^2_{\th_k}(\Sigma_t \cap C_{u_k}^k)$ norm by
	$$\| f\|_{L^2_{\th_k}(\Sigma_t \cap C_{u_k}^k)} :=  (\int_{\mathbb R} |f|^2(u_k,\th_k,t_k) \, d\th_k)^{\f 12}.$$
	\end{defn}

\section{Initial data assumptions on $\Sigma_0$} \label{data}

In this section, we give the precise assumptions on the initial data for our theorem. Recall from Section~\ref{sec:intro.delta} that we will consider both impulsive wave data and $\de$-impulsive wave data, which are approximation of impulsive wave data on a length scale $\de>0$.

In \textbf{Section~\ref{sec:data}}, we first recall the notion of initial data in \cite{HL.elliptic}, which in particular involves the constraint equations. The precise assumptions on the impulsive wave data and the $\de$-impulsive data will be stated in \textbf{Section~\ref{dataroughsection}} and \textbf{Section~\ref{datasmoothsection}} respectively.

\subsection{Choice of admissible initial data}\label{sec:data}

Before we proceed, we need to fix a cutoff function for the rest of the paper:
\begin{defn}[Cutoff function $\omega$]\label{def.cutoff}
	From now on, fix a smooth cutoff function $\omega:\mathbb R\to [0,1]$ such that $\omega(\tau) \equiv 0$ for $\tau\leq 1$ and $\omega(\tau) \equiv 1$ for $\tau\geq 2$
\end{defn}

We are now ready to define the notion of an admissible initial data set (c.f.~\cite{HL.elliptic}).
 \begin{defn}[Admissible initial data]\label{datachoice}
An {\bf admissible initial data set} with respect to the elliptic gauge for the system \eqref{Einsteinricci}, \eqref{Einsteinwave} is a quadruple $(\phi,\phi',\gamma,K)$, where
	\begin{enumerate}
		\item $(\phi, \phi') \in W^{1,4}(\RR^2)\times L^4(\RR^2)$ (where $\phi'$ is the prescribed initial value for $\n \phi$) are a pair of real-valued compactly supported functions,
		\item $\gamma$ is a real-valued function with the decomposition 
		$\gamma=-\gamma_{asymp} \cdot  \omega(|x|)  \cdot \log (|x|) +\widetilde{\gamma},$
		where $\gamma_{asymp} \geq 0$ is a constant, $\omega(|x|)$ is as in Definition~\ref{def.cutoff}, and 
		$\widetilde{\gamma} \in H^2_{-\f 18}(\RR^2)$, and
		\item $K_{ij} \in H^1_{\f 78}(\RR^2)$ is a symmetric traceless (with respect to $\de^{ij}$) $2$-tensor,
		\end{enumerate}
which satisfy
\begin{enumerate}
\item the {\bf constraint equations} \eqref{Kellipticequation} and\footnote{Note that \eqref{eq:gamma.constraint} is simply a restatement of \eqref{gammaellipticequation}, but in terms of $K$ (instead of $\bt$).}
\begin{equation}\label{eq:gamma.constraint}
\Delta\gamma = -\de^{il} (\rd_i \phi)(\rd_l\phi) - \f{e^{-2\gamma}}{2} |K|^2 -  e^{2\gamma} \cdot (\n \phi)^2,
\end{equation}
 and
\item the {\bf integral compatibility condition} 
\begin{equation} \label{compatibility}
	\int_{\Sigma_0} e^{2\gamma} \phi' \cdot \partial_j \phi =0.
	\end{equation}
	\end{enumerate}

\end{defn}

\subsection{Assumptions on impulsive waves data} \label{dataroughsection}

In this subsection, we define the precise notion of impulsive wave data. That such initial data sets exist require solving the constraint equations; this will be carried out in Appendix~\ref{sec:appendix.constraint.for.rough}; see Lemma~\ref{lem:data.exist!}.

In the statement of the following theorem, $u_k$, $X_k$ and $E_k$ are to be understood as their values on $\Sigma_0$, according to \eqref{eikonalinit}, \eqref{X^i(0)formula}, and \eqref{initialEj}.
\begin{defn} \label{roughdef}
Let $\ep>0$, $0< s'' < s' < \f 12$ with $s'-s'' <\f 13$, $R\geq 10$ and $\upkappa_0 >0$. We say that $(\phi,\phi',\gamma,K)$ is an \textbf{admissible initial data set featuring three impulsive waves with parameters $(\ep,s',s'',R,\upkappa_0)$} if the following conditions are satisfied:

 \begin{enumerate}
	
	\item \label{data1} $(\phi,\phi',\gamma,K)$ is an admissible initial data set according to Definition~\ref{datachoice}.
		
	\item \label{data.transverse} The transversality condition \eqref{cangle2} hold with the parameter $\upkappa_0$.
	
	\item \label{data2} We have the decomposition $\phi= \rphi + \sum_{k=1}^3 \widetilde{\phi}_k$ and $\phi'= \rphi' + \sum_{k=1} \tphi'$ on $\Sigma_0$, where for every $k = 1,2,3$, $\mathrm{supp}(\rphi),\,\mathrm{supp}(\rphi'),\,\mathrm{supp}(\phi_k),\, \mathrm{supp}(\phi'_k) \subseteq B(0,\f R2):= \{ (x^1,x^2) \in \Sigma_0 ,\ \sqrt{(x^1)^2+ (x^2)^2} < \f R 2\}$. Moreover, for every $k = 1,2,3$, $\mathrm{supp}(\tphi) \cup \mathrm{supp}(\tphi') \subseteq \{u_k \geq 0\}$.

	\item \label{datareg} $\rphi$ and $\rphi'$ satisfy the following estimates:
	\begin{equation}\label{eq:datareg}
	\|\rphi\|_{H^{2+s'}(\Sigma_0)} + \|\rphi' \|_{H^{1+s'}(\Sigma_0)} \leq \ep.
	\end{equation}
	
	\item\label{data5} For $k=1,2,3$, $\tphi$ and $\tphi'$ satisfy the following estimates: 
		\begin{subequations}
	\begin{align}
	\label{eq:assumption.rough.energy}
	\|\tphi\|_{W^{1,\infty}(\Sigma_0)} + \|\tphi\|_{H^{1+s'}(\Sigma_0)} + \|\tphi' \|_{L^\infty(\Sigma_0)} + \|\tphi'\|_{H^{s'}(\Sigma_0)} \leq &\: \ep, \\
	\label{eq:assumption.rough.energy.commuted}
	\|E_k\tphi\|_{H^{1+s''}(\Sigma_0)} + \|E_k \tphi' \|_{H^{s''}(\Sigma_0)} + \| \tphi' - X_k \tphi \|_{H^{1+s''}(\Sigma_0)} \leq &\: \ep.
	\end{align}
	\end{subequations}

\item \label{datameasure} For $k = 1,2,3$, there exist signed Radon measures $T_{i j,k}$, $T_{i,k}'$, $T_{ijE,k}$, $T'_{iE,k}$ and $T_{ijL,k}$ such that 
\begin{equation}\label{eq:T.assumptions.1}
\begin{split}
\| \rd^2_{ij}\tphi - T_{i j,k}\|_{L^2(\Sigma_0)} + \| \rd_{i}\tphi' - T_{i,k}'\|_{L^2(\Sigma_0)} + \| \rd^2_{ij}E_k \tphi - T_{ijE,k}\|_{L^2(\Sigma_0)} &\\
+ \| \rd_{i}E_k\tphi' - T_{iE,k}'\|_{L^2(\Sigma_0)} + \| \rd^2_{ij}(\tphi' - X_k \tphi) - T_{ijL,k}\|_{L^2(\Sigma_0)} & \leq \ep,
\end{split}
\end{equation}
\begin{equation} \label{eq:T.assumptions.2}
\mathrm{supp}(T_{i j,k}) \cup \mathrm{supp}(T_{i,k}') \cup \mathrm{supp}(T_{i j E,k}) \cup \mathrm{supp}(T_{iE,k}') \cup  \mathrm{supp}(T_{i j L,k}) \subseteq \{ u_k=0\},
\end{equation}  
and 
\begin{equation}\label{eq:T.assumptions.3}
T.V.(T_{i j,k })+  T.V.(T_{i,k }') + T.V.(T_{i j E,k })+  T.V.(T_{i E,k }') + T.V.(T_{i j L,k }) \leq \ep,
\end{equation} 
where $T.V.$ is the total variation norm of Radon measures. 
	\item \label{nondegeneracy} The following lower bound holds:
	\begin{equation}\label{eq:nondegeneracy.in.thm}
	\left\| \rd_1 \phi - \f{\la \rd_1\phi,\, \rd_2\phi \ra_{L^2(\Sigma_0, dx)}}{\|\rd_2\phi \|_{L^2(\Sigma_0)}^2} \rd_2\phi \right\|_{H^{-3}(\Sigma_0)} \times \left\| \rd_2 \phi - \f{\la \rd_1\phi,\, \rd_2\phi \ra_{L^2(\Sigma_0, dx)}}{\|\rd_1\phi \|_{L^2(\Sigma_0)}^2} \rd_1\phi \right\|_{H^{-3}(\Sigma_0)} \geq  \ep^{\f 52}.
	\end{equation}
\end{enumerate}
\end{defn}

The following remarks clarify Definition~\ref{roughdef}.

\begin{rmk}
It may be helpful to rephrase the main points of our assumptions in words:
\begin{enumerate}
\item $\rphi$ and $\rphi'$ are the regular parts of $\phi$ and $\phi'$.
\item $\tphi$ and $\tphi'$ are singular. In particular, the second derivatives of $\tphi$ and first derivatives of $\tphi'$ are Radon measures with singular parts supported on $\{u_k = 0\}$. 
\item However, $E_k \tphi$, $E_k \tphi'$ and $L_k \tphi$ are better behaved. (Note that $\tphi' - X_k \tphi$ corresponds to $L_k \tphi$.)
\end{enumerate}

Notice in particular that all these bounds are consistent with $X_k \tphi$ having a jump discontinuity of amplitude $O(\epsilon)$ across $\{u_k=0\}$.

\end{rmk}

\begin{rmk}
In Definition~\ref{roughdef}, we assumed that $\mathrm{supp}(\tphi) \cup \mathrm{supp}(\tphi') \subseteq \{u_k \geq 0\}$. This is not a severe restriction: we can remove this condition as long as we assume instead that $(\tphi,\tphi')_{|\{u_k <0\}}$ can be extended to a pair of functions with $H^{2+s'}(\Sigma_0)\times H^{1+s'}(\Sigma_0)$ norms of $O(\ep)$. In this case, it is easy to redefine the decomposition so that the new decomposition obeys the assumptions in Definition~\ref{roughdef} (including those for the support properties), after allowing $\ep$ and $R$ to increase by a constant multiplicative factor.
\end{rmk}

\begin{rmk}
As part of the proof, we will show that the three singularities propagate along $u_k = 0$. As a result, the three impulsive waves interact at the spacetime point characterized by $u_1 = u_2 = u_3 = 0$. As we will see, defining $\underline{u}_k=t + a_k + c_{kj} x^j$, we have $u_k = \underline{u}_k + O(\ep^{\f 34})$ on the support of $\phi$. Therefore, while our theorem applies to any choice of $a_k$'s and $c_{kj}$'s, in the particular case where $\underline{u}_1 = \underline{u}_2 = \underline{u}_3 = 0$ corresponds to a spacetime point\footnote{Notice that the corresponding spacetime point $(\underline{t}, \underline{x}^1,\underline{x}^2)$ can be given explicitly by
$$\begin{bmatrix}
\underline{t} \\ \underline{x}^1 \\ \underline{x}^2
\end{bmatrix} 
= -\begin{bmatrix}
1 & c_{11} & c_{12} \\
1 & c_{21} & c_{22} \\
1 & c_{31} & c_{32} 
\end{bmatrix}^{-1}
\begin{bmatrix}
a^1 \\ a^2 \\ a^3
\end{bmatrix}
$$
whenever inverse matrix above is well-defined.} in $[0,1]\times \RR^2$ within $\cup_{k=1}^3 \mathrm{supp}(\tphi)$, the theorem indeed features the interaction of three impulsive waves within the time interval $t\in [0,1]$.
\end{rmk}

\begin{rmk}\label{rmk:nondegeneracy}
The condition \eqref{eq:nondegeneracy.in.thm} can be thought of as a non-degeneracy assumption, which is used only to solve the constraints to obtain data for $\de$-impulsive waves; see the proof of Lemma~\ref{lem:data.approx}. Notice that while $\rd_1\phi$ and $\rd_2\phi$ are $O(\ep)$ in $L^2(\Sigma_0)$, we only need a weaker lower bound of order $\ep^{\f 52}$ (as opposed to $\ep^2$). In particular, it can be checked that \eqref{eq:nondegeneracy.in.thm} can always be guaranteed after adding, say, an $O(\ep^{\f 65})$ smooth perturbation.

We remark also that for any non-zero compactly supported $H^1(\Sigma_0)$ function $\phi$, LHS of  \eqref{eq:nondegeneracy.in.thm} is non-zero; see Lemma~\ref{lem:some.nontrivial.lower.bd}.
\end{rmk}

\subsection{Assumptions on $\de$-impulsive waves data} \label{datasmoothsection}

In this section we present a choice of smooth data, which are not strictly speaking impulsive, but which obey scaled estimates \textit{consistent} with the data being a smooth approximation of the data of Definition \ref{roughdef}. Note (c.f.\ introduction) that such data are less idealized, perhaps more realistic representations of impulsive gravitational waves. Most of the paper concerns the propagation of low regularity norms for such smooth data, a result from which we \textit{ultimately} obtain local existence for the rough data of section \ref{dataroughsection}.

\begin{defn} \label{smoothdef}
Let $\ep>0$, $0< s'' < s' < \f 12$ with $s'-s'' <\f 13$, $R\geq 10$, $\upkappa_0 >0$ and $\de>0$. We say that $(\phi,\phi') \in C^{\infty}(\Sigma_0) \times C^{\infty}(\Sigma_0)$ is \textbf{admissible initial data set featuring three $\de$-impulsive waves with parameters $(\ep,s',s'',R,\upkappa_0)$} if the following holds:	
\begin{itemize}
\item The conditions \ref{data1} and \ref{data.transverse} of Definition~\ref{roughdef} are satisfied.
\item $\phi$ and $\phi'$ admit decompositions as in \ref{data2} of Definition~\ref{roughdef}, and $\rphi$, $\rphi'$, $\phi_k$, $\phi'_k$ are supported in $B(0,\f R2)$ for $k=1,2,3$. Unlike in Definition~\ref{roughdef}, however, for each $k = 1,2,3$, $\mathrm{supp}(\tphi) \cup \mathrm{supp}(\tphi') \subseteq \{u_k \geq -\de\}$.
\item The estimates in conditions \ref{datareg} and \ref{data5} of Definition~\ref{roughdef} are satisfied.
\item For $k= 1,2,3$, $(\tphi,\tphi')$ satisfy the following additional bounds (recall definitions from \eqref{defStwosided}, \eqref{defS}): 
	\begin{align} 
		 \label{eq:delta.waves.1}
		 \|  \tphi\|_{H^2(\Sigma_0)}+ \| \tphi'\|_{H^1(\Sigma_0)}+  \| E_k \tphi\|_{H^2(\Sigma_0)}+ \| E_k  \tphi'\|_{H^1(\Sigma_0)} + \|\tphi' - X_k\tphi \|_{H^2(\Sigma_t)} \  \leq &\:  \epsilon \cdot \delta^{-\frac{1}{2}}, \\
		 \label{eq:delta.waves.2}
		 \| \tphi \|_{H^2(\Sigma_0 \setminus S^k(-\de,0))} + \| \tphi' \|_{H^1(\Sigma_0 \setminus S^k(-\de,0))} \leq &\: \ep. 
	\end{align}
\end{itemize}
\end{defn}

\begin{rmk}
	The conditions \eqref{eq:delta.waves.1} and \eqref{eq:delta.waves.2} can be viewed as a smoothed-out version of \eqref{eq:T.assumptions.1}--\eqref{eq:T.assumptions.3}. Here, the second derivatives of $\tphi$ (and first derivatives of $\tphi'$, etc.) can be thought of as a Radon measure smoothed out at a length scale $\de$. 
	
	The $\de$-impulsive waves of Definition~\ref{smoothdef} can indeed by constructed by smoothing out the impulsive waves of Definition~\ref{roughdef}; see Lemma~\ref{lem:data.approx}.
\end{rmk}

\begin{rmk}
Notice that the condition \ref{nondegeneracy} in Definition~\ref{roughdef} is only needed for solving the constraint equations in the approximation argument. In particular, the class of data we can handle for interaction of $\de$-impulsive waves (as defined in Definition~\ref{smoothdef}) do not require such a condition.
\end{rmk}

\section{Precise statement of the main theorems}\label{sec:precise.statements}

In this section, we present the precise version of the main results. 

In parallel with the definitions in Section~\ref{data}, we give two versions of the main theorem. The first version (\textbf{Theorem~\ref{roughtheorem}}) concerns existence of impulsive waves (see Definition~\ref{roughdef}), while the second version (\textbf{Theorem~\ref{smooththeorem}}) concerns existence and uniqueness of $\de$-impulsive waves (see Definition~\ref{smoothdef}). We recall again that (see Section~Section~\ref{sec:intro.delta}) our proof of existence for impulsive waves relies on first understanding $\de$-impulsive waves and taking limits.

\subsection{Three (rough) impulsive gravitational waves}

We first begin with a notion of weak solutions.\footnote{It should be remarked that there is a geometric notion of weak solutions which requires only the metric to be continuous and the Christoffel symbols to be $L^2_{loc}$; see for instance \cite[Definition~2.1]{LR3}. In particular, this notion does not require any symmetry and gauge assumptions. We will however take advantage of the symmetry and gauge conditions in our definition. Note that weak solutions in the sense of Definition~\ref{def:weak.solution} are automatically weak solutions in the sense of \cite[Definition~2.1]{LR3}.}

\begin{definition}\label{def:weak.solution}
Let $\gamma, \bt^1, \bt^2,N,\phi$ be functions on $[0,T)\times \mathbb R^2$, where $N$ is everywhere non-vanishing.
\begin{enumerate}
\item We say that $(\gamma, \bt^j, N, \phi)$ is a \textbf{weak solution to the Einstein vacuum equations in polarized $\mathbb U(1)$ symmetry under elliptic gauge} if
\begin{enumerate}
\item The following regularity conditions hold: 
$$\gamma,\,\bt,\,N \in (C^0_{loc})_{t,x},\quad \rd_i \gamma,\,\rd_i\bt^j,\,\rd_i N\in (C^0_{loc})_{t,x},\quad \rd_t\gamma \in (C^0_{loc})_{t,x},\quad \rd_t\bt^j,\,\rd_t N \in (L^\i_{loc})_{t,x},$$
$$\phi \in (C^0_{loc})_t (H^1_{loc})_x,\quad \rd_t\phi \in  (C^0_{loc})_t (L^2_{loc})_x.$$
\item The following maximality condition holds pointwise:
\begin{equation}\label{eq:main.thm.maximality}
-2e_0\gamma + \rd_i\bt^i = 0.
\end{equation}
\item The elliptic equations \eqref{Kellipticequation}--\eqref{gammaellipticequation} hold weakly in the sense that for every $t\in [0,T)$, and every $\varsigma \in C^\infty_c(\mathbb R^2)$, we have\footnote{We remark that these conditions correspond respectively to $Ric(\n, \rd_i) = 2(\n\phi)(\rd_i\phi)$, $Ric(\n,\n)= 2(\n\phi)^2$ and $\de^{ij}Ric(\rd_i,\rd_j)+e^{2\gamma}Ric(\n,\n) = 2\de^{ij} (\rd_i\phi)(\rd_j\phi)+ 2e^{2\gamma}(\n\phi)^2$.}:
\begin{equation}\label{eq:main.thm.beta}
-\int_{\{t\}\times \mathbb R^2} \delta^{i k } K_{i j} \rd_k \varsigma\, dx = \int_{\{t\}\times \mathbb R^2} (\mbox{RHS of \eqref{Kellipticequation}})\times \varsigma \, dx.
\end{equation}
\begin{equation}\label{eq:main.thm.N}
-\int_{\{t\}\times \mathbb R^2} \delta^{i k } \rd_i N \rd_k \varsigma\, dx = \int_{\{t\}\times \mathbb R^2} (\mbox{RHS of \eqref{Nellipticequation}})\times \varsigma \, dx,
\end{equation}
\begin{equation}\label{eq:main.thm.gamma}
-\int_{\{t\}\times \mathbb R^2} \delta^{i k } \rd_i \gamma \rd_k \varsigma\, dx = \int_{\{t\}\times \mathbb R^2} (\mbox{RHS of \eqref{gammaellipticequation}})\times \varsigma \, dx,
\end{equation}
\item The evolution equation\footnote{We remark that this corresponds to $Ric(\rd_i,\rd_j) - \f 12 \de_{ij} \de^{kl} Ric(\rd_k,\rd_l) = 2(\rd_i\phi)(\rd_j\phi) - \de_{ij} \de^{kl}(\rd_k \phi)(\rd_l\phi)$.} \eqref{Khypequation} for $K_{ij}$ (where $K_{ij}$ is defined in terms of $\gamma$, $\bt^j$, $N$ by \eqref{maximality3}) holds  weakly in the sense that for every $\varsigma \in C^\infty_c((0,T) \times \mathbb R^2)$:
\begin{equation}\label{eq:main.thm.K}
\begin{split}
&\: \int_{(0,T) \times \mathbb R^2} (- K_{ij} (e_0 \varsigma) + K_{ij} (\rd_l \bt^l) \varsigma + \f 12 \rd_i N \rd_j \varsigma + \f 12 \rd_j N \rd_i \varsigma - \f 12 \de_{ij} \de^{lk} \rd_l N \rd_k \varsigma) \, dx \\
= &\: \int_{(0,T) \times \mathbb R^2} N \cdot (\mbox{RHS of \eqref{Khypequation}})\times \varsigma \, dx.
\end{split}
\end{equation}
\item The wave equation \eqref{Einsteinwave} holds weakly in the sense that for every $\varsigma \in C^\infty_c((0,T)\times \mathbb R^2)$,
\begin{equation}\label{eq:main.thm.wave}
 \int_{(0,T)\times \mathbb R^2} (g^{-1})^{\alp\sigma} \rd_\alp \varsigma \rd_\sigma \phi \, N e^{2\gamma} \,dx^1\, dx^2\, dt = 0.
 \end{equation}
\end{enumerate}
\item Given a weak solution $(\gamma, \bt^j, N, \phi)$ to the Einstein vacuum equations in polarized $\mathbb U(1)$ symmetry under elliptic gauge (as defined in part 1), we moreover say that the solution \textbf{achieves initial data $(\gamma_0,K_0,\phi_0, \phi'_0)$} if 
\begin{enumerate}
\item $\gamma$, $K$ and $\phi$ converges to the prescribed initial value pointwise, i.e.~
$$\lim_{t \to 0} (\gamma, K, \phi)(t,x) = (\gamma, K, \phi)(0,x).$$
\item The initial data for $\n \phi$ is achieved in an $L^2$ sense, i.e.
\begin{equation}\label{eq:def.achieving.nphi.data}
\lim_{t \to 0} \|\n \phi(t,\cdot) - \phi_0'(\cdot)\|_{L^2(\mathbb R^2)} =0.
\end{equation}
\end{enumerate}
\end{enumerate}
\end{definition}

We now state our main result for the interaction of three impulsive waves (recall the definition for the data in Definition~\ref{roughdef}).
\begin{thm} \label{roughtheorem}
For every $0< s'' < s' < \f 12$ with $s'-s'' <\f 13$, $R\geq 10$ and $\upkappa_0 > 0$, there exists $\ep_0 = \ep_0 (s',s'',R,\upkappa_0)>0$ such that the following holds.

Let $(\phi_{0},\phi'_{0},\gamma_0,K_{0})$ be an admissible initial data set featuring three impulsive waves with parameters $(\ep,s',s'',R,\upkappa_0)$ as in Definition \ref{roughdef}.

Then, whenever $\epsilon \in (0,\epsilon_0]$, there exists a Lorentzian metric $g$ 
$$g = -N^2 dt^2 + e^{2\gamma} \de_{i j} (dx^i + \beta^i dt) (dx^j + \beta^j dt)$$ 
on the manifold $M:= [0,1]\times \RR^2$ and a scalar function $\phi:M\to \RR$ such that $(\gamma,\bt^1,\bt^2,N,\phi)$ is a weak solution to the Einstein vacuum equation in polarized $\mathbb U(1)$ symmetry under elliptic gauge, with the initial data $(\phi_{0},\phi'_{0},\gamma_0,K_{0})$ (see Definition~\ref{def:weak.solution}).

	Moreover, $\phi = \rphi + \sum_{k=1}^3 \tphi$, where each of $\rphi$ and $\tphi$ is defined to satisfy the wave equation weakly in the sense of \eqref{eq:main.thm.wave}, with initial data as given in Definition~\ref{roughdef} (understood as in part 2 of Definition~\ref{def:weak.solution}). Furthermore, each of $\rphi$ and $\tphi$ is supported in $B(0,R)$ for every $t\in [0,1]$.

	Additionally, the following estimates are satisfied for all $k = 1,2,3$ and all $t\in [0,1]$, for some implicit constants depending only on $s'$, $s''$, $R$ and $\upkappa_0$, and for $\alp = 10^{-2}$: 
	\begin{enumerate}
		\item 
 The following $L^2$ estimates for $\phi$ hold:
 	\begin{subequations}
	\begin{align} 
	\label{eq:main.thm.rphi} 
 		\|\rphi \|_{H^{2+s'}(\Sigma_t)} + \|\rd_t \rphi \|_{H^{1+s'}(\Sigma_t)} \ls &\: \ep,\\
	\label{eq:main.thm.tphi.1}
		\| \tphi \|_{H^{1+s'}(\Sigma_t)} +  \| \rd_t \tphi \|_{H^{s'}(\Sigma_t)} \lesssim &\: \epsilon, \\
	\label{eq:main.thm.tphi.2}
		\|    L_k  \tphi \|_{H^{1+s'}(\Sigma_t)}+  \|  E_k  \tphi \|_{H^{1+s'}(\Sigma_t)} +   \| \rd_t L_k \tphi \|_{H^{s'}(\Sigma_t)}+ \| \rd_t E_k \tphi \|_{H^{s'}(\Sigma_t)}  \lesssim &\: \epsilon.
	\end{align}
	\end{subequations} 
	\item \label{roughb} There exist signed Radon measures $T_{\mu\nu,k} \in \mathcal{M}(\Sigma_t)$ for all $t\in [0,1]$ such that 
	 \begin{equation} 
	 \begin{split}
	\| \partial^2_{\mu\nu}\tphi -  T_{\mu\nu,k}\|_{L^2(\Sigma_t)} \ls \ep,
	\end{split}
	\end{equation} 
	Moreover, the following holds for $T_{\mu\nu,k}$:
	\begin{equation} 
	\begin{split}
			\mathrm{supp}(T_{\mu\nu,k}) \subseteq \{ u_k=0\},	\quad
				T.V._{|\Sigma_t}(T_{\mu\nu,k }) \ls \ep.
	\end{split}
	\end{equation}
\item The following Lipschitz and improved H\"older estimates hold for $\phi$:
	\begin{equation}\label{eq:roughtheorem.Lipschitz}
	\|\rd\phi\|_{L^\infty(\Sigma_t)} \ls \ep,\quad \|\rd\rphi\|_{C^{0,\f{s''}{2}}(\Sigma_t)} \ls \ep.
	\end{equation}
	Moreover, the following improved H\"older estimates in each half space\footnote{Note that it then follows easily that $\rd\tphi$ admits an extension from $\{u_k > 0\}$ to $\{u_k \geq 0\}$, which is moreover H\"older continuous on the closed subset $\{u_k \geq 0\}$. Similarly, there is a H\"older continuous extension from $\{u_k < 0\}$ to $\{u_k \leq 0\}$. These two extensions are in general different on $\{u_k = 0\}$.} away from $\{(t,x): u_k = 0\}$:
	\begin{equation}\label{eq:main.thm.Holder.away}
		\|\rd \tphi \|_{C^{0,\f{s''}2}(\Sigma_t \cap \{u_k > 0\})} + \|\rd \tphi \|_{C^{0,\f{s''}2}(\Sigma_t \cap \{u_k < 0\})} \ls \ep.
	\end{equation}
\item The wavefronts $u_k$ of the waves $\tphi$ are $(C^{1,1}_{loc})_{t,x}$ with the following estimates:
	\begin{equation}\label{eq:wavefront}
	\|\rd_i u_k\|_{W^{1,\i}(\Sigma_t)}  \ls 1,
	\end{equation}
	and the components of the vector fields $(L_k, E_k, X_k)$ adapted to the wavefronts are $(C^{0,1}_{loc})_{t,x}$ and satisfy
	\begin{equation}\label{eq:frames.in.thm}
	|L^{\beta}_k|+	|X^i_k|+ |E^i_k| \lesssim  \la x\ra^{\epsilon},\quad  |\rd_t L_k^t| + \la x \ra (|\rd_t L_k^i| + |\rd_i L_k^\bt| + |\partial X^i_k|+ |\partial E^i_k|)  \lesssim  \epsilon^{\frac{5}{4}}  \cdot\la x\ra^{4\alp}.
	\end{equation}
\item \label{roughtheorem.metric}Finally, the metric components $\gamma$ and $N$ admit a decomposition for every $(t,x) \in [0,1]\times \RR^2$
\begin{equation}\label{eq:metric.decomposition.in.thm}
\begin{split}
\gamma(t,x) = &\: -\gamma_{asymp}\,\omega(|x|)\log |x| + \widetilde{\gamma}(t,x), \\
N(t,x) = &\: 1+ N_{asymp}(t)\,\omega(|x|)\log |x| + \widetilde{N}(t,x),
\end{split}
\end{equation}
where $\gamma_{asymp}\geq 0$ is a constant, $N_{asymp}(t)\geq 0$ is a Lipschitz function of $t$, and $\omega$ is the cutoff function in Definition~\ref{def.cutoff}. Moreover, $\gamma$, $\bt^i$ and $N$ satisfy the following estimates for all $t\in [0,1]$:
\begin{equation}  \begin{split}  \label{metric.est}
|\gamma_{asymp}| + |N_{asymp}|(t) \ls \ep^{\f 32} ,\\
\sum_{\widetilde{\mfg} \in \{\widetilde{\gamma},\bt^i,\widetilde{N}\}} (\|\widetilde{\mfg} \|_{W^{1,\infty}_{1-\alp}(\Sigma_t)} + \|\rd^2_x \widetilde{\mfg} \|_{L^\i_{2-\alp}(\Sigma_t)}  ) \ls \ep^{\f 32},
\end{split}
\end{equation}
and the following estimates hold for a.e.~$t\in [0,1]$:
\begin{equation}\label{metric.est.t}
 |\rd_t N_{asymp}|(t) + \sum_{\widetilde{\mfg} \in \{\widetilde{\gamma},\bt^i,\widetilde{N}\}} \|\rd_t \widetilde{\mfg} \|_{W^{1,\f{2}{s'-s''}}_{1-s'+s''-2\alp}(\Sigma_t)} \ls \ep^{\f 32}.
\end{equation}
	\end{enumerate}

\end{thm}

The proof of Theorem~\ref{roughtheorem} can be found in Section~\ref{sec:approximation}. We will show there that Theorem~\ref{roughtheorem} follows from the theorem on $\de$-impulsive gravitational waves (Theorem~\ref{smooththeorem} below), after a suitable approximation and limiting argument.

We give a few technical remarks regarding the estimates.

\begin{rmk}
The proof of the theorem gives a few other estimates, which are not stated explicitly in Theorem~\ref{roughtheorem}. For instance, we have additional bounds for the commuting vector fields $L_k$, $E_k$, as well as for the metric components. 
\end{rmk}

\begin{rmk}
Notice that while we assume $\rd^2_{ij}E_k \tphi$, $\rd_{i} \n E_k \tphi$ and $\rd^2_{ij}L_k \tphi$ to be Radon measures initially, the theorem does not guarantee that this is propagated. We only propagate (see point 2 in Theorem~\ref{roughtheorem}) that $\rd^2_{ij}\tphi$ is a Radon measure.
\end{rmk}

\begin{rmk}
We can impose, in addition to \eqref{eq:assumption.rough.energy}--\eqref{eq:assumption.rough.energy.commuted}, the stronger assumption that for \underline{all} $s\in (0,\f 12)$,
$$ \|\tphi\|_{H^{1+s}(\Sigma_0)} + \|\tphi'\|_{H^{s}(\Sigma_0)} + \|E_k\tphi\|_{H^{1+s}(\Sigma_0)} + \|E_k \tphi' \|_{H^{s}(\Sigma_0)} + \| \tphi' - X_k \tphi \|_{H^{1+s}(\Sigma_0)} \leq \frac{\epsilon}{\sqrt{1-2s}}.$$
(Note that this is still consistent with $X_k\tphi$ having a jump discontinuity.) In this case, one can in principle also show a posteriori that the stronger estimate is propagated. Moreover, using Theorem~\ref{thm:bootstrap.Li}, one also have that $\rd\tphi \in \cap_{\th \in [0,\f 14)} C^{0,\th}$ in the sets $\{u_k >0\}$ and $\{u_k <0\}$ (as opposed to only being in $C^{0,\f{s''}{2}}$).
\end{rmk}

\subsection{Three $\de$-impulsive impulsive waves}

Now we present our result for \textit{smooth}, quantitatively impulsive data as in Definition \ref{smoothdef}. It is, in fact, the following theorem that we prove in most of the paper, and we use this theorem to obtain Theorem \ref{roughtheorem} eventually.
\begin{thm} \label{smooththeorem}
For every $0< s'' < s' < \f 12$ with $s'-s'' <\f 13$, $R\geq 10$ and $\upkappa_0 > 0$, there exists $\ep_0 = \ep_0 (s',s'',R,\upkappa_0)>0$ such that the following holds.

Let $(\phi_{0},\phi'_{0},\gamma_0,K_{0})$ be an admissible initial data set featuring three $\de$-impulsive waves with parameters $(\ep,s',s'',R,\upkappa_0)$ as in Definition \ref{smoothdef} for some $\ep>0$ and $\de >0$.

Then, whenever $\epsilon\in (0,\epsilon_0]$, there exists $\delta_0 = \de_0(\epsilon,s',s'',R,\upkappa_0)>0$ such that for all $0< \delta< \delta_0$, there exists a unique smooth Lorentzian metric $g$ 
$$g = -N^2 dt^2 + e^{2\gamma} \de_{i j} (dx^i + \beta^i dt) (dx^j + \beta^j dt)$$ 
on the manifold $M:= [0,1]\times \RR^2$ and a unique smooth scalar function $\phi: M\to \RR$ such that $(g,\phi)$ satisfy the Einstein vacuum equations in polarized $\mathbb U(1)$ symmetry under elliptic gauge \eqref{Einsteinwave}, \eqref{Kellipticequation}--\eqref{gammaellipticequation} and \eqref{Khypequation} all hold, with initial data $(\phi_{0},\phi'_{0},\gamma_0,K_{0})$, in the classical sense.	

Moreover, $\phi = \rphi + \sum_{k=1}^3 \tphi$, where $\rphi$ and $\tphi$ are defined to satisfy 
\begin{equation}\label{eq:separate.wave.equations}
\Box_g \rphi = 0,\quad \Box_g \tphi = 0,
\end{equation}
with initial data as prescribed by the corresponding decomposition in Definition~\ref{smoothdef}. Furthermore, each of $\rphi$ and $\tphi$ is supported in $B(0,R)$ for every $t\in [0,1]$.
	
Additionally, the following estimates hold for all $k = 1,2,3$ and all $t\in [0,1]$, for some implicit constants depending only on $s'$, $s''$, $R$ and $\upkappa_0$: 
\begin{enumerate}
\item The wave estimates in \eqref{eq:main.thm.rphi}--\eqref{eq:main.thm.tphi.2}, the wavefront estimate \eqref{eq:wavefront}, the vector fields estimates \eqref{eq:frames.in.thm}, and the metric estimates \eqref{eq:metric.decomposition.in.thm}--\eqref{metric.est.t} all hold.
\item The following higher-order estimates hold:
\begin{equation}\label{eq:smooththeorem.1}
\|\rd^2 \tphi\|_{L^2(\Sigma_t)} + \sum_{ \substack{ Y_k^{(1)}, Y_k^{(2)}, Y_k^{(3)} \in \{ X_k, E_k, L_k\} \\ \exists i, Y_k^{(i)} \neq X_k} } \|Y_k^{(1)} Y_k^{(2)} Y_k^{(3)} \tphi \|_{L^2(\Sigma_t)} \ls \ep\cdot \de^{-\f 12},
\end{equation}
\begin{equation}\label{eq:smooththeorem.top}
\|\phi \|_{H^3(\Sigma_t)} + \|\n \phi\|_{H^2(\Sigma_t)} \ls \ep \cdot \de^{-\f 12} + \|\phi \|_{H^3(\Sigma_0)} + \|\n \phi\|_{H^2(\Sigma_0)},
\end{equation}
together with the following improvement away from $S^k_\de$ (recall $S^k_\de = S^k(-\de,\de)$ in Definition~\ref{def:sets.eikonal}):
\begin{equation}\label{eq:smooththeorem.2}
\| \partial^2  \tphi \|_{ L^2(\Sigma_t \setminus S^k_\de)}  \lesssim \epsilon.
\end{equation}
\item The Lipschitz and improved H\"older estimates \eqref{eq:roughtheorem.Lipschitz} hold. Moreover,
\begin{equation}\label{eq:smooththeorem.Lipschitz}
\| \partial \tphi \|_{ C^{\frac{s''}{2}}(\Sigma_t \setminus S^k_\de) }\ls \ep.
\end{equation}
\end{enumerate}

\end{thm}

The high-level proof of Theorem~\ref{smooththeorem} can be found in Section~\ref{sec:pf.of.smooththeorem}. (It relies in particular on the a priori estimates stated in Section~\ref{sec:aprioriestimates}, whose proof will occupy most of the remainder of the series.)

\section{Approximation argument and the proof of Theorem~\ref{roughtheorem}}\label{sec:approximation}

In this section, we assume the validity of Theorem~\ref{smooththeorem} and prove Theorem~\ref{roughtheorem}. We first approximate the impulsive wave data in Theorem~\ref{roughtheorem} by $\de$-impulsive waves data, then use Theorem~\ref{smooththeorem} to obtain solutions for $\de$-impulsive waves, and finally pass to the $\de\to 0$ limit.

\textbf{For the remainder of this section, we assume the validity of Theorem~\ref{smooththeorem}.}

\subsection{Approximating the initial data}\label{sec:approx.data}

Our first step is to show that data in Definition~\ref{roughdef} can be approximated by data in Definition~\ref{smoothdef}. This is given by the following lemma, whose proof will be postponed to Section~\ref{sec:data.approx}. Notice that both the size and the support are allowed to be slightly larger for the approximate data.

\begin{lemma}\label{lem:data.approx}
For every $0< s'' < s' < \f 12$, $R\geq 10$ and $\upkappa_0 > 0$, there exists $\ep_0 = \ep_0 (s',s'',R,\upkappa_0) \in (0,1]$ such that the following holds for all $\ep \in (0, \ep_0]$.

Let $(\phi, \phi', \gamma, K)$ be an initial data set featuring three impulsive waves with parameters $(\ep,s',s'',R,\upkappa_0)$ as in Definition~\ref{roughdef} . Then, for every $\de \in (0,\ep^{\f 1{s'}}]$, there exists an initial data set $(\phi^{(\de)},(\phi')^{(\de)},\gamma^{(\de)},K^{(\de)})$ such that
\begin{enumerate}
\item $(\phi^{(\de)}, (\phi')^{(\de)}, \gamma^{(\de)}, K^{(\de)})$ correspond to data for three $\de$-impulsive waves with parameters $(3\ep,s',s'',2R,\upkappa_0)$ in Definition~\ref{smoothdef}, and
\item $(\phi^{(\de)}, (\phi')^{(\de)}, \gamma^{(\de)}, K^{(\de)})$ is an approximation of $(\phi, \phi', \gamma, K)$ in the sense that as $\de \to 0$, $\rphi^{(\de)} \to \rphi$, $\tphi^{(\de)} \to \tphi$ in $H^{1+s'}(\Sigma_0)$, $(\rphi')^{(\de)} \to \rphi'$, $(\tphi')^{(\de)} \to \tphi'$ in $H^{s'}(\Sigma_0)$, and, after writing $\gamma^{(\de)} = -\gamma_{asymp}^{(\de)}\om(|x|) \log |x| + \widetilde{\gamma}^{(\de)}$ and $\gamma = -\gamma_{asymp} \om(|x|) \log |x| + \widetilde{\gamma}$, it holds that $(\gamma_{asymp}^{(\de)}, \widetilde{\gamma}^{(\de)},K^{(\de)}) \to (\gamma_{asymp}, \widetilde{\gamma},K)$ in $\RR \times H^2_{-\f 18}(\Sigma_0) \times H^1_{\f 78}(\Sigma_0)$.
\end{enumerate}
\end{lemma}

Since in what follows we will need to consider the data and the solution for both the limit and the approximations, let us introduce the following conventions for the remainder of the section: \textbf{(1) we will use subscripts ${}_0$ to denote \emph{data} quantities, and quantities without the ${}_0$ subscripts corresponds to those in the \emph{solution}, and (2) a superscript ${}^{(\de)}$ denotes quantities from the approximating $\de$-impulsive waves, while a superscript ${}^{(0)}$ denotes quantities from the limit impulsive waves (for both the data and the solution).}

Suppose now we are given data $(\phi_0, \phi'_0, \gamma_0, K_0)$ as in Definition~\ref{roughdef} with parameters $\ep$, $s'$, $s''$, $R$, $\upkappa_0$. Take also $\alp = 10^{-2}$. Assuming $\ep \in (0, \ep_0]$ for a sufficiently small $\ep_0$, we apply Lemma~\ref{lem:data.approx} to obtain a $1$-parameter family of $\de$-impulsive wave data $(\phi_0^{(\de)}, (\phi')_0^{(\de)}, \gamma_0^{(\de)}, K_0^{(\de)})$.

By Theorem~\ref{smooththeorem} (local existence for $\de$-impulsive waves), there exists $\de_0>0$ such that for all $\de \in (0,\de_0]$, the initial data set $(\phi_0^{(\de)}, (\phi')_0^{(\de)}, \gamma_0^{(\de)}, K_0^{(\de)})$ gives rise to a unique solution in $[0,1]\times \RR^2$, which we denote as $(\gamma^{(\de)}, (\bt^i)^{(\de)}, N^{(\de)},\phi^{(\de)})$. \textbf{For the remainder of the section, consider such one-parameter family of $\de$-impulsive wave solutions $(\gamma^{(\de)}, (\bt^i)^{(\de)}, N^{(\de)},\phi^{(\de)})$.}

In order to prove Theorem~\ref{roughtheorem}, our goal will be to show that there exists a sequence $\de_j \to 0$ such that 
\begin{enumerate}
\item one can take appropriate limits of $(\gamma^{(\de_j)}, (\bt^i)^{(\de_j)}, N^{(\de_j)},\phi^{(\de_j)})$ as $j\to \infty$ (\textbf{Section~\ref{sec:extract.limit}}),
\item and the limit is a weak solution which satisfies all the bounds stated in Theorem~\ref{roughtheorem} (\textbf{Section~\ref{sec:limit.is.solution}}).
\end{enumerate}
See \textbf{Section~\ref{sec:conclusion.rough.theorem}} for the conclusion of the proof of Theorem~\ref{roughtheorem}.

\subsection{Extracting a limit}\label{sec:extract.limit}

Let $(\gamma^{(\de)}, (\bt^i)^{(\de)}, N^{(\de)},\phi^{(\de)})$ to be as in the end of Section~\ref{sec:approx.data}.

The goal of this subsection is to extract a suitable limit. We will combine the bounds for the $\de$-impulsive waves with various compactness results and the following standard Aubin--Lions lemma:
\begin{lemma}[Aubin--Lions lemma]\label{lem:AL}
Let $X_0\subseteq X \subseteq X_1$ be three Banach spaces such that the embedding $X_0\subseteq X$ is compact and the embedding $X \subseteq X_1$ is continuous. For $T>0$ and $q>1$, let
$$W:= \{ v \in L^\infty([0,T]; X_0): \dot{v} \in L^q([0,T]; X_1) \},$$
where $\dot{}$ denotes the (weak) derivative in the variable on $[0,T]$. 

Then $W$ embeds compactly into $C^0([0,T]; X)$.
\end{lemma}

We now begin extracting a limit of $(\gamma^{(\de)}, (\bt^i)^{(\de)}, N^{(\de)},\phi^{(\de)})$. For the remainder of the subsection, we will repeatedly extract subsequences of $\de$, which will always be denoted by $\de_j$ without relabelling.

\begin{proposition}[Limiting metric]\label{prop:limit.metric}
There exists a sequence $\de_j \to 0$ and $\gamma^{(0)} = -\gamma_{asymp}^{(0)}\om(|x|) \log |x| + \widetilde{\gamma}^{(0)}$, $(\bt^i)^{(0)}$, and $N^{(0)} = 1+ N_{asymp}^{(0)}(t) \om(|x|) \log |x| + \widetilde{N}^{(0)}$ such that after writing $\gamma^{(\de)} = -\gamma_{asymp}^{(\de)}\om(|x|) \log |x| + \widetilde{\gamma}^{(\de)}$ and $N^{(\de)} = 1+ N_{asymp}^{(\de)}(t) \om(|x|) \log |x| + \widetilde{N}^{(\de)}$, it holds that
$$(\gamma_{asymp}^{(\de_j)}, \widetilde{\gamma}^{(\de_j)}, (\bt^i)^{(\de_j)}, N_{asymp}^{(\de_j)}, \widetilde{N}^{(\de_j)}) \to (\gamma_{asymp}^{(0)}, \widetilde{\gamma}^{(0)}, (\bt^i)^{(0)}, N_{asymp}^{(0)}, \widetilde{N}^{(0)}) $$
in $\RR \times C^0([0,1]; C^1(\RR))\times C^0([0,1]; C^1(\RR)) \times C^0([0,1];\RR) \times C^0([0,1]; C^1(\RR))$.
Moreover, the following additional estimates hold for all $t \in [0,1]$:
\begin{equation}\label{eq:limit.metric.est.asymp}
0\leq \gamma_{asymp}^{(0)} \ls \ep^{\f 32},\quad |N_{asymp}^{(0)}|(t) \ls \ep^{\f 32},
\end{equation}
\begin{equation}\label{eq:limit.metric.est}
\sum_{\widetilde\mfg^{(0)} \in \{\widetilde{\gamma}^{(0)}, (\bt^i)^{(0)}, \widetilde{N}^{(0)} \} } (\|\widetilde\mfg^{(0)} \|_{W^{1,\i}_{1-\alp}(\Sigma_t)} + \|\rd^2_x \widetilde\mfg^{(0)} \|_{L^\i_{2-\alp}(\Sigma_t)} ) \ls \ep^{\f 32},
\end{equation}
and the following holds for a.e.~$t\in [0,1]$:
\begin{equation}\label{eq:limit.metric.est.2}
|\rd_t N_{asymp}^{(0)}|(t) + \|\rd_t \widetilde{\mfg}^{(0)} \|_{W^{1,\f{2}{s'-s''}}_{1-s'+s''-2\alp}(\Sigma_t)} \ls \ep^{\f 32}.
\end{equation}
\end{proposition}
\begin{proof}
By Theorem~\ref{smooththeorem}, $\gamma^{(\de)}_{asymp}$ is a set of bounded numbers obeying \eqref{eq:limit.metric.est.asymp}, which by the Bolzano--Weierstrass theorem has a limit for some $\de_j$. In particular, the bound for $\gamma_{asymp}^{(0)}$ in \eqref{eq:limit.metric.est.asymp} holds.

By Theorem~\ref{smooththeorem}, $N_{asymp}^{(\de_j)}$ are bounded $C^1$ functions of $t$, and thus by the Arzel\`a--Ascoli theorem on $C^0([0,1])$, has a subsequential strong limit in $C^0([0,1])$. Hence, the bound for $N_{asymp}^{(0)}$ in \eqref{eq:limit.metric.est.asymp} holds. Finally, the $C^1$ bound and the uniform convergence imply the $L^\i$ bound for $|\rd_t N_{asymp}^{(0)}|(t)$ in \eqref{eq:limit.metric.est.2}. 

For the convergences of $\widetilde{\gamma}$, $\bt^i$ and $\widetilde{N}$, we apply Lemma~\ref{lem:AL} with $T=1$, $X_0 = C^2(\RR^2)$, $X = C^{1}(\RR^2)$, $X_1 = W^{1,4}_{-1}(\RR^2)$, and $W$ as in Lemma~\ref{lem:AL}. (The compactness of $X_0 \subseteq X$ follows from the Arzel\`a--Ascoli theorem.) The uniform boundedness of $\widetilde{\gamma}^{(\de_j)}$, $(\bt^i)^{(\de_j)}$ and $\widetilde{N}^{(\de_j)}$ in $W$ is an immediate consequence of the estimates in  Theorem~\ref{smooththeorem}. Finally, note that by Theorem~\ref{smooththeorem}, for all sufficiently small $\de$, $\gamma^{(\de)}$, $\bt^{(\de)}$ and $N^{(\de)}$ satisfy analogous bounds as \eqref{eq:limit.metric.est.asymp}--\eqref{eq:limit.metric.est.2}, which imply the estimates in \eqref{eq:limit.metric.est.asymp}--\eqref{eq:limit.metric.est.2}. \qedhere
\end{proof}

\begin{proposition}[Limiting wavefront]\label{prop:limit.wavefront}
There exists a subsequence $\de_j$ (of that in Proposition~\ref{prop:limit.metric}, but not relabelled) such that for $k=1,2,3$, the following holds:
\begin{enumerate}
\item The eikonal functions $u_k^{(\de_j)}$ converges in $(C^1_{loc})_{t,x}([0,1]\times \mathbb R^2)$ to a limit $u_k^{(0)}$. Moreover, $u_k^{(0)}$ is a $(C^{1,1}_{loc})_{t,x}$ function satisfying the following estimate:
\begin{equation}\label{eq:limit.u}
\|\rd_i u_k^{(0)} \|_{W^{1,\infty}(\Sigma_t)} \ls 1.
\end{equation}
\item The vector fields\footnote{Here, we use the obvious notation that for every $\de$, $L_k^{(\de)}$, $E_k^{(\de)}$ and $X_k^{(\de)}$ denote respectively the $L_k$, $E_k$ and $X_k$ vector fields arising from the initial data set $(\phi^{(\de)},(\phi')^{(\de)},\gamma^{(\de)},K^{(\de)})$.} $L_k^{(\de_j)}$, $E_k^{(\de_j)}$ and $X_k^{(\de_j)}$ converge uniformly on compact sets to limiting vector fields $L_k^{(0)}$, $E_k^{(0)}$ and $X_k^{(0)}$. Moreover, on any compact set, $L_k^{(0)}$, $E_k^{(0)}$ and $X_k^{(0)}$ are $(C_{loc}^{0,1})_{t,x}$ vector fields\footnote{Notice that higher regularity holds if we separately consider different components. For instance, $E_k^{(0)}$ and $X_k^{(0)}$ are in $W^{2,\infty}_{loc}$ in spacetime; and while our estimates do not necessarily show that $L_k^{(0)}$ is $W^{2,\infty}_{loc}$ in spacetime, it does show that $\rd_i (L_k^j)^{(0)}$, $\rd_t (L_k^j)^{(0)}$ and $\rd_i (L_k^t)^{(0)}$ are spacetime-Lipschitz, and that $\rd_t (L_k^t)^{(0)}$ is Lipschitz in space.} with the following estimates (recall that $\alp = 10^{-2}$):
\begin{equation}\label{eq:limit.frames}
|L^{\beta}_k|+	|X^i_k|+ |E^i_k| \lesssim  \la x\ra^{\epsilon},\quad  |\rd_t L_k^t| + \la x \ra (|\rd_t L_k^i| + |\rd_i L_k^\bt| + |\partial X^i_k|+ |\partial E^i_k|)  \lesssim  \epsilon^{\frac{5}{4}}  \cdot\la x\ra^{4\alp}.
\end{equation}
\end{enumerate}
\end{proposition}
\begin{proof}
\pfstep{Step~1: Limit for $u_k$} By the Arzel\`a--Ascoli Theorem, in order to prove $(C^1_{loc})_{t,x}$ convergence, it suffices to prove a $C^{2}_{t,x}$ bound for $u_j^{(\de)}$ on any compact set, uniformly in $\de$. The bounds for the spatial derivatives follow from the fact that \eqref{eq:wavefront} holds for all small enough $\de$ by Theorem~\ref{smooththeorem}. To obtain the bounds for $\rd_t^2 u_k^{(\de)}$ and $\rd_x \rd_t u_k^{(\de)}$, we combine $L_k^{(\de)} u_k^{(\de)} = 0$ and $(L_k^{(\de)})^t = \f 1{N^{(\de)}}$ to write $\rd_t u_k^{(\de)} =-  N^{(\de)} \cdot (L_k^{(\de)})^i \cdot \rd_i u_k^{(\de)}$, and use the bounds in \eqref{metric.est}, \eqref{metric.est.t} and \eqref{eq:frames.in.thm} (which hold for small $\de$ by Theorem~\ref{smooththeorem}) with the estimates on $\rd_i u_k^{(\de)}$.

The estimate \eqref{eq:limit.u} then follows from the $(C^1_{loc})_{t,x}$ convergence and the estimate \eqref{eq:wavefront} (for $\de>0$).

\pfstep{Step~2: Limit for $L_k$, $E_k$ and $X_k$} Using the Arzel\`a--Ascoli Theorem, in order to prove the convergence statement, we need uniform $C^1_{t,x}$ bounds for $L_k^{(\de)}$, $E_k^{(\de)}$ and $X_k^{(\de)}$ on compact sets. These bounds are consequences of Theorem~\ref{smooththeorem}, which states that \eqref{eq:frames.in.thm} holds for small $\de>0$. Using the estimate \eqref{eq:frames.in.thm} again then implies \eqref{eq:limit.frames}. \qedhere
\end{proof}

\begin{proposition}[Limiting $\phi$]\label{prop:limit.phi}
There exists a subsequence $\de_j$ (of that in Proposition~\ref{prop:limit.wavefront}, but not relabelled) such that the following holds:
\begin{enumerate}
\item There exists $\rphi^{(0)}$ such that $(\phi_{reg}^{(\de_j)}, \rd_t \rphi^{(\de_j)}) \to (\rphi^{(0)}, \rd_t \rphi^{(0)})$ in the $C^0([0,1], H^{1+s}(\mathbb R^2)) \times C^0([0,1], H^{s}(\mathbb R^2))$ norm for all $s<s'$. Moreover, for every $t\in [0,1]$,
\begin{equation}\label{eq:phi.limit.final.1}
\| \rphi^{(0)}\|_{H^{2+s'}(\Sigma_t)} + \| \rd_t \rphi^{(0)} \|_{H^{1+s'}(\Sigma_t)} \ls \ep.
\end{equation}
\item For each $k=1,2,3$, there exists $\tphi^{(0)}$ such that $(\tphi^{(\de_j)}, \rd_t \tphi^{(\de_j)}) \to (\tphi^{(0)}, \rd_t \tphi^{(0)})$ in the $C^0([0,1], H^{1+s}(\mathbb R^2)) \times C^0([0,1], H^{s}(\mathbb R^2))$ norm for all $s<s'$. Moreover, for every $t\in [0,1]$,
\begin{equation}\label{eq:phi.limit.final.2}
\| \tphi^{(0)}\|_{H^{1+s'}(\Sigma_t)} + \| \rd_t \tphi^{(0)} \|_{H^{s'}(\Sigma_t)} \ls \ep.
\end{equation}
\item For each $k=1,2,3$, the second distributional derivatives of $\tphi^{(0)}$ satisfy the following properties:
\begin{enumerate}
\item For every $t\in [0,1]$, $\rd^2_{\mu\nu} \tphi^{(0)} = T_{\mu\nu,k} + f_{\mu\nu,k}$, where $T_{\mu\nu,k}$ is a signed Radon measure with $\mathrm{supp}(T_{\mu\nu,k}) \subseteq \{ u_k^{(0)} = 0\}$, and 
$$T.V._{|\Sigma_t}(T_{\mu\nu,k}) +  \|f_{\mu\nu,k} \|_{L^2(\Sigma_t)} \ls \ep.$$
\item For the vector fields $L_k^{(0)}$ and $E_k^{(0)}$ as in part 2 of Proposition~\ref{prop:limit.wavefront}, $L_k^{(0)} \tphi$ and $E_k^{(0)} \tphi$ are more regular and satisfy
$$\| L_k^{(0)} \tphi^{(0)} \|_{H^{1+s''}(\Sigma_t)} + \| E_k^{(0)} \tphi^{(0)} \|_{H^{1+s''}(\Sigma_t)} + \| \rd_t L_k^{(0)} \tphi^{(0)} \|_{H^{s''}(\Sigma_t)} + \| \rd_t E_k^{(0)} \tphi^{(0)} \|_{H^{s''}(\Sigma_t)} \ls \ep.$$
\end{enumerate}
\item For each $k=1,2,3$, $\rphi^{(0)}$ and $\tphi^{(0)}$ satisfy the following Lipschitz estimates for every $t\in [0,1]$:
\begin{equation}\label{eq:phi.limit.final.Lipschitz}
\|\rd\tphi^{(0)} \|_{L^\i(\Sigma_t)} \ls \ep,\quad \|\rd\rphi^{(0)} \|_{C^{0,\f{s''}2}(\Sigma_t)} \ls \ep.
\end{equation}
Moreover, for each $k=1,2,3$ 
and for every $t\in [0,1]$,
\begin{equation}\label{eq:phi.limit.final.Holder}
		\|\rd \tphi^{(0)} \|_{C^{0,\f{s''}2}(\Sigma_t \cap \{u_k > 0\})} + \|\rd \tphi^{(0)}  \|_{C^{0,\f{s''}2}(\Sigma_t \cap \{u_k < 0\})} \ls \ep.
\end{equation}
\item For each $k=1,2,3$, $\rphi^{(0)}$ and $\tphi^{(0)}$ are supported in $B(0,R)$ for every $t\in [0,1]$.
\end{enumerate}
\end{proposition}
\begin{proof}
\pfstep{Step~1: Limits for $\phi_{reg}$, $\rd_t\rphi$, $\tphi$ and $\rd_t \tphi$ and support properties (Statements 1, 2, 5)} All of $\phi_{reg}$, $\rd_t\rphi$, $\tphi$ and $\rd_t\tphi$ can be treated in a similar manner, except for the different regularity (i.e.~$H^{2+s'}$ for $\rphi$, $H^{1+s'}$ for $\rd_t \rphi$ and $\tphi$, and $H^{s'}$ for $\rd_t \tphi$). We will thus only discuss $\rphi$ in detail.

Denote, for the purpose only on this proof, by $H^{1+s}_{\mathfrak c}(\RR^2)$ the closed subspace of $H^{1+s}(\RR^2)$  where the function is compactly supported in $B(0,R)$; similarly for $H^{1+s'}_{\mathfrak c}(\RR^2)$. With this definition, for $s\in (0,s')$, $H^{1+s'}_{\mathfrak c}(\RR^2) \subseteq H^{1+s}_{\mathfrak c}(\RR^2)$ is compact (for instance using \cite[Lemma~2.1]{yCBdC1981} and Plancherel's theorem). 

For $s \in (0,s')$, we thus apply Lemma~\ref{lem:AL} with $T=1$, $X_0 = H^{1+s'}_{\mathfrak c}(\RR^2)$, $X = H^{1+s}_{\mathfrak c}(\RR^2)$, $X_1 = H^{s'}(\RR^2)$, together with the fact that Theorem~\ref{smooththeorem} guarantees the bound \eqref{eq:main.thm.rphi} holds for all small $\de$, to obtain the desired convergence.

To show that the first estimate in \eqref{eq:phi.limit.final.1} for $\rphi^{(0)}$ holds. Note that since \eqref{eq:main.thm.rphi} holds uniformly for all small $\de$, it admits a weak limit which also satisfies \eqref{eq:main.thm.rphi}. The weak limit necessarily coincides with $\rphi^{(0)}$, implying the desired bound.

Finally, it is clear from the definition that the limits are indeed supported in $B(0,R)$ as stated.

 \pfstep{Step~2: Regularity for $\rd^2_{\mu\nu} \tphi^{(0)}$ (Statement 3(a))} Fix $k$, $\mu$ and $\nu$. Let $\widetilde{\rho}:\mathbb R\to [0,1]$ be a smooth function with $\widetilde{\rho} \equiv 0$ on $\RR\setminus [-2,2]$, and $\widetilde{\rho} \equiv 1$ on $[-1, 1]$.
 
For each $\de>0$ sufficiently small, define $\rho_k^{(\de)}(u_k,\theta_k,t)=\widetilde{\rho}(\frac{u_k}{\delta})$. Introduce the decomposition 
\begin{equation}\label{eq:decomposition.to.take.limit.to.a.measure}
\rd^2_{\mu\nu}\tphi^{(\de)} = \rho_k^{(\de)}\cdot \rd^2_{\mu\nu}\tphi^{(\de)} + ( 1 - \rho_k^{(\de)} ) \cdot \rd^2_{\mu\nu}\tphi^{(\de)}.
\end{equation}

Using the Cauchy--Schwarz inequality, \eqref{eq:smooththeorem.1}, the support properties of $\rho_k^{(\de)}$, and Corollary~\ref{cor:diffeo}, we obtain $\| \rho_k^{(\de)} \cdot \rd^2_{\mu\nu}\tphi^{(\de)}\|_{L^1(\Sigma_t)} \ls \|\rho_k^{(\de)} \|_{L^2(\Sigma_t)} \|\rd^2_{\mu\nu} \tphi^{(\de)}\|_{L^2(\Sigma_t)} \ls \de^{\f 12} \cdot (\ep \de^{-\f 12}) = \ep.$ Combining this uniform $L^1$ bound with the Banach--Alaoglu theorem, it follows that there exists a subsequence of $\rho_k^{(\de)} \cdot \rd^2_{\mu\nu}\tphi^{(\de_j)}$ (not relabelled) which converges in weak-* to a signed Radon measure $T_{\mu\nu,k}$ with $T.V._{|\Sigma_t}(T_{\mu\nu,k}) \ls \ep$.

Note that indeed $\mathrm{supp}(T_{\mu\nu,k}) \subseteq \{(t,x): u_k^{(0)} = 0\}$. This is because of $\de_j \to 0$, Proposition~\ref{prop:limit.wavefront}, and $\mathrm{supp}(\rho_k^{(\de_j)} \cdot \rd^2_{\mu\nu} \tphi^{(\de_j)}) \subseteq \{(t,x): u_k^{(\de_j)} \in [-\de_j,\de_j] \}$.

On the other hand, note that $\| ( 1 - \rho_k^{(\de)} ) \cdot \rd^2_{\mu\nu}\tphi^{(\de)} \|_{L^2(\Sigma_t)} \ls \ep$ using \eqref{eq:smooththeorem.2} and $\mathrm{supp}(1-\rho_k^{(\de)}) \cap S^k_\de = \emptyset$. Therefore, by the Banach--Alaoglu theorem, there is a further subsequence $\de_j$ such that $( 1 - \rho_k^{(\de_j)} ) \cdot \rd^2_{\mu\nu}\tphi^{(\de_j)}$ converges weakly in $L^2(\Sigma_t)$ to some $f_{\mu\nu,k}$ satisfying $\| f_{\mu\nu, k} \|_{L^2(\Sigma_t)} \ls \ep$.

Finally, since by definition, $T_{\mu\nu,k} + f_{\mu\nu,k}$ is the distributional limit of $\rd^2_{\mu\nu}\tphi^{(\de_j)}$, it must hold that $\rd^2_{\mu\nu} \tphi^{(0)} = T_{\mu\nu,k} + f_{\mu\nu,k}$.

\pfstep{Step 3: Regularity for second derivatives of $\tphi$ with one good derivative (Statement 3(b))} Arguing as in Step~1, but using that $E_k^{(\de_j)} \tphi^{(\de_j)}$, $L_k^{(\de_j)} \tphi^{(\de_j)}$ satisfy uniformly the estimates \eqref{eq:main.thm.tphi.2}, it follows that (for a further subsequence) $E_k^{(\de_j)} \tphi^{(\de_j)}$, $L_k^{(\de_j)} \tphi^{(\de_j)}$, $\rd_t E_k^{(\de_j)} \tphi^{(\de_j)}$, $\rd_t L_k^{(\de_j)} \tphi^{(\de_j)}$ have (say, distributional) limits which again obey again the estimates \eqref{eq:main.thm.tphi.2}.

It thus remains to show that the distributional limit of $E_k^{(\de_j)} \tphi^{(\de_j)}$ (respectively $L_k^{(\de_j)} \tphi^{(\de_j)}$) is indeed $E_k^{(0)} \tphi^{(0)}$ (respectively $L_k^{(0)} \tphi^{(0)}$). To see this, it suffices to note that $(E_k^i)^{(\de_j)} \to (E_k^i)^{(0)}$ uniformly (by part 2 of Proposition~\ref{prop:limit.wavefront}) and that $\rd\tphi^{(\de_j)} \to \rd\tphi^{(0)}$ in $L^\i([0,1]; L^2(\RR^2))$ (by part 2 of this proposition).

\pfstep{Step~4: Lipschitz and improved H\"older bounds (Statement 4)} To show the first estimate in \eqref{eq:phi.limit.final.Lipschitz}, note that Step~1 in particular implies $\rd\tphi^{(\de_j)} \to \rd\tphi^{(0)}$ in $L^\i([0,1]; L^2(\RR^2))$. Thus, after passing to a further subsequence, $\rd\tphi^{(\de_j)} \to \rd\tphi^{(0)}$ almost everywhere. In particular, the desired $L^\infty$ bound follows from that for $\rd\tphi^{(\de_j)}$.

For the second estimates in \eqref{eq:phi.limit.final.Lipschitz} (for the $C^{1,\f{s''}{2}}$ bound for $\rphi$), we first note that the uniform $C^{1,\f{s''}{2}}$ estimates, together with Lemma~\ref{lem:AL}, imply that (for a further subsequence) $\rd\rphi^{(\de_j)}\to \rd\rphi^{(0)}$ uniformly on compact sets. As a result, the desired $C^{1,\f{s''}{2}}$ bound follows from that for $\rd\rphi^{(\de_j)}$.

Finally, we prove the H\"older estimates in \eqref{eq:phi.limit.final.Holder}. It suffices to consider $u_k > 0$ since $\tphi \equiv 0$ when $u_k <0$ by the finite speed of propagation. For every $\updelta_0 > 0$, define $\mathscr R_{\updelta_0}^k:= \{(t,x) \in [0,1]\times \RR^2 : u_k^{(0)}(t,x) \geq \updelta_0\}$. Now, by part 1 of Proposition~\ref{prop:limit.wavefront}, there exists $J \in \mathbb N$ such that $\mathscr R_{\updelta_0}^k \subseteq \{(t,x) \in [0,1]\times \RR^2 : u_k^{(\de_j)}(t,x)\geq \de_j \}$ for all $j\geq J$. Thus, using \eqref{eq:smooththeorem.Lipschitz}, we argue as for the improved H\"older estimate in \eqref{eq:phi.limit.final.Lipschitz} to obtain
\begin{equation}\label{eq:phi.limit.final.Lipschitz.almost}
\| \rd \tphi^{(0)} \|_{C^{0,\f{s''}{2}}(\Sigma_t \cap \mathscr R_{\updelta_0}^k)} \ls \ep.
\end{equation}
Importantly, since \eqref{eq:phi.limit.final.Lipschitz.almost} is independent of $\updelta_0$, we deduce the estimate in \eqref{eq:phi.limit.final.Holder} for $u_k >0$. \qedhere
\end{proof}

\subsection{The limit is a desired solution of Theorem~\ref{smooththeorem}}\label{sec:limit.is.solution}

We continue to take $(\phi_0, \phi'_0, \gamma_0, K_0)$, $(\phi_0^{(\de)}, (\phi')_0^{(\de)}, \gamma_0^{(\de)}, K_0^{(\de)})$ and $(\gamma^{(\de)}, (\bt^i)^{(\de)}, N^{(\de)},\phi^{(\de)})$ to be as in the end of Section~\ref{sec:approx.data}, and let $(\gamma^{(0)}, (\bt^i)^{(0)}, N^{(0)}, \phi^{(0)})$ be as given by Propositions~\ref{prop:limit.metric} and \ref{prop:limit.phi}.

\begin{proposition}\label{prop:limit.is.solution}
\begin{enumerate}
\item The limit $(\gamma^{(0)}, (\bt^i)^{(0)}, N^{(0)}, \phi^{(0)})$ given by Propositions~\ref{prop:limit.metric} and \ref{prop:limit.phi} is a weak solution to the Einstein vacuum equations in polarized $\mathbb U(1)$ symmetry under elliptic gauge (in the sense of part 1 Definition~\ref{def:weak.solution}). Moreover, each of $\rphi^{(0)}$ and $\tphi^{(0)}$ satisfies the wave equation weakly in the sense of \eqref{eq:main.thm.wave}.
\item The solution in part 1 achieves the given data $(\phi_0,\phi'_0, \gamma_0, K_0)$ in Theorem~\ref{roughtheorem}, and moreover, each of $\rphi^{(0)}$ and $\tphi^{(0)}$ achieve the initial data as given in Definition~\ref{roughdef}  (in the sense of part 2 of Definition~\ref{def:weak.solution}).
\end{enumerate}
\end{proposition}
\begin{proof}
\pfstep{Step~1: The limit is a weak solution} For every $\de>0$ sufficiently small, we have a smooth solution and thus the equations \eqref{eq:main.thm.maximality}--\eqref{eq:main.thm.wave} all hold (either directly or after integrating by parts). Moreover, Theorem~\ref{smooththeorem} stated that each of $\rphi^{(0)}$ and $\tphi^{(0)}$ satisfies the wave equation. In order to pass to the limit, it suffices to have, for $\mfg \in \{\gamma, \bt^i, N\}$, $\mfg$, $\rd_i\mfg$ and $e_0 \gamma$ converge uniformly on compact sets, and $\rd_t \mfg$ and $\rd \phi$ both converge in $L^2([0,1];L^2_{loc}(\mathbb R^2))$. These convergences (in fact much stronger ones) follow from Propositions~\ref{prop:limit.metric} and \ref{prop:limit.phi}.

\pfstep{Step~2: The limit achieves the given initial data} In view of part 2 of Lemma~\ref{lem:data.approx}, it suffices to prove quantitative convergence for $\de>0$. 

Note that $\rd_t \gamma^{(\de)}$, $\rd_t K^{(\de)}$, $\rd_t \rphi^{(\de)}$ and $\rd_t \tphi^{(\de)}$ are all uniformly bounded on $B(0,R_0)$ for any $R_0>0$, which implies that
\begin{equation}\label{eq:achieving.data.easy.1}
\begin{split}
&\: \| \gamma^{(\de)}_{|\Sigma_t} - \gamma_0^{(\de)} \|_{L^\i(\RR^2 \cap B(0,R_0))} + \| K^{(\de)}_{|\Sigma_t} - K_0^{(\de)} \|_{L^\i(\RR^2 \cap B(0,R_0))} \\
&\: \qquad + \| (\rphi^{(\de)})_{|\Sigma_t} - (\rphi^{(\de)})_0 \|_{L^\i(\RR^2 \cap B(0,R_0))} + \| (\tphi^{(\de)})_{|\Sigma_t} - (\tphi^{(\de)})_0 \|_{L^\i(\RR^2 \cap B(0,R_0))} \ls \ep t.
\end{split}
\end{equation}
On the other hand, by Lemma~\ref{lem:data.approx} and Sobolev embedding ($H^{1+s'}_{loc}(\RR^2) \hookrightarrow L^\i_{loc}(\RR^2)$), we know that for any $R_0>0$,
\begin{equation}\label{eq:achieving.data.easy.2}
\lim_{\de\to 0} (\| \gamma^{(\de)}_0 - \gamma_0 \|_{L^\i(\Sigma_0 \cap B(0,R_0))} + \| K^{(\de)}_0 - K_0 \|_{L^\i(\Sigma_0 \cap B(0,R_0))} + \| \phi^{(\de)}_0 - \phi_0 \|_{L^\i(\Sigma_0 \cap B(0,R_0))}  = 0.
\end{equation}
Combining \eqref{eq:achieving.data.easy.1} and \eqref{eq:achieving.data.easy.2}, and using the triangle inequality give part 2(a) of Definition~\ref{def:weak.solution}.

It remains to check \eqref{eq:def.achieving.nphi.data} for each of $\rphi'$ and $\tphi'$. First, by Proposition~\ref{prop:limit.phi},
\begin{equation}\label{eq:def.achieving.nphi.data.1}
 \lim_{\de\to 0} \sup_{t\in [0,1]} ( \| (\n\rphi)^{(0)} - (\n\rphi)^{(\de)} \|_{L^2(\Sigma_t)} + \sum_{k=1}^3 \| (\n\tphi)^{(0)} - (\n\tphi)^{(\de)} \|_{L^2(\Sigma_t)} ) = 0,
\end{equation}
where by Lemma~\ref{lem:data.approx}, we also have
\begin{equation}\label{eq:def.acieving.nphi.data.1.1}
(\n\rphi)^{(0)}_{|\Sigma_0} = (\rphi')_0,\quad (\n\tphi)^{(0)}_{|\Sigma_0} = (\tphi')_0.
\end{equation}

Note now that for all $\de$ sufficiently small, we have the uniform bound $\|(L_k \n \tphi)^{(\de)} \|_{L^2(\Sigma_t)} \ls \ep^{\f 34}$ by Theorem~\ref{smooththeorem} (which implies \eqref{eq:main.thm.tphi.1}, \eqref{eq:main.thm.tphi.2}, \eqref{metric.est} and \eqref{metric.est.t} hold for small $\de>0$).
Working in the $(u_k^{(\de)},\th_k^{(\de)},t_k^{(\de)})$ coordinates, and recalling $\rd_{t_k^{(\de)}} = N^{(\de)} L_k^{(\de)}$, the fundamental theorem of calculus gives
$$\n\tphi^{(\de)}(u_k,\th_k,t_k) - (\tphi')^{(\de)}_0(u_k,\th_k) = \int_0^{t_k}  N^{(\de)} \cdot (L_k \n \tphi)^{(\de)}(u_k,\th_k,s) \, \, ds.$$
Thus the bound $\|L_k \n \tphi^{(\de)} \|_{L^2(\Sigma_t)} \ls \ep^{\f 34}$, together with the compact support of $\tphi^{(\de)}$ and Minkowski's inequality, imply that
\begin{equation}\label{eq:def.achieving.nphi.data.2}
 \sum_{k=1}^3 \| (\n\tphi)^{(\de)}_{|\Sigma_t}  - (\n\tphi)^{(\de)}_{|\Sigma_0}\|_{L^2(\RR^2)} \ls \ep^{\f 34} t.
\end{equation}
Also, an analogous (and indeed easier) statement as \eqref{eq:def.achieving.nphi.data.2} holds for $\rphi$ instead of $\tphi$ (which can be proven by replacing the use of \eqref{eq:main.thm.tphi.1}, \eqref{eq:main.thm.tphi.2} by that of \eqref{eq:main.thm.rphi}):
\begin{equation}\label{eq:def.achieving.nphi.data.3}
 \| (\n\rphi)^{(\de)}_{|\Sigma_t}  - (\n\rphi)^{(\de)}_{|\Sigma_0}\|_{L^2(\RR^2)} \ls \ep^{\f 34} t.
\end{equation}
Combining \eqref{eq:def.achieving.nphi.data.1}, \eqref{eq:def.acieving.nphi.data.1.1}, \eqref{eq:def.achieving.nphi.data.2} and \eqref{eq:def.achieving.nphi.data.3} yields the desired statement
$$\lim_{t \to 0} ( \| (\n\rphi)_{|\Sigma_t}^{(0)} - (\rphi')_0 \|_{L^2(\Sigma_0)} + \sum_{k=1}^3 \| (\n\tphi)^{(0)}_{|\Sigma_t}  - (\tphi')_0 \|_{L^2(\RR^2)} ) = 0. \qedhere$$
\end{proof}

\subsection{Proof of Theorem~\ref{roughtheorem}}\label{sec:conclusion.rough.theorem}
We now conclude the proof of Theorem~\ref{roughtheorem} (under the assumption of the validity of Theorem~\ref{smooththeorem}).
\begin{proof}[Proof of Theorem~\ref{roughtheorem}]
We start with an initial data set $(\phi_0, \phi'_0, \gamma_0, K_0)$ as in Definition~\ref{roughdef} with parameters $\ep$, $s'$, $s''$, $R$, $\upkappa_0$. For $\ep_0$ sufficiently small, we use the approximation procedure described in Section~\ref{sec:approx.data} and the limiting procedure in  Proposition~\ref{prop:limit.metric}--\ref{prop:limit.phi} to obtain a limiting quadruple $(\gamma, \bt^i, N,\phi)$ (called $(\gamma^{(0)}, (\bt^i)^{(0)}, N^{(0)},\phi^{(0)})$ above). By Proposition~\ref{prop:limit.is.solution}, $(\gamma, \bt^i, N,\phi)$ is a solution arising from the given initial data set $(\phi_0, \phi'_0, \gamma_0, K_0)$ (in the sense of Definition~\ref{def:weak.solution}).

We claim that $(\gamma, \bt^i, N,\phi)$ is the desired solution as asserted in Theorem~\ref{roughtheorem}. To see this, it remains to check all the estimates.
\begin{itemize}
\item Support properties of the scalar wave and the wave estimates \eqref{eq:main.thm.rphi}--\eqref{eq:main.thm.Holder.away} follow from Proposition~\ref{prop:limit.phi}.
\item The regularity properties for $u_k$, $L_k$, $E_k$, $X_k$ and the estimates \eqref{eq:wavefront}--\eqref{eq:frames.in.thm} follow from Proposition~\ref{prop:limit.wavefront}.
\item The estimates in \eqref{eq:metric.decomposition.in.thm}--\eqref{metric.est.t} follow from Proposition~\ref{prop:limit.metric}.
\end{itemize}
This concludes the proof. \qedhere
\end{proof}

\section{Bootstrap argument and the proof of Theorem~\ref{smooththeorem}}\label{sec:proof.smooth.theorem}

In this section, we outline the structure for the proof of the main theorem on $\de$-impulsive gravitational waves, i.e.~Theorem~\ref{smooththeorem}.

The proof of Theorem~\ref{smooththeorem} is based on a priori estimates proven in a bootstrap argument:
\begin{itemize}
\item The bootstrap assumptions will be set up in \textbf{Section~\ref{bootstrapsection}}.
\item Under the bootstrap assumptions of Section~\ref{bootstrapsection}, we prove a priori estimates in \textbf{Section~\ref{sec:aprioriestimates}}.
\begin{itemize}
\item The a priori estimates are split into three steps: the metric estimates will be stated in \textbf{Section~\ref{sec:aprioriestimates1}}, the Lipschitz and improved H\"older estimates for the wave variables will be stated in \textbf{Section~\ref{sec:aprioriestimates2}}, and finally, the $L^2$-based energy estimates can be found in \textbf{Section~\ref{sec:aprioriestimates3}}.
\end{itemize}
The proof of these estimates will occupy most of the remainder of this paper and \cite{LVdM2}.
\item In \textbf{Section~\ref{sec:pf.of.smooththeorem}}, we conclude the proof of Theorem~\ref{smooththeorem} assuming the a priori estimates of Section~\ref{sec:aprioriestimates}.
\end{itemize}

\subsection{Main bootstrap assumptions}\label{bootstrapsection}

Let $\ep$, $\de$, $s'$, $s''$, $R$ and $\upkappa_0$ be as in Theorem~\ref{smooththeorem}. We continue to take $\alp = 10^{-2}$ as in Theorem~\ref{roughtheorem}.

We now introduce the main bootstrap assumptions. In the setting of a bootstrap argument (see Section~\ref{sec:aprioriestimates}), we will assume that there is a $T_B \in (0,1]$ such that all the estimates below hold on $[0,T_B) \times \RR^2$. The definitions of all the norms below can be found in Section~\ref{sec:norms}.

\bigskip

\noindent\underline{\textbf{Estimates for the metric components in the elliptic gauge.}}
Let $\omega$ be a cutoff function as in Definition~\ref{def.cutoff} . Assume that the metric components $\gamma$ and $N$ admit the following decomposition\footnote{That such a decomposition exists is a consequence of the local theory; see \cite[Theorem~5.4]{HL.elliptic}.}:
\begin{equation}\label{eq:metric.decomposition}
\begin{split}
\gamma(t,x) = &\: -\gamma_{asymp}\,\omega(|x|)\log |x| + \widetilde{\gamma}(t,x), \\
N(t,x) = &\: 1+ N_{asymp}(t)\,\omega(|x|)\log |x| + \widetilde{N}(t,x),
\end{split}
\end{equation}
where $\gamma_{asymp}\geq 0$ is a constant, and $N_{asymp}(t)\geq 0$ is a function of $t$ alone. Moreover, $\gamma$, $\bt^i$ and $N$ satisfy
\begin{subequations}
\begin{align}
\label{eq:BA.g.asymp}
|\gamma_{asymp}| + \sup_{0\leq t < T_B} (|N_{asymp}|(t) + |\rd_t N_{asymp}|(t)) \leq &\: \ep^{\f 54}, \\
\label{eq:BA.g.Li}
\sup_{0\leq t < T_B} \sum_{\widetilde{\mfg} \in \{\widetilde{\gamma}, \bt^i,\widetilde{N} \} } ( \| \widetilde{\mfg} \|_{W^{1,\infty}_{1 - \alp}(\Sigma_t)} + \| \rd_t\widetilde{\mfg} \|_{L^{\infty}_{1 - 2\alp}(\Sigma_t)} ) \leq &\: \ep^{\f 54}, \\
\label{eq:BA.g.L4}
\sup_{0\leq t < T_B} \sum_{\widetilde{\mfg} \in \{\widetilde{\gamma}, \bt^i,\widetilde{N} \} } \| \widetilde{\mfg} \|_{W^{2,4}_{\f 12 - \alp}(\Sigma_t)}  \leq &\: \ep^{\f 54}. 
\end{align}
\end{subequations}

\noindent\underline{\textbf{Estimates for the Ricci coefficients.}}
\begin{subequations}
\begin{align}
\label{bootstrapK} \sup_{0 \leq t < T_B}  (\| K \|_{L^{\infty}_{2-\alpha}(\Sigma_t)} + \|\rd_x K \|_{L^\i_{2-\alp}(\Sigma_t)}) \leq \epsilon^{\f 54},\\
\label{bootstrapricci}
\max_k \sup_{0 \leq t < T_B} ( \| \chi_k \| _{ L^{\infty}_{1-\alpha}(\Sigma_t)}+ \| \eta_k \| _{ L^{\infty}_{1-\alpha}(\Sigma_t)} ) \leq \epsilon^{\f 54},\\ \label{bootstrapnablaricci}
\max_k  \sup_{0 \leq t < T_B, u_k \in \RR} (\| \partial_x  \chi_k \|_{L^2_{\theta_k}(\Sigma_t \cap C^k_{u_k})} + \| E_k \eta_k \|_{ L^2_{\theta_k}(\Sigma_t \cap C^k_{u_k})}) \leq \epsilon^{\f 54},\\
 \label{bootstrapmu}
\max_k \sup_{0 \leq t < T_B} (\|  \log \mu_k - \gamma_{asymp} \om(|x|) \log |x| \|_{L^{\infty}_{1-\alpha}(\Sigma_t)} +\| \partial_x \mu_k  \|_{L^{\infty}_{1-\alpha}(\Sigma_t)} ) \leq \epsilon^{\f 54} , \\
 \label{bootstrapvarTheta}
 \max_k \sup_{0 \leq t < T_B} (\| \log(\varTheta_k) - \gamma_{asymp} \om(|x|) \log|x|\|_{L^{\infty}_{1-2\alp}(\Sigma_t)} + \sup_{u_k \in \RR} \| \la x \ra^{-\alp} \rd_x \log \varTheta_k \|_{L^2_{\th_k}(\Sigma_t\cap C^k_{u_k})}) \leq \ep^{\f 54}.
\end{align}
\end{subequations}

\noindent\underline{\textbf{Energy estimates for $\rphi$.}}
\begin{align}  
\sup_{0 \leq t < T_B} (\| \phi_{reg} \|_{H^{s'}(\Sigma_t)}+\|\partial \phi_{reg} \|_{H^{s'}(\Sigma_t)} + \| \partial^2 \phi_{reg} \|_{H^{s'}(\Sigma_t)}) \leq &\: \epsilon^{\frac{3}{4}}. \label{BA:rphi}
\end{align}

\noindent\underline{\textbf{Energy estimates for $\tphi$.}}
\begin{subequations}
\begin{align} 
\label{bootstrapsmallnessenergy}
\max_k \sup_{0 \leq t < T_B} (\| \partial \tphi \|_{L^2(\Sigma_t)} + \sum_{Z_k \in \{L_k,\,E_k\}}\| Z_k \partial \tphi\|_{L^2(\Sigma_t)} ) \leq &\: \epsilon^{\frac{3}{4}}, \\
 \label{tphiH2bootstrap}
\max_k \sup_{0 \leq t < T_B} \|  \partial^2 \tphi \|_{L^2(\Sigma_t)} \leq &\: \epsilon^{\frac{3}{4}} \cdot \delta^{-\frac{1}{2}}, \\
 \label{tphiH3/2bootstrap}
\max_k \sup_{0 \leq t < T_B} \|  \rd \Db^{s'} \tphi \|_{L^2(\Sigma_t)} \leq &\: \epsilon^{\frac{3}{4}}, \\
\label{EtphiH2bootstrap}
\max_k \sup_{0 \leq t < T_B} \| \partial E_k \partial_x \tphi \|_{L^2(\Sigma_t)} \leq &\: \epsilon^{\frac{3}{4}} \cdot \delta^{-\frac{1}{2}}.	
\end{align}
\end{subequations}

\noindent\underline{\textbf{Improved energy estimates for $\tphi$.}}

\begin{subequations}\begin{align}
 \label{bootstrapbadunlocenergyhyp}
\max_k \sup_{0 \leq t < T_B} ( \|\partial \tphi\|_{L^2(\Sigma_t \cap S_{2\de}^k)} +\sum_{Z_k \in \{L_k,\,E_k\}}\| Z_k \partial \tphi\|_{L^2(\Sigma_t \cap S_{2\de}^k)} ) \leq &\: \epsilon^{\frac{3}{4}} \cdot \sdelta, \\
\label{BA:away.from.singular} \max_k \sup_{0 \leq t < T_B}  \| \partial^2 \tphi\|_{L^2(\Sigma_t \setminus \Sd^k)}  \leq  &\: \epsilon^{\frac{3}{4}}.
\end{align}
\end{subequations}

\noindent\underline{\textbf{Flux estimates for the wave variables.}}
\begin{subequations}
\begin{align}
 \max_k \sup_{u_k \in \mathbb R} \sum_{Z_k \in \{L_k, E_k\}} (\| Z_k \rd_x \rphi\|_{L^2(C^k_{u_k}([0,T_B)))} + \| Z_k \rphi\|_{L^2(C^k_{u_k}([0,T_B)))})  \leq &\: \ep^{\f 34}, \label{BA:flux.for.rphi}\\ 
\max_{k,k'} \sup_{u_{k'} \in\RR} \sum_{Z_{k'} \in \{L_{k'},\,E_{k'}\}} (  \| Z_{k'} \rd_x \tphi\|_{L^2(C^{k'}_{u_{k'}}([0,T_B))\setminus \Sd^k)} +   \| Z_{k'} \tphi\|_{L^2(C^{k'}_{u_{k'}}([0,T_B)))}) \leq &\: \epsilon^{\frac{3}{4}}, \label{BA:flux.for.tphi.improved} \\
\max_{k} \sup_{u_{k} \in\RR} (\| L_k \rd_x \tphi\|_{L^2(C^{k}_{u_k}([0,T_B)))} + \| E_k^2 \tphi\|_{L^2(C^{k}_{u_k}([0,T_B)))} )  \leq &\: \epsilon^{\frac{3}{4}}, \label{BA:flux.for.tphi.improved.2} \\
\max_{k,k'} \sup_{u_{k'} \in\RR} \sum_{Z_{k'} \in \{L_{k'},\,E_{k'}\}} \| Z_{k'} \rd_x \tphi\|_{L^2(C^{k'}_{u_{k'}}([0,T_B)))} \leq  &\:\epsilon^{\frac{3}{4}} \cdot  \delta^{-\frac{1}{2}}. \label{BA:flux.for.tphi}
\end{align}
\end{subequations}

\noindent\underline{\textbf{Besov and $L^\i$ estimates for the wave variables.}}
\begin{subequations}
\begin{align}  
\sup_{0 \leq t < T_B} \max_{(k,k'): k\neq k'} \| \partial \rphi \|_{ \Bes} \leq &\:\epsilon^{\frac{3}{4}}, \label{rphiBbootstrap} \\  
\sup_{0 \leq t < T_B}\max_k  \max_{k': k' \neq k} \| \partial \tphi \|_{ \Bes} \leq &\: \epsilon^{\frac{3}{4}}, \label{tphiBbootstrap} \\
\sup_{0 \leq t < T_B}(\| \partial \rphi \|_{ L^\infty(\Sigma_t)} + \max_k \| \partial \tphi \|_{ L^\infty(\Sigma_t)}) \leq &\: \ep^{\f 34}. \label{BA:Li}
\end{align}
\end{subequations} 

\subsection{Main a priori estimates}\label{sec:aprioriestimates}

\subsubsection{Estimates for metric components and other geometric quantities}\label{sec:aprioriestimates1}

\begin{theorem}\label{thm:bootstrap.metric}
Assume that for some $T_B \in [0,1]$, 
\begin{enumerate}
\item there is a smooth solution $(\gamma, \bt^i, N, \phi)$ to \eqref{Einsteinwave}, \eqref{Kellipticequation}--\eqref{gammaellipticequation} and \eqref{Khypequation} on $[0,T_B)\times \RR^2$,
\item arising from an initial data set as in Theorem~\ref{smooththeorem},
\item which satisfies all the bootstrap assumptions \eqref{bootstrapK}--\eqref{BA:Li}.
\end{enumerate}

Then, after choosing $\ep_0$ and $\de_0$ to be sufficiently small, there exists $C = C(s',s'',R,\upkappa_0) >0$ such that the following estimates hold for all $t\in [0, T_B)$:
\begin{enumerate}
\item The estimates \eqref{eq:BA.g.asymp}--\eqref{eq:BA.g.L4} and \eqref{bootstrapK}--\eqref{bootstrapvarTheta} all hold with $\ep^{\f 54}$ replaced by $C \ep^{\f 32}$.
\item Moreover, the following additional estimates hold uniformly for all $\de \in (0, \de_0]$:
\begin{enumerate}
\item The wave front and vector field estimates \eqref{eq:wavefront} and \eqref{eq:frames.in.thm} hold.
\item The metric estimates in \eqref{metric.est} hold.
\end{enumerate}
\end{enumerate}
\end{theorem}

\subsubsection{The main Lipschitz and improved H\"older bounds}\label{sec:aprioriestimates2}

\begin{defn}\label{def:Lipschitz.control.norm}
Define $\mathcal E(t)$ to be the following norm, 
\begin{equation*}
\begin{split}
\mathcal E(t) := &\: \| \partial \Db^{s'} \tphi\|_{L^2(\Sigma_t)} + \| E_k \rd \tphi\|_{L^2(\Sigma_t)} + \| \rd E_k \Db^{s''} \tphi\|_{L^2(\Sigma_t)} + \de^{\f 12} (\| \partial^2 \tphi\|_{L^2(\Sigma_t)} + \|\rd E_k \rd \tphi\|_{L^2(\Sigma_t)}) \\
&\: + \de^{-\f 12} \|\partial \tphi\|_{L^2(\Sigma_t \cap S_{2\de}^k)} + \de^{-\f 12} \| E_k \partial \tphi\|_{L^2(\Sigma_t \cap S_{2\de}^k)}  +\| \partial^2 \tphi\|_{L^2(\Sigma_t \setminus \Sd^k)} + \|\rd^2 \Db^{s'} \rphi \|_{L^2(\Sigma_t)}.
\end{split}
\end{equation*}
\end{defn}

The following theorem is achieved through an anisotropic Sobolev-type embedding result:
\begin{theorem}\label{thm:bootstrap.Li}
Under the assumptions of Theorem~\ref{thm:bootstrap.metric}, and after choosing $\ep_0$ and $\de_0$ smaller, the following holds for some $C = C(s',s'',R,\upkappa_0) >0$ and all $t\in [0,T_B)$:
$$\mbox{LHSs of \eqref{rphiBbootstrap}--\eqref{BA:Li}} + \|\rd\rphi\|_{C^{0,\f{s''}{2}}(\Sigma_t)} + \| \partial \tphi \|_{ C^{0,\frac{s''}{2}}(\Sigma_t \cap C^k_{\geq \delta}) } \ls \mathcal E.$$
\end{theorem}

The proof of Theorem~\ref{thm:bootstrap.Li} will be carried out in \cite{LVdM2}.

\subsubsection{Wave estimates}\label{sec:aprioriestimates3}

\begin{theorem}\label{thm:energyest}
Under the assumptions of Theorem~\ref{thm:bootstrap.metric}, and after choosing $\ep_0$ and $\de_0$ smaller, the following\footnote{Note that some bounds in 1~and 2~are repeated.} holds for some $C = C(s',s'',R,\upkappa_0) >0$ and all $t\in [0,T_B)$:
\begin{enumerate}
\item The wave energy estimates \eqref{BA:rphi}--\eqref{BA:flux.for.tphi} all hold with $\ep^{\f 34}$ replaced by $C \ep$.
\item The wave energy estimates stated in \eqref{eq:main.thm.rphi}--\eqref{eq:main.thm.tphi.2}, \eqref{eq:smooththeorem.1}--\eqref{eq:smooththeorem.2} hold.
\item In addition, the norm $\mathcal E$ satisfies the estimate $\mathcal E \ls \ep$.
\end{enumerate}

\end{theorem}
The proof of Theorem~\ref{thm:energyest} will be carried out in \cite{LVdM2}.

\subsection{Local existence and the proof of Theorem~\ref{smooththeorem}}\label{sec:pf.of.smooththeorem}

\begin{thm}[Local existence]\label{thm:local}
\begin{enumerate}
\item Given an initial data set in Theorem~\ref{smooththeorem}, there exists a time $T_{local} \in (0,1)$ (potentially depending on the initial data profile and not only the parameters involved) and a smooth solution $(\gamma, \bt^i, N, \phi)$ to \eqref{Einsteinwave}, \eqref{Kellipticequation}--\eqref{gammaellipticequation} and \eqref{Khypequation} on $[0,T_{local})\times \RR^2$
\item Moreover, there exists a universal $\ep_{local}>0$ such that if $[0,T_*)$ is the maximal time interval for which the solution exists for some $T_*\in (0,1)$, then at least one of the following holds:
\begin{enumerate}
\item $\liminf_{t\to T_*} (\|\phi\|_{H^3(\Sigma_t)} + \|\n \phi\|_{H^2(\Sigma_t)} ) = \infty$,
\item $\liminf_{t\to T_*} (\|\n \phi\|_{L^\infty(\Sigma_t)} + \|\rd_i \phi\|_{L^\infty(\Sigma_t)} ) \geq \ep_{local}$.
\end{enumerate}
\end{enumerate}
\end{thm}
\begin{proof}
This is an immediate consequence of the local existence result in \cite[Theorem~5.4]{HL.elliptic}. \qedhere
\end{proof}

\begin{lemma}\label{lem:local}
Taking $T_{local}\in (0,1)$ smaller if necessary, the solution on $(0,T_{local})\times \RR^2$ given by Theorem~\ref{thm:local} obeys all the estimates in \eqref{bootstrapK}--\eqref{BA:Li} with $T_B$ replaced by $T_{local}$.
\end{lemma}
\begin{proof}
By continuity, it suffices to check that the estimates \eqref{bootstrapK}--\eqref{BA:Li} are satisfied initially for $t = 0$. It is easy to verify that the following \emph{stronger} estimates all hold when $t = 0$:
\begin{enumerate}
\item The wave energy estimates \eqref{BA:rphi}--\eqref{BA:away.from.singular} hold at $t=0$ with $\ep^{\f 34}$ replaced by $C\ep$. For \eqref{BA:rphi}--\eqref{EtphiH2bootstrap} and \eqref{BA:away.from.singular}, this is an immediate consequence of the assumptions in Definition~\ref{smoothdef}. The initial estimate for \eqref{bootstrapbadunlocenergyhyp} follows from the other estimates; see \cite[Proposition~10.1]{LVdM2} for the proof of this fact.
\item The wave flux estimates \eqref{BA:flux.for.rphi}--\eqref{BA:flux.for.tphi} trivially hold when (a) $T_B$ is replaced by $0$ and (b) $\ep^{\f 34}$ is replaced by $0$. That this is the case is because $C^k_{u_k}(\{0\})$ has measure zero (with respect to $du_k\, d\th_k$).
\item The pointwise wave estimates \eqref{rphiBbootstrap}--\eqref{BA:Li} hold when $t = 0$ with $\ep^{\f 34}$ replaced by $C\ep$. For \eqref{BA:Li}, this is directly assumed in Definition~\ref{smoothdef}. For \eqref{rphiBbootstrap} and \eqref{tphiBbootstrap}, this is a consequence of the embedding in Theorem~\ref{thm:bootstrap.Li} together with the $L^2$-based assumptions in Definition~\ref{smoothdef}.
\item The estimates \eqref{eq:BA.g.asymp}--\eqref{eq:BA.g.L4} all hold at $t=0$ with $\ep^{\f 54}$ replaced by $C\ep^2$. That this is the case can be proven using the elliptic equations for the metric components, in the same manner as Propositions~\ref{prop:elliptic.easy} and \ref{prop:elliptic.dtg}. Moreover, since the Besov estimates \eqref{rphiBbootstrap}--\eqref{tphiBbootstrap} hold for $\phi$ initially (see point 3.~above), we have the estimates $\sum_{\mfg \in \{\gamma, \bt^l,N\}} \| \rd^2_{ij} \mfg \|_{L^\i_{2-\alp}(\Sigma_0)} \ls \ep^2$ (proved in the same way as Proposition~\ref{prop:elliptic.Besov.weighted}). (Note that the estimates for $C\ep^2$ instead of $C \ep^{\f 32}$ as in Propositions~\ref{prop:elliptic.easy}, \ref{prop:elliptic.Besov.weighted} and \ref{prop:elliptic.dtg} because we have the better initial wave estimates from points 1~and 3~above.) 
\item Finally, the bounds \eqref{bootstrapK}--\eqref{bootstrapvarTheta} for $K$, $\chi_k$, $\eta_k$, $\mu_k$ and $\varTheta_k$ hold at $t=0$ with $\ep^{\f 54}$ replaced by $C\ep^2$. That this is the case is due to (a) the formulas \eqref{maximality3} and Lemma~\ref{riccisigma0expression} for their initial values at $t=0$, and (b) the metric estimates \eqref{eq:BA.g.asymp}--\eqref{eq:BA.g.L4} and $\sum_{\mfg \in \{\gamma, \bt^l,N\}} \| \rd^2_{ij} \mfg \|_{L^\i_{2-\alp}(\Sigma_0)} \ls \ep^2$ discussed in point 4~above. 
\end{enumerate}
\end{proof}

We are now ready to combine the local existence results (Theorem~\ref{thm:local} and Lemma~\ref{lem:local}) with the bootstrap results (Theorems~\ref{thm:bootstrap.metric}, \ref{thm:bootstrap.Li} and \ref{thm:energyest}) to conclude the proof of Theorem~\ref{smooththeorem}.
\begin{proof}[Proof of Theorem~\ref{smooththeorem}]
\pfstep{Step~1: Solution exists in the time interval $[0,1]$} Suppose for the sake of contradiction that there exists $T_* \in (0,1)$ such that $[0,T_*)$ is the maximal time interval for which the solution exists. 

We claim that the bootstrap assumptions \eqref{bootstrapK}--\eqref{BA:Li} hold for all $t\in [0,T_*)$. Suppose not, then by continuity and Lemma~\ref{lem:local}, there exists $T_{**} \in (0, T_*)$ such that \eqref{bootstrapK}--\eqref{BA:Li} all hold for $t\in [0, T_{**} ]$, and that when $t = T_{**}$,  in at least one of the estimates \eqref{bootstrapK}--\eqref{BA:Li}, ``$\leq$'' can be replaced by ``$=$''.

We now apply Theorems~\ref{thm:bootstrap.metric}, \ref{thm:bootstrap.Li} and \ref{thm:energyest} with $T_B = T_{**}$. In particular, \eqref{bootstrapK}--\eqref{BA:Li} all hold with $\ep^{\f 54}$ replaced by $C\ep^{\f 32}$ and $\ep^{\f 34}$ replaced by $C\ep$ for all $t\in [0, T_{**})$. Choosing $\ep_0$ smaller, we have $C\ep^{\f 32} \leq \f 12 \ep^{\f 54}$ and $C\ep \leq \f 12 \ep^{\f 34}$. However, this contradicts that fact that at least one of the estimates \eqref{bootstrapK}--\eqref{BA:Li} is an equality at $t = T_{**}$.

We have thus established that the bootstrap assumptions \eqref{bootstrapK}--\eqref{BA:Li} hold for all $t\in [0,T_*)$. Apply now again Theorems~\ref{thm:bootstrap.metric}, \ref{thm:bootstrap.Li} and \ref{thm:energyest}, but with $T_B = T_*$, we obtain the following:
\begin{enumerate}[(a)]
\item Since the initial data are smooth, \eqref{eq:smooththeorem.top} implies $\liminf_{t\to T_*} (\|\phi\|_{H^3(\Sigma_t)} + \|\n \phi\|_{H^2(\Sigma_t)} ) < \infty$.
\item By \eqref{eq:roughtheorem.Lipschitz}, $\|\n \phi\|_{L^\infty(\Sigma_t)} + \|\rd_i \phi\|_{L^\infty(\Sigma_t)} \leq C\ep$.
\end{enumerate}
Therefore, after taking $\ep_0$ smaller if necessary (so that $C\ep_0 \leq \ep_{local}$), part 2 of Theorem~\ref{thm:local} implies that $[0,T_*)$ is not a maximal time interval, which leads to a contradiction.

\pfstep{Step~2: Estimates in Theorem~\ref{smooththeorem}} Repeating the argument in Step~1, we also show that the bootstrap assumptions \eqref{bootstrapK}--\eqref{BA:Li} hold for all $t\in [0,1]$. As a result, all the estimates stated in the conclusions of Theorems~\ref{thm:bootstrap.metric}, \ref{thm:bootstrap.Li} and \ref{thm:energyest} hold. This thus implies all the estimates stated in Theorem~\ref{smooththeorem}. \qedhere
\end{proof}

\section{Preliminary estimates resulting from the bootstrap assumptions} \label{preliminary.estimates.section}
From this section onwards until Section~\ref{sec:Ricci.coeff}, we will prove Theorem~\ref{thm:bootstrap.metric} (see Section~\ref{sec:thm.bootstrap.metric.conclusion} for the conclusion of the proof). In particular, we work under the assumptions of Theorem~\ref{thm:bootstrap.metric}.

Moreover, we will allow \textbf{all constants $C$ or implicit constants in $\ls$ to depend on $s'$, $s''$, $R$ and $\upkappa_0$} (as in the statement of Theorem~\ref{thm:bootstrap.metric}). Whenever necessary, we will also \textbf{take $\ep_0$ smaller (depending on $s'$, $s''$, $R$ and $\upkappa_0$) without further comments.}

In this section, we prove preliminary estimates which follow directly from the bootstrap assumptions. Before we begin, we prove an easy finite speed of propagation lemma in \textbf{Section~\ref{sec:support}} that will be useful for the remainder of the paper. Turning to the preliminary estimates, the first group of estimates involve the coefficients of $(X_k,E_k,L_k)$ in the $(\rd_t, \rd_1, \rd_2)$ basis; see \textbf{Section~\ref{dgeomsection}}. The second group of estimates can be viewed as quantitative bounds on the transversality between the three waves; see \textbf{Section~\ref{controlsection}}.

\subsection{Support properties}\label{sec:support}

\begin{lemma}\label{lem:support}
The following holds on $\Sigma_t$ for all $t\in [0,T_B)$ and for $k=1,2,3$:
\begin{enumerate}
\item $\mathrm{supp}(\rphi),\,\mathrm{supp}(\tphi) \subseteq B(0,R)$,
\item $\mathrm{supp}(\tphi) \subseteq \{(t,x): u_k(t,x) \geq - \de \}$.
\end{enumerate}
\end{lemma}
\begin{proof}
Both assertions are standard finite speed of propagation statements for the wave equation. The first statement follows from the facts that the initial supports are in $B(0,\f R2)$ and that for $\ep_0$ sufficiently small, the metric is $O(\ep^{\f 54})$-close to the Minkowski metric in the $C^1_{t,x}$ norm on compact sets (by \eqref{eq:BA.g.asymp} and \eqref{eq:BA.g.Li}). The second statement follows from the facts that $u_k\geq -\de$ on $\mathrm{supp}(\tphi)$ initially, and that $\{(t,x): u_k = -\de\}$ is a null hypersurface with $e_0 u_k >0$ (by Definition~\ref{def:eikonal}). \qedhere
\end{proof}

\subsection{Estimates for the coefficients of $(X_k,E_k,L_k)$ in the $(\rd_t, \rd_1, \rd_2)$ basis} \label{dgeomsection}

\begin{lem}\label{lem:L.X.E}
	The following estimates hold for the coefficients of the vector fields $L_k, E_k, X_k$:
\begin{equation} \label{Yibounded} 
 |L^{\beta}_k|+	|X^i_k|+ |E^i_k| \lesssim  \la x\ra^{\epsilon}.
\end{equation}
\end{lem}
\begin{proof}
Using $g(X_k,X_k)=g(E_k,E_k)=1$ (by \eqref{XELframecondition}), $g_{ij} = e^{2\gamma}\de_{ij}$ (by \eqref{gauge}), \eqref{eq:BA.g.asymp} and \eqref{eq:BA.g.Li}, we know that 
\begin{equation}\label{eq:Xki.Eki.pf}
 |X_k^i|+ |E_k^i| \lesssim e^{-2\gamma} \lesssim e^{C \cdot  \epsilon^{\f 54} \cdot  \log(1+ \la x\ra)} \lesssim  \la x\ra^{ C \ep^{\f 54}}.
 \end{equation}
Using \eqref{nXEL} and \eqref{def:e0}, and then \eqref{eq:BA.g.asymp}, \eqref{eq:BA.g.Li}, \eqref{eq:Xki.Eki.pf} and the inequality 
\begin{equation}\label{eq:trivial.calculus}
\log y \leq \f 1 \rho y^\rho\quad  \mbox{for all $\rho >0$ and $y\geq 1$},
\end{equation}
we obtain
\begin{equation}\label{eq:Xkalp.pf}
 |L^{\alpha}_k| \lesssim  |N^{-1} | + \max_{i=1,2} ( |X_k^i| + |\f {\beta^i}{N}| )\lesssim \ep^{\f 54} \log (1+ \la x\ra) + \la x\ra^{C\ep^{\f 54}} \ls \la x\ra^{C\ep^{\f 54}}.
\end{equation}

After choosing $\ep_0$ to be sufficiently small, \eqref{eq:Xki.Eki.pf} and \eqref{eq:Xkalp.pf} obviously imply \eqref{Yibounded}.
\end{proof}

\begin{lem}\label{lem:rd.in.terms.of.XEL}
For any sufficiently regular function $f$,
\begin{align} 
	\label{spatialintermsofEX}
 		|\partial_i f| \ls &\: \la x \ra^{\ep} \left( |E_k f |+ |X_k f |\right), \\
	\label{timeintermsofLandspace}
		|\partial_t f| \ls &\: \la x \ra^{\ep} (|L_k f| + |\rd_x f| ), \\
 	\label{timeintermsofLEX}
 		|\partial_t f| \ls &\: \la x\ra^{\ep} \left( |L_k f |+ |X_k f |\right) + \la x \ra^{-1+\ep} |E_k f |.
 \end{align}
\end{lem}
\begin{proof}
Starting with the formula \eqref{partial12EX}, estimating $e^{2\gamma}$ with \eqref{eq:BA.g.asymp}, \eqref{eq:BA.g.Li}, and bounded $X_k^i$, $E_k^i$ by \eqref{eq:Xki.Eki.pf}, \eqref{eq:Xkalp.pf}, we immediately obtain \eqref{spatialintermsofEX}.

		By \eqref{nXEL} and \eqref{defnormal},
		\begin{equation}\label{eq:partialtLEX}
		\partial_t = N L_k +N X_k+ \beta^i \partial_i.
		\end{equation} 
		Using \eqref{eq:partialtLEX} and applying \eqref{eq:BA.g.asymp}, \eqref{eq:BA.g.Li}, \eqref{eq:Xki.Eki.pf}, \eqref{eq:Xkalp.pf} and \eqref{eq:trivial.calculus}, we obtain
		\begin{equation}\label{eq:pf.timeintermsofLandspace}
 		|\partial_t f| \ls  N \cdot \left( |L_k f |+ |X_k f |\right) + |\bt| \cdot |\rd_x f| \ls \la x \ra^{\ep} (|L_k f| + |\rd_x f| ),
 		\end{equation}
		which implies \eqref{timeintermsofLandspace}. 
		
		Finally, arguing as in \eqref{eq:pf.timeintermsofLandspace} but controlling $ |\bt| \cdot |\rd_x f|$ by \eqref{eq:BA.g.asymp}, \eqref{eq:BA.g.Li} and \eqref{spatialintermsofEX}, we obtain \eqref{timeintermsofLEX}. \qedhere
\end{proof}

\begin{lem} \label{dgeomvflemma}
	
	The following estimates hold for the derivatives of the coefficients of the vector fields $L_k, E_k, X_k$:
\begin{align}
	 \label{dYibounded}
	|L_k L_k^t | + \la x \ra (|L_k L_k^i| + \sum_{\substack {Y_k,\,Z_k \in \{ L_k, E_k, X_k\} \\ (Y_k,Z_k) \neq (L_k, L_k)} } |Y_k Z^{\beta}_k|) \lesssim &\: \epsilon^{\frac{5}{4}}  \cdot\la x\ra^{3\alp} \\
	 \label{rdYibounded}
	|\rd_t L_k^t| + \la x \ra (|\rd_t L_k^i| + |\rd_i L_k^\bt| + |\partial X^i_k|+ |\partial E^i_k|)  \lesssim &\: \epsilon^{\frac{5}{4}}  \cdot\la x\ra^{4\alp}.
	\end{align}

\end{lem}
\begin{proof}

\pfstep{Step~1: Proof of \eqref{dYibounded}} The derivatives on the LHS of \eqref{dYibounded} were computed in \eqref{LEi}--\eqref{YLi} in Lemma~\ref{dXELellipticprop}. Using those formulas, and plugging in the bootstrap assumptions \eqref{bootstrapK}, \eqref{bootstrapricci}, \eqref{eq:BA.g.asymp}, \eqref{eq:BA.g.Li}, as well as the estimates \eqref{eq:Xki.Eki.pf} and \eqref{eq:Xkalp.pf}, we obtain the desired result.

\pfstep{Step~2: Proof of \eqref{rdYibounded}} Finally, \eqref{rdYibounded} follows from \eqref{dYibounded} and Lemma~\ref{lem:rd.in.terms.of.XEL}. \qedhere

 \end{proof}

\subsection{Control of the angle between the impulsive waves} \label{controlsection}
\begin{prop}\label{prop:angle}
The following estimate holds:
\begin{equation}\label{partialukcontrol}
\sup_{0 \leq t < T_B}\|\partial_i u_k - c_{i k}\|_{L^{\infty}_{\f 12}(\Sigma_t)} \ls  \epsilon^{\f 54}.
\end{equation}
Moreover, for any $k \neq k'$ we have the following pointwise estimates (recall $\upkappa_0$ in \eqref{cangle2}):
\begin{equation}\label{anglecontrol} 
\begin{split}
\frac{ \upkappa_0}{2} \leq|g(E_k,X_{k'}) | \leq 2,
\end{split}
\end{equation}
\begin{equation} \label{anglecontrol2}
\frac{\upkappa_0^2}{4} \leq |\partial_{t_k} u_{k'}|_{|B(0,3R)} \leq 2
\end{equation}
\end{prop}

\begin{proof} 
\pfstep{Step~0: Preliminaries} To prove \eqref{partialukcontrol}, we will compare the value of $\partial_i u_k$ with its initial value along an integral curve of $L_k$.

First, we need an easy bound that $\la x \ra$ is comparable at any two points along the integral curve of $L_k$. For this, we simply use \eqref{Yibounded}  to obtain
$$ |L_k \la x \ra^2| = |2\de^{ij} L_k^i x_j| \ls \la x\ra^2,$$
and apply Gr\"onwall's inequality. In what follows, we will silently assume the comparability of $\la x \ra$.

\pfstep{Step~1: Estimates for $\rd_i u_k$} Recalling the formula for $\rd_i u_k$ in \eqref{ucartderivative}, we compare each of the factors $e^{2\gamma}$, $X^i_k$ and $\mu_k^{-1}$ at $(u_k,\theta_k,t_k)$ (recall the coordinates in Section~\ref{relationXELgeocoordinatesection}) with their (initial) values at $(u_k, \th_k,0)$. For this purpose, we will consider the $L_k$ derivative of each of these quantities. Recall that by \eqref{thetaE}, $\partial_{t_k} = N \cdot L_k$, where $\rd_{t_k}$ is the coordinate derivative in the $(u_k,\theta_k,t_k)$ coordinate system.
	
	For $e^{2\gamma}$, we use \eqref{eq:BA.g.asymp}, \eqref{eq:BA.g.Li} and \eqref{Yibounded} to obtain
	\begin{equation}\label{eq:e2gamma.diff.est}
	\begin{split}
	|e^{2\gamma(u_k,\theta_k,t_k)}-e^{2\gamma(u_k,\theta_k,0)}| = &\:  2 |\int_0^{t_k} (e^{2\gamma} \rd_{t_k}\gamma) (u_k,\theta_k,s)\, ds| \\
	= &\: 2 |\int_0^{t_k} (e^{2\gamma} N L_k\gamma) (u_k,\theta_k,s)\, ds| \ls \epsilon^{\f 54} \la x \ra^{-1+3\epsilon} .
	\end{split}
	\end{equation}
	
For $X^i_k$, we use Lemmas~\ref{lem:L.X.E} and \ref{dgeomvflemma} with  \eqref{eq:BA.g.asymp} and \eqref{eq:BA.g.Li} to get that
\begin{equation}\label{eq:X.diff.est}
 | X^i_k(u_k,\theta_k,t_k)-X^i_k(u_k,\theta_k,0)| = | \int_0^{t_k} (N L_k X^i_k) (u_k,\theta_k,s)\, ds |\ls \epsilon^{\f 54} \la x \ra^{-1+5\alp}.
 \end{equation}

	Finally, we control the difference of $\mu_k^{-1}$. By the equation \eqref{Lmu} and the estimates in \eqref{bootstrapK}, \eqref{bootstrapmu}, \eqref{eq:BA.g.asymp}, \eqref{eq:BA.g.Li} and Lemma~\ref{lem:L.X.E}, we obtain
	\begin{equation}\label{eq:mu.diff.est}
	\begin{split} 
	|\mu_k^{-1} (u_k,\theta_k,t_k)-\mu_k^{-1} (u_k,\theta_k,0)|  =  |\int_0^{t_k} (\mu_k^{-1} N L_k\log \mu_k) (u_k,\theta_k,s)\, ds| \ls \la x\ra^{-1+3\alp}.
	\end{split}
	\end{equation} 
	
	By \eqref{eikonalinit}, $(\partial_i u_k)_{|\Sigma_0} = c_{ik}$. Thus, we write
	\begin{equation*} \begin{split}
\partial_i u_k(u_k,\theta_k,t_k)-c_{i k} = &\: \delta_{i j} \cdot [e^{2\gamma(u_k,\theta_k,t_k)}-e^{2\gamma(u_k,\theta_k,0)}] \cdot \mu^{-1}_k (u_k,\theta_k,t_k) \cdot X^j_k(u_k,\theta_k,t_k)  \\ & + \delta_{i j} \cdot e^{2\gamma(u_k,\theta_k,0)} \cdot [\mu^{-1}_k (u_k,\theta_k,t_k) -\mu^{-1}_k(u_k,\theta_k,0)]\cdot X^j_k(u_k,\theta_k,t_k)  \\ & + \delta_{i j} \cdot e^{2\gamma(u_k,\theta_k,0)} \cdot \mu^{-1}_k (u_k,\theta_k,0) \cdot[X^j_k(u_k,\theta_k,t_k)-X^j_k(u_k,\theta_k,0)] ,
	\end{split}
	\end{equation*}
	and combining \eqref{eq:e2gamma.diff.est}--\eqref{eq:mu.diff.est} with bootstrap assumptions \eqref{bootstrapmu}, \eqref{eq:BA.g.asymp}, \eqref{eq:BA.g.Li} and Lemma~\ref{lem:L.X.E} yields \eqref{partialukcontrol}.
	
	\pfstep{Step~2: Proof of \eqref{anglecontrol}} In view of \eqref{gauge} and \eqref{EXinellipticcoord}, we have 
	$$ g(E_k,X_{k'}) = e^{2\gamma} \cdot ( E_k^1 \cdot   X_{k'}^1+ E_k^2 \cdot   X_{k'}^2) = e^{2\gamma}\cdot (-X_k^2 \cdot X_{k'}^1 + X_k^1 \cdot X_{k'}^2).$$ 
	Hence, by \eqref{gauge}, the initial condition \eqref{X^i(0)formula}, and the estimates \eqref{eq:e2gamma.diff.est} and \eqref{eq:X.diff.est}, we obtain 
	\begin{equation*}
	\begin{split}
	&\: |g(E_k,X_{k'}) (u_k,\theta_k,t_k)-g(E_k,X_{k'}) (u_k,\theta_k,0)|\\
	= &\: |g(E_k,X_{k'}) (u_k,\theta_k,t_k)  + c_{k 2} \cdot c_{k' 1} - c_{k 1} \cdot  c_{k' 2} |\ls  \la x\ra^{-1+2\epsilon} \cdot \epsilon^{\f 54}.
	\end{split}
	\end{equation*}
	The lower bounds in \eqref{anglecontrol} then follow immediately from \eqref{cangle2}, after taking $0< \epsilon_0 \ll \upkappa_0$, while the upper bounds follow from \eqref{cnormalization}.
	
	\pfstep{Step~3: Proof of \eqref{anglecontrol2}} To fix the notation we take $k'=2$ and $k\neq 2$. By \eqref{defnormal}, \eqref{nXEL} and \eqref{thetaE}, we have
	$$  \partial_{t_k}= N^{-1}  \cdot (L_{2}+ X_2-X_k).$$ 
	Hence, since $L_2 u_2=0$ by definition, we have 
	\begin{equation}\label{eq:dtku2}
	 \partial_{t_k} u_2 = N^{-1} \cdot (X_2^i -X_k^i) \partial_i u_2.
	 \end{equation}

	Using \eqref{eq:X.diff.est} and \eqref{X^i(0)formula}, we have 
	\begin{equation}\label{eq:X.diff.est.local}
	|X_k^i(u_k,\th_k,t_k) - e^{-\gamma(u_k,\th_k,0)} \de^{iq} \cdot c_{kq}|_{|B(0,R)} \ls \ep^{\f 54},\quad |X_2^i(u_2,\th_2,t_2) - e^{-\gamma(u_2,\th_2,0)} \de^{iq} \cdot c_{2q}|_{|B(0,R)} \ls \ep^{\f 54}.
	\end{equation}
	Therefore, combining \eqref{partialukcontrol}, \eqref{eq:X.diff.est.local} with \eqref{eq:BA.g.asymp}, \eqref{eq:BA.g.Li} to estimate the RHS of \eqref{eq:dtku2}, we obtain 
	\begin{equation}\label{eq:dtku2.in.constants}
	|(X_2^i -X_k^i) \partial_i u_2 - \delta^{iq} ( c_{2 q}-c_{k q}) \cdot c_{2 i}|_{|B(0,3R)} \ls \ep^{\f 54}.
	\end{equation}
	In view of \eqref{cnormalization}, we have $\delta^{iq} \cdot c_{2 q} \cdot c_{2 i} =1$. By \eqref{cnormalization} and the Cauchy--Schwarz inequality, we also have $|\delta^{iq} \cdot c_{k q} \cdot c_{2 i}| \leq 1$. Hence, by the triangle inequality and \eqref{cangle2}, we have $|1-\delta^{iq} \cdot c_{k q} \cdot c_{2 i}| \geq 1 - |\delta^{iq} \cdot c_{k q} \cdot c_{2 i}|  \geq \f 12 (1+ |\delta^{iq} \cdot c_{k q} \cdot c_{2 i}|)(1- |\delta^{iq} \cdot c_{k q} \cdot c_{2 i}|)\geq \f 12 \upkappa_0^{2}$. Hence, after choosing $\ep$ smaller, we obtain \eqref{anglecontrol2} by using \eqref{eq:dtku2} and \eqref{eq:dtku2.in.constants}. 	\qedhere

\end{proof}

The following is an immediate consequence of Proposition~\ref{prop:angle}.
\begin{cor}\label{cor:diffeo}
For any $k\neq k'$, the map $(x^1,x^2)\mapsto (u_k, u_{k'})$ is a $C^1$-diffeomorphism with entry-wise pointwise estimates independent of $\de$:
$$\left|\begin{bmatrix}  \f{\rd u_k}{\rd x^1} & \f{\rd u_{k'}}{\rd x^1}  \\  \f{\rd u_k}{\rd x^2} & \f{\rd u_{k'}}{\rd x^2} \end{bmatrix} \right|\ls 1,\quad \left| \begin{bmatrix}  \f{\rd u_k}{\rd x^1} & \f{\rd u_{k'}}{\rd x^1}  \\  \f{\rd u_k}{\rd x^2} & \f{\rd u_{k'}}{\rd x^2} \end{bmatrix}^{-1} \right| \ls 1.$$

\end{cor}

\begin{proposition}\label{prop:dx2u}
For any $k=1,2,3$, 
$$|\rd^2_{ij} u_k| \ls \ep^{\f 54}.$$
\end{proposition}
\begin{proof}
We start with \eqref{ucartderivative} and differentiate by $\rd_j$. Then estimate the resulting terms by \eqref{bootstrapmu}, \eqref{eq:BA.g.asymp}, \eqref{eq:BA.g.Li} and Lemmas~\ref{lem:L.X.E} and \ref{dgeomvflemma}. (Notice that the $\rd_j$ derivative of $e^{2\gamma}$, $\mu_k^{-1}$ or $X_k$ would give sufficient $\la x\ra$ decay to compensate the growth in weights in the other factors.) \qedhere
\end{proof}

\section{Estimates for the metric components in elliptic gauge} \label{metricsection}

We continue to work under the assumptions of Theorem~\ref{thm:bootstrap.metric}.

Our goal in this section is to prove estimates for the metric components $(\gamma,\bt^1,\bt^2,N)$ in elliptic gauge. In particular, we improve the bootstrap assumptions \eqref{eq:BA.g.asymp}--\eqref{eq:BA.g.L4}.

We begin with some preliminaries in \textbf{Sections~\ref{sec:analytic.prelim.metric}} and \textbf{Section~\ref{sec:first.remaks.elliptic}}. The main estimates are given in the following sections:
\begin{itemize}
\item In \textbf{Section~\ref{sec:g.purely.spatial}}, we prove elliptic estimates for purely spatial derivatives of $\mfg$. 
\item In \textbf{Section~\ref{sec:dtg.elliptic}}, we prove the elliptic estimates for $\rd_t \mfg$.
\item In \textbf{Section~\ref{sec:dtg.elliptic.top}}, we prove the top (fractional) order elliptic estimates for $\rd_t \mfg$.
\item In \textbf{Section~\ref{sec:third.der.metric}}, we carry out the elliptic estimates for third derivatives of the metric.
\item In \textbf{Section~\ref{sec:easy.consequence.from.g.est}}, we deduce some estimates for $K$ which follow directly from earlier subsections.
\end{itemize}

\subsection{Analytic preliminaries}\label{sec:analytic.prelim.metric}

\subsubsection{Embeddings of weighted spaces}

Most of the following results can be found in \cite{yCB09}; see also \cite[Lemmas~A2--A3]{HL.elliptic}. The only part not in \cite{yCB09} is the compactness statement in 1(b), which follows readily from the first part of 1(b) together with the Kondrachov compactness theorem (for compact domains).
\begin{proposition}\label{prop:Sobolev.weighted}
\begin{enumerate}
\item Let $m \in \mathbb N \cup \{0\}$, $p\in (1,+\infty)$. 
	\begin{enumerate}
		\item If $m>\f 2p$ and $\sigma \leq \sigma' + \f 2p$, then $\|f\|_{C^0_\sigma(\mathbb R^2)} \ls_{m,p} \| f\|_{W^{m,p}_{\sigma'}(\mathbb R^2)}$.				\item If $m< \f 2p$, then for any $\sigma$, $\|f\|_{L^{\f{2p}{2-mp}}_{\sigma + m}(\mathbb R^2)} \ls_{m,p,\sigma} \|f\|_{W^{m,p}_{\sigma}(\mathbb R^2)}$. For $m< \f 2p$ and $\sigma < \sigma'$, the embedding $W^{m,p}_{\sigma'}(\mathbb R^2) \hookrightarrow L^{\f{2p}{2-mp}}_{\sigma + m}(\mathbb R^2)$ is moreover compact.
	\end{enumerate}
\item Let $1\leq p_1\leq p_2 \leq \infty$ and $\sigma_2 - \sigma_1 > 2(\f 1{p_1} - \f 1{p_2})$. Then $\|f\|_{L^{p_1}_{\sigma_1}(\RR^2)} \ls_{p_1,p_2,\sigma_1,\sigma_2} \|f\|_{L^{p_2}_{\sigma_2}(\RR^2)}$.
\end{enumerate}
\end{proposition}

\subsubsection{Preliminaries on fractional derivatives}

We cite a standard lemma regarding fractional derivatives. 
\begin{lem}[(2.1) in \cite{cMwS2013}]\label{lem:frac.product}
	For any $\th > 0$ and $2\leq p_1,\, p_1',\, p_2,\, p_2'\leq +\infty$ such that $\f 1{p_1} + \f 1{p_2} = \f 12 = \f 1{p_1'} + \f 1{p_2'}$,
	$$\|\Db^\th (fh)\|_{L^2(\RR^2)}\ls_{p_1,p_2,p_1',p_2',\th} \|\Db^\th f\|_{L^{p_1}(\Sigma_t)} \|h\|_{L^{p_2}(\RR^2)} + \|f\|_{L^{p_1'}(\RR^2)} \| \Db^\th h\|_{L^{p_2'}(\RR^2)}.$$
\end{lem}

\subsubsection{Standard facts about elliptic estimates}

\begin{defn} \label{def:Delta-1}
	Let $p \in (1,+\infty)$, $-\f 2p < \sigma < 1-\f 2p$. Define $\Delta^{-1}: L^{p}_{\sigma+2}(\RR^2) \rightarrow \mathcal S'(\RR^2)$ by 
	$$\Delta^{-1}f(x)= \f 1{2\pi}\int_{\RR^2} f(y) \log|x-y| dy.$$
\end{defn}

We need a result regarding mapping properties of $\Delta^{-1}$ in weighted Sobolev space, which essentially follows from \cite{Huneau.constraints, rcM1979}. Since we cannot find the exact statement we need, we include a reduction to \cite{Huneau.constraints, rcM1979} for completeness.
\begin{proposition}\label{prop:basic.elliptic.2}
	Let $p \in (1,+\infty)$ and $-\f 2p < \sigma < 1-\f 2p$. Then for every $f\in L^p_{\sigma+2}(\mathbb R^2)$, 
	$$(\Delta^{-1} f)(x) = \f 1{2\pi} (\int_{\mathbb R^2} f(y) \, \ud y)\om(|x|) \log |x| + \widetilde{v}(x),\quad \widetilde{v} \in W^{2,p}_\sigma(\mathbb R^2),$$
	where $\om:\mathbb R\to [0,1]$ is a cutoff function such that $\om(s) \equiv 0$ for $s\leq 1$ and $\om(s) \equiv 1$ for $s\geq 2$.
	
	Moreover, there exists $C = C(p,\sigma)>0$ such that for every $f\in L^p_{\sigma+2}(\mathbb R^2)$, 
	$$\|\Delta^{-1} f - \f 1{2\pi} (\int_{\mathbb R^2} f(y) \, \ud y)\om(|x|) \log |x| \|_{W^{2,p}_\sigma(\mathbb R^2)} \leq C \|f\|_{L^p_{\sigma+2}(\mathbb R^2)}.$$
\end{proposition}
\begin{proof}
	By \cite[Corollary~2.7]{Huneau.constraints}\footnote{Notice that technically, \cite[Corollary~2.7]{Huneau.constraints} only gives the result when $p=2$. However, the general case for $p \in (1,+\infty)$ follows with an identical proof. }, given $f \in L^p_{\sigma+2}(\mathbb R^2)$, there exists a function $v$ such that 
	\begin{equation}\label{eq:elliptic.inversion.u}
	(\Delta v)(x) = f(x),\quad v(x) = \f 1{2\pi} (\int_{\mathbb R^2} f(y) \,\ud y)\om(|x|) \log|x| + \widetilde{v}(x),\quad \widetilde{v} \in W^{2,p}_\sigma,
	\end{equation}
	and a constant $C = C(p,\sigma)>0$ such that for every $f\in L^p_{\sigma+2}(\mathbb R^2)$, 
	$$\|\widetilde{v} \|_{W^{2,p}_\sigma(\mathbb R^2)} \leq C \|f\|_{L^p_{\sigma+2}(\mathbb R^2)}.$$
	
	Let (recall Definition~\ref{def:Delta-1})
	\begin{equation}\label{eq:elliptic.inversion.u0}
	v_0(x) := \Delta^{-1} f = \f 1{2\pi}\int_{\mathbb R^2} f(y) \log|x-y| \, \ud y
	\end{equation}
	so that $\Delta v_0 = f$ in the sense of distribution. In order to prove the proposition, it suffices to show that $v_0 = v$.
	
	To achieve this, first note that $h = v_0 -v$ is a harmonic function. In particular, it is a bounded function on $|x|\leq 3$. Moreover, using \eqref{eq:elliptic.inversion.u} and \eqref{eq:elliptic.inversion.u0}, we obtain that for $|x|\geq 2$,
	\begin{equation}\label{eq:elliptic.inversion.diff}
	h(x) = v_0(x) - v(x) = - \widetilde{v}(x) + \f 1{2\pi} \int_{\mathbb R^2} f(y) (\log|x-y| - \log|x|) \,\ud y.
	\end{equation}
	
	Notice that by \cite[Corollary~2]{rcM1979},
	$$\log |x - y| - \log |x| = \int_0^1 \f{\ud}{\ud t} \log |x-ty| \, \ud t = - \int_0^1 \f{y\cdot (x-ty)}{|x-ty|^2} \, \ud t =:\widetilde{R}(x,y)$$
	satisfies, for some $C>0$,
	\begin{equation}\label{eq:elliptic.inversion.R.bddness}
	\sup_{\|w\|_{L^p_{\sigma+2}(\mathbb R^2)} = 1} \left\| \int_{\mathbb R^2} w(y) \widetilde{R}(x,y)\, \ud y \right\|_{L^p_\sigma(\mathbb R^2)} \leq C.
	\end{equation}
	
	By \eqref{eq:elliptic.inversion.diff}, \eqref{eq:elliptic.inversion.R.bddness} and the fact that $f \in L^p_{\sigma+2}(\RR^2) $, $\widetilde{v} \in L^p_\sigma(\RR^2)$, it follows that $h \in L^{p}_\sigma(\RR^2)$. Writing in polar coordinates, this means
	$$ \int_0^{+\infty} \int_0^{2\pi}  \la r \ra^{\sigma p} |h|^p(r,\vartheta) \, \ud \vartheta \, r \, \ud r <+\infty.$$
	In particular, by the mean value theorem, there exists a sequence $\{r_i\}_{i=1}^{+\infty}$, $r_i\to +\infty$ such that as $i\to +\infty$,
	$$\int_0^{2\pi} |h|^p(r_i,\vartheta)\, \ud \vartheta \ls r^{-\sigma p-2}_i \to 0$$
	
	Let $i$ be sufficiently large so that $r_i \geq 2$. Recall that $h$ is harmonic. Then, by Poisson's integral formula,
	$$\sup_{y \in B(0,1)} |h|(y) \ls \int_0^{2\pi} |h|(r_i,\vartheta) \, \ud \vartheta \ls (\int_0^{2\pi} |h|^p(r_i,\vartheta) \, \ud \vartheta)^{\f 1p} \to 0.$$
	
	Hence, $h$ is a harmonic function which is identically $0$ on $B(0,1)$. Unique continuation implies that $h\equiv 0$, which is what we wanted to show. \qedhere
\end{proof}

\begin{proposition}\label{prop:elliptic.1st.der}
	Let $p \in (1,+\infty)$, $\sigma \in (-\f 2 p, 1-\f 2p)$ and $\sigma' > 0$. Then
	$$\rd_i \Delta^{-1} : L^p_{\sigma+2}(\mathbb R^2) \to W^{1,p}_{1-\f 2p-\sigma'}(\mathbb R^2)$$
	is a bounded map.
\end{proposition}
\begin{proof}
	By Proposition~\ref{prop:basic.elliptic.2},
	$$\|\Delta^{-1} f - \f 1{2\pi} (\int_{\mathbb R^2} f(y) \, \ud y)\om(|x|) \log |x| \|_{W^{2,p}_\sigma(\mathbb R^2)} \ls \|f\|_{L^p_{\sigma+2}(\mathbb R^2)}.$$
	It follows that 
	\begin{equation}\label{eq:diD-1.1}
	\|\rd_i \Delta^{-1} f - \f 1{2\pi} (\int_{\mathbb R^2} f(y) \, \ud y)\rd_i (\om(|x|) \log |x|) \|_{W^{1,p}_{\sigma+1}(\mathbb R^2)} \ls \|f\|_{L^p_{\sigma+2}(\mathbb R^2)}.
	\end{equation}
	
	Notice now that for $p$, $\sigma$, $\sigma'$ as given, we have
	\begin{equation}\label{eq:diD-1.2}
	\|\rd_i (\om(|x|) \log |x|)\|_{W^{1,p}_{1-\f 2p - \sigma'}(\mathbb R^2)} \ls 1, \quad |\int_{\mathbb R^2} f(y) \, \ud y|\ls \|f \|_{L^p_{\sigma+2}(\mathbb R^2)}.
	\end{equation}
	
	Note that $1-\f 2p - \sigma' < \sigma +1$ since $-\f 2p <\sigma$. Therefore, using the triangle inequality, \eqref{eq:diD-1.1} and \eqref{eq:diD-1.2}, we obtain
	\begin{equation*}
	\begin{split}
	&\: \|\rd_i \Delta^{-1} f \|_{W^{1,p}_{1-\f 2p-\sigma'}(\mathbb R^2)} \\
	\leq &\: \|\rd_i \Delta^{-1} f - \f 1{2\pi} (\int_{\mathbb R^2} f(y) \, \ud y)\rd_i (\om(|x|) \log |x|) \|_{W^{1,p}_{\sigma+1}(\mathbb R^2)} + |\int_{\mathbb R^2} f(y) \, \ud y| \| \rd_i (\om(|x|) \log |x|)\|_{W^{1,p}_{1-\f 2p-\sigma'}(\mathbb R^2)} \\
	\ls &\: \|f \|_{L^p_{\sigma+2}(\mathbb R^2)},
	\end{split}
	\end{equation*}
	as desired.
\end{proof}

\subsection{Some remarks and conventions}\label{sec:first.remaks.elliptic}

In the remainder of this section, we obtain the desired bounds using the Poisson equations that the metric components $N$, $\bt^i$ and $\gamma$ satisfy. Recall that the metric components admit decompositions as in \eqref{eq:metric.decomposition} and so they are given by 
\begin{equation}\label{eq:metric.Delta-1.def}
\gamma = \Delta^{-1} \Delta \gamma,\quad \bt^i = \Delta^{-1}\Delta \bt^i,\quad N = 1 + \Delta^{-1} \Delta N, 
\end{equation}
where $\Delta^{-1}$ is as in Definition~\ref{def:Delta-1}. Note particularly the constant $1$ term built into the definition of $N$. Moreover, it is useful to note that in the decomposition in \eqref{eq:metric.decomposition}, $\bt^i$ does not have a logarithmic contribution and $\gamma_{asymp}$ is independent of $t$.

We introduce some schematic notations to be used in this section. 
\begin{itemize}
\item We will use $\mfg$ to denote a metric component, i.e.~$\mfg \in \{\gamma, \bt^i, N\}$.
\item Denote by $\Omg(\mfg)$ a smooth function of $\mfg$ such that
\begin{equation}\label{eq:def.Omega.g}
|\Omega(\mfg)(x)|\ls \la x\ra^{{\f {\alp}{10}}},\quad |\rd_t (\Omega(\mfg))(x)|\ls \ep^{\f 54}\la x\ra^{{\f {\alp}{10}}},\quad |\rd_i (\Omega(\mfg))(x)|\ls \ep^{\f 54} \la x \ra^{-1+{\f {\alp}{10}}}.
\end{equation}
This holds for instance for $\f{ e^{2\gamma}}{N^2}$, $\f{e^{2\gamma}}{N}$, $\f{e^{-2\gamma}}{N}$ and $e^{-2\gamma}\log N$, etc.
\item When considering terms on the RHSs of \eqref{Nellipticequation}--\eqref{betaellipticequation}, if the precise structure of the term is unimportant, we will denote the terms schematically by $\Omega(\mfg) \rd_\alp \phi \rd_\bt \phi$ and $\Omega(\mfg) \rd_i \mfg \rd_j \mfg$, where it is understood that $\Omg(\mfg)$ satisfies \eqref{eq:def.Omega.g}.
\end{itemize}

It will also useful introduce the following cutoff function.
\begin{definition}\label{def:varpi}
Fix a cutoff function $\varpi\in C^\i_c(\RR^2)$ such that $\varpi \equiv 1$ on $B(0,2R)$ and $\mathrm{supp}(\varpi) \subseteq B(0,3R)$.
\end{definition}

\subsection{Elliptic estimates for purely spatial derivatives of $\mfg$}\label{sec:g.purely.spatial}

In this subsection, we obtain various estimates for the purely spatial derivatives of $\mfg$: the simple up to $W^{2,4}$ estimates (with weights) in \textbf{Section~\ref{sec:elliptic.W24}}, the higher order $W^{2,2+s'}$ estimates in \textbf{Section~\ref{sec:elliptic.W22+s'}}, and finally, the most difficult $W^{2,\infty}$ estimate (which is a Besov space end-point elliptic estimate; recall Section~\ref{sec:intro.elliptic}) in \textbf{Section~\ref{sec:elliptic.W2i}}.

\subsubsection{The weighted $W^{2,4}$ estimates}\label{sec:elliptic.W24}
 
We begin with the simplest elliptic estimates, which is a direct application of Proposition~\ref{prop:basic.elliptic.2}.
\begin{proposition}\label{prop:elliptic.easy}
	$\gamma$, $\bt^i$ and $N$ admit a decomposition as in \eqref{eq:metric.decomposition}. Moreover, for all $t\in [0,T_B)$:
	\begin{equation}\label{eq:d2g.L4}
	|\gamma_{asymp}| + |N_{asymp}| \ls \ep^{\f 32},\quad \| \widetilde{\gamma}\|_{W^{2,4}_{\f 12 - \alp}(\Sigma_t)} + \| \widetilde{N}\|_{W^{2,4}_{\f 12 - \alp}(\Sigma_t)} + \| \beta \|_{W^{2,4}_{\f 12 - \alp}(\Sigma_t)}  \ls \ep^{\f 32}.
	\end{equation}
	Furthermore, it holds that 
	\begin{equation}\label{eq:d2g.Linfty}
	\| \widetilde{\gamma}\|_{W^{1,\infty}_{1 - \alp}(\Sigma_t)} + \| \widetilde{N}\|_{W^{1,\infty}_{1 - \alp}(\Sigma_t)} + \| \beta \|_{W^{1,\infty}_{1 - \alp}(\Sigma_t)}  \ls \ep^{\f 32}.
	\end{equation}
\end{proposition}
\begin{proof}
	\pfstep{Step~0: The logarithmic term as $|x|\to \infty$} To show that the decomposition \eqref{eq:metric.decomposition} holds, we need that $\gamma_{asymp}$ is a constant, $\gamma_{asymp},\,N_{asymp} \geq 0$, and $\bt^i$ has no logarithmic terms. All these follow from the local existence result in \cite[Theorem~5.4]{HL.elliptic}. From now on, we thus focus on the estimates.

	\pfstep{Step~1: Estimates in \eqref{eq:d2g.L4} for $\gamma$ and $N$} Recall that $\gamma$ and $N$ are given by \eqref{eq:metric.Delta-1.def}. Using Proposition~\ref{prop:basic.elliptic.2}, it thus suffices show that the RHSs of \eqref{Nellipticequation} and \eqref{gammaellipticequation} can be bounded in $L^4_{\f 12 - \alp +2}$ by $\ls \ep^{\f 32}$. We will now prove such a bound.
		
	We start with the scalar field terms. The precise structure of scalar field terms in the RHSs of \eqref{Nellipticequation} and \eqref{gammaellipticequation} is unimportant, and we control general terms of the schematic form $\Omega(\mfg) \rd_\alp \phi \rd_\bt \phi$ (recall Section~\ref{sec:first.remaks.elliptic}). Since $\mathrm{supp}(\phi) \subseteq B(0,R)$, we can ignore all the weights and use \eqref{eq:def.Omega.g} and \eqref{BA:Li} to obtain
	\begin{equation}\label{eq:dphi.dphi.L4}
	\|\Omega(\mfg) \rd_\alp \phi \rd_\bt \phi \|_{L^4_{\f 12 - \alp +2}(\Sigma_t)} \ls \|\rd \phi\|_{L^\infty(\Sigma_t)}^2 \ls \ep^{\f 32}.
	\end{equation}
	
	The remaining terms in \eqref{Nellipticequation} and \eqref{gammaellipticequation} take the form $\f{e^{2\gamma}}{N}|\mathfrak L \bt|^2$ and $\f{e^{2\gamma}}{N^2}|\mathfrak L \bt|^2$. Note that the sign properties of $\gamma_{asymp}$ and $N_{asymp}$ means that $\f{e^{2\gamma}}{N}$ and $\f{e^{2\gamma}}{N^2}$ are favorable in terms of $\la x \ra$ weights. Hence, by H\"older's inequality and the bootstrap assumptions \eqref{eq:BA.g.Li} and \eqref{eq:BA.g.L4}, we have
	\begin{equation}\label{eq:dbt.dbt.L4}
	\begin{split}
	&\: \| (\mathfrak L \bt) (\mathfrak L \bt) \|_{L^4_{\f12 - \alp +2}(\Sigma_t)} \ls  \|\mathfrak L \bt\|_{L^4_{\f12 - \alp + 1}(\Sigma_t)} \|\mathfrak L \bt\|_{L^\i_1(\Sigma_t)} \ls \|\bt\|_{W^{1,4}_{\f 12 - \alp}(\Sigma_t)} \|\bt\|_{W^{1,\i}_{1-\alp}(\Sigma_t)} \ls \ep^{\f 52}.
	\end{split}
	\end{equation}
	
	Combining \eqref{eq:dphi.dphi.L4} and \eqref{eq:dbt.dbt.L4}, we obtain the desired bound
	$$\| \Delta^{-1} (\mbox{RHS of \eqref{Nellipticequation}})\|_{L^4_{\f 12 - \alp +2}(\Sigma_t)} + \| \Delta^{-1} (\mbox{RHS of \eqref{gammaellipticequation}})\|_{L^4_{\f 12 - \alp +2}(\Sigma_t)} \ls \ep^{\f 32}.$$
	
	\pfstep{Step~2: Estimates in \eqref{eq:d2g.L4} for $\bt$} To obtain the bound in \eqref{eq:d2g.L4} for $\bt$, we argue as in Step~1 to bound the RHS of \eqref{betaellipticequation} in $L^4_{\f 12 - \alp +2}$. Clearly, the scalar field term can be bounded exactly as in \eqref{eq:dphi.dphi.L4}. For the remaining terms, we use H\"older's inequality and the bootstrap assumptions \eqref{eq:BA.g.Li} and \eqref{eq:BA.g.L4} to obtain
	\begin{equation}\label{eq:dbt.dgamma.L4}
	\begin{split}
	&\: \| \rd_i \bt \rd_j \gamma \|_{L^4_{\f12 - \alp +2}(\Sigma_t)} \ls  \|\rd_i \bt\|_{L^4_{\f12 - \alp + 1}(\Sigma_t)} \|\rd_j \gamma\|_{L^\i_{1}(\Sigma_t)} \ls  \| \bt \|_{W^{1,4}_{\f12 - \alp}(\Sigma_t)} \|\rd_j \gamma\|_{L^\i_{1}(\Sigma_t)} \ls \ep^{\f 52}.
	\end{split}
	\end{equation}
	Note that the estimate \eqref{eq:dbt.dgamma.L4} saturates the $\la x \ra$ weights. For the $\f{\rd_i \bt \rd_j N}N$, we note that $\f 1N$ is favorable in terms of weight, so it suffices to show the following bound (which can be proven as in \eqref{eq:dbt.dgamma.L4}):
	\begin{equation}\label{eq:dbt.dN.L4}
	\| \rd_i \bt \rd_j N \|_{L^4_{\f12 - \alp +2}(\Sigma_t)} \ls \ep^{\f 52}.
	\end{equation}
	
	Combining \eqref{eq:dphi.dphi.L4}, \eqref{eq:dbt.dgamma.L4}, \eqref{eq:dbt.dN.L4} gives $\| \Delta^{-1} (\mbox{RHS of \eqref{betaellipticequation}})\|_{L^4_{\f 12 - \alp +2}(\Sigma_t)} \ls \ep^{\f 32}$, which gives the desired bound for $\bt$ in \eqref{eq:d2g.L4}.
	
	\pfstep{Step~3: Proof of \eqref{eq:d2g.Linfty}} Finally, \eqref{eq:d2g.Linfty} follows from Sobolev embedding (1(a) of Proposition~\ref{prop:Sobolev.weighted}) and the estimate \eqref{eq:d2g.L4} that we have already obtained. \qedhere
\end{proof}

\subsubsection{The estimates for $2+s'$ derivatives of metric coefficients in $L^2$}\label{sec:elliptic.W22+s'}

\begin{proposition}\label{prop:g.est.top.L2}
	The following estimate holds for all $t\in [0,T_B)$:
	$$\sum_{\mfg \in \{\gamma, \bt^i,N\}} \|\rd^2_{x} \Db^{s'} \mfg \|_{L^2(\Sigma_t)} \ls \ep^{\f 32}.$$
\end{proposition}
\begin{proof}
	Using \eqref{eq:metric.Delta-1.def} and the $L^2$-boundedness of the operator $\rd^2_{ij}\Delta^{-1}$, it suffices to prove that 
	\begin{equation}\label{eq:H2+s'.metric.1}
	\sum_{\mfg \in \{\gamma, \bt^i,N\}} \| \Db^{s'} (\Delta \mfg) \|_{L^2(\Sigma_t)} \ls \ep^{\f 32}.
	\end{equation}
	(Note in particular that we do not demand weights in this bound.)
	
	Let $\Omega(\mfg)$ be a smooth function of $\mfg$ as in \eqref{eq:def.Omega.g}. Schematically\footnote{We emphasize that this equation is schematic so that the components on the RHS may be different from the component on the LHS.}, 
	\begin{equation}\label{eq:mfgellipticequation.schematic}
	\Delta \mfg = \Omega(\mfg) \rd_\alp \phi \rd_\sigma \phi + \Omega(\mfg) \rd_i \mfg \rd_j \mfg.
	\end{equation}
	
	Since $\mathrm{supp}(\phi) \subseteq B(0,R)$, $\Omega(\mfg) \rd_\alp \phi \rd_\sigma \phi = \varpi  \Omega(\mfg) \rd_\alp \phi \rd_\sigma \phi$ (for $\varpi$ in Definition~\ref{def:varpi}). Thus, by Lemma~\ref{lem:frac.product}, \eqref{eq:def.Omega.g} and the bootstrap assumptions \eqref{BA:rphi}, \eqref{tphiH3/2bootstrap} and \eqref{BA:Li}, we have
	\begin{equation}\label{eq:H2+s'.metric.2}
	\begin{split}
	&\: \|\Db^{s'} (\Omega(\mfg) \rd_\alp \phi \rd_\sigma \phi) \|_{L^2(\Sigma_t)} \\
	\ls &\: \|\Db^{s'} (\varpi \Omega(\mfg))\|_{L^\i(\Sigma_t)} \| \rd \phi\|_{L^4(\Sigma_t)}^2 + \|\varpi \Omega(\mfg)\|_{L^\i(\Sigma_t)} \| \rd\Db^{s'} \phi\|_{L^2(\Sigma_t)} \|\rd\phi\|_{L^\i(\Sigma_t)} \\
	\ls &\: \|\varpi\Omega(\mfg)\|_{C^1(\Sigma_t)} ( \|\rd\phi\|_{L^\i(\Sigma_t)}^2 + \| \rd\Db^{s'} \phi\|_{L^2(\Sigma_t)} \|\rd\phi\|_{L^\i(\Sigma_t)}) \ls \ep^{\f 32}.
	\end{split}
	\end{equation}
	
	Using Lemma~\ref{lem:frac.product} after distributing the weights, and applying \eqref{eq:def.Omega.g} and the bootstrap assumptions \eqref{eq:BA.g.asymp}--\eqref{eq:BA.g.L4} (and recalling $\alp = 10^{-2}$), we obtain
	\begin{equation}\label{eq:H2+s'.metric.3}
	\begin{split}
	&\: \|\Db^{s'} (\Omega(\mfg) \rd_i \mfg \rd_i \mfg )\|_{L^2(\Sigma_t)} \\
	\ls &\: \|\Db^{s'} (\la x\ra^{-\f \alp 2} \Omega(\mfg))\|_{L^\i(\Sigma_t)} \|\la x\ra^{\f \alp 4}\rd_i \mfg\|_{L^4(\Sigma_t)}^2 \\
	&\: + \|\la x\ra^{-\f \alp 2} \Omega(\mfg)\|_{L^\i(\Sigma_t)} \|\Db^{s'} (\la x\ra^{\f \alp 4}\rd_i \mfg)\|_{L^4(\Sigma_t)} \|\la x\ra^{\f \alp 4} \rd_i \mfg\|_{L^4(\Sigma_t)} \\
	\ls &\: \|\Omega(\mfg) \|_{W^{1,\i}_{-\f \alp 2}(\Sigma_t)} \|\rd_i \mfg \|_{W^{1,4}_{\f 12-\alp}(\Sigma_t)}^2 \ls \ep^{\f 52}.
	\end{split}
	\end{equation}
	
	Combining \eqref{eq:mfgellipticequation.schematic}, \eqref{eq:H2+s'.metric.2} and \eqref{eq:H2+s'.metric.3}, we obtain \eqref{eq:H2+s'.metric.1}, as desired. \qedhere
\end{proof}

\subsubsection{The estimates for second spatial derivatives of $\mfg$}\label{sec:elliptic.W2i}

Finally, we control the second spatial derivatives of $\mfg$. This could be thought of as an $L^\i$-endpoint elliptic estimates in Besov space. Notice that the scalar field obeys Besov space estimates in the $(u_a, u_b)$ coordinate system, while the elliptic operator that we need to invert is a constant coefficient operator only in the $(x^1, x^2)$ coordinate system. We will treat this by using the physical space representation of the kernel.

\begin{lem}\label{lem:decomposition.into.diff.Besov}
For $\mfg \in \{N,\,\bt,\,\gamma\}$, $\Delta \mfg$ (c.f.~\eqref{Nellipticequation}--\eqref{betaellipticequation}) admits the decomposition
$$\Delta\mfg = \sum_{ 1\leq a< b\leq 3} F^{(ab)}_{\mfg},$$
where $F_{\mfg}^{(ab)}$ obeys\footnote{Recall the definition of the Besov space in Definition~\ref{def:Besov}.} the following estimates for all $t\in [0,T_B)$:
	\begin{equation}\label{eq:hard.Besov.2}
	\| F_{\mfg}^{(ab)} \|_{B^{u_a,u_b}_{\infty,1}(\Sigma_t)} + \|F_{\mfg}^{(ab)} \|_{L^3(\Sigma_t)} \ls \ep^{\f 32},
	\end{equation}
\end{lem}
\begin{proof}
The desired estimate relies only on the schematic form of the equation \eqref{eq:mfgellipticequation.schematic}. After defining the decomposition in Step~1, we first prove the Besov space estimates in Step~2. The $L^3$ estimates are simpler, and the decomposition plays no role. This will be carried out in Step~3.

\pfstep{Step~1: The decomposition} We now define the decomposition. For definiteness, we put all the metric terms and the quadratic $\rphi$ terms in $F_{\mfg}^{(12)}$. The other terms require a more precise decomposition: 
\begin{itemize}
\item $\Omg(\mfg) \rd_\alp \widetilde{\phi}_a \rd_\bt \rphi$ will be put in $F_{\mfg}^{(1a)}$ if $a\neq 1$, and in $F_{\mfg}^{(a2)}$ if $a = 1$;
\item $\Omg(\mfg) \rd_\alp \widetilde{\phi}_a \rd_\bt \widetilde{\phi}_b$ will be put in $F_{\mfg}^{(ab)}$ if $a<b$, in $F_{\mfg}^{(ba)}$ if $b<a$, in $F_{\mfg}^{(1a)}$ if $1<a = b$, and in $F_{\mfg}^{(12)}$ if $a=b=1$.
\end{itemize}
For concreteness, we explicitly give the decomposition when $\mfg = N$. In this case, we decompose RHS of \eqref{Nellipticequation} as
\begin{align}
F_N^{(12)}:= &\: \f{e^{2\gamma}}{4N} |\mathfrak L \bt|^2 + 2 N e^{2\gamma} \cdot [ (\n \rphi)^2 + (\n \widetilde{\phi}_1)^2  + (\n \widetilde{\phi}_2)^2 ] \notag \\
&\: + 4 N e^{2\gamma} \cdot \{ [(\n \widetilde{\phi}_1) + (\n \widetilde{\phi}_2)] \cdot (\n \rphi) +(\n \widetilde{\phi}_1)\cdot (\n \widetilde{\phi}_2) \},\\
F_N^{(13)}:= &\: 2 N e^{2\gamma} \cdot (\n \widetilde{\phi}_3)^2 + 4 N e^{2\gamma} \cdot (\n \widetilde{\phi}_3) \cdot [(\n \rphi) + (\n \widetilde{\phi}_1)],\\
F_N^{(23)}:= &\: 4 N e^{2\gamma} \cdot (\n \widetilde{\phi}_2) \cdot (\n \widetilde{\phi}_3).
\end{align}

\pfstep{Step~2: The Besov estimates} An important ingredient for the estimate is that the Besov space $B^{u_a,u_b}_{\infty,1}$ is an algebra and obeys the estimate
\begin{equation}\label{eq:Besov.algebra}
\|f\cdot h \|_{B^{u_a,u_b}_{\infty,1}} \ls \|f \|_{B^{u_a,u_b}_{\infty,1}} \|h \|_{B^{u_a,u_b}_{\infty,1}}.
\end{equation}
This is obvious using the definition of $B^{u_a,u_b}_{\infty,1}$ and Young's convolution inequality.

\pfstep{Step~2(a): The metric terms}
We first bound terms schematically of the form $\Omg(\mfg) \rd_i \mfg \rd_j \mfg$ (which are in $F^{(12)}_{\mfg}$ of the decomposition). Note the standard Sobolev embedding $W^{1,4}(\mathbb R^2)\hookrightarrow B_{\infty,1}(\mathbb R^2)$. Hence, using also Corollary~\ref{cor:diffeo}, we have $\|f\|_{B^{u_a,u_b}_{\infty,1}(\Sigma_t)} \ls \|f \|_{W^{1,4}(\Sigma_t)}$. In particular, using \eqref{eq:def.Omega.g}, \eqref{eq:BA.g.asymp} and \eqref{eq:BA.g.L4}, it follows that (for any $a\neq b$) $\|\la x \ra^{-\f{\alp}{10}}\Omg(\mfg)\|_{B^{u_a,u_b}_{\infty,1}(\Sigma_t)}\ls 1$ and $\| \la x\ra^{\f{\alp}{20}} \rd_x \mfg\|_{B^{u_a,u_b}_{\infty,1}(\Sigma_t)}\ls \ep^{\f 54}$.

Hence, by \eqref{eq:Besov.algebra}, we have, for any $a\neq b$ (and in particular $(a,b) = (1,2)$),
\begin{equation}
\begin{split}
\| \Omg(\mfg) \rd_i \mfg \rd_j \mfg \|_{B^{u_a,u_b}_{\infty,1}(\Sigma_t)} \ls \| \la x \ra^{-\f{\alp}{10}} \Omg(\mfg)  \|_{B^{u_a,u_b}_{\infty,1}(\Sigma_t)} \| \la x\ra^{\f{\alp}{20}} \rd_x \mfg  \|_{B^{u_a,u_b}_{\infty,1}(\Sigma_t)}^2 \ls \ep^{\f {5}2}.
\end{split}
\end{equation}

\pfstep{Step~2(b): The scalar field terms} Since $\mathrm{supp}(\phi) \subseteq B(0,R)$, we have $\Omg(\mfg) \rd_\alp \phi \rd_\bt \phi = \varpi \Omg(\mfg) \rd_\alp \phi \rd_\bt \phi$. Arguing as in Step~2(a), we have (for any $a\neq b$)
\begin{equation}\label{eq:varpi.Omg.B.est}
\| \varpi \Omg(\mfg) \|_{B^{u_a,u_b}_{\infty,1}(\Sigma_t)} \ls 1.
\end{equation}

To proceed, we need to control $\phi$, for which we use the decomposition $\phi = \rphi + \sum_{k=1}^3 \tphi$. The quadratic term would give the following three types of contributions:
$$\underbrace{\rd \rphi \cdot \rd\rphi}_{=:I},\quad \underbrace{\rd\rphi \cdot \rd \widetilde{\phi}_a}_{=:II},\quad \underbrace{\rd \widetilde{\phi}_a\cdot \rd \widetilde{\phi}_b}_{=:III}.$$
Each of these terms has $B^{u_a,u_b}_{\infty,1}$ norm $\ls \ep^{\f 32}$ after \emph{choosing} suitable $a$ and $b$. More precisely, $\| I \|_{B^{u_a,u_b}_{\infty,1}} \ls \ep^{\f 32}$ for any $(a,b)$ such that $a\neq b$ (by \eqref{eq:Besov.algebra} and \eqref{rphiBbootstrap}); $\| II \|_{B^{u_a,u_b}_{\infty,1}} \ls \ep^{\f 32}$ for $a$ as in the term and any $b\neq a$ (by \eqref{eq:Besov.algebra}, \eqref{rphiBbootstrap} and \eqref{tphiBbootstrap}); $\| III \|_{B^{u_a,u_b}_{\infty,1}} \ls \ep^{\f 32}$ for $a$, $b$ as in the term (by \eqref{eq:Besov.algebra}, \eqref{rphiBbootstrap} and \eqref{tphiBbootstrap}).

Combining this with \eqref{eq:varpi.Omg.B.est}, and using \eqref{eq:Besov.algebra}, it follows that that all the scalar field contributions for $F_{\mfg}$ obey the desired Besov bound. Finally, combining Steps~2(a) and 2(b), we conclude the proof of the Besov bounds in \eqref{eq:hard.Besov.2}.

\pfstep{Step~3: $L^3$ estimates} We begin with the estimates for the metric terms. We have more than enough regularity; the key issue is thus the decay at infinity. Noting that $\la x\ra^{-\f 23 - \alp} \in L^3(\Sigma_t)$, we have, by H\"older's inequality, \eqref{eq:BA.g.asymp}, \eqref{eq:BA.g.Li} and \eqref{eq:def.Omega.g}, that
$$\|\Omg(\mfg) \rd_i \mfg \rd_j \mfg \|_{L^3(\Sigma_t)} \ls \|\la x\ra^{-\f 23 - \alp}\|_{L^3(\Sigma_t)} \|\la x\ra^{- \alp} \Omg(\mfg)\|_{L^\infty(\Sigma)} \|\la x \ra^{\f 13+\alp}\rd_x \mfg \|_{L^\infty(\Sigma_t)}^2 \ls \ep^{\f 52}.$$

Turning to the scalar field terms, we use \eqref{eq:def.Omega.g}, \eqref{BA:Li}, and that $\mathrm{supp}(\phi) \subseteq B(0,R)$ to obtain
\begin{equation*}
\begin{split}
&\: \|\Omg(\mfg) \rd_\alp\rphi \rd_\bt \rphi\|_{L^3(\Sigma_t)},\,\|\Omg(\mfg) \rd_\alp\rphi \rd_\bt \widetilde{\phi}_a\|_{L^3(\Sigma_t)},\,\|\Omg(\mfg) \rd_\alp \widetilde{\phi}_a \rd_\bt \widetilde{\phi}_b \|_{L^3(\Sigma_t)} \\
\ls &\: \|\Omg(\mfg)\|_{ L^\i(\Sigma_t \cap B(0,R))}( \|\rd \rphi\|_{L^\i(\Sigma_t)} + \|\rd \tphi\|_{L^\i(\Sigma_t)})^2 \ls \ep^{\f 32}.
\end{split}
\end{equation*}

Recalling the decomposition in Step~1, and combining the above estimates, we obtain the desired $L^3$ bound in \eqref{eq:hard.Besov.2}. \qedhere

\end{proof}

Using the decomposition in Lemma~\ref{lem:decomposition.into.diff.Besov}, we prove our main elliptic estimate for $\rd^2_{ij}\mfg$:
\begin{proposition}\label{prop:elliptic.Besov}
	The following estimates hold for all $t\in [0,T_B)$:
	\begin{equation} \label{secondderivative.g.estimate}
	\sum_{\mfg \in \{N,\, \bt^l,\, \gamma\}} \max_{i,j} \|\rd^2_{ij} \mfg \|_{L^\i(\Sigma_t)} \ls \ep^{\f 32}.
	\end{equation}
\end{proposition}
\begin{proof}
	By Lemma~\ref{lem:decomposition.into.diff.Besov}, each $\mfg$ satisfies a Poisson equation 
	\begin{equation} \label{Laplaceg}
	\Delta \mfg =  \sum_{ 1\leq a< b\leq 3} F^{(ab)}_{\mfg},
	\end{equation}
	where 
	the inhomogeneous terms $F_{\mfg}^{(ab)}$ obey \eqref{eq:hard.Besov.2}.
	
	Define auxiliary functions $\mfg^{(ab)}_k$ by
	\begin{equation}\label{eq:hard.Besov.def.of.g}
	 \mfg^{(ab)}_k = \Delta^{-1} P_k^{u_a,u_b}  F^{(ab)}_{\mfg}.
	\end{equation}

	By \eqref{eq:metric.Delta-1.def}, it follows that in order to obtain the desired conclusion of the proposition, it suffices to prove
	$$\sum_{ 1\leq a< b\leq 3}\sum_{k\geq 0} \| \rd^2_{ij} \mfg^{(ab)}_k \|_{L^\i(\Sigma_t)} \ls \ep^{\f 32}.$$
	
	From now on, fix $1\leq a < b\leq 3$. We will bound each piece in the sum, treating the $k=0$ and $k>0$ cases separately. (See Steps~1 and 2 below.)
	
	Before we proceed, note that by (the second term in) \eqref{eq:hard.Besov.2}, we also have
	\begin{equation}\label{eq:hard.Besov.3}
	\sup_{k\geq 0} \|P_k^{u_a,u_b} F_{\mfg}^{(ab)} \|_{L^3(\Sigma_t)}  \ls \ep^{\f 32}.
	\end{equation}

	\pfstep{Step~1: The case $k=0$} This is the easy case. Clearly,
	\begin{equation}\label{eq:hard.Besov.easy.1}
	\| P_0 F_{\mfg}^{(ab)} \|_{L^3(\Sigma_t)} + \|\rd_{\ell} (P_0 F_{\mfg}^{(ab)})\|_{L^3(\Sigma_t)} \ls \ep^{\f 32}.
	\end{equation}
	(By definition of $P_0$, $P_0 F_{\mfg}^{(ab)} \in W^{1,3}$ in the $(u_a,u_b)$ coordinates. \eqref{eq:hard.Besov.easy.1} then follows from Corollary~\ref{cor:diffeo}.)
	Using the definition \eqref{eq:hard.Besov.def.of.g}, the bound \eqref{eq:hard.Besov.easy.1}, standard $L^3$ elliptic estimates and then Sobolev embedding, we obtain immediately that
	\begin{equation}\label{eq:hard.Besov.easy.2}
	\max_{i,j}\| \rd^2_{ij} \mfg_0^{(ab)} \|_{L^\i(\Sigma_t)} \ls \ep^{\f 32}.
	\end{equation}

	\pfstep{Step~2: The case $k >0$}
	
	\pfstep{Step~2(a): Extracting information from \eqref{eq:hard.Besov.2} and \eqref{eq:hard.Besov.3}} By \eqref{eq:hard.Besov.2} and the frequency support information of $P_k^{u_a,u_b}  F_{\mfg}^{(ab)}$, we know that 
	\begin{equation}\label{eq:hard.Besov.4}
	\| \srd_{u_a} (P_k^{u_a,u_b} F_{\mfg}^{(ab)}) \|_{L^\i(\Sigma_t)} +  \| \srd_{u_b} (P_k^{u_a,u_b} F_{\mfg}^{(ab)}) \|_{L^\i(\Sigma_t)}  \ls 2^k \| P_k^{u_a,u_b} F_{\mfg}^{(ab)}\|_{L^\i(\Sigma_t)}.
	\end{equation}
	Hence, by Corollary~\ref{cor:diffeo}, 
	\begin{equation}\label{eq:hard.Besov.5}
	\| \rd_i (P_k^{u_a,u_b} F_{\mfg}^{(ab)}) \|_{L^\i(\Sigma_t)}  \ls 2^k \| P_k^{u_a,u_b} F_{\mfg}^{(ab)}\|_{L^\i(\Sigma_t)}.
	\end{equation}
	
	On the other hand, since the Fourier transform of $P_k^{u_a,u_b} F_{\mfg}^{(ab)}$ is by definition supported away from $0$, we can introduce a partition of unity in the angular Fourier directions to deduce that there exist $\widetilde{F}_{\mfg}^{(ab)}$ and $\widetilde{\widetilde{F}}_{\mfg}^{(ab)}$ such that
	\begin{equation}\label{eq:hard.Besov.6}
	P_k^{u_a,u_b} F_{\mfg}^{(ab)} = \srd_{u_a} \widetilde{F}_{\mfg}^{(ab)} + \srd_{u_b} \widetilde{\widetilde{F}}_{\mfg}^{(ab)}.
	\end{equation}
	Moreover, the frequency support of $P_k^{u_a,u_b} F_{\mfg}^{(ab)}$ implies that $\widetilde{F}_{\mfg}^{(ab)}$ and $\widetilde{\widetilde{F}}_{\mfg}^{(ab)}$ can be chosen so that
	\begin{equation}\label{eq:hard.Besov.6.1}
	\| \widetilde{F}_{\mfg}^{(ab)} \|_{L^\i(\Sigma_t)} +  \| \widetilde{\widetilde{F}}_{\mfg}^{(ab)} \|_{L^\i(\Sigma_t)}  \ls 2^{-k} \| P_k^{u_a,u_b} F_{\mfg}^{(ab)}\|_{L^\i(\Sigma_t)}\quad  \| \widetilde{F}_{\mfg}^{(ab)} \|_{L^3(\Sigma_t)} +  \| \widetilde{\widetilde{F}}_{\mfg}^{(ab)} \|_{L^3(\Sigma_t)} \ls 2^{-k} \ep^{\f 32},
	\end{equation}
	where we have used Corollary~\ref{cor:diffeo} (to compare volume forms) and \eqref{eq:hard.Besov.3}.
	
	We now rewrite $\srd_{u_a} \widetilde{F}_{\mfg}^{(ab)} = \sum_{\ell=1}^2 \{\rd_\ell [(\srd_{u_a} x^\ell)  \widetilde{F}_{\mfg}^{(ab)}] - [\rd_\ell (\srd_{u_a} x^\ell) ] \widetilde{F}_{\mfg}^{(ab)} \}$ and similarly for $\srd_{u_b} \widetilde{\widetilde{F}}_{\mfg}^{(ab)}$.
	Therefore, using Corollary~\ref{cor:diffeo} and Proposition~\ref{prop:dx2u}, we deduce from \eqref{eq:hard.Besov.6.1} that 
	\begin{equation}\label{eq:hard.Besov.7}
	P_k^{u_a,u_b} F_{\mfg}^{(ab)} = \sum_{\ell =1,2} \rd_{\ell} H_{\mfg,k,\ell}^{(ab)} + \widetilde{H}_{\mfg,k}^{(ab)}, \quad \| H_{\mfg,k,\ell}^{(ab)} \|_{L^\i(\Sigma_t)} \ls 2^{-k} \| P_k^{u_a,u_b} F_{\mfg}^{(ab)}\|_{L^\i(\Sigma_t)},\quad  \| \widetilde{H}_{\mfg,k}^{(ab)} \|_{L^3(\Sigma_t)}  \ls 2^{-k} \ep^{\f 32} .
	\end{equation}
	
	\pfstep{Step~2(b): Estimating the kernel} We now bound \eqref{eq:hard.Besov.def.of.g}. Differentiating the kernel in Definition~\ref{def:Delta-1}, we have\begin{equation}\label{eq:stupid.kernel}
	\rd^2_{ij} \mfg^{(ab)}_k = \f 1{2\pi} \int_{\mathbb R^2} \frac{\de_{ij} - \f{2(x-y)_i (x-y)_j}{ |x -y|^2 } }{ |x-y|^2 } (P_k^{u_a,u_b} F^{(ab)}_{\mfg})(y)\, \ud y = \f 1{2\pi} \int_{\mathbb R^2} \rd_i (\frac{(x-y)_j}{ |x-y|^2 }) (P_k^{u_a,u_b} F^{(ab)}_{\mfg})(y)\, \ud y.
	\end{equation}
	
	We estimate separately the contributions from $|x-y|\leq 2^{-k}$ and $|x-y|\geq 2^{-k}$. For $|x-y|\leq 2^{-k}$, we use the second representation in \eqref{eq:stupid.kernel}, integrate by parts and use \eqref{eq:hard.Besov.5}, 
	\begin{equation}\label{eq:hard.Besov.8}
	\begin{split}
	&\:  \left| \int_{\{y\in \mathbb R^2:|x-y|\leq 2^{-k}\} } \rd_i (\frac{(x-y)_j}{ |x-y|^2 }) (P_k^{u_a,u_b} F^{(ab)}_{\mfg})(y)\, \ud y \right| \\
	\ls &\: \int_{\{y\in \mathbb R^2:|x-y|\leq 2^{-k}\} }  \f 1{ |x-y| }  |\rd_i(P_k^{u_a,u_b} F^{(ab)}_{\mfg})| (y)\, \ud y +  \| P_k^{u_a,u_b} F^{(ab)}_{\mfg} \|_{L^\i(\Sigma_t)} \\
	\ls &\: 2^{-k}  \|\rd_i(P_k^{u_a,u_b} F^{(ab)}_{\mfg})\|_{L^\i(\Sigma_t)} +  \| P_k^{u_a,u_b} F^{(ab)}_{\mfg} \|_{L^\i(\Sigma_t)}  \ls \| P_k^{u_a,u_b} F_{\mfg}^{(ab)}\|_{L^\i(\Sigma_t)}.
	\end{split}
	\end{equation}
	For $|x-y| \geq 2^{-k}$, we use the first representation in \eqref{eq:stupid.kernel} together with \eqref{eq:hard.Besov.7}. More precisely, we integrate by parts (for the $\rd_{\ell} H_{\mfg,k,\ell}^{(ab)}$ terms), apply Young's inequality and use the bounds in \eqref{eq:hard.Besov.7} to obtain
	\begin{equation}\label{eq:hard.Besov.9}
	\begin{split}
	&\: \left| \int_{\{y\in \mathbb R^2: |x-y|\geq 2^{-k}\} } \frac{\de_{ij} - \f{2(x-y)_i (x-y)_j}{ |x -y|^2 } }{ |x-y|^2 } (P_k^{u_a,u_b} F^{(ab)}_{\mfg})(y)\, \ud y \right| \\
	= &\: \left| \int_{\{y\in \mathbb R^2: |x-y|\geq 2^{-k}\} } \frac{\de_{ij} - \f{2(x-y)_i (x-y)_j}{ |x -y|^2 } }{ |x-y|^2 } (\sum_{\ell =1,2} \rd_{\ell} H_{\mfg,k,\ell}^{(ab)} + \widetilde{H}_{\mfg,k}^{(ab)})(y)\, \ud y \right| \\
	\ls &\: \sum_{\ell=1,2} (\int_{\{y\in \mathbb R^2: |x-y|\geq 2^{-k}\} } \f{|H_{\mfg,k,\ell}^{(ab)}|(y)}{|x-y|^3}  \,\ud y + 2^{k} \|H_{\mfg,k,\ell}^{(ab)}\|_{L^\i(\Sigma_t)} ) \\
	&\: + \int_{\{y\in \mathbb R^2: |x-y|\geq 2^{-k}\} } \f{|\widetilde{H}_{\mfg,k}^{(ab)}|(y)}{|x-y|^2}  \,\ud y \\
	\ls &\: (\int_{\{z\in \mathbb R^2: |z|\geq 2^{-k}\}} \f{1}{|z|^3} \,\ud z + 2^{k} ) (\sum_{\ell = 1,2} \|H_{\mfg,k,\ell}^{(ab)} \|_{L^\i(\Sigma_t)}) + (\int_{\{z\in \mathbb R^2: |z|\geq 2^{-k}\}} \f{1}{|z|^3} \,\ud z)^{\f 23} \|\widetilde{H}_{\mfg,k}^{(ab)}\|_{L^3(\Sigma_t)} \\
	\ls &\: 2^{k}  \sum_{\ell = 1,2} \|H_{\mfg,k,\ell}^{(ab)}\|_{L^\i(\Sigma_t)} + 2^{\f{2k}{3}}  \|\widetilde{H}_{\mfg,k}^{(ab)}\|_{L^3(\Sigma_t)} \ls \| P_k^{u_a,u_b} F_{\mfg}^{(ab)}\|_{L^\i(\Sigma_t)} + 2^{-\f k3} \ep^{\f 32}.
	\end{split}
	\end{equation}
	
	Combining \eqref{eq:stupid.kernel}, \eqref{eq:hard.Besov.8} and \eqref{eq:hard.Besov.9}, and then using \eqref{eq:hard.Besov.2}, we obtain
	\begin{equation}\label{eq:hard.Besov.10}
	\max_{i,j} \sum_{k\geq 1} \|\rd^2_{ij} \mfg^{(ab)}_k\|_{L^\i(\Sigma)} \ls \sum_{k\geq 1} (\| P_k^{u_a,u_b} F_{\mfg}^{(ab)}\|_{L^\i(\Sigma_t)} + 2^{-\f k3} \ep^{\f 32} ) \ls \ep^{\f 32}.
	\end{equation}
	
	Finally, combining \eqref{eq:hard.Besov.10} with \eqref{eq:hard.Besov.easy.2}, we obtain
	\begin{equation}
	\|\rd^2_{ij} \mfg \|_{L^\i(\Sigma)} \ls \max_{i,j,a,b} \sum_{k\geq 0} \|\rd^2_{ij} \mfg^{(ab)}_k\|_{L^\i(\Sigma)} \ls \ep^{\f 32}.
	\end{equation}
\end{proof}

We now improve Proposition~\ref{prop:elliptic.Besov} to obtain some decaying weights at infinity. This is much easier given Proposition~\ref{prop:elliptic.Besov} since we only need to improve the weights in regions away from the support of $\phi$.
\begin{proposition}\label{prop:elliptic.Besov.weighted}
The following estimates hold for all $t\in [0,T_B)$:
	\begin{equation} \label{secondderivative.g.estimate.weighted}
	\sum_{\mfg \in \{N,\, \bt^l,\, \gamma\}} \max_{i,j} \|\rd^2_{ij} \mfg \|_{L^\i_{2-\f \alp 2}(\Sigma_t)} \ls \ep^{\f 32}.
	\end{equation}
\end{proposition}
\begin{proof}
Using Proposition~\ref{prop:elliptic.Besov}, we only need a bound when $|x| \geq 3R$. By Definition~\ref{def:varpi}, it thus suffices to bound $|\rd^2_{ij}((1-\varpi) \mfg) |$. In fact, using Sobolev embedding (1(a) of Proposition~\ref{prop:Sobolev.weighted}), it in turn suffices to show
\begin{equation}\label{eq:elliptic.Besov.weighted.goal}
\max_{i,j} \|\rd^2_{ij} (1-\varpi) \mfg \|_{L^4_{\f 32-\f \alp 2}(\Sigma_t)}+ \max_{i,j,l} \| \rd^3_{ijl}((1-\varpi) \mfg) \|_{L^4_{\f 52 - \f \alp 2}(\Sigma_t)} \ls \ep^{\f 32}.
\end{equation}

The key is now to derive an equation $\Delta \rd_l ((1-\varpi) \mfg)$, and use the fact $\mathrm{supp}(\phi) \subseteq B(0,R)$, which guarantees that there is no scalar field contribution after multiplying by the cutoff $1-\varpi$.

We now commute the derivatives with the cutoff. Notice that when at least one derivative falls on $\varpi$, we can put as much $\la x\ra^{-1}$ weights as we need. Thus, we obtain the pointwise bound
\begin{equation}\label{eq:Delta.d.1-pi.mfg}
 |\Delta \rd_l ((1-\varpi) \mfg)| \ls \la x \ra^{-10} (|\mfg| + |\rd_x \mfg| + |\rd_x \rd_x \mfg|) + (1-\varpi)|\Delta \rd_l \mfg|.
 \end{equation}
Using the estimates in \eqref{eq:d2g.L4}, it follows that 
\begin{equation}\label{eq:Delta.d.1-pi.mfg.1}
\sum_{\mfg \in \{\gamma, \bt^i, N\} } \| \la x \ra^{-10} (|\mfg| + |\rd_x \mfg| + |\rd_x \rd_x \mfg|) \|_{L^4_{\f 72 -\alp}(\Sigma_t)} \ls \ep^{\f 32}.
\end{equation}
We now consider $(1-\varpi)|\Delta \rd_l \mfg|$. For this, we recall \eqref{Nellipticequation}--\eqref{betaellipticequation}. Notice that the scalar field terms drop out since $\mathrm{supp}(\phi) \cap \mathrm{supp}(1-\varpi) = \emptyset$. Hence, we only need to control the derivatives of $\f{e^{2\gamma}}{N}|\mathfrak L \bt|^2$, $\rd_x \gamma \rd_x \bt$ and $\f{\rd_x N \rd_x \bt}{N}$. These terms can be controlled in a similar manner as \eqref{eq:dbt.dbt.L4}, \eqref{eq:dbt.dgamma.L4} and \eqref{eq:dbt.dN.L4}, except that since we have an additional $\rd_l$ derivative, we control these terms also using \eqref{eq:d2g.L4}, and get an additional $\la x\ra^{-1}$ weight. In other words,
\begin{equation}\label{eq:Delta.d.1-pi.mfg.2}
\begin{split}
&\: \sum_{\mfg \in \{\gamma, \bt^i, N\}} \| (1-\varpi) \Delta \rd_l \mfg \|_{L^4_{\f 72 -\alp}(\Sigma_t)} \\
\ls &\: \| \rd_l (\f{e^{2\gamma}}{N}|\mathfrak L \bt|^2) \|_{L^4_{\f 72 -\alp}(\Sigma_t)} + \| \rd_l (\rd_x \gamma \rd_x \bt) \|_{L^4_{\f 72 -\alp}(\Sigma_t)} + \| \rd_l (\f{\rd_x N \rd_x \bt}{N}) \|_{L^4_{\f 72 -\alp}(\Sigma_t)} \ls \ep^{\f 52}.
\end{split}
\end{equation}
Notice that we also have $\int_{\Sigma_t} \Delta \rd_i ((1-\varpi) \mfg)\, dx^1\, dx^2 = 0$. (This can be proven by noting that $\Delta \rd_i ((1-\varpi) \mfg)$ is an exact divergence, and then using the compact support of $\phi$ together with the $x$-decay given by \eqref{eq:d2g.Linfty}.) Hence, by Proposition~\ref{prop:basic.elliptic.2} (with $p =4$, $\sigma = \f 12 - \f \alp 2$) and the estimates in \eqref{eq:Delta.d.1-pi.mfg}, \eqref{eq:Delta.d.1-pi.mfg.1} and \eqref{eq:Delta.d.1-pi.mfg.2}, we obtain 
\begin{equation}\label{eq:elliptic.Besov.weighted.goal.almost}
\| \rd_x ( (1-\varpi) g) \|_{W^{2,4}_{\f 12 - \f \alp 2}(\Sigma_t)} \ls \ep^{\f 32}.
\end{equation}
In particular, \eqref{eq:elliptic.Besov.weighted.goal.almost} implies \eqref{eq:elliptic.Besov.weighted.goal}. \qedhere
\end{proof}

\subsection{Elliptic estimates for $\rd_t \mfg$}\label{sec:dtg.elliptic}

We now turn to the elliptic estimates for $\rd_t\mfg$ and its spatial derivatives. The main estimate will be proven in Proposition~\ref{prop:elliptic.dtg}. These estimates should be compared to Proposition~\ref{prop:elliptic.easy} and \ref{prop:elliptic.Besov.weighted}. Notice however, that the estimate is weak: while we control $\rd^2_{ij} \mfg$ in a (weighted) $L^\i(\Sigma_t)$ space, we only prove that $\rd^2_{it} \mfg$ belongs to a (weighted) $L^{\f{2}{s'-s''}}$ space\footnote{In fact, one can replace $\f{2}{s'-s''}$ with an arbitrarily large $p <\infty$ (as long as one takes $\ep$ smaller). (Note, however, that our argument does not give an estimate for $p=\infty$.) The particular estimate we prove here is sufficient for our later applications. }.

As we explained in Section~\ref{sec:intro.elliptic} in the introduction, the main difficulty is that after differentiating the elliptic equations \eqref{Nellipticequation}--\eqref{betaellipticequation}, we have terms the involve second derivatives of $\tphi$, which by \eqref{tphiH2bootstrap} may appear to have $L^2(\Sigma_t)$ norm of size $\de^{-\f 12}$. These terms in particular require a careful analysis using the transversality of the different waves. We will first control the inhomogeneous terms in the elliptic equations for $\rd_t \mfg$: in \textbf{Section~\ref{sec:metric.inho.1}}, we carry out the more straightforward bounds, and the estimates corresponding to the interaction of the different waves are proven in \textbf{Section~\ref{sec:metric.inho.2}}. The main weighted $W^{1,\f{2}{s'-s''}}$ estimates will then be proven in \textbf{Section~\ref{sec:weighted.rdtg}}.

\subsubsection{Estimates for inhomogeneous terms I: the easy terms}\label{sec:metric.inho.1}

\begin{lemma}\label{lem:rodnianski.trick.no.need}
	Let $\Omg(\mfg)$ be as in \eqref{eq:def.Omega.g}. Then the following estimate holds for all $t\in [0,T_B)$:
	$$\|\rd_t [ \Omega(\mfg) (\rd_\alp \phi_{reg}) (\rd_\sigma \phi_{reg}) ] \|_{L^2(\Sigma_t)} \ls \ep^{\f 32}.$$
\end{lemma}
\begin{proof}
	This follows from the H\"older inequality, \eqref{eq:def.Omega.g}, compact support of $\rphi$ and the bootstrap assumptions \eqref{BA:rphi} and \eqref{BA:Li}. \qedhere
\end{proof}

\begin{lemma}\label{lem:rodnianski.trick.1}
	Let $\Omega(\mfg)$ be as in \eqref{eq:def.Omega.g}. Then
	\begin{equation}\label{eq:rodnianski.trick.1.1}
	\|\rd_t [ \Omega(\mfg) (\rd_\alp \tphi) (\rd_\sigma \tphi) ] - \rd_i [ \Omega(\mfg)(N X_k^i +\bt^i) (\rd_\alp \tphi) (\rd_\sigma \tphi)] \|_{L^2(\Sigma_t)} \ls \ep^{\f 32},
	\end{equation}
	and
	\begin{equation}\label{eq:rodnianski.trick.1.2}
	\|\rd_t [ \Omega(\mfg) (\rd_\alp \tphi) (\rd_\sigma \phi_{reg}) ] - \rd_i [ \Omega(\mfg)(N X_k^i +\bt^i) (\rd_\alp \tphi) (\rd_\sigma \phi_{reg})] \|_{L^2(\Sigma_t)} \ls \ep^{\f 32},
	\end{equation}
\end{lemma}
\begin{proof}
	\pfstep{Step~1: Proof of \eqref{eq:rodnianski.trick.1.1}} By \eqref{nXEL} and \eqref{def:e0}, $L_k = \n - X_k$, $\n = \f 1N(\rd_t -\bt^i \rd_i)$. Therefore, we can decompose $\rd_t$ as follows:
	\begin{equation}\label{eq:dt.decompose}
	\rd_t = N\n + \bt^i \rd_i = N L_k + (N X_k^i +\bt^i)\rd_i.
	\end{equation}
	
	Now we express the term on LHS of \eqref{eq:rodnianski.trick.1.1} using \eqref{eq:dt.decompose} as follows:
	\begin{equation}\label{eq:rodnianski.trick.1.3}
	\begin{split}
	&\: \rd_t [ \Omega(\mfg) (\rd_\alp \tphi) (\rd_\sigma \tphi) ] -\rd_i  [\Omg(\mfg)(N X_k^i +\bt^i) (\rd_\alp \tphi) (\rd_\sigma \tphi)] \\
	=&\: \underbrace{N L_k [ \Omega(\mfg) (\rd_\alp \tphi) (\rd_\sigma \tphi) ]}_{=: I} - \underbrace{\{\rd_i [(N X_k^i +\bt^i)]\} \Omega(\mfg) (\rd_\alp \tphi) (\rd_\sigma \tphi)}_{=: II}.
	\end{split}
	\end{equation}
	
	To estimate term $I$ in \eqref{eq:rodnianski.trick.1.3}, we use $\mathrm{supp}(\tphi)\subseteq B(0,R)$, \eqref{eq:def.Omega.g}, Lemma~\ref{lem:L.X.E}, and the bootstrap assumptions \eqref{eq:BA.g.Li}, \eqref{bootstrapsmallnessenergy} and \eqref{BA:Li} to obtain
	\begin{equation}\label{eq:rodnianski.trick.1.3.1}
	\begin{split}
	\| I \|_{L^2(\Sigma_t)} 
	\ls &\: \| N L_k (\Omg (\mfg))\|_{L^\i(\Sigma_t\cap B(0,R))}  \|\rd \tphi\|_{L^\i(\Sigma_t)}^2 \\
	&\: + \| N \Omg(\mfg)\|_{L^\i(\Sigma_t\cap B(0,R)}  \|L_k \rd \tphi\|_{L^2(\Sigma_t)} \|\rd\tphi\|_{L^\i(\Sigma_t)} \ls \ep^{\f 32}.
	\end{split}
	\end{equation}
	
	The term $II$ can be treated similarly. Using $\mathrm{supp}(\tphi)\subseteq B(0,R)$, \eqref{eq:def.Omega.g}, Lemmas~\ref{lem:L.X.E}, \ref{dgeomvflemma} and the bootstrap assumptions \eqref{eq:BA.g.Li} and \eqref{BA:Li}, we obtain
	\begin{equation}\label{eq:rodnianski.trick.1.3.2}
	\begin{split}
	\| II \|_{L^2(\Sigma_t)} \ls \| \Omega(\mfg) \rd_i [(N X_k^i +\bt^i)]\} \|_{L^\i(\Sigma_t\cap B(0,R))} \|\rd \tphi\|_{L^\i(\Sigma_t)}^2 \ls \ep^{\f 32}.
	\end{split}
	\end{equation}
	
	Combining \eqref{eq:rodnianski.trick.1.3}--\eqref{eq:rodnianski.trick.1.3.2} yields \eqref{eq:rodnianski.trick.1.1}.
	
	\pfstep{Step~2: Proof of \eqref{eq:rodnianski.trick.1.2}} The estimate \eqref{eq:rodnianski.trick.1.2} can be proven in a similar manner. The main only difference is that $\|N L_k [ \Omega(\mfg) (\rd_\alp \tphi) (\rd_\sigma \phi_{reg}) ]\|_{L^2(\Sigma_t)} \ls \ep^{\f 32}$ (c.f.~term $I$ in \eqref{eq:rodnianski.trick.1.3}) has to be proved slightly differently and we use additionally the bootstrap assumption \eqref{BA:rphi} for $\rphi$. The rest of the argument proceeds similarly. \qedhere
\end{proof}

\begin{lemma}\label{lem:dt.dg.dg.est}
The following estimate holds for all $t\in [0,T_B)$:
	\begin{equation}\label{eq:dt.db.db.est}
	\|\rd_t \{ \f{e^{2\gamma}}{N^2} | \mathfrak L \bt |^2 \}\|_{L^{\f{2}{1+s'-s''}}_{-s'+s''-2\alp+2}(\Sigma_t)} \ls \ep^{\f {15}4} + \ep^{\f 54}\sum_{\widetilde{\mfg} \in \{\widetilde{\gamma}, \bt^i,\widetilde{N}\}} \|\rd_t \widetilde{\mfg}\|_{W^{1,\f{2}{1+s'-s''}}_{-s'+s''-2\alp}(\Sigma_t)},
	\end{equation}
	\begin{equation}\label{eq:dt.dN.db.est}
	\|\rd_t \{  \rd_i (\log (Ne^{-\gamma})) (\mathfrak L \bt)_{jl} \}\|_{L^{\f{2}{1+s'-s''}}_{-s'+s''-2\alp+2}(\Sigma_t)} \ls \ep^{\f {15}4} + \ep^{\f 54}\sum_{\widetilde{\mfg} \in \{\widetilde{\gamma}, \bt^i,\widetilde{N}\}} \|\rd_t \widetilde{\mfg}\|_{W^{1,\f{2}{1+s'-s''}}_{-s'+s''-2\alp}(\Sigma_t)}. 
	\end{equation}
\end{lemma}
\begin{proof}
\pfstep{Step~1: Preliminaries} Note that $\rd_t \gamma_{asymp} = 0$ (since $\gamma_{asymp}$ is a constant; see \eqref{eq:metric.decomposition}), and thus both $\rd_t\gamma$ and $\rd_t \bt^i$ does not have a logarithmic growing contribution. Hence,
\begin{equation}\label{eq:dt.g.elliptic.proof.1}
\sum_{\mfg\in \{\gamma, \bt^i\}} \| \rd_t \rd_j \mfg \|_{L^{\f{2}{1+s'-s''}}_{-s'+s''-2\alp+1}(\Sigma_t)} \ls \sum_{\widetilde{\mfg} \in \{\widetilde{\gamma}, \bt^i,\widetilde{N}\}} \|\rd_t \widetilde{\mfg}\|_{W^{1,\f{2}{1+s'-s''}}_{-s'+s''-2\alp}(\Sigma_t)}.
\end{equation}
For $\rd_t N$, the logarithmic terms give a worse decay for large $\la x\ra$, but after using \eqref{eq:BA.g.asymp}, we still have
\begin{equation}\label{eq:dt.g.elliptic.proof.2}
\begin{split}
\|\rd_t \rd_i N\|_{L^{\f{2}{1+s'-s''}}_{-s'+s''-2\alp}(\Sigma_t)} \ls &\: \|\rd_t \rd_i \widetilde{N} \|_{L^{\f{2}{1+s'-s''}}_{-s'+s''-2\alp}(\Sigma_t)}+ |\rd_t N_{asymp}|(t) \|\la x\ra^{-1}\|_{L^{\f{2}{1+s'-s''}}_{-s'+s''-\alp}(\Sigma_t)} \\
\ls &\: \|\rd_t \widetilde{N} \|_{W^{1,\f{2}{1+s'-s''}}_{-s'+s''-2\alp}(\Sigma_t)} + \ep^{\f 54}.
\end{split}
\end{equation}

Additionally, Proposition~\ref{prop:Sobolev.weighted}, \eqref{eq:BA.g.asymp} and \eqref{eq:BA.g.Li} imply that
\begin{equation}\label{eq:dt.g.elliptic.proof.3}
\sum_{\mfg \in \{\gamma, \bt^i, N\}} \|\rd_t \mfg \|_{L^{\f{2}{1+s'-s''}}_{-1-s'+s''-\alp}(\Sigma_t)} \ls \sum_{\mfg \in \{\widetilde{\gamma}, \bt^i, \widetilde{N}\}}  \|\rd_t \widetilde{\mfg} \|_{L^{\i}(\Sigma_t)} + \ep^{\f 54} \|\log (2+|x|) \|_{L^{\f{2}{1+s'-s''}}_{-1-s'+s''-\alp}(\Sigma_t)} \ls \ep^{\f 54}.
\end{equation}

\pfstep{Step~2: Proof of \eqref{eq:dt.db.db.est}} Note that because of the signs of $\gamma_{asymp}$ and $N_{asymp}$, the factors $e^{2\gamma}$ and $\f 1N$ are favorable in terms of the $\la x \ra$ weights. Therefore, by H\"older's inequality, \eqref{eq:dt.g.elliptic.proof.1}, \eqref{eq:dt.g.elliptic.proof.3} and \eqref{eq:BA.g.Li}, we obtain
	\begin{equation*}
	\begin{split}
	&\: \|\rd_t \{ \f{e^{2\gamma}}{N^2} |\mathfrak L \bt|^2 \}\|_{L^{\f{2}{1+s'-s''}}_{-s'+s''-2\alp+2}(\Sigma_t)} \\
	\ls &\: \|\rd_t \mathfrak L \bt\|_{L^{\f{2}{1+s'-s''}}_{-s'+s''-2\alp+1}(\Sigma_t)} \| \mathfrak L \bt\|_{L^\i_{1}(\Sigma_t)} + \| \rd_t (\f{e^{2\gamma}}{N^2}) \|_{L^{\f{2}{1+s'-s''}}_{-1-s'+s''-2\alp}(\Sigma_t)} \|\mathfrak L \bt\|_{L^{\i}_{\f 32}(\Sigma_t)}^2 \\
	\ls &\: \|\rd_t \bt\|_{W^{1,\f{2}{1+s'-s''}}_{-s'+s''-2\alp}(\Sigma_t)} \| \bt\|_{W^{1,\i}_{0}(\Sigma_t)}  + \max_{\mfg\in \{\gamma, N\}} \|\rd_t\mfg\|_{L^{\f{2}{1+s'-s''}}_{-1-s'+s''-2\alp}(\Sigma_t)} \|\bt \|_{W^{1,\i}_{\f 12}(\Sigma_t)}^2\\
	\ls &\: \ep^{\f {15}4} + \ep^{\f 54}\sum_{\widetilde{\mfg} \in \{\widetilde{\gamma}, \bt^i,\widetilde{N}\}} \|\rd_t \widetilde{\mfg}\|_{W^{1,\f{2}{1+s'-s''}}_{-s'+s''-2\alp}(\Sigma_t)}.
	\end{split}
	\end{equation*}
	
\pfstep{Step~3: Proof of \eqref{eq:dt.dN.db.est}} This is similar to \eqref{eq:dt.db.db.est}, except we use also \eqref{eq:dt.g.elliptic.proof.2} and have less room in the weights in the $(\rd_i \rd_t N) (\mathfrak L \bt)_{jl}$ and $(\rd_i N) (\rd_t N) (\mathfrak L \bt)_{jl}$ terms. More precisely, 
	\begin{equation*}
	\begin{split}
	&\: \|\rd_t \{  \rd_i (\log (Ne^{-\gamma})) (\mathfrak L \bt)_{jl} \}\|_{L^{\f{2}{1+s'-s''}}_{-s'+s''-2\alp+2}(\Sigma_t)} \\
	\ls &\: ( \| \rd_i N\|_{L^\i_{1}} + \| \rd_i \gamma \|_{L^\i_{1}(\Sigma_t)} ) \| \rd_t \mathfrak L \bt\|_{L^{\f{2}{1+s'-s''}}_{-s'+s''-2\alp+1}(\Sigma_t)} + \| \rd_i \rd_t N \|_{L^{\f{2}{1+s'-s''}}_{-s'+s''-\alp}(\Sigma_t)} \|\mathfrak L \bt\|_{L^{\i}_{2-\alp}(\Sigma_t)} \\
	&\: +  \|\rd_t N \|_{L^{\f{2}{1+s'-s''}}_{-1-s'+s''-\alp}(\Sigma_t)} \| \rd_i N\|_{L^\i_1(\Sigma_t)} \|\mathfrak L \bt\|_{L^{\i}_{2-\alp}(\Sigma_t)} + \| \rd_i \rd_t \gamma \|_{L^{\f{2}{1+s'-s''}}_{-s'+s''-2\alp+1}(\Sigma_t)}  \|\mathfrak L \bt\|_{L^{\i}_{1}(\Sigma_t)} \\
	\ls &\: \ep^{\f {15}4} + \ep^{\f 54}\sum_{\widetilde{\mfg} \in \{\widetilde{\gamma}, \bt^i,\widetilde{N}\}} \|\rd_t \widetilde{\mfg}\|_{W^{1,\f{2}{1+s'-s''}}_{-s'+s''-2\alp}(\Sigma_t)}.
	\end{split}
	\end{equation*}
\end{proof}

\subsubsection{Estimates for inhomogeneous terms II: the interaction terms}\label{sec:metric.inho.2}

We now analyze the contribution coming from two different waves, say $\tphi$ and $\widetilde{\phi}_j$ (for $k \neq j$). Before we prove our main estimate (Lemma~\ref{lem:rodnianski.trick.2}), we need some preliminary observations making use of the transversality of the two singular zones (see Lemmas~\ref{lem:rodnianski.prelim.1} and \ref{lem:rodnianski.prelim.2}).

Given $k \neq j$, we now construct a polar coordinate system. Let $p \in \Sigma_t$ be the point\footnote{Note that such a point is indeed uniquely defined since $(u_k,u_j)$ forms a coordinate system.} corresponding to $u_k = u_j = 0$, and let $z = (z_1,z_2)$ be the elliptic gauge coordinates of the point $p$. Introduce the polar coordinates $(r,\vartheta)$ be the polar coordinates corresponding to the elliptic gauge coordinate system centered at $(z_1,z_2)$, with (recall our convention $(c_{k1}^\perp, c_{k2}^\perp) = (-c_{k2}, c_{k1})$)
\begin{equation}\label{eq:x.in.terms.of.polar}
x-z = r (\cos \vartheta \begin{bmatrix} c_{k1}^\perp \\  c_{k2}^\perp \end{bmatrix} + \sin \vartheta \begin{bmatrix} c_{k1} \\  c_{k2} \end{bmatrix} ).
\end{equation}
In particular, $(r=1, \vartheta = 0)$ corresponds to $x = z + \begin{bmatrix} c_{k1}^\perp \\  c_{k2}^\perp \end{bmatrix}$ in elliptic gauge coordinates. Using moreover \eqref{partialukcontrol}, one sees that $\{(r,\vartheta): \vartheta = 0\}$ is an approximation of the curve $\{x: u_k(t,x) = 0\}$.

Define $\vartheta_0 \in (-\pi, \pi)$ so that $(r=1, \vartheta = \vartheta_0)$ corresponds to $z + \begin{bmatrix} c_{j1}^\perp \\  c_{j2}^\perp \end{bmatrix}$ in elliptic gauge coordinates (recall the $\begin{bmatrix} c_{j1}^\perp \\  c_{j2}^\perp \end{bmatrix}$ has unit length by \eqref{cnormalization}). In other words, we impose
\begin{equation}\label{eq:th_0.def}
\begin{bmatrix} c_{j1}^\perp \\  c_{j2}^\perp \end{bmatrix} = \cos \vartheta_0 \begin{bmatrix} c_{k1}^\perp \\  c_{k2}^\perp \end{bmatrix} + \sin \vartheta_0 \begin{bmatrix} c_{k1} \\  c_{k2} \end{bmatrix}.
\end{equation}
Note that \eqref{cnormalization} implies $\left| \det \begin{bmatrix} c_{k1} & c_{k1}^\perp \\  c_{k2} & c_{k2}^\perp \end{bmatrix} \right| = 1$. Hence, combining this with \eqref{eq:th_0.def} and using \eqref{cangle2}, we obtain
\begin{equation}\label{eq:transversality.in.angle}
\begin{split}
|\sin\vartheta_0| = \left| \det \begin{bmatrix} 0 & \sin \vartheta_0 \\  1 & \cos \vartheta_0 \end{bmatrix} \right| = \left| \det \begin{bmatrix} c_{k1} & c_{k1}^\perp \\  c_{k2} & c_{k2}^\perp \end{bmatrix} \begin{bmatrix} 0 & \sin \vartheta_0 \\  1 & \cos \vartheta_0 \end{bmatrix} \right| = \left| \det \begin{bmatrix} -c_{k2} & -c_{j2} \\  c_{k1} & c_{j1} \end{bmatrix} \right| \geq \upkappa_0.
\end{split}
\end{equation}

The next lemma shows that the singular region from the point of view of $\tphi$ is localized in the region where $\sin\vartheta$ for the above polar coordinates system.
\begin{lemma}\label{lem:rodnianski.prelim.1}
	For $\ep>0$ sufficiently small, $u_k \notin (-\de,\de)$ in the set $\{(r,\vartheta) : r \geq 16\upkappa_0^{-1}\de,\, |\sin \vartheta| \geq \f{\upkappa_0}{8} \}$.
\end{lemma}
\begin{proof}
	Take a point $y \in \Sigma_t$ such that its $(r,\vartheta)$ coordinates satisfy $r \geq 16\upkappa_0^{-1}\de$ and $|\sin \vartheta| \geq \f{\upkappa_0}{8}$. We want to show that $|u_k(y)|\geq\de$. To this end, we integrate along the radial line $\gamma:[0,1] \to \Sigma_t$ (connecting\footnote{Here, $z$ is as defined above before the lemma, which corresponds to the center for the polar coordinates.} $z$ and $y$) given by 
	$$\gamma(s) = z + s r (\cos \vartheta \begin{bmatrix} c_{k1}^\perp \\  c_{k2}^\perp \end{bmatrix} + \sin \vartheta \begin{bmatrix} c_{k1} \\  c_{k2} \end{bmatrix} ).$$
	To proceed, we use the fundamental theorem of calculus and \eqref{partialukcontrol} to obtain
	\begin{equation}\label{eq:lower.bound.u.angular}
	\begin{split}
	&\: u_k(y) = u_k (\gamma(1))  = u_k(\gamma(1)) - u_k(\gamma(0))  = \int_0^1 \f{d}{ds} u_k(\gamma(s)) \, ds \\
	= &\: \int_0^1 \{ r [ -(\cos \vartheta) c_{k2} + (\sin \vartheta) c_{k1}]  (\rd_1 u_k)(\gamma(s)) +  r [ (\cos \vartheta) c_{k1} + (\sin \vartheta) c_{k2}]  (\rd_2 u_k)(\gamma(s)) \}\, \ud s \\
	= &\: r \int_0^1 (\sin\vartheta)(c_{k1}^2 + c_{k2}^2)\,\ud s+ O(\ep^{\f 54} r) = r\sin\vartheta + O(\ep^{\f 54} r),
	\end{split}
	\end{equation}
	since $c_{k1}^2 + c_{k2}^2 = 1$ by \eqref{cnormalization}. By \eqref{eq:lower.bound.u.angular}, it is clear that since $r\geq 16\upkappa_0^{-1}\de$ and $|\sin \vartheta| \geq \f{\upkappa_0}{8}$, if we choose $\ep$ to be sufficiently small, then $|u_k(y)| \geq \f 12 r |\sin\vartheta| \geq \f 12 \cdot 16\upkappa_0^{-1} \de \cdot \f{\upkappa_0}{8} = \de$, as desired. \qedhere

\end{proof}

The following lemma is related to Lemma~\ref{lem:rodnianski.prelim.1}, but adapted for $u_j$.

\begin{lemma}\label{lem:rodnianski.prelim.2}
	For $\ep>0$ sufficiently small, $u_j \notin (-\de,\de)$ in the set $\{(r,\vartheta) : r \geq 16\upkappa_0^{-1}\de,\, |\sin \vartheta| \leq \f {\upkappa_0}{4} \}$.
\end{lemma}
\begin{proof}
	In an entirely analogous manner as Lemma~\ref{lem:rodnianski.prelim.1}, we can show that 
	\begin{equation}\label{eq:condition.for.uj.away.from.singular}
	\mbox{if $r \geq 16\upkappa_0^{-1} \de$ and $|\sin(\vartheta - \vartheta_0)|\geq \f{\upkappa_0}{8}$, then $u_j \notin (-\de,\de)$.}
	\end{equation} 
	
	Now, given a point $(r,\vartheta) \in \{(r,\vartheta) : r \geq 16\upkappa_0^{-1}\de,\, |\sin \vartheta| \leq \f {\upkappa_0}{4} \}$, we know that
	\begin{itemize}
	\item $|\sin\vartheta\cos\vartheta_0|\leq |\sin\vartheta| \leq \f{\upkappa_0}{4}$, and 
	\item $|\sin\vartheta_0\cos\vartheta| \geq \f {\upkappa_0}2$ (using \eqref{eq:transversality.in.angle} and the fact $|\sin\vartheta|\leq \f{\upkappa_0}{4} \implies |\sin\vartheta|\leq \f 12 \implies |\cos\vartheta|\geq \f 12$).
	\end{itemize}
	Therefore, $|\sin(\vartheta - \vartheta_0)| = |\sin\vartheta\cos\vartheta_0 - \sin\vartheta_0\cos\vartheta| \geq \f{\upkappa_0}{4}$. Consequently, it follows from \eqref{eq:condition.for.uj.away.from.singular} that $u_j \notin (-\de,\de)$ at the given point. \qedhere
\end{proof}

Before we control the interaction terms, we need one more simple lemma.

\begin{lemma}\label{lem:mixed.norm}
For any $k$ and any $k'\neq k$, the following estimate holds for all $t\in [0,T_B)$:
$$\|\rd^2 \tphi\|_{L^2_{u_k} L^\i_{u_{k'}}(\Sigma_t)} \ls \ep^{\f 34} \cdot \de^{-\f 12}. $$
\end{lemma}
\begin{proof}
Using the wave equation if necessary, we only need to estimate $\rd \rd_x \tphi$. Using the $1$-dimensional Sobolev embedding, we have
$$ \|\rd \rd_x\tphi \|_{L^2_{u_k} L^\i_{u_{k'}}(\Sigma_t)} \ls \|\partialukp \rd \rd_x\tphi\|_{L^2(\Sigma_t)} \ls \|\rd E_k \rd_x \tphi\|_{L^2(\Sigma_t)} \ls \ep^{\f 34} \cdot \de^{-\f 12},$$
where we have used \eqref{partialukpEX}, \eqref{bootstrapmu}, \eqref{eq:BA.g.asymp}, \eqref{eq:BA.g.Li}, Lemma~\ref{lem:L.X.E} and \eqref{anglecontrol} to compare $\partialukp$ and $E_k$, and used Lemma~\ref{dgeomvflemma} to commute $[\rd, E_k]$. Finally, we apply the bootstrap assumption \eqref{EtphiH2bootstrap}. \qedhere
\end{proof}

We are now ready to prove the main estimate for the interaction terms. 

\begin{lemma}\label{lem:rodnianski.trick.2}
	Let $\Omega(\mfg)$ be as in \eqref{eq:def.Omega.g}. Then, for $k\neq j$, there exist $t$-independent functions $\widetilde{\zeta}_{\mathrm{int}}$ and $\widetilde{\zeta}_{\mathrm{ang}}$ with $L^\i$ norms $\ls 1$ (defined precisely in the proof) such that for any $p \in [1,2)$, the following estimate holds for all $t\in [0,T_B)$:
	\begin{equation*}
	\begin{split} 
	\|\rd_t [ \Omega(\mfg) (\rd_\alp \tphi) (\rd_\sigma \widetilde{\phi}_j) ] - &\rd_i [ (1-\widetilde{\zeta}_{\mathrm{int}}) \widetilde{\zeta}_{\mathrm{ang}} \Omega(\mfg)(N X_k^i +\bt^i) (\rd_\alp \tphi) (\rd_\sigma \widetilde{\phi}_j)] \\
	&- \rd_i [ (1-\widetilde{\zeta}_{\mathrm{int}}) (1-\widetilde{\zeta}_{\mathrm{ang}}) \Omega(\mfg)(N X_j^i +\bt^i) (\rd_\alp \tphi) (\rd_\sigma \widetilde{\phi}_j)] \|_{L^{p}(\Sigma_t)} \ls \f{\ep^{\f 32}}{(2-p)^{\f 1p}}.
	\end{split}
	\end{equation*}
\end{lemma}
\begin{proof}

	\pfstep{Step~1: Defining the decomposition}
	Recall the polar coordinates $(r,\vartheta)$ in \eqref{eq:x.in.terms.of.polar}. We introduce two cut-off functions. First, define a radial cut-off function $\widetilde{\zeta}_{\mathrm{int}} = \widetilde{\zeta}_{\mathrm{int}}(r)$ be a non-negative function which $=1$ when $r \leq 16\upkappa_0^{-1}\de$ and $=0$ when $r \geq 20\upkappa_0^{-1}\de$. $\widetilde{\zeta}_{\mathrm{int}}$ can chosen so that
	\begin{equation}\label{eq:zeta.int.est}
	|\widetilde{\zeta}_{\mathrm{int}}|\ls 1,\quad |\rd_i \widetilde{\zeta}_{\mathrm{int}}| \ls \de^{-1}.
	\end{equation}
	Second, define an angular cut-off function $\widetilde{\zeta}_{\mathrm{ang}} = \widetilde{\zeta}_{\mathrm{ang}}(\vartheta)$ to be a non-negative function, smooth in $\vartheta$, which $=1$ when $|\sin\vartheta| \leq \f{\upkappa_0}{8}$ and $=0$ when $|\sin\vartheta|\geq \f {\upkappa_0}{4}$. Note that while the derivatives of $\widetilde{\zeta}_{\mathrm{ang}}$ with respect to $\vartheta$ are $\de$-independent, the derivative $\rd_i \vartheta$ is unbounded and obeys only $|\rd_i \vartheta|\ls \f 1r$. As a result, $\widetilde{\zeta}_{\mathrm{ang}}$ can only be chosen to obey the following bounds:
	\begin{equation}\label{eq:zeta.ang.est}
	|\widetilde{\zeta}_{\mathrm{ang}}| \ls 1,\quad |\rd_i \widetilde{\zeta}_{\mathrm{ang}}|(x) \ls \f 1{|x-z|}.
	\end{equation}
	
	Using the above cutoffs and \eqref{eq:dt.decompose}, for $\mathfrak I :=  \Omega(\mfg) (\rd_\alp \tphi) (\rd_\sigma \widetilde{\phi}_j)$, we can rewrite
	\begin{equation}\label{eq:rodnianski.trick.trans.decomp}
	\begin{split}
	\rd_t \mathfrak I
	= &\: \widetilde{\zeta}_{\mathrm{int}} \rd_t \mathfrak I + (1-\widetilde{\zeta}_{\mathrm{int}}) \widetilde{\zeta}_{\mathrm{ang}}\rd_t \mathfrak I + (1-\widetilde{\zeta}_{\mathrm{int}}) (1 - \widetilde{\zeta}_{\mathrm{ang}}) \rd_t \mathfrak I \\
	= &\: \underbrace{\widetilde{\zeta}_{\mathrm{int}} \rd_t \mathfrak I}_{=:I} + \underbrace{ (1-\widetilde{\zeta}_{\mathrm{int}}) \widetilde{\zeta}_{\mathrm{ang}} (N L_k + (N X_k^i +\bt^i)\rd_i) \mathfrak I}_{=:II} + \underbrace{ (1-\widetilde{\zeta}_{\mathrm{int}}) (1 - \widetilde{\zeta}_{\mathrm{ang}}) (N L_j + (N X_j^i +\bt^i)\rd_i) \mathfrak I}_{=:III}.
	\end{split}
	\end{equation}
	
	In the following steps, we consider each of terms $I$, $II$ and $III$.

	\pfstep{Step~2: The region near the interaction zone (Term $I$ in \eqref{eq:rodnianski.trick.trans.decomp})}
	The key here is to use the smallness of the interaction zone. We have
	\begin{equation}\label{eq:rodnianski.trick.int.1}
	\begin{split}
	&\:\|I\|_{L^2(\Sigma_t)} =  \| \widetilde{\zeta}_{\mathrm{int}} \rd_t [\Omg(\mfg) (\rd_\alp \tphi) (\rd_\sigma \widetilde{\phi}_j)] \|_{L^2(\Sigma_t)}\\
	\ls &\: \| \rd_t \Omg(\mfg)\|_{L^\i(\Sigma_t \cap B(0,R))} \|\rd \tphi\|_{L^\i(\Sigma_t)} \| \rd \widetilde{\phi}_j \|_{L^\i(\Sigma_t)} + \| \rd^2 \tphi\|_{L^2(\Sigma_t\cap \{x:|x-z|\ls \de\})}  \| \rd \widetilde{\phi}_j \|_{L^\i(\Sigma_t)}\\
	&\: + \| \rd \tphi\|_{L^\i(\Sigma_t)} \| \rd^2 \widetilde{\phi}_j \|_{L^2(\Sigma_t\cap \{x:|x-z|\ls \de\})},
	\end{split}
	\end{equation}
	where we have used that $\mathrm{supp} (\widetilde{\zeta}_{\mathrm{int}}) \subseteq  \{x:|x-z|\ls \de\}$.
	
	The first term in \eqref{eq:rodnianski.trick.int.1} is obviously $\ls \ep^{\f {11}4}$ using \eqref{eq:def.Omega.g} and the bootstrap assumption \eqref{BA:Li}.
	
	For the second term in \eqref{eq:rodnianski.trick.int.1}, we start by noting that by Corollary~\ref{cor:diffeo},
	$$(\sup_{y,\,y' \in \{x:|x-z|\ls \de\}} |u_k(y) - u_k(y')|) + (\sup_{y,\,y' \in \{x:|x-z|\ls \de\}} |u_j(y) - u_j(y')|) \ls \de.$$
	As a result, by Corollary~\ref{cor:diffeo}, H\"older's inequality and Lemma~\ref{lem:mixed.norm}, we have
	$$\| \rd^2 \tphi\|_{L^2(\Sigma_t\cap \{x:|x-z|\ls \de\})} \ls \|\rd^2\tphi \|_{L^2_{u_k} L^\i_{u_j}(\Sigma_t)} \|1\|_{L^\i_{u_k} L^2_{u_j} (\{x:|x-z|\ls \de\}) } \ls (\ep^{\f 34} \de^{-\f 12}) \de^{\f 12} = \ep^{\f 34}.$$
	In particular, using also the bootstrap assumption \eqref{BA:Li}, we obtain
	$$\| \rd^2 \tphi\|_{L^2(\Sigma_t\cap \{x:|x-z|\ls \de\})}  \| \rd \widetilde{\phi}_j \|_{L^\i(\Sigma_t)} \ls \ep^{\f 32}.$$
	
	The third term in \eqref{eq:rodnianski.trick.int.1} can be treated similarly as the second term so that altogether we have 
	\begin{equation}\label{eq:rodnianski.trick.int.1.final}
	\begin{split}
	\|I \|_{L^2(\Sigma_t)} = \| \widetilde{\zeta}_{\mathrm{int}} \rd_t [\Omg(\mfg) (\rd_\alp \tphi) (\rd_\sigma \widetilde{\phi}_j)] \|_{L^2(\Sigma_t)} \ls \ep^{\f 32}.
	\end{split}
	\end{equation}
	
	\pfstep{Step~3: The remaining region (Terms $II$ and $III$ in \eqref{eq:rodnianski.trick.trans.decomp})} We first consider term $II$ of \eqref{eq:rodnianski.trick.trans.decomp}. The key observation is that by Lemma~\ref{lem:rodnianski.prelim.2}, on the support of $(1-\widetilde{\zeta}_{\mathrm{int}}) \widetilde{\zeta}_{\mathrm{ang}}$, $u_j \notin [-\de,\de]$. As a result, we have $\| (1-\widetilde{\zeta}_{\mathrm{int}}) \widetilde{\zeta}_{\mathrm{ang}} \rd^2 \widetilde{\phi}_j \|_{L^2(\Sigma_t)} \ls \ep^{\f 34}$ by \eqref{BA:away.from.singular}.
	
	We now move onto the details. We write
	\begin{equation}\label{eq:rodnianski.trick.hardest.term}
	\begin{split}
	II = &\: (1-\widetilde{\zeta}_{\mathrm{int}}) \widetilde{\zeta}_{\mathrm{ang}} (N L_k + (N X_k^i +\bt^i)\rd_i) [\Omg(\mfg) (\rd_\alp \tphi) (\rd_\sigma \widetilde{\phi}_j)] \\
	=&\: \underbrace{(1-\widetilde{\zeta}_{\mathrm{int}}) \widetilde{\zeta}_{\mathrm{ang}} N L_k  [\cdots ] }_{=:II_1} + \underbrace{(1-\widetilde{\zeta}_{\mathrm{int}}) \widetilde{\zeta}_{\mathrm{ang}}  (N X_k^i +\bt^i)\rd_i [\cdots]}_{=:II_2}.
	\end{split}
	\end{equation}
	
	For $II_1$ in \eqref{eq:rodnianski.trick.hardest.term}, we compute
	\begin{equation}\label{eq:rodnianski.trick.hardest.term.I}
	\begin{split}
	II_1 = &\: \underbrace{(1-\widetilde{\zeta}_{\mathrm{int}}) \widetilde{\zeta}_{\mathrm{ang}} N \{ (L_k \Omg(\mfg)) (\rd_\alp \tphi) (\rd_\sigma \widetilde{\phi}_j) +  \Omg(\mfg) (L_k \rd_\alp \tphi) (\rd_\sigma \widetilde{\phi}_j)  \} }_{=:II_{1,1}}\\
	&\: + \underbrace{ (1-\widetilde{\zeta}_{\mathrm{int}}) \widetilde{\zeta}_{\mathrm{ang}} N\Omg(\mfg) (\rd_\alp \tphi) (L_k \rd_\sigma \widetilde{\phi}_j) }_{=:II_{1,2}}.
	\end{split}
	\end{equation}
	
	The term $II_{1,1}$ is easy, particularly because $L_k$ is a regular vector field for $\tphi$. More precisely, using \eqref{eq:def.Omega.g}, \eqref{bootstrapsmallnessenergy} and \eqref{BA:Li}, we obtain $\|II_{1,2}\|_{L^2(\Sigma_t)} \ls \ep^{\f 32}$.
	
	For $II_{1,2}$ in \eqref{eq:rodnianski.trick.hardest.term.I}, the key is that Lemma~\ref{lem:rodnianski.prelim.2} implies that $\mathrm{supp}(II_{1,2}) \subseteq \Sigma_t \setminus S^j_\de$. Therefore, we use \eqref{eq:def.Omega.g}, \eqref{BA:away.from.singular} and \eqref{BA:Li} to obtain $\|II_{1,2}\|_{L^2(\Sigma_t)} \ls \ep^{\f 32}$.

	For $II_2$ in \eqref{eq:rodnianski.trick.hardest.term}, we compute
	\begin{equation}
	\begin{split}
	II_2 =&\: \underbrace{\rd_i  [(1-\widetilde{\zeta}_{\mathrm{int}}) \widetilde{\zeta}_{\mathrm{ang}}  (N X_k^i +\bt^i) \Omg(\mfg) (\rd_\alp\tphi)(\rd_\sigma\widetilde{\phi}_{j})]}_{=:II_{2,1}}  - \underbrace{ (1-\widetilde{\zeta}_{\mathrm{int}}) \widetilde{\zeta}_{\mathrm{ang}} [\rd_i(N X_k^i +\bt^i)] \Omg(\mfg) (\rd_\alp\tphi)(\rd_\sigma\widetilde{\phi}_{j}) }_{=:II_{2,2}}\\
	&\: + \underbrace{ (\rd_i \widetilde{\zeta}_{\mathrm{int}}) \widetilde{\zeta}_{\mathrm{ang}} (N X_k^i +\bt^i) \Omg(\mfg) (\rd_\alp\tphi)(\rd_\sigma\widetilde{\phi}_{j}) }_{=:II_{2,3}} - \underbrace{ (1-\widetilde{\zeta}_{\mathrm{int}}) (\rd_i\widetilde{\zeta}_{\mathrm{ang}}) (N X_k^i +\bt^i) \Omg(\mfg) (\rd_\alp\tphi)(\rd_\sigma\widetilde{\phi}_{j}) }_{=:II_{2,4}}.
	\end{split}
	\end{equation}
	
	$II_{2,1}$ is one of the main terms we have in the statement of the lemma. $II_{2,2}$ can be handled just as term $II$ in \eqref{eq:rodnianski.trick.1.3} so that $\|II_{2,2}\|_{L^2(\Sigma_t)} \ls \ep^{\f 32}$ by \eqref{eq:rodnianski.trick.1.3.2}. The term $II_{2,3}$ has $L^\i$ norm $\ls \ep^{\f 32} \de^{-1}$ (using \eqref{eq:BA.g.Li}, \eqref{BA:Li}, Lemma~\ref{lem:L.X.E}, \eqref{eq:def.Omega.g}, \eqref{eq:zeta.int.est} and \eqref{eq:zeta.ang.est}), but $\rd_i \widetilde{\zeta}_{\mathrm{int}}$ is supported in $\{|x-z|\ls \de\}$. Thus, using H\"older's inequality,
	$$\|II_{2,3} \|_{L^2(\Sigma_t)} \ls \|II_{2,3}\|_{L^\i(\Sigma_t)} \|1 \|_{L^2(\Sigma_t\cap \{|x-z|\ls \de\})} \ls \ep^{\f 32} \de^{-1} (\de^2)^{\f 12} \ls \ep^{\f 32}.$$
	Now $II_{2,4}$ is compactly supported in $B(0,R)$, and is bounded in $L^\i$ above by $\ls \f{\ep^{\f 32}}{|x-z|}$ (by \eqref{eq:BA.g.Li}, \eqref{BA:Li}, Lemma~\ref{lem:L.X.E}, \eqref{eq:def.Omega.g}, \eqref{eq:zeta.int.est} and \eqref{eq:zeta.ang.est}). It follows that\footnote{Note that $\f{1}{|x-z|} \notin L^2_{loc}$ in two dimensions.} for $p\in [1,2)$,
	$$\|II_{2,4} \|_{L^p(\Sigma_t)} \ls \ep^{\f 32}(\int_{B(0,R)} \, \f{\ud x}{|x-z|^p})^{\f 1p} \ls \f{\ep^{\f 32}}{(2-p)^{\f 1p}}. $$
	
	Putting all the estimates above (using also that the $L^2$ norm controls the $L^p$ norm since the support of each term $\subseteq B(0,R)$), we obtain that for every $p\in [1.2)$,
	\begin{equation}\label{eq:rodnianski.trick.int.2}
	\|II - \rd_i  [(1-\widetilde{\zeta}_{\mathrm{int}}) \widetilde{\zeta}_{\mathrm{ang}}  (N X_k^i +\bt^i) \Omg(\mfg) (\rd_\alp\tphi)(\rd_\sigma\widetilde{\phi}_{j})]\|_{L^p(\Sigma_t)} \ls \f{\ep^{\f 32}}{(2-p)^{\f 1p}}.
	\end{equation}
	
	Finally, term $III$ in \eqref{eq:rodnianski.trick.trans.decomp} can be treated in a similar way as term $II$, except we use Lemma~\ref{lem:rodnianski.prelim.1} instead of Lemma~\ref{lem:rodnianski.prelim.2}. Hence, we have
	\begin{equation}\label{eq:rodnianski.trick.int.3}
	\|III - \rd_i [ (1-\widetilde{\zeta}_{\mathrm{int}}) (1-\widetilde{\zeta}_{\mathrm{ang}}) \Omega(\mfg)(N X_j^i +\bt^i) (\rd_\alp \tphi) (\rd_\sigma \widetilde{\phi}_j)] \|_{L^p(\Sigma_t)}  \ls \f{\ep^{\f 32}}{(2-p)^{\f 1p}};
	\end{equation}
	we omit the details. 
	
	Combining \eqref{eq:rodnianski.trick.trans.decomp}, \eqref{eq:rodnianski.trick.int.1.final}, \eqref{eq:rodnianski.trick.int.2} and \eqref{eq:rodnianski.trick.int.3} yields the lemma. \qedhere
\end{proof}

\subsubsection{The main weighted $\protect W^{1,\f 2{s'-s''}}$ estimates for $\rd_t \mfg$}\label{sec:weighted.rdtg}

\begin{proposition}\label{prop:elliptic.dtg}
	Decomposing $\gamma, \bt^i, N$ as in \eqref{eq:metric.decomposition}, the following estimates hold for all $t\in [0,T_B)$:
	\begin{equation}\label{eq:dtg.main}
	|\rd_t N_{asymp}|(t) + \sum_{\widetilde{\mfg} \in \{\widetilde{\gamma}, \bt^i,\widetilde{N}\}}  \|\rd_t \widetilde{\mfg} \|_{W^{1,\f 2{s'-s''}}_{1-s'+s''-2\alp}(\Sigma_t)} \ls \ep^{\f 32}.
	\end{equation}
	Using also Proposition~\ref{prop:Sobolev.weighted}, it follows moreover that
	\begin{equation}\label{eq:dtg.improve.bootstrap}
	|\rd_t N_{asymp}|(t) + \sum_{\widetilde{\mfg} \in \{\widetilde{\gamma}, \bt^i,\widetilde{N}\}}  \|\rd_t \widetilde{\mfg} \|_{L^\i_{1-2\alp}(\Sigma_t)} \ls \ep^{\f 32}.
	\end{equation}
\end{proposition}
\begin{proof}
	The fact that $\gamma, \bt^i, N$ admit the decomposition \eqref{eq:metric.decomposition}, and that $\gamma_{asymp}$ being a constant (and hence $\rd_t \gamma_{asymp} = 0$), is again a consequence of the local existence result in \cite[Theorem~5.4]{HL.elliptic}. From now on, we focus on deriving the estimates using \eqref{Nellipticequation}--\eqref{betaellipticequation}.
	
	\pfstep{Step~1: Decomposition of $\Delta \rd_t \mfg$} Differentiating \eqref{Nellipticequation}--\eqref{betaellipticequation} by $\rd_t$, we obtain, for $\mfg \in \{\gamma, \bt^i, N\}$,
	\begin{equation}\label{eq:eqn.for.dtg}
	\Delta \rd_t \mfg = \mathfrak G_{\mfg} + \mathfrak T_{\mfg},
	\end{equation}
	where $\mathfrak G_{\mfg}$ are the metric terms, given explicitly by
	\begin{equation}\label{eq:def.mfG_mfg}
	\mathfrak G_\gamma := -\rd_t [\f{e^{2\gamma}}{8N} |\mathfrak L\bt|^2],\quad \mathfrak G_N := \rd_t[\f{e^{2\gamma}}{4N} |\mathfrak L\bt|^2],\quad \mathfrak G_{\bt^j} := 2 \rd_t[ \delta^{i k } \delta^{j l } \partial_k (\log (N  e^{-2\gamma})) (\mathfrak L \bt)_{i l }] ;
	\end{equation}
	and, for any $\mfg$, $\mathfrak T_{\mfg}$ takes the schematic form $\mathfrak T_{\mfg} = \rd_t \{ \Omg(\mfg) \rd_\alp \phi \rd_\sigma\phi \}$.
	
	By \eqref{eq:metric.Delta-1.def} and \eqref{eq:eqn.for.dtg}, 
	\begin{equation}\label{eq:eqn.for.dtg.inverted}
	\rd_t\mfg = \Delta^{-1} (\mathfrak G_{\mfg} + \mathfrak T_{\mfg}).
	\end{equation}
	
	\pfstep{Step~1(a): The metric term $\mathfrak G_{\mfg}$ in \eqref{eq:eqn.for.dtg.inverted}} By Lemma~\ref{lem:dt.dg.dg.est}, the terms in \eqref{eq:def.mfG_mfg} can be bounded as follows:
	\begin{equation}\label{eq:dtg.est.4}
	\sum_{\mfg\in \{\gamma, \bt^i, N\}} \|\mathfrak G_\mfg \|_{L^{\f{2}{1+s'-s''}}_{-s'+s''-2\alp+2}(\Sigma_t)} \ls \ep^{\f {15}4} + \ep^{\f 54}\sum_{\widetilde{\mfg} \in \{\widetilde{\gamma}, \bt^i,\widetilde{N}\}} \|\rd_t \widetilde{\mfg}\|_{W^{1,\f{2}{1+s'-s''}}_{-s'+s''-2\alp}(\Sigma_t)}.
	\end{equation}
	
	\pfstep{Step~1(b): The scalar field term $\mathfrak T_{\mfg}$ in \eqref{eq:eqn.for.dtg.inverted}} Expanding 
	$$\mathfrak T_{\mfg} = \rd_t \{ \Omg(\mfg) \rd_\alp \phi \rd_\sigma\phi \} = \rd_t \{ \Omg(\mfg) (\rd_\alp \rphi + \sum_{k=1}^3 \rd_\alp \tphi) (\rd_\sigma\rphi + \sum_{\ell=1}^3 \rd_\sigma \widetilde{\phi}_\ell \},$$ 
	and using Lemmas~\ref{lem:rodnianski.trick.no.need}, \ref{lem:rodnianski.trick.1} and \ref{lem:rodnianski.trick.2}, we obtain a decomposition
	\begin{equation}\label{eq:dtg.est.1}
	\mathfrak T_{\mfg} =  \mathfrak F_{\mfg} + \rd_i \mathfrak H_{\mfg}^i,
	\end{equation}
	where $\mathfrak F_{\mfg}$ and $\mathfrak H_{\mfg}^i$ ($i=1,2$) are smooth and compactly supported in $B(0,R)$ (for each $t$) and 
	\begin{equation}\label{eq:dtg.est.2}
	\| \mathfrak F_{\mfg}\|_{L^{\f 2{1+s'-s''}}(\Sigma_t)} \ls \ep^{\f 32},\quad \| \mathfrak H_{\mfg}^i\|_{L^\i(\Sigma_t)} \ls \ep^{\f 32}.
	\end{equation} 
	
	\pfstep{Step~2: Bounding $\rd_t N_{asymp}$} By \eqref{eq:eqn.for.dtg.inverted} and Proposition~\ref{prop:basic.elliptic.2}, $\rd_t N_{asymp} = \f 1{2\pi} \int_{\Sigma_t} (\mathfrak F_N + \mathfrak G_N + \rd_i \mathfrak H^i_N) \, dx = \f 1{2\pi} \int_{\Sigma_t} (\mathfrak F_N + \mathfrak G_N) \, dx$ (since $\mathrm{supp}(\mathfrak H^i_N)\subseteq B(0,R)$). Hence, by part 2 of Proposition~\ref{prop:Sobolev.weighted}, \eqref{eq:dtg.est.4} and \eqref{eq:dtg.est.2},
	\begin{equation}\label{eq:dtg.est.5}
	\begin{split}
	|\rd_t N_{asymp}|(t) = &\: \left| \f 1{2\pi} \int_{\Sigma_t} (\mathfrak F_N + \mathfrak G_N ) \, dx \right| \ls \| \mathfrak F_N \|_{L^{\f 2{1+s'-s''}}_{-s'+s''+\alp+1}(\Sigma_t)} + \|\mathfrak G_N \|_{L^{\f 2{1+s'-s''}}_{-s'+s''+\alp+1}(\Sigma_t)} \\
	\ls &\: \ep^{\f 32} + \ep^{\f 54}\sum_{\widetilde{\mfg} \in \{\widetilde{\gamma}, \bt^i,\widetilde{N}\}} \|\rd_t \widetilde{\mfg}\|_{W^{1,\f{2}{1+s'-s''}}_{-s'+s''-2\alp}(\Sigma_t)}.
	\end{split}
	\end{equation}
	
	\pfstep{Step~3: Bounding $\Delta^{-1}(\mathfrak F_{\mfg} + \mathfrak G_{\mfg})$} Using the obvious notation $\rd_t\mfg_{asymp} = \rd_t N_{asymp}$ for $\mfg = N$, and $\rd_t \mfg_{asymp} = 0$ for $\mfg \in \{ \gamma, \bt^i \}$. By Proposition~\ref{prop:basic.elliptic.2}, 	
	\begin{equation}\label{eq:dtg.est.6}
	\begin{split}
	&\: \sum_{\mfg\in \{\gamma,\bt^i,N\}} \| \Delta^{-1} ( \mathfrak F_{\mfg} + \mathfrak G_\mfg) - \rd_t \mfg_{asymp}(t) \om(|x|)\log|x| \|_{W^{2,\f{2}{1+s'-s''}}_{-s'+s''-2\alp}(\Sigma_t)} \\
	\ls &\: \|\mathfrak F_\mfg + \mathfrak G_\mfg \|_{L^{\f{2}{1+s'-s''}}_{-s'+s''-2\alp+2}(\Sigma_t)} 
	\ls \ep^{\f 32} + \ep^{\f 54}\sum_{\widetilde{\mfg} \in \{\widetilde{\gamma}, \bt^i,\widetilde{N}\}} \|\rd_t \widetilde{\mfg}\|_{W^{1,\f{2}{1+s'-s''}}_{-s'+s''-2\alp}(\Sigma_t)},
	\end{split}
	\end{equation}
	where we have used \eqref{eq:dtg.est.2}, the support properties of $\mathfrak F_{\mfg}$, and \eqref{eq:dtg.est.4}.
	
	Sobolev embedding (1(b) of Proposition~\ref{prop:Sobolev.weighted}) applied to \eqref{eq:dtg.est.6} gives additionally that
	\begin{equation}\label{eq:dtg.est.7}
	\begin{split}
	&\: \|\Delta^{-1} ( \mathfrak F_{\mfg} + \mathfrak G_\mfg) - \rd_t \mfg_{asymp}(t) \om(|x|)\log|x| \|_{W^{1,\f{2}{s'-s''}}_{1-s'+s''-2\alp}(\Sigma_t)}\\
	\ls &\: \ep^{\f 32} + \ep^{\f 54}\sum_{\widetilde{\mfg} \in \{\widetilde{\gamma}, \bt^i,\widetilde{N}\}} \|\rd_t \widetilde{\mfg}\|_{W^{1,\f{2}{1+s'-s''}}_{-s'+s''-2\alp}(\Sigma_t)}.
	\end{split}
	\end{equation}
	
	\pfstep{Step~3: Bounding $\Delta^{-1} \rd_t \mathfrak H_{\mfg}^i$} Using Proposition~\ref{prop:elliptic.1st.der}, \eqref{eq:dtg.est.2}, $\mathrm{supp}(\mathfrak H_{\gamma}^i)\subseteq B(0,R)$, as well as part 2 of Proposition~\ref{prop:Sobolev.weighted}, we have
	\begin{equation}\label{eq:dtg.est.8}
	\sum_{\mfg\in \{\gamma,\bt^j,N\}} \|\rd_i \Delta^{-1} \mathfrak H^i_\mfg\|_{W^{1,\f{2}{1+s'-s''}}_{-s'+s''-2\alp}(\Sigma_t)} \ls \sum_{\mfg\in \{\gamma,\bt^j,N\}} \|\rd_i \Delta^{-1} \mathfrak H^i_\mfg \|_{W^{1,\f{2}{s'-s''}}_{1-s'+s''-\alp}(\Sigma_t)}\ls \ep^{\f 32}.
	\end{equation}
	\pfstep{Step~4: Obtaining the $W^{1,\f{2}{1+s'-s''}}_{-s'+s''-2\alp}(\Sigma_t)$ estimates} We now combine \eqref{eq:eqn.for.dtg.inverted}, \eqref{eq:dtg.est.1}, \eqref{eq:dtg.est.5}, \eqref{eq:dtg.est.6} and (first term in) \eqref{eq:dtg.est.8} to obtain\footnote{The reader may have noted that we have not used the estimate \eqref{eq:dtg.est.7} proven above. It will be used in Step~5 below.}
	\begin{equation}\label{eq:easy.elliptic.dtg.almost.final}
	|\rd_t N_{asymp}|(t) + \sum_{\widetilde{\mfg} \in \{\widetilde{\gamma}, \bt^i,\widetilde{N}\}}  \|\rd_t \widetilde{\mfg} \|_{W^{1,\f{2}{1+s'-s''}}_{-s'+s''-2\alp}(\Sigma_t)} \ls \ep^{\f 32} + \ep^{\f 54}\sum_{\widetilde{\mfg} \in \{\widetilde{\gamma}, \bt^i,\widetilde{N}\}} \|\rd_t \widetilde{\mfg}\|_{W^{1,\f{2}{1+s'-s''}}_{-s'+s''-2\alp}(\Sigma_t)}.
	\end{equation}
	Choosing $\ep$ smaller if necessary, we can absorb the term $\ep^{\f 54}\sum_{\widetilde{\mfg} \in \{\widetilde{\gamma}, \bt^i,\widetilde{N}\}} \|\rd_t \widetilde{\mfg}\|_{W^{1,\f{2}{1+s'-s''}}_{-s'+s''-2\alp}(\Sigma_t)}$ on the RHS of \eqref{eq:easy.elliptic.dtg.almost.final} by the corresponding term on the LHS, giving
	\begin{equation}\label{eq:easy.elliptic.dtg.almost.almost.final}
	|\rd_t N_{asymp}|(t) + \sum_{\widetilde{\mfg} \in \{\widetilde{\gamma}, \bt^i,\widetilde{N}\}}  \|\rd_t \widetilde{\mfg} \|_{W^{1,\f{2}{1+s'-s''}}_{-s'+s''-2\alp}(\Sigma_t)} \ls \ep^{\f 32}.
	\end{equation}
	
	\pfstep{Step~5: Obtaining the $W^{1,\f{2}{s'-s''}}_{1-s'+s''-2\alp}(\Sigma_t)$ estimates} Plugging \eqref{eq:easy.elliptic.dtg.almost.almost.final} into \eqref{eq:dtg.est.7}, and combining it with \eqref{eq:eqn.for.dtg.inverted}, \eqref{eq:dtg.est.1}, and (second term in) \eqref{eq:dtg.est.8}, we thus obtain the desired estimate \eqref{eq:dtg.main}. \qedhere

\end{proof}

\subsection{Elliptic estimates for $\Db^{s'}\rd_i$ derivatives of $\rd_t\mfg$}\label{sec:dtg.elliptic.top}

\begin{proposition}
	Let $\mathcal P_0$ be a cutoff in frequency (corresponding to the elliptic gauge coordinates) to $|\xi| \ls 1$. Then, for every $t\in [0,T_B)$,
	\begin{equation}\label{eq:DsdidtN.main}
	\|\Db^{s'} [\rd_i \rd_t N - (\rd_t N_{asymp})(t) \mathcal P_0  \rd_i (\om (|x|) \log |x|)] \|_{L^2(\Sigma_t)} \ls \ep^{\f 32},
	\end{equation}
	\begin{equation}\label{eq:Dsdidtgammabt.main}
	\|\Db^{s'} \rd_i \rd_t \gamma \|_{L^2(\Sigma_t)} + \|\Db^{s'} \rd_i \rd_t \bt \|_{L^2(\Sigma_t)} \ls \ep^{\f 32}.
	\end{equation}
\end{proposition}
\begin{proof}
	We only prove \eqref{eq:DsdidtN.main} since it features a low-frequency correction which is not in $L^2$ (coming from $\rd_tN_{asymp}$ potentially non-vanishing). The estimate \eqref{eq:Dsdidtgammabt.main} is similar but slightly simpler; we omit the details.
	
	By \eqref{eq:eqn.for.dtg.inverted} and \eqref{eq:dtg.est.1}, we write
	\begin{equation}\label{eq:DsdidtN.1}
	\begin{split}
	&\: \Db^{s'} \rd_i \rd_t  N -\Db^{s'} \mathcal P_0  [(\rd_t N_{asymp})\rd_i (\om (|x|) \log |x|) ]\\
	= &\: \Db^{s'} \mathcal P_0 \rd_i \rd_t  N -\Db^{s'} \mathcal P_0 [  (\rd_t N_{asymp})\rd_i (\om (|x|) \log |x|)] +  \Db^{s'} (I -\mathcal P_0) \rd_i \rd_t  N\\
	= &\: \underbrace{ \Db^{s'} \mathcal P_0 [\rd_i \Delta^{-1} (\mathfrak F_N + \mathfrak G_N + \rd_j \mathfrak H_N^j) - (\rd_t N_{asymp})\rd_i (\om (|x|) \log |x|)]}_{=:I} + \underbrace{ \Db^{s'} (I - \mathcal P_0) \rd_i \rd_t  N}_{=:II}.
	\end{split}
	\end{equation}
	
	For $I$, we use bounded frequency, i.e.~the fact $\Db^{s'} \mathcal P_0: L^2(\Sigma_t) \to L^2(\Sigma_t)$ is bounded, H\"older's inequality and \eqref{eq:dtg.main} to obtain
	\begin{equation}\label{eq:DsdidtN.2}
	\begin{split}
	\| I \|_{L^2(\Sigma_t)} \ls &\: \| \rd_i \Delta^{-1} (\mathfrak F_N + \mathfrak G_N + \rd_j \mathfrak H_N^j) - (\rd_t N_{asymp})\rd_i (\om (|x|) \log |x|) \|_{L^2(\Sigma_t)} \\
	\ls &\:  \|\rd_i \rd_t \widetilde{N} \|_{L^2(\Sigma_t)} \ls \| \la x\ra^{2-s'+s''-\alp} \rd_i \rd_t \widetilde{N} \|_{L^{\f 2{s'-s''}}(\Sigma_t)} \|\la x\ra^{-2+s'-s''+\alp} \|_{L^{\f{2}{1-s'+s''}}(\Sigma_t)} \ls \ep^{\f 32}.
	\end{split}
	\end{equation}

	For $II$ in \eqref{eq:DsdidtN.1}, we use that the frequency is bounded away from $0$ so that by Plancherel's theorem,
	\begin{equation}\label{eq:DsdidtN.3}
	\|II\|_{L^2(\Sigma_t)} = \|\Db^{s'} (I - \mathcal P_0) \rd_i \rd_t  N \|_{L^2(\Sigma_t)} \ls \| \Db^{-1+s'} \Delta \rd_t N \|_{L^2(\Sigma_t)}.
	\end{equation}
	
	The remaining of the proof concerns bounding \eqref{eq:DsdidtN.3}. First, by \eqref{eq:eqn.for.dtg} and \eqref{eq:dtg.est.1}, Sobolev embedding ($\Db^{-1+s'}: L^{\f 32}(\Sigma_t) \to L^2(\Sigma_t)$ is bounded) and Plancherel's theorem,
	\begin{equation}\label{eq:DsdidtN.4}
	\| \Db^{-1+s'} \Delta \rd_t N \|_{L^2(\Sigma_t)} \ls \|\mathfrak F_N\|_{L^{\f 32}(\Sigma_t)} + \|\mathfrak G_N\|_{L^{\f 32}(\Sigma_t)} + \max_i \|\Db^{s'} \mathfrak H_N^i\|_{L^2(\Sigma_t)}.
	\end{equation}
	
	The $\mathfrak F_N$ and $\mathfrak G_N$ terms are easier. Since $\mathrm{supp}(\mathfrak F_N)\subseteq B(0,R)$, by H\"older's inequality and \eqref{eq:dtg.est.2},
	\begin{equation}\label{eq:DsdidtN.5}
	\|\mathfrak F_N\|_{L^{\f 32}(\Sigma_t)} \ls \ep^{\f 32}.
	\end{equation}
	Using H\"older's inequality, $s'-s'' < \f 13$, \eqref{eq:dtg.est.4} and \eqref{eq:easy.elliptic.dtg.almost.almost.final}, we also have
	\begin{equation}\label{eq:DsdidtN.6}
	\|\mathfrak G_N \|_{L^{\f 32}(\Sigma_t)} \ls \|\mathfrak G_N \|_{L^{\f{2}{1+s'-s''}}_{-s'+s''-2\alp+2}} \|\la x\ra^{s'-s''+2\alp-2}\|_{L^{\f 6{1-3s'+3s''}}} \ls \|\mathfrak G_N \|_{L^{\f{2}{1+s'-s''}}_{-s'+s''-2\alp+2}} \ls \ep^{\f 32}.
	\end{equation}
	
	To handle $\mathfrak H_N^i$, we need a more explicit form of $\mathfrak H_N^i$. Going back to Lemmas~\ref{lem:rodnianski.trick.no.need}, \ref{lem:rodnianski.trick.1} and \ref{lem:rodnianski.trick.2}, we see that schematically $\mathfrak H_N^i$ takes one of the following four forms
	\begin{equation}\label{eq:DsdidtN.7}
	\begin{split}
	\Omega(\mfg)(N X_k^i +\bt^i) (\rd_\alp \tphi) (\rd_\sigma \tphi),\quad \Omega(\mfg)(N X_k^i +\bt^i) (\rd_\alp \tphi) (\rd_\sigma \phi_{reg}),\\
	\underbrace{(1-\widetilde{\zeta}_{\mathrm{int}}) \widetilde{\zeta}_{\mathrm{ang}} \Omega(\mfg)(N X_k^i +\bt^i) (\rd_\alp \tphi) (\rd_\sigma \widetilde{\phi}_j)}_{=:*},\quad (1-\widetilde{\zeta}_{\mathrm{int}}) (1-\widetilde{\zeta}_{\mathrm{ang}}) \Omega(\mfg)(N X_j^i +\bt^i) (\rd_\alp \tphi) (\rd_\sigma \widetilde{\phi}_j).
	\end{split}
	\end{equation}
	
	They can all be handled similarly; with the last two terms being slightly harder due to the cutoffs $\widetilde{\zeta}_{\mathrm{int}}$ and $\widetilde{\zeta}_{\mathrm{ang}}$. We will take the $*$ term as an example. We first handle the fractional derivatives of the cutoffs. First, by interpolation, \eqref{eq:zeta.int.est} and the support of $\rd_x\widetilde{\zeta}$,
	\begin{equation}\label{eq:DsdidtN.cutoff.1}
	\| \Db^{s'} [\varpi (1-\widetilde{\zeta}_{\mathrm{int}})] \|_{L^2(\Sigma_t)} \ls \|\varpi (1-\widetilde{\zeta}_{\mathrm{int}})\|_{H^1(\Sigma_t)} \ls 1.
	\end{equation}
	Also, by Sobolev embedding and \eqref{eq:zeta.ang.est},
	\begin{equation}\label{eq:DsdidtN.cutoff.2}
	\| \Db^{s'} (\varpi \widetilde{\zeta}_{\mathrm{ang}}) \|_{L^2(\Sigma_t)} \ls \|\varpi \widetilde{\zeta}_{\mathrm{ang}} \|_{W^{1,\f 32}(\Sigma_t)}  \ls 1.
	\end{equation}
	
	Notice now that since $\mathrm{supp}(\phi)\subseteq B(0,R)$, we have $* = \varpi^2 *$. Therefore, by repeated applications of Lemma~\ref{lem:frac.product}, we have
	\begin{equation}\label{eq:DsdidtN.8}
	\begin{split}
	\|\Db^{s'} (*)\|_{L^2(\Sigma_t)} \ls &\: \|\Db^{s'} [\varpi (1-\widetilde{\zeta}_{\mathrm{int}})]\|_{L^2(\Sigma_t)} \|\varpi \widetilde{\zeta}_{\mathrm{ang}}\|_{L^\i(\Sigma_t)} \|\rd \tphi\|_{L^\i(\Sigma_t)} \|\rd \widetilde{\phi}_j\|_{L^\i(\Sigma_t)} \\
	&\: + \| \varpi (1-\widetilde{\zeta}_{\mathrm{int}}) \|_{L^\i(\Sigma_t)} \|\Db^{s'}(\varpi \widetilde{\zeta}_{\mathrm{ang}}) \|_{L^2(\Sigma_t)} \|\rd \tphi\|_{L^\i(\Sigma_t)} \|\rd \widetilde{\phi}_j\|_{L^\i(\Sigma_t)} \\
	&\: + \| \varpi (1-\widetilde{\zeta}_{\mathrm{int}}) \|_{L^\i(\Sigma_t)} \|\varpi \widetilde{\zeta}_{\mathrm{ang}}\|_{L^\i(\Sigma_t)} \|\rd \Db^{s'} \tphi\|_{L^2(\Sigma_t)} \|\rd \widetilde{\phi}_j\|_{L^\i(\Sigma_t)} \\
	&\: + \| \varpi (1-\widetilde{\zeta}_{\mathrm{int}}) \|_{L^\i(\Sigma_t)} \|\varpi \widetilde{\zeta}_{\mathrm{ang}}\|_{L^\i(\Sigma_t)} \|\rd \tphi\|_{L^\i(\Sigma_t)} \|\rd \Db^{s'} \widetilde{\phi}_j\|_{L^2(\Sigma_t)} \ls \ep^{\f 32},
	\end{split}
	\end{equation}
	where in the last estimate we have used \eqref{eq:DsdidtN.cutoff.1}, \eqref{eq:DsdidtN.cutoff.2}, and the bootstrap assumptions \eqref{tphiH3/2bootstrap} and \eqref{BA:Li}.
	
	We can handle the other terms in \eqref{eq:DsdidtN.7} in a similar manner as \eqref{eq:DsdidtN.8}, so that when combined with \eqref{eq:DsdidtN.3}, \eqref{eq:DsdidtN.4}, \eqref{eq:DsdidtN.5} and \eqref{eq:DsdidtN.6}, we obtain the following bound for $II$ in \eqref{eq:DsdidtN.1}:
	\begin{equation}\label{eq:DsdidtN.9}
	\|II\|_{L^2(\Sigma_t)} \ls \ep^{\f 32}.
	\end{equation}
	
	Finally, combining \eqref{eq:DsdidtN.1}, \eqref{eq:DsdidtN.2} and \eqref{eq:DsdidtN.9}, we obtain \eqref{eq:DsdidtN.main}. \qedhere

\end{proof}

 \subsection{Estimate of three derivatives of the metric}\label{sec:third.der.metric}
 
 Our final elliptic estimate concerns third derivatives of the metric; see Proposition~\ref{dnablanablagwithloss.prop}. Note that
 \begin{enumerate}
 \item the estimate allows for at most one $\rd_t$ derivative, and
 \item the bound blows up as $\de\to 0$.
 \end{enumerate}
 \begin{prop} \label{dnablanablagwithloss.prop}
 The following estimate holds for all $t\in [0,T_B)$:
 	\begin{equation} \label{dnablanablagwithloss.eq}
 	 \sum_{\mfg \in \{ \gamma,\,\bt^i,\,N\}} \| \partial \partial_x^2 \mfg \|_{L^2(\Sigma_t)} \lesssim \epsilon^{\frac{3}{2}} \cdot \delta^{-\frac{1}{2}}.
 	\end{equation}
 \end{prop}
 	\begin{proof}
 		By the $L^2$-boundedness of $\rd^2_{ij} \Delta^{-1}$, it suffices to show that $\| \rd \Delta \mfg\|_{L^2(\Sigma_t)} \ls \ep^{\f 32} \cdot \de^{-\f 12}$. 
		
		Differentiating \eqref{Nellipticequation}--\eqref{betaellipticequation} by $\rd$, it follows from \eqref{eq:BA.g.asymp}, \eqref{eq:BA.g.Li} and \eqref{BA:Li} that
		$$\sum_{\mfg \in \{\gamma, \bt^i ,N\} } |\rd \Delta\mfg| \ls \underbrace{\ep^{\f 32}\la x\ra^{-2-\alp} }_{=:I} + \underbrace{\ep^{\f 54} \sum_{\mfg \in \{\gamma, \bt^i ,N \} } \la x\ra^{-\f 12} |\rd \rd_x \mfg| }_{=:II} + \underbrace{ \ep^{\f 34} |\rd^2 \phi| }_{=:III}.$$
		We control the $L^2(\Sigma_t)$ norm of each term. Obviously, $\| I \|_{L^2(\Sigma_t)} \ls \ep^{\f 32}$. By \eqref{eq:BA.g.L4} (with Proposition~\ref{prop:Sobolev.weighted}) and Proposition~\ref{prop:elliptic.dtg}, $\|II \|_{L^2(\Sigma_t)} \ls \ep^{\f 52}$. Finally, by \eqref{BA:rphi} and \eqref{tphiH2bootstrap}, $\|III \|_{L^2(\Sigma_t)} \ls \ep^{\f 32}\cdot \de^{-\f 12}$. \qedhere
		
 	\end{proof}

\subsection{Estimate for $K$}\label{sec:easy.consequence.from.g.est}

\begin{proposition}\label{prop:K}
The following estimate holds for all $t\in [0,T_B)$:
\begin{equation}\label{Kmainestimate}
\| K \|_{L^{\infty}_{2-\alpha}(\Sigma_t)} + \| \rd_x K \|_{L^{\infty}_{2-\alpha}(\Sigma_t)} + \|\rd_t K \|_{L^{\f 2{s'-s''}}_{2-s'+s''+\alp}(\Sigma_t)} \ls \ep^{\f 32}.
\end{equation}
\end{proposition}
\begin{proof}
We use the formula \eqref{maximality3} to write $K$ in terms $\gamma$, $\bt$ and $N$. Notice that $\f{e^{2\gamma}}N$ is favorable in terms of the $\la x\ra$ weights. Hence, the estimates follow from Propositions~\ref{prop:elliptic.easy}, \ref{prop:elliptic.Besov.weighted} and \ref{prop:elliptic.dtg}. \qedhere
\end{proof}

\section{Estimates for the Ricci coefficients and related geometric quantities}\label{sec:Ricci.coeff}
 
 We continue to work under the assumptions of Theorem~\ref{thm:bootstrap.metric}.

Our goal in this section is to control the remaining geometric quantities, particularly those related to the eikonal functions $u_k$. In \textbf{Section~\ref{sec:Ricci.coefficients}}, we bound the Ricci coefficients $\chi_k$, $\eta_k$ and their derivatives. In \textbf{Section~\ref{sec:mu.Theta}}, we bound the metric coefficients $\mu_k$ and $\varTheta_k$ (in the $(u_k,\th_k,t_k)$ coordinates. Finally, in \textbf{Section~\ref{sec:second.derivative.X.E.L}}, we estimate the second derivatives of the commutation fields.
 
 \subsection{Estimates for the Ricci coefficients and their derivatives}\label{sec:Ricci.coefficients}

 In this subsection we bound the Ricci coefficients and their derivatives. estimates, which require a treatment of the quadratic interaction between two impulsive waves:
 \begin{prop}\label{prop:Ricci}
 	The following estimates hold for all $t\in [0,T_B)$ and all $u_k \in \RR$:

 \begin{align} 
 \label{ricci.pointwise.main.estimate}
\| \chi_k \| _{ L^{\infty}_{1-\alpha}(\Sigma_t)}+ \| \eta_k \| _{ L^{\infty}_{1-\alpha}(\Sigma_t)} \lesssim \ep^{\f 32},\\
 \label{nablachi.estimate}
\| \partial_x  \chi_k \|_{L^2_{\theta_k}(\Sigma_t \cap C^k_{u_k})} \lesssim \ep^{\f 32}, \\
 \label{Eeta.estimate}
\| E_k \eta_k \|_{ L^2_{\theta_k}(\Sigma_t \cap C^k_{u_k})} \lesssim \ep^{\f 32}, \\
 \label{nablaeta.estimate}
\|\rd_x \eta_k \|_{L^2(\Sigma_t \cap B(0,3R))} \ls \ep^{\f 32}.
 \end{align}

 \end{prop}
 \begin{proof}
In this proof, we prove estimates by solving transport equations and integrating along the integral curves of $L_k$. Recall in particular that in the coordinate system $(u_k,\theta_k,t_k)$,  $\partial_{t_k} = N \cdot L_k$ by \eqref{thetaE}.
 
\pfstep{Step~1: Controlling $\chi_k$ and $\eta_k$ (Proof of \eqref{ricci.pointwise.main.estimate})} Using the transport equations \eqref{Lchi}, \eqref{Leta},  the bootstrap assumptions \eqref{bootstrapK}, \eqref{bootstrapricci}, \eqref{BA:Li}, and the estimates in \eqref{Yibounded}, we obtain
\begin{equation}\label{eq:Lchi.Leta.est}
|L_k \chi_k| +  |L_k \eta_k|\ls \ep^{\f 32} \cdot \la x\ra^{-2+3\alp} \ls \ep^{\f 32} \cdot \la  x\ra^{-1+\alpha}.
\end{equation}
Note that
\begin{itemize}
\item the initial $\chi_k$ and $\eta_k$ are bounded by $\la x\ra^{-1+\alp}$ (see point 5~in the proof of Lemma~\ref{lem:local}), and
\item that $\la x\ra$ are comparable between any two points on the integral curve of $L_k$ (see Step~0 of Proposition~\ref{prop:angle}).
\end{itemize}
Hence, integrating \eqref{eq:Lchi.Leta.est} along the integral curve of $L_k$, we obtain \eqref{ricci.pointwise.main.estimate}.

\pfstep{Step~2: Controlling derivatives of $\chi_k$ (Proof of \eqref{nablachi.estimate})} 
\pfstep{Step~2(a): Preliminary reductions}
First, we commute \eqref{Lchi} with $\partial_i$, and rewrite $L_k = N^{-1} \cdot \rd_{t_k}$ (using \eqref{thetaE}): 
\begin{equation}\label{eq:L.di.chik}
\begin{split}
\rd_{t_k} \rd_i \chi_k 
= &\: \underbrace{N [L_k,\rd_i]\chi_k}_{=:A} - \underbrace{4 N([\rd_i, L_k] \phi)(L_k \phi )}_{=:B} \\
&\: - \underbrace{4 N(L_k \rd_i\phi)(L_k \phi )}_{=:D} - \underbrace{ N \rd_i (\chi^2_k - ( K( X_k,X_k) -X_k \log(N)) \cdot \chi_k)}_{=:E}.
\end{split}
\end{equation}

We first control $A$, $B$ and $E$  of \eqref{eq:L.di.chik} in the $L^2_{\th_k}(\Sigma_t\cap C_{u_k}^k)$ norm (see Definition~\ref{def:L2.in.th_k}). 

Using Lemma~\ref{dgeomvflemma}, \eqref{eq:BA.g.asymp}--\eqref{eq:BA.g.Li}, \eqref{eq:trivial.calculus}, \eqref{timeintermsofLandspace} and \eqref{eq:Lchi.Leta.est} in order, we obtain $|A| \ls \ep^{\f 54}\la x\ra^{-1+4\alp+\ep} |\rd\chi_k| \ls \ep^{\f 54}\la x\ra^{-1+4\alp+2\ep}(|L_k \chi_k| + |\rd_x \chi_k |) \ls \ep^{\f {11}4}\la x\ra^{-3+7\alp+2\ep} + \ep^{\f 54}\la x\ra^{-1+4\alp+2\ep}|\rd_x \chi_k |$. Note that using $L_k \th_k = 0$, $|(\th_k)_{|\Sigma_0}|\ls \la x\ra$ (by \eqref{thetainit}), and the comparability of $\la x\ra$ along integral curves of $L_k$, we have $\la x\ra^{-1} \ls \la \th_k\ra^{-1}$. Hence, using also \eqref{bootstrapnablaricci},
$$ \|A\|_{L^2_{\th_k}(\Sigma_t\cap C_{u_k}^k)}  \ls \ep^{\f {11}4} + \ep^{\f 54}  \| \la x \ra^{-\f 12 -\alp} \rd_x \chi_k\|_{L^2_{\th_k}(\Sigma_t\cap C_{u_k}^k)}  \ls \ep^{\f 52}.$$

Using Lemmas~\ref{lem:L.X.E} and \ref{dgeomvflemma}, $\mathrm{supp}(\phi)\subseteq B(0,R)$ and \eqref{BA:Li}, it follows easily that $|B|\ls \ep^{\f 32} \la x\ra^{-2+\alp}$. Using Lemma~\ref{lem:L.X.E}, \eqref{bootstrapK}--\eqref{bootstrapnablaricci}, \eqref{eq:BA.g.asymp}--\eqref{eq:BA.g.Li}, we have $|E| \ls \ep^{\f 52}\la x \ra^{-2+2\alp}$. Arguing as for term $A$, both $B$ and $E$ can be controlled in $L^2_{\th_k}(\Sigma_t\cap C_{u_k}^k)$ by $\ls \ep^{\f 32}$.

Combining all the above estimates, it follows that (with $D$ as in \eqref{eq:L.di.chik})
\begin{equation}\label{eq:Lk.rdi.chik}
\begin{split}
 \|\partial_{t_k} \partial_i \chi_k + D\|_{L^2_{\th_k}(\Sigma_t\cap C_{u_k}^k)} \ls \ep^{\f 32}.
\end{split}
\end{equation} 

We now turn to the term $D$ in \eqref{eq:L.di.chik} (and \eqref{eq:Lk.rdi.chik}). Using the decomposition $\phi= \sum_{q=1}^3 \widetilde{\phi}_q + \rphi$, the $L^\infty$ bootstrap assumption \eqref{BA:Li} for $\rd\phi$, and Lemma~\ref{lem:L.X.E}, we obtain the following pointwise bounds for $D$:
\begin{equation}\label{eq:Lk.rdi.chik.main.term}
|D| \ls |L_k \rd_i \rphi| + |L_k \rd_i \tphi| + \sum_{q\neq k} |L_k \partial_i \phi_{q}|.
\end{equation}

We now bound $\rd_i \chi_k$ using \eqref{eq:Lk.rdi.chik}, by first integrating along the integral curve of $t_k$ for every fixed $(u_k, \th_k)$, and then taking the $\| \la \th_k \ra^{-\f 12 - \alp} \cdot \|_{L^2_{\th_k}}$ norm. Writing in the $(u_k, \th_k, t_k)$ coordinate system, \eqref{eq:Lk.rdi.chik}, \eqref{eq:Lk.rdi.chik.main.term} and the initial data bound (obtained in part 5~of the proof of Lemma~\ref{lem:local}) imply that
\begin{equation}\label{eq:error.in.dxchi}
\begin{split}
	&\:\| \rd_i \chi_k \|_{L^2(\Sigma_t\cap C_{u_k}^k)} \\
	\ls &\: \| \rd_i \chi_k \|_{L^2(\Sigma_0\cap C_{u_k}^k)} + \| \int_0^t |D|(u_k,\cdot,t') \, dt' \|_{L^2(\Sigma_0\cap C_{u_k}^k)}\\
	\ls &\:  \ep^2 + \underbrace{\ep^{\f 34} \| \int_{0}^t |L_k \partial_i \rphi|(u_k,\cdot,t') dt' \|_{L^2((\Sigma_t\cap C_{u_k}^k))} }_{=:I} + \underbrace{ \ep^{\f 34} \| \int_{0}^t |L_k \partial_i \phi_k|(u_k,\cdot,t') dt' \|_{L^2(\Sigma_t\cap C_{u_k}^k)} }_{=:II}\\
	&\: + \underbrace{ \ep^{\f 34} \sum_{q\neq k} \| \int_{0}^t |L_k \partial_i \phi_{q}|(u_k,\cdot,t') dt' \|_{L^2(\Sigma_t\cap C_{u_k}^k)} }_{=:III}.
\end{split}
\end{equation}
We will bound the terms $I$, $II$ and $III$ in the following substeps.

\pfstep{Step~2(b): The easy terms $I$ and $II$ in \eqref{eq:error.in.dxchi}} To handle the terms $I$ and $II$ in \eqref{eq:error.in.dxchi}, we first use Minkowski's inequality in the $\th_k$ variable, and then use the Cauchy--Schwarz inequality in $t_k$ to obtain that
\begin{equation}\label{eq:error.in.dxchi.easy}
I \ls \ep^{\f 34} \|L_k \partial_i \rphi \|_{L^2(C_{u_k}^k([0,T_B)))}, \quad II \ls \ep^{\f 34} \|L_k \partial_i \tphi \|_{L^2(C_{u_k}^k([0,T_B)))}.
\end{equation}
The terms in \eqref{eq:error.in.dxchi.easy} are bounded above by $\ep^{\f 32}$ by \eqref{BA:flux.for.rphi} and \eqref{BA:flux.for.tphi.improved.2} respectively.

\pfstep{Step~2(c): The main term $III$ in \eqref{eq:error.in.dxchi}} We now turn to term $III$ in \eqref{eq:error.in.dxchi}, which is more delicate and requires the transversality of the different waves.

Fix $q\neq k$. Take a constant-$(u_k, \theta_k)$ curve (parametrized by $t_k$) which passes through the support of $\phi$ for some $t_k \in [0,1]$. Using the bootstrap assumptions \eqref{eq:BA.g.asymp}--\eqref{eq:BA.g.Li} on the metric, and the fact that $\mathrm{supp}(\phi) \subseteq B(0,R)$, it is easy to check that for $t \in [0,T_B)\subseteq [0,1)$, the whole curve is contained in $B(0,3R)$.

Define $T^{\mp}_{k,q}$ (depending on the chosen constant-$(u_k, \theta_k)$ curve) by
$$T^{-}_{k,q}(u_k,\theta_k):= \inf \{ t \geq 0, (u_k,\theta_k,t) \in S^q_{\delta}\},\quad T^{+}_{k,q}(u_k,\theta_k):= \sup \{ t \geq 0, (u_k,\theta_k,t) \in S^q_{\delta}\}$$ (recall the definition \eqref{defS}). 

Let us consider only the case that $0<T^{-}_{k,q}(u_k,\theta_k)<T^{+}_{k,q}(u_k,\theta_k)<t$ (if not the proof is even easier). In view of the fact that $\partial_{t_k} u_q \in (\frac{\upkappa_0^2}{4},2)$ on $B(0,3R)$ by \eqref{anglecontrol2}, and that (by definition) $ \int_{T^{-}_{k,q}(u_k,\theta_k)}^{T^{+}_{k,q}(u_k,\theta_k)} \partial_{t_k} u_q \, dt_k = 2\de$, we get that 
\begin{equation} \label{sizesingzone}
\delta \leq T^{+}_{k,q}(u_k,\theta_k)-T^{-}_{k,q}(u_k,\theta_k) \leq \f{8\delta}{\upkappa_0^2}. 
\end{equation} 

We split the integral $\int_0^t$ term $III$ in \eqref{eq:error.in.dxchi} into an integral in $(S_\de^k)^c$, i.e.\ $\int_0^{T^{-}_{k,q}(u_k,\theta_k)}+ \int_{T^{+}_{k,q}(u_k,\theta_k)}^t $, and another integral in $S_\de^k$, i.e.\  $\int_{T^{-}_{k,q}(u_k,\theta_k)}^{T^{+}_{k,q}(u_k,\theta_k)} $. 

Note that, since $\widetilde{\phi}_{q} \equiv 0$ on 
$C_{\leq -\delta}^q$ (by Lemma~\ref{lem:support}), the $\int_0^{T^{-}_{k,q}(u_k,\theta_k)}$ integral is trivial. Using the Cauchy--Schwarz inequality and the bootstrap assumption \eqref{BA:flux.for.tphi.improved}, we obtain
$$ \|\int_{T^{+}_{k,q}(u_k,\theta_k)}^t |L_k  \partial_i\widetilde{\phi}_q|(u_k,\cdot,t') dt'  \|_{L^2_{\theta_k}(\Sigma_t \cap C^k_{u_k})} \ls \ep \cdot  \|L_k  \partial_i \tilde{\phi}_q \|_{L^2( C^k_{u_k}([0,T_B)) \cap S^q_{\de})} \ls \ep^{\f 34}.$$ 
It therefore follows that
\begin{equation}\label{proof.chi.4}
 \|\int_0^{T^{-}_{k,q}(u_k,\th_k)} |L_k  \partial_i\widetilde{\phi}_q|(u_k,\cdot,t') dt'  \|_{L^2_{\theta_k}(\Sigma_t \cap C^k_{u_k})} + \|\int_{T^{+}_{k,q}(u_k,\theta_k)}^t |L_k  \partial_i\widetilde{\phi}_q|(u_k,\cdot,t') dt'  \|_{L^2_{\theta_k}(\Sigma_t \cap C^k_{u_k})} \ls \ep^{\f 34}.
\end{equation}
Then we turn to the integral on the singular zone, whose smallness we will exploit. This time, Cauchy--Schwarz gives, in view of \eqref{sizesingzone}: 
\begin{equation}\label{proof.chi.5}
  \|\int_{T^{-}_{k,q}(u_k,\th_k)}^{T^{+}_{k,q}(u_k,\th_k)} |L_k  \partial_i\tilde{\phi}_q|  (u_k,\cdot,t') dt'  \|_{L^2_{\theta_k}(\Sigma_t \cap C^k_{u_k})} \ls \sdelta \cdot  \| L_k  \partial_i\tilde{\phi}_q    \|_{L^2( C^k_{u_k} \cap [0,T_B])}  \ls  \sdelta \cdot (\ep^{\f 34} \de^{-\f 12}) \ls \epsilon^{\f 34}, 
 \end{equation}
where we used \eqref{BA:flux.for.tphi}. 

Combining \eqref{proof.chi.4} and \eqref{proof.chi.5}, term $III$ in \eqref{eq:error.in.dxchi} is bounded by $III \ls \ep^{\f 32}$. 

\pfstep{Step~2(d): Putting everything together} Combining the estimates in Steps~2(b) and 2(c), we have thus shown that the terms $I$, $II$ and $III$ in \eqref{eq:error.in.dxchi} are bounded above by $\ls \ep^{\f 32}$ (for all $t\in [0,T_B)$ and $u_k \in \mathbb R$). Thus, using \eqref{eq:error.in.dxchi}, we obtain the desired estimate \eqref{nablachi.estimate}.

\pfstep{Step~3: Controlling $E_k\eta_k$ (Proof of \eqref{Eeta.estimate})} The proof is broadly similar to that of \eqref{nablachi.estimate} so we only explain the difference. By \eqref{Leta} and using similar arguments as in Step~2, we get  
\begin{equation} \label{proof.eta.0}
\sup_{0 \leq t < T_B, u_k \in \RR}\|  \partial_{t_k} E_k \eta_k + 2 L_k E_k\phi \cdot E_k \phi + 2 L_k \phi \cdot E_k^2 \phi\|_{L^2_{\theta_k}(\Sigma_t \cap C^k_{u_k})}  \ls \ep^{\f 32}.
\end{equation} 
Using \eqref{BA:Li}, Lemmas~\ref{lem:L.X.E} and \ref{dgeomvflemma}, we have $|L_k E_k\phi \cdot E_k \phi + L_k \phi \cdot E_k^2 \phi| \ls \ep^{\f 32} + \ep^{\f 34} |L_k \rd_x \phi| + \ep^{\f 34} |E_k^2 \phi|$. The $|L_k \rd_x \phi|$ term can be treated exactly like the terms in \eqref{eq:error.in.dxchi} to obtain
\begin{equation} \label{proof.eta.1}
\sup_{0 \leq t < T_B, u_k \in \RR} \ep^{\f 34} \|\int_0^{t} |L_k \rd_x \phi| (u_k,\cdot,t') dt'  \|_{L^2_{\theta_k}(\Sigma_t \cap C^k_{u_k})} \ls \ep^{\f 32}.
\end{equation} 
For the $E_k^2 \phi$ term, we split it into $|E_k^2 \rphi|$, $|E_k^2 \tphi|$ and $|E_k^2 \tilde{\phi}_q|$ for $q \neq k$ (c.f.~\eqref{eq:error.in.dxchi}). 
\begin{itemize}
\item The $|E_k^2 \rphi|$ term can be controlled similar to term $I$ in \eqref{eq:error.in.dxchi}.
\item The $|E_k^2 \tphi|$ term can be addressed like term $II$ in \eqref{eq:error.in.dxchi} (except for using the second, instead of the first, term in \eqref{BA:flux.for.tphi.improved.2}).
\item The $|E_k^2 \tilde{\phi}_q|$ (with $q\neq k$) can be treated as term $III$ in \eqref{eq:error.in.dxchi}.
\end{itemize} 
Altogether this gives
\begin{equation} \label{proof.eta.2}
\sup_{0 \leq t < T_B, u_k \in \RR} \|\int_0^{t} |E_k^2 \phi|(u_k,\cdot,t') dt'  \|_{L^2_{\theta_k}(\Sigma_t \cap C^k_{u_k})} \ls \ep^{\f 32}.
\end{equation}

Combining \eqref{proof.eta.0}, \eqref{proof.eta.1} and \eqref{proof.eta.2} with the initial data bound (obtained in part 5~of the proof of Lemma~\ref{lem:local})  gives \eqref{Eeta.estimate}.

\pfstep{Step~4: Controlling $\rd_x\eta_k$ (Proof of \eqref{nablaeta.estimate})} Using \eqref{Lvartheta}, we rewrite \eqref{Leta} as
\begin{equation}\label{Leta.2}
L_k (\varTheta_k \eta_k) = -2 \varTheta_k L_k \phi \cdot E_k \phi - \varTheta_k \chi_k  \cdot  ( K(E_k,X_k) - E_k\log N).
\end{equation}
(This rewriting absorbs the linear $\eta_k$ term on the RHS, so that when differentiating the equation by $\rd_i$, we do not have a linear $\chi_k \cdot \rd_i \eta_k$ term.) Differentiating \eqref{Leta.2} by $\rd_i$ and arguing as in Steps~2 and 3, we obtain 
\begin{equation} \label{proof.dxeta.0}
\sup_{0 \leq t < T_B, u_k \in \RR}\|  \partial_{t_k} \rd_i (\varTheta_k \eta_k) + 2 \varTheta_k \cdot L_k \rd_i \phi \cdot E_k \phi + 2 \varTheta_k \cdot L_k \phi \cdot E_k \rd_i \phi\|_{L^2_{\theta_k}(\Sigma_t \cap C^k_{u_k})}  \ls \ep^{\f 32}.
\end{equation} 
After putting $\varTheta_k \cdot E_k\phi$ in $L^\infty$ (by \eqref{bootstrapvarTheta}, $\mathrm{supp}(\tphi) \subseteq B(0,R)$, \eqref{BA:Li} and Lemma~\ref{lem:L.X.E}), we have $|\varTheta_k \cdot L_k \rd_i \phi \cdot E_k \phi |\ls |L_k \rd_i \phi|$. We can then proceed as with the terms in \eqref{eq:error.in.dxchi} in Step~2.

Now for the $\varTheta_k \cdot L_k \phi \cdot E_k \rd_i \phi$ term in \eqref{proof.dxeta.0}, we first bound it pointwise by $|E_k \rd_i \phi|$ (by \eqref{bootstrapvarTheta}, $\mathrm{supp}(\tphi) \subseteq B(0,R)$, \eqref{BA:Li} and Lemma~\ref{lem:L.X.E}), and then decompose into the terms $|E_k \rd_i \rphi|$, $|E_k \rd_i \tphi|$ and $|E_k \rd_i \widetilde{\phi}_q|$ where $q\neq k$. The terms $|E_k \rd_i \rphi|$ and $|E_k \rd_i \widetilde{\phi}_q|$ are similar to terms $I$ and $III$ in \eqref{eq:error.in.dxchi}.

It remains to control $|E_k \rd_i \tphi|$. The key issue is that using the bootstrap assumption \eqref{BA:flux.for.tphi} for the flux, we only have an estimate that is large of order $\de^{-\f 12}$.
\begin{equation} \label{proof.dxeta.2}
\ep^{\f 34} \sup_{0 \leq t < T_B} \|\int_0^{t}  |E_k \rd_i \tphi| (u_k,\cdot,t')\, dt'  \|_{L^2_{\theta_k}(\Sigma_t \cap C^k_{u_k})} \ls \ep^{\f 32}\cdot \de^{-\f 12}.
\end{equation}
The important point, however, is that for $u_k \not\in [-\de,\de]$, we have a better estimates using \eqref{BA:flux.for.tphi.improved}:
\begin{equation} \label{proof.dxeta.3}
\ep^{\f 34} \sup_{0 \leq t < T_B, u_k \not\in [-\de,\de]} \|\int_0^{t}  |E_k \rd_i \tphi| (u_k,\cdot,t') \,dt'  \|_{L^2_{\theta_k}(\Sigma_t \cap C^k_{u_k})} \ls \ep^{\f 32}.
\end{equation}

Combining \eqref{proof.dxeta.0}, \eqref{proof.dxeta.2} and \eqref{proof.dxeta.3}, the bounds which are similar to Step~2, as well as the initial data bound $\|  \rd_x (\varTheta_k \cdot \eta_k)\|_{L^2_{\th_k}(\Sigma_0 \cap C^k_{u_k})} \ls \ep^2$ (which can be proven as in point 5~of Lemma~\ref{lem:local}), we obtain that 
\begin{equation}\label{proof.dxeta.almost.almost.final}
\sup_{0 \leq t < T_B, u_k \in \RR}\| \rd_x (\varTheta_k \cdot \eta_k) \|_{ L^2_{\theta_k}(\Sigma_t \cap C^k_{u_k})} \lesssim \ep^{\f 32} \cdot \de^{-\f 12}, \quad\sup_{0 \leq t < T_B, u_k \not\in [-\de,\de]}\| \rd_x (\varTheta_k \cdot \eta_k) \|_{ L^2_{\theta_k}(\Sigma_t \cap C^k_{u_k})} \lesssim \ep^{\f 32}.
\end{equation}

We now integrate \eqref{proof.dxeta.almost.almost.final} over $u_k$, where when we use the weaker estimate when $u_k \in [-\de,\de]$, we combine it with the short length scale. Hence,
\begin{equation}\label{proof.dxeta.almost.final}
\sup_{0 \leq t < T_B}\|  \rd_x (\varTheta_k \cdot \eta_k) \|_{ L^2_{u_k} L^2_{\theta_k}(\Sigma_t )} \lesssim \ep^{\f 32}.
\end{equation}

Finally, we change variables from $(u_k, \th_k)$ to $(x^1, x^2)$. Comparing \eqref{volelliptict=0} with \eqref{volthetau}, we see that $dx^1 \, dx^2 = e^{-2\gamma} \mu_k^{-2} \varTheta_k^{-2} \, du_k\, d\th_k$. The estimates in \eqref{bootstrapmu}, \eqref{bootstrapvarTheta} and \eqref{eq:BA.g.Li} imply that (for $B(0,3R)$ to be understood in the $(x^1,x^2)$ coordinates) $\|\rd_x(\varTheta_k \eta_k)\|_{L^2_{x^1,x^2}(\Sigma_t \cap B(0,3R))} \ls \|\rd_x (\varTheta_k \eta_k)\|_{L^2_{u_k} L^2_{\th_k}(\Sigma_t \cap B(0,3R))}$. Combining this estimate with \eqref{proof.dxeta.almost.final}, and using also \eqref{bootstrapvarTheta} and  \eqref{ricci.pointwise.main.estimate} (established above), we obtain \eqref{nablaeta.estimate}. \qedhere

 \end{proof}

In the course of the proof of Proposition~\ref{prop:Ricci}, we also proved estimates for $\chi_k$ and $\eta_k$ with $L_k$ derivatives, which we collect in the following proposition. In particular, while we have no control over general second derivatives of $\chi_k$ and $\eta_k$, we do bound the combinations of second derivatives with at least one $L_k$. This will also turn out to be important in \cite{LVdM2}.
\begin{prop}\label{prop:L.chi.eta}
The following estimates hold for all $t\in [0,T_B)$:
\begin{equation}\label{eq:Lchi.Leta}
\|L_k \chi_k \|_{L^\i_{1-\alp}(\Sigma_t)} +  \|L_k \eta_k \|_{L^\i_{1-\alp}(\Sigma_t)} \ls \ep^{\f 32},
\end{equation}
\begin{equation}\label{eq:L2chi.etc}
 \| L_k^2 \chi_k \|_{L^2(\Sigma_t \cap B(0,R))} + \|L_k^2 \eta_k \|_{L^2(\Sigma_t \cap B(0,R))} + \|L_k \rd_x \chi_k \|_{L^2(\Sigma_t \cap B(0,R))} + \|L_k E_k \eta_k \|_{L^2(\Sigma_t \cap B(0,R))} \ls \ep^{\f 32}.
 \end{equation}
\end{prop}
\begin{proof}
The estimates for $L_k \chi_k$ and $L_k \eta_k$ follows from \eqref{eq:Lchi.Leta.est}. The estimates for $L_k \rd_x \chi_k$ and $L_k E_k \eta_k$ follow from combining \eqref{eq:Lk.rdi.chik}, \eqref{proof.eta.0} with \eqref{bootstrapsmallnessenergy} and \eqref{BA:Li}. Finally, the estimates for $L_k^2\eta_k$ and $L_k^2\chi_k$ follow from differentiating \eqref{Leta} and \eqref{Lchi} by $L_k$, and then controlling the resulting terms using \eqref{bootstrapsmallnessenergy}, \eqref{BA:Li}, Lemmas~\ref{lem:L.X.E} and \ref{dgeomvflemma}, Propositions~\ref{prop:elliptic.Besov.weighted}, \ref{prop:elliptic.dtg} and \eqref{prop:K}, and \eqref{eq:Lchi.Leta}. \qedhere
\end{proof}

 \subsection{Estimates for $\mu_k$ and $\varTheta_k$}\label{sec:mu.Theta}

 We next consider the estimates for $\mu_k$ and $\varTheta_k$.
 
 \begin{prop}\label{prop:mu.varTheta}
 	 The following estimates hold for all $t\in [0,T_B)$ and all $u_k \in \RR$:
 	\begin{align} 
	\label{mu.main.estimate}
 	\|  \log \mu_k - \gamma_{asymp} \om(|x|) \log |x| \|_{L^{\infty}_{1-\alpha}(\Sigma_t)} +\| \partial_x \mu_k  \|_{L^{\infty}_{1-\alpha}(\Sigma_t)}  \lesssim \ep^{\f 32} , \\
  	 \label{vartheta.main.estimate}
	\| \log(\varTheta_k) - \gamma_{asymp} \om(|x|) \log|x|\|_{L^{\infty}_{1-2\alp}(\Sigma_t)} + \| \la x \ra^{-\alp} \rd_x \log \varTheta_k \|_{L^2_{\th_k}(\Sigma_t\cap C^k_{u_k})} \lesssim \ep^{\f 32}.
\end{align}
 \end{prop}

\begin{proof}
By initial condition \eqref{mu(0)formula}, and the bounds in Proposition~\ref{prop:elliptic.easy}, 
$$\| \log \mu_k - \gamma_{asymp} \om(|x|) \log |x| \|_{L^\i_{1-\alp}(\Sigma_0)} \ls \ep^{\f 32}.$$
Integrating the transport equation \eqref{Lmu}, and using the estimates in Lemma~\ref{lem:L.X.E} regarding $|X^i_k|$ together with \eqref{eq:d2g.L4}, \eqref{eq:d2g.Linfty}, \eqref{Kmainestimate}, and the comparability of $\la x\ra$ along integral curves of $\rd_{t_k}$, we get 
$$\sup_{0 \leq t < T_B}\| \log \mu_k - \gamma_{asymp} \om(|x|) \log |x| \|_{L^{\infty}_{1-\alpha}(\Sigma_t)} \ls \ep^{\f 32},$$
which controls the first term in  \eqref{mu.main.estimate}.

Similarly, after commuting \eqref{Lmu} with $\partial_i$, we can integrate the transport equation using Lemma~\ref{dgeomvflemma}, \eqref{secondderivative.g.estimate.weighted} and \eqref{Kmainestimate}, in addition to \eqref{mu(0)formula}, Lemma~\ref{lem:L.X.E} and \eqref{eq:d2g.Linfty}. This bounds the second term in \eqref{mu.main.estimate}. 

The estimate \eqref{vartheta.main.estimate} can be obtained in a similar manner. Here, we use the transport equation \eqref{Lvartheta}, and the initial condition is given by \eqref{initialvarTheta}. To bound the initial value, we use Proposition~\ref{prop:elliptic.easy}, . In order to control the inhomogeneous term $\chi_k$ in \eqref{Lvartheta} (and its $\rd_x$ derivative), we use the estimates in Proposition~\ref{prop:Ricci}, and the comparability of $\la x\ra$ along integral curves of $\rd_{t_k}$. \qedhere
\end{proof}

\subsection{Second derivatives of the commutation vector fields}\label{sec:second.derivative.X.E.L}

\begin{lem}\label{lem:rdrd.in.terms.of.geometric}
For every sufficiently regular function $f$,
\begin{align}
\label{eq:rdrdx.in.terms.of.geometric}
\|\rd \rd_x f \|_{L^2(\Sigma_t\cap B(0,3R))} \ls &\: \sum_{Y_k \in \{L_k,X_k,E_k \}} \sum_{Z_k \in \{X_k,E_k\} }\| Y_k Z_k f \|_{L^2(\Sigma_t\cap B(0,3R))} + \| \rd_x f \|_{L^2(\Sigma_t\cap B(0,3R))}, \\
\label{eq:rdrd.in.terms.of.geometric}
\|\rd^2 f \|_{L^2(\Sigma_t\cap B(0,3R))} \ls &\: \sum_{Y_k,Z_k \in \{L_k,X_k,E_k \}} \| Y_k Z_k f \|_{L^2(\Sigma_t\cap B(0,3R))} + \| \rd f \|_{L^2(\Sigma_t\cap B(0,3R))}.
\end{align}
\end{lem}
\begin{proof}
Acting the vector fields in \eqref{partial12EX} on $f$, and further differentiating by $\rd$, and using the estimates in Lemmas~\ref{lem:L.X.E} and \ref{dgeomvflemma} and \eqref{eq:d2g.L4}--\eqref{eq:d2g.Linfty}, we obtain the pointwise bound in $B(0,3R)$.
$$ |\rd \rd_i f| \ls |\rd E_k f| + |\rd X_k f| + |E_k f| + |X_k f|,$$
which implies \eqref{eq:rdrdx.in.terms.of.geometric} after also using Lemma~\ref{lem:rd.in.terms.of.XEL}.

To prove \eqref{eq:rdrd.in.terms.of.geometric}, we need to additionally control $\rd \rd_t f$, which can be done with a similar argument except for starting with \eqref{eq:partialtLEX}; we omit the details. \qedhere
\end{proof}

\begin{prop} \label{secondderivativeYprop} The following estimates\footnote{We remark that some components of the second derivatives are in fact better. For instance, the derivatives of $L_k E^i_k$ and $E_k E_k^i$ obey better bounds (in terms of the $L^p$ space) than $X_k E_k^i$. Such improvements will not be made precise, nor will they be useful.} for the second derivatives of the coefficients of the vector fields $L_k$, $E_k$ and $X_k$ hold for all $t\in [0,T_B)$:
$$\|\rd^2 E^i_k\|_{L^2(\Sigma_t\cap B(0,3R))} + \|\rd^2 X^i_k\|_{L^2(\Sigma_t\cap B(0,3R))} + \|\rd\rd_x L^\alp_k\|_{L^2(\Sigma_t\cap B(0,3R))}  \ls \ep^{\f 54}.$$	
\end{prop}
\begin{proof} \pfstep{Step~1: Estimates for $E_k^i$ and $X_k^i$}  By Lemma~\ref{lem:rdrd.in.terms.of.geometric}, to obtain the estimates for $E_k^i$ and $X_k^i$, it suffices to bound $\|\rd Z_k E_k^i\|_{L^2(\Sigma_t\cap B(0,3R))}$ and $\|\rd Z_k X_k^i\|_{L^2(\Sigma_t\cap B(0,3R))}$ for $Z_k \in \{L_k,X_k,E_k\}$ . Moreover, in view of \eqref{YXj}, it in fact suffices to bound only $\|\rd Z_k E_k^i\|_{L^2(\Sigma_t\cap B(0,3R))}$. 

To control $\|\rd Z_k E_k^i\|_{L^2(\Sigma_t\cap B(0,3R))}$, we differentiate the equations \eqref{LEi}--\eqref{EEi} in Proposition~\ref{dXELellipticprop}. To treat the resulting terms, note that we only need estimates in $B(0,3R)$. We then use Lemmas~\ref{lem:L.X.E} and \ref{dgeomvflemma} to bound $L_k^\alp$, $X_k^i$, $E_k^i$ and their first derivatives in $L^\i$, use \eqref{eq:d2g.L4}, \eqref{eq:d2g.Linfty} and \eqref{eq:dtg.improve.bootstrap} to control the metric and its first derivatives in $L^\i$, and use \eqref{Kmainestimate} and \eqref{ricci.pointwise.main.estimate} to bound $K$, $\chi_k$ and $\eta_k$ in $L^\i$. Thus, on the set $\Sigma_t\cap B(0,3R)$, we have the pointwise bound
\begin{equation}\label{eq:second.der.of.Y.pointwise}
\sum_{Z_k \in \{L_k, X_k, E_k\}} |\rd Z_k E_k^i| \ls \ep^{\f 32} + \sum_{\mfg \in \{\gamma,\bt^i,N\} } |\rd \rd_x \mfg| + |\rd K| + |\rd \chi_k| + |\rd \eta_k|.
\end{equation}
We take the $L^2(\Sigma_t\cap B(0,3R))$ norm \eqref{eq:second.der.of.Y.pointwise}: the metric terms are bounded by \eqref{secondderivative.g.estimate}, and \eqref{eq:dtg.main} (and the fact that $s'-s'' <\f 12$), $|\rd K|$ is bounded by \eqref{Kmainestimate}, and $|\rd \chi_k|$ and $|\rd \eta_k|$ are bounded by Lemma~\ref{lem:rd.in.terms.of.XEL}, \eqref{eq:Lchi.Leta}, \eqref{nablachi.estimate} and \eqref{nablaeta.estimate}. We thus obtain $\sum_{Z_k \in \{L_k, X_k, E_k\}} \|\rd Z_k E_k^i \|_{L^2(\Sigma_t\cap B(0,3R))} \ls \ep^{\f 32}$, as desired.

\pfstep{Step~2: Estimates for $L^\alp_k$ and $L^i_k$} Using Lemma~\ref{lem:rdrd.in.terms.of.geometric}, it suffices to bound $\sum_{Z_k \in \{X_k, E_k\}} \|\rd Z_k L_k^\alp\|_{L^2(\Sigma_t\cap B(0,3R))}$ (note in particular that $Z_k \neq L_k$). Differentiating \eqref{YLt} and \eqref{YLi}, and using the bound \eqref{eq:second.der.of.Y.pointwise} above together with \eqref{eq:d2g.L4}, \eqref{eq:d2g.Linfty} and \eqref{eq:dtg.improve.bootstrap}, we obtain that on $\Sigma_t \cap B(0,3R)$,
\begin{equation}
\sum_{Z_k \in \{X_k, E_k\} } |\rd Z_k L_k^\alp| \ls \ep^{\f 32} + \sum_{\mfg \in \{\gamma,\bt^i,N\} } |\rd \rd_x \mfg| + |\rd K| + |\rd \chi_k| + |\rd \eta_k|.
\end{equation}
Notice that all the terms have appeared in \eqref{eq:second.der.of.Y.pointwise}, and we then proceed as in Step~1. \qedhere

\end{proof}

\section{Conclusion of the proof of Theorem~\ref{thm:bootstrap.metric}}\label{sec:thm.bootstrap.metric.conclusion}

We continue to work under the assumptions of Theorem~\ref{thm:bootstrap.metric}. We now conclude the proof of Theorem~\ref{thm:bootstrap.metric}.

\begin{proof}[Proof of Theorem~\ref{thm:bootstrap.metric}]
To conclude the proof of Theorem~\ref{thm:bootstrap.metric}, we only need to collect already proven facts:
\begin{itemize}
\item \eqref{eq:BA.g.asymp}--\eqref{eq:BA.g.L4} all hold with $\ep^{\f 54}$ replaced by $C \ep^{\f 32}$ thanks to Proposition~\ref{prop:elliptic.easy} and \eqref{eq:dtg.improve.bootstrap}.
\item \eqref{bootstrapK} is improved by \eqref{Kmainestimate}, \eqref{bootstrapricci} and \eqref{bootstrapnablaricci} are improved by estimates in Proposition~\ref{prop:Ricci}, and \eqref{bootstrapmu} and \eqref{bootstrapvarTheta}  are improved by estimates in Proposition~\ref{prop:mu.varTheta}. All these estimates are improved from $\ep^{\f 54}$ to $C \ep^{\f 32}$.
\item The estimates in \eqref{eq:wavefront} for $u_k$ follow from Corollary~\ref{cor:diffeo} and Proposition~\ref{prop:dx2u}.
\item The estimates in \eqref{eq:frames.in.thm} for $L_k^\bt$, $E_k^i$ and $X_k^i$ follow from Lemmas~\ref{lem:L.X.E} and \ref{dgeomvflemma}.
\item The estimates in \eqref{metric.est} for the metric components follow from Propositions~\ref{prop:elliptic.easy}, \ref{prop:elliptic.Besov.weighted}, and \ref{prop:elliptic.dtg}. \qedhere
\end{itemize}
\end{proof}

\appendix

\section{Solving the constraint equations}\label{sec:appendix}

This appendix concerns the constraint equations for the initial data. In our setting, we cannot directly use the result in \cite{Huneau.constraints} (or \cite{HL.elliptic}) to solve the constraints. We will instead need a modification which we sketch in this appendix.

We first explain why \cite{Huneau.constraints} is not applicable in our setting. In \cite{Huneau.constraints, HL.elliptic}, one prescribes $\phi$ and $\dot \phi:= e^{2\gamma} \n\phi$ so that one can directly impose the integrability condition
\begin{equation}\label{eq:integrability.appendix}
\int_{\Sigma_0} e^{2\gamma} \,(\n\phi)\, (\rd_j \phi) \, dx^1\, dx^2 = \int_{\Sigma_0} \dot\phi \, (\rd_j \phi) \, dx^1\, dx^2 = 0,\quad j=1,2.
\end{equation}
Using this condition, one can solve for $\gamma$ and $K$ as a coupled system of nonlinear elliptic equations. However, in our case, we need to impose initial data with the additional condition that $\n \tphi - X_k \tphi$ is better than generic first derivatives of $\tphi$. In terms of $\tphi$ and $\dot{\widetilde{\phi}}_k$, this corresponds $e^{-2\gamma} \dot{\widetilde{\phi}}_k - e^{-\gamma} \cdot \de^{iq} \cdot c_{kq} \cdot \rd_i \tphi$ being better. This, however, cannot be imposed with the scheme in \cite{Huneau.constraints} since $\gamma$ is not known a priori.

Instead, we prescribe $\phi$ and $\underline{\phi} := e^\gamma \n \phi$. In order to impose the condition \eqref{eq:integrability.appendix}, we need to introduce a two-parameter family of data and show that there exists a choice of parameters such that \eqref{eq:integrability.appendix} holds.

In \textbf{Section~\ref{sec:general.lemma.constraint}}, we prove a general lemma for solving the constraint equations. In \textbf{Section~\ref{sec:nondegeneracy}}, we then expand on the non-degeneracy condition \eqref{eq:nondegeneracy.in.thm}, in preparation of solving the constraint equations in our setting. In \textbf{Section~\ref{sec:appendix.constraint.for.rough}}, we then solve the constraint equations to construct examples of initial data sets satisfying the assumptions of Definition~\ref{roughdef}. Finally, in \textbf{Section~\ref{sec:data.approx}}, we solve the constraint equations in order to construct $\de$-impulsive wave data that approximate impulsive wave data, hence proving Lemma~\ref{lem:data.approx}.

\subsection{A general lemma for solving the constraint equations}\label{sec:general.lemma.constraint}

Let $\mathscr K\subseteq \RR^2$ be a compact, convex set and consider a two parameter family $\underline{\phi}^{{\bm \lambda}}$ parametrized by $\bm \lambda = (\lambda_1, \lambda_2) \in \mathscr K$ such that the following holds:
\begin{enumerate}
\item For any $\bm\lambda \in \mathscr K$, $(\phi, \uphi^{{\bm \lambda}})$ satisfies $\mathrm{supp}(\phi), \, \mathrm{supp}(\uphi^{{\bm \lambda}}) \subseteq B(0,R)$ and obeys the estimate 
\begin{equation}\label{eq:phi.uphi.basic.bound}
\|\phi\|_{W^{1,4}(\RR^2)} + \|\uphi^{\bm \lambda} \|_{L^4(\RR^2)} \leq \ep.
\end{equation}
\item $\bm \lambda \mapsto \underline{\phi}^{\bm \lambda}$ is a continuous map $\mathscr K \to L^4(\RR^2)$.
\item For any $\underline{\gamma} = -\underline{\gamma}_{asymp}\om(|x|) \log|x| + \widetilde{\underline{\gamma}}$ with $|\underline{\gamma}_{asymp}|\leq \ep$ and $\| \widetilde{\underline{\gamma}} \|_{W^{1,4}_{\f 14}(\Sigma_0)}\leq \ep$, there exists $\bm \lambda \in \mathscr K$ such that 
\begin{equation}\label{eq:integrability.in.appendix}
\int_{\Sigma_0} e^{\underline{\gamma}} \underline{\phi}^{\bm \lambda} (\rd_j \phi) \, dx^1 dx^2 = 0,\quad j=1,2.
\end{equation}
\end{enumerate}

\begin{lemma}\label{lem:constraint}
For any $R>0$, there exists $\ep_0 = \ep_0(R) >0$ such that the following holds.

Suppose $(\phi, \uphil)$ satisfies the conditions 1--3 above. Then, if $\ep \in (0, \ep_0]$, there exist $\bm \lambda_0 \in \mathscr K$ and functions $(\gamma, K)$ such that $(\phi, \phi' = e^{-\gamma} \uphi^{\bm \lambda_0},\gamma, K)$ is an admissible initial data set (see Definition~\ref{datachoice}).
\end{lemma}
\begin{proof}
Denote by $\overline{B}_{\mathscr Y}(0,\ep)$ the closed ball of radius $\ep$ around $0$ in a Banach space $\mathscr Y$. Let $\Phi: \mathscr K \times [0,\ep]\times \overline{B}_{W^{1,4}_{\f 14}(\Sigma_0)}(0,\ep) \times \overline{B}_{L^4_{\f 54}(\Sigma_0)}(0,\ep) \to \mathscr K \times [0,\ep]\times \overline{B}_{W^{1,4}_{\f 14}(\Sigma_0)}(0,\ep) \times \overline{B}_{L^4_{\f 54}(\Sigma_0)}(0,\ep)$ be given by $(\bm \lambda, \gamma_{asymp}, \widetilde{\gamma}, K) \mapsto (\bm \lambda^*, \gamma_{asymp}^*, \widetilde{\gamma}^*, K^*)$, where
\begin{enumerate}
\item $\bm \lambda^*$ is chosen so that \eqref{eq:integrability.in.appendix} holds.
\item $\gamma^* = -\gamma_{asymp}^* \om(|x|) \log|x| + \widetilde{\gamma}^*$ is given by 
\begin{equation}\label{eq:gamma*.def}
\gamma^* = \Delta^{-1} ( -\de^{il} (\rd_i \phi)(\rd_l\phi) - \f{e^{-2\gamma}}{2} |K|^2 -  (\uphi^{\bm \lambda^*})^2 ),
\end{equation}
where $\Delta^{-1}$ is as in Definition~\ref{def:Delta-1}. 
\item Define $\sigma^i := \de^{ij} \cdot e^{\gamma} \cdot \uphi^{\bm \lambda^*} \cdot \partial_j \phi$. Impose $K^*$ to be
\begin{equation}\label{eq:K*.def}
 K^*_{i j} = 2 [\mathfrak L \Delta^{-1} \sigma]_{ij},
 \end{equation}
with $\bm \lambda^*$ as in 1~above, $\Delta^{-1}$ as in Definition~\ref{def:Delta-1}, and $\mathfrak L$ the conformal Killing operator as in \eqref{maximality3}.
\end{enumerate}

The equations \eqref{eq:gamma*.def} and \eqref{eq:K*.def} are easy to solve: for $\ep_0 >0$ sufficiently small and $\ep\in (0,\ep_0]$, if $(\gamma_{asymp}, \widetilde{\gamma}, K) \in [0,\ep] \times \overline{B}_{W^{1,4}_{\f 14}(\Sigma_0)}(0,\ep) \times \overline{B}_{L^4_{\f 54}(\Sigma_0)}(0,\ep)$, it follows from Proposition~\ref{prop:basic.elliptic.2}, \eqref{eq:integrability.in.appendix} and Proposition~\ref{prop:Sobolev.weighted} that
\begin{equation}\label{eq:constraint.bounds.ep2}
|\gamma_{asymp}|,\, \| \widetilde{\gamma}^* \|_{H^2_{-\f 18}(\Sigma_0)},\, \| K^* \|_{H^1_{\f 78}(\Sigma_0)} \ls \ep^2
\end{equation}
 for some implicit constants depending only on $R$. (This can for instance be proven as \cite[Lemma~7.1]{HL.elliptic}, with $\de = -\f 14$ in the notation there, and noting that there is in fact extra room in the weights.) In particular, using 1(b) of Proposition~\ref{prop:Sobolev.weighted}, it follows, after choosing $\ep_0$ smaller if necessary, that $\Phi$ indeed maps into $\mathscr K \times [0,\ep]\times \overline{B}_{W^{1,4}_{\f 14}(\Sigma_0)}(0,\ep) \times \overline{B}_{L^4_{\f 54}(\Sigma_0)}(0,\ep)$ (as stated above), and moreover, $\Phi$ is compact.

By Schauder's fixed point theorem, $\Phi$ has a fixed point $(\bm \lambda_0, \gamma_{asymp}, \widetilde{\gamma}, K)$. As a result, $(\phi, \phi' := e^{-\gamma}\uphi^{\bm \lambda_{0}}, \gamma:= -\gamma_{asymp} \om(|x|) \log|x| + \widetilde{\gamma}, K)$ constitutes an admissible initial data set. \qedhere
\end{proof}

\subsection{Lemmas on the non-degeneracy assumption}\label{sec:nondegeneracy}

In this subsection, we prove two lemmas related to the non-degeneracy condition \eqref{eq:nondegeneracy.in.thm}. First, in Lemma~\ref{lem:some.nontrivial.lower.bd}, we prove the assertion in Remark~\ref{rmk:nondegeneracy} that LHS of \eqref{eq:nondegeneracy.in.thm} is non-zero for non-identically zero, compactly supported $\phi$. Then, in Lemma~\ref{lem:nondegeneracy.consequence} we deduce a consequence of the condition \eqref{eq:nondegeneracy.in.thm} which will be used when solving the constraint equations for the impulsive and $\de$-impulsive gravitational waves.

\begin{lem}\label{lem:some.nontrivial.lower.bd}
Let $\phi \in H^1(\Sigma_0)$ be compactly supported and non-identically vanishing. Then 
\begin{equation}\label{eq:some.nontrivial.lower.bd}
 \left\| \rd_1 \phi - \f{\la \rd_1\phi,\, \rd_2\phi \ra_{L^2(\Sigma_0, dx)}}{\|\rd_2\phi \|_{L^2(\Sigma_0)}^2} \rd_2\phi \right\|_{H^{-3}(\Sigma_0)}\neq 0,\quad \left\| \rd_2 \phi - \f{\la \rd_1\phi,\, \rd_2\phi \ra_{L^2(\Sigma_0, dx)}}{\|\rd_1\phi \|_{L^2(\Sigma_0)}^2} \rd_1\phi \right\|_{H^{-3}(\Sigma_0)} \neq 0.
 \end{equation}
\end{lem}
\begin{proof}
Take $\phi$ as in the assumption of the lemma. By the compact support assumption, both $\rd_1\phi$ and $\rd_2\phi$ are not identically $0$. The same argument shows that it is impossible to have $\rd_1 \phi + a \rd_2 \phi = 0$ or $a \rd_1 \phi + \rd_2\phi = 0$ for some constant $a\in \RR$. It thus follows from the Cauchy--Schwarz inequality that 
$$ \left\| \rd_1 \phi - \f{\la \rd_1\phi,\, \rd_2\phi \ra_{L^2(\Sigma_0, dx)}}{\|\rd_2\phi \|_{L^2(\Sigma_0)}^2} \rd_2\phi \right\|_{L^2(\Sigma_0)}\neq 0, \quad \left\| \rd_2 \phi - \f{\la \rd_1\phi,\, \rd_2\phi \ra_{L^2(\Sigma_0, dx)}}{\|\rd_1\phi \|_{L^2(\Sigma_0)}^2} \rd_1\phi \right\|_{L^2(\Sigma_0)} \neq 0.$$
Clearly, since a function with non-zero $L^2$ norm must be non-vanishing (and hence has non-zero $H^{-3}$ norm), we obtain \eqref{eq:some.nontrivial.lower.bd}. \qedhere
\end{proof}

\begin{lem}\label{lem:nondegeneracy.consequence}
Suppose that $\mathrm{supp}(\phi) \subseteq B(0,\f R2)$, and
\begin{equation}\label{eq:nondegenerate.in.terms.of.H-3}
\left\| \rd_1 \phi - \f{\la \rd_1\phi,\, \rd_2\phi \ra_{L^2(\Sigma_0, dx)}}{\|\rd_2\phi \|_{L^2(\Sigma_0)}^2} \rd_2\phi \right\|_{H^{-3}(\Sigma_0)} \times \left\| \rd_2 \phi - \f{\la \rd_1\phi,\, \rd_2\phi \ra_{L^2(\Sigma_0, dx)}}{\|\rd_1\phi \|_{L^2(\Sigma_0)}^2} \rd_1\phi \right\|_{H^{-3}(\Sigma_0)} \geq  \ep^{\f 52}.
\end{equation}

Then there exist smooth functions $\psi_1$, $\psi_2$ compactly supported in $B(0,R)$ such that
\begin{equation}\label{eq:nondegeneracy}
\max_{j=1,2} \| \psi_j \|_{H^3(\Sigma_0)} \leq 1, \quad \left| \det  \begin{bmatrix}
 		\int_{\Sigma_0} \psi_1 \rd_1 \phi \, dx^1\, dx^2 & 	\int_{\Sigma_0} \psi_1 \rd_2 \phi \, dx^1\, dx^2   \\
 		\int_{\Sigma_0} \psi_2 \rd_1 \phi \, dx^1\, dx^2 & 	\int_{\Sigma_0} \psi_2 \rd_2 \phi \, dx^1\, dx^2  
 		\end{bmatrix} \right| \gtrsim \ep^{\f52}.
\end{equation}
\end{lem}
\begin{proof}
For this proof, let $\|\cdot \|_{H^3(\Sigma_0)}:= \|\Db^3 (\cdot) \|_{L^2(\Sigma_t)}$ (which is equivalent to that in Definition~\ref{def:Sobolev.norm}). Given $\rd_1 \phi$ and $\rd_2 \phi$, notice that
$$\sup_{\substack{ \psi \in C^\infty_c(\Sigma_0) \setminus \{0\},\\
\int_{\Sigma_0} \psi \rd_2 \phi \, dx^1\, dx^2 =0} } \f{ \int_{\Sigma_0} \psi \rd_1 \phi \, dx^1\, dx^2 }{ \|\psi \|_{H^3(\Sigma_0)}} 
= 
\sup_{\substack{ \psi \in C^\infty_c(\Sigma_0) \setminus \{0\},\\
\int_{\Sigma_0} \psi \rd_2 \phi \, dx^1\, dx^2 =0} } \f{ \int_{\Sigma_0} \psi (\rd_1 \phi - \f{\la \rd_1\phi,\, \rd_2\phi \ra_{L^2(\Sigma_0,dx)}}{\|\rd_2\phi \|_{L^2(\Sigma_0)}^2} \rd_2\phi) \, dx^1\, dx^2 }{ \|\psi \|_{H^3(\Sigma_0)}}.$$
It is thus easy to see that the supremum is achieved by $\psi = \Db^{-6}(\rd_1 \phi - \f{\la \rd_1\phi,\, \rd_2\phi \ra_{L^2(\Sigma_0,dx)}}{\|\rd_2\phi \|_{L^2(\Sigma_0)}^2} \rd_2\phi)$, and that the supremum is 
$$\sup_{\substack{ \psi \in C^\infty_c(\Sigma_0) \setminus \{0\},\\
\int_{\Sigma_0} \psi \rd_2 \phi \, dx^1\, dx^2 =0} }
 \f{ \int_{\Sigma_0} \psi \rd_1 \phi \, dx^1\, dx^2 }{ \|\psi \|_{H^3(\Sigma_0)}} = \left\| \rd_1 \phi - \f{\la \rd_1\phi,\, \rd_2\phi \ra_{L^2(\Sigma_0, dx)}}{\|\rd_2\phi \|_{L^2(\Sigma_0)}^2} \rd_2\phi \right\|_{H^{-3}(\Sigma_0)}.$$
A similar statement holds after switching $\rd_1\phi$ and $\rd_2 \phi$. Therefore, using also that $\mathrm{supp}(\phi) \subseteq B(0,\f R2)$, we deduce that there are smooth functions $\psi_1$, $\psi_2$ compactly supported in $B(0,R)$ with $\max_{j=1,2} \|\psi_j\|_{H^3(\Sigma_0)}\leq 1$ such that $\la \psi_1, \rd_2\phi\ra_{L^2(\Sigma_0,dx)} = 0 = \la \psi_2, \rd_1\phi\ra_{L^2(\Sigma_0,dx)}$ and for $i,j=1,2$, $i\neq j$, \textbf{without summation},
$$  \int_{\Sigma_0} \psi_j \rd_j \phi \, dx^1\, dx^2 \geq \f 12\left\| \rd_j \phi - \f{\la \rd_j\phi,\, \rd_i\phi \ra_{L^2(\Sigma_0, dx)}}{\|\rd_i\phi \|_{L^2(\Sigma_0)}^2} \rd_i\phi \right\|_{H^{-3}(\Sigma_0)}.$$
Assuming also that \eqref{eq:nondegenerate.in.terms.of.H-3} holds, this implies that \eqref{eq:nondegeneracy} holds for this choice of $\psi_j$. \qedhere
\end{proof}

\subsection{Construction of impulsive wave data}\label{sec:appendix.constraint.for.rough}

It is now straightforward to apply Lemma~\ref{lem:constraint} to construct initial data set satisfying conditions in Definition~\ref{roughdef}. We will simply be content with constructing some --- instead of classifying all --- such initial data sets. To simplify the exposition, let us construct special examples\footnote{This particular choice we make here allows \eqref{eq:integrability.in.appendix} to be checked more easily.} such that $\tphi$ and $\tphi'$ are of size $O(\ep^2)$, $\rphi$ is of size $O(\ep)$, and $\rphi'$ is of size $O(\ep^{\f 32})$.

\begin{lem}\label{lem:data.exist!}
There exist a large class of admissible initial data sets $(\gamma, K, \phi, \phi')$ satisfying the assumptions of Definition~\ref{roughdef}.
\end{lem}
\begin{proof}
\pfstep{Step~1: Prescribing $\rphi$, $\tphi$ and $\underline{\widetilde{\phi}}_k$} Impose $\rphi$, $\tphi$ and $\underline{\widetilde{\phi}}_k$ satisfying the following conditions:
\begin{itemize}
\item The transversality condition in \ref{data.transverse} and the support properties in \ref{data2} of Definition~\ref{roughdef} hold. Moreover, $\rphi$ and $\tphi$ satisfies the stronger support assumptions $\mathrm{supp}(\rphi),\,\mathrm{supp}(\tphi)\subseteq B(0,\f R4)$.
\item The following estimates hold for $k=1,2,3$:
\begin{subequations}
	\begin{align}
	\label{eq:constraint.data.assumption.1}
	\|\rphi\|_{H^{2+s'}(\Sigma_0)} \leq &\: 0.1 \ep, \\
	\label{eq:constraint.data.assumption.2}
	\|\tphi\|_{W^{1,\infty}(\Sigma_0)} + \|\tphi\|_{H^{1+s'}(\Sigma_0)} + \|\underline{\widetilde{\phi}}_k \|_{L^\infty(\Sigma_0)} + \|\underline{\widetilde{\phi}}_k \|_{H^{s'}(\Sigma_0)} \leq &\: \ep^2, \\
	\label{eq:constraint.data.assumption.3}
	\| \de^{iq} \cdot c^{\perp}_{kq} \cdot \rd_i  \tphi\|_{H^{1+s''}(\Sigma_0)} + \|\de^{iq} \cdot c^{\perp}_{kq} \cdot \rd_i \underline{\widetilde{\phi}}_k\|_{H^{s''}(\Sigma_0)} + \|  \underline{\widetilde{\phi}}_k - \de^{iq} \cdot c_{kq} \cdot \rd_i  \tphi \|_{H^{1+s''}(\Sigma_0)} \leq &\: \ep^2.
	\end{align}
	\end{subequations}
\item For $k = 1,2,3$, there exist signed Radon measures $\mathfrak T_{i j,k}$, $\underline{\mathfrak  T}_{i,k}$, $\mathfrak T_{ijE,k}$, $\underline{\mathfrak  T}_{iE,k}$ and $\mathfrak T_{ijL,k}$ such that 
\begin{equation}
\begin{split}
\| \rd^2_{ij}\tphi - \mathfrak T_{i j,k}\|_{L^2(\Sigma_0)} + \| \rd_{i}\underline{\widetilde{\phi}}_k - \underline{\mathfrak  T}_{i,k}\|_{L^2(\Sigma_0)} + \| \rd^2_{ij}(\de^{lq} \cdot c^{\perp}_{kq} \cdot \rd_l  \tphi) - \mathfrak T_{ijE,k}\|_{L^2(\Sigma_0)} &\\
+ \| \rd_{i}(\de^{lq} \cdot c^{\perp}_{kq} \cdot \rd_l  \tphi') - \underline{\mathfrak T}_{iE,k}\|_{L^2(\Sigma_0)} + \| \rd^2_{ij}(\underline{\widetilde{\phi}}_k - \de^{lq} \cdot c_{kq} \cdot \rd_l  \tphi) - \mathfrak T_{ijL,k}\|_{L^2(\Sigma_0)} & \leq \ep^2,
\end{split}
\end{equation}
\begin{equation} 
\mathrm{supp}(\mathfrak T_{i j,k}) \cup \mathrm{supp}(\underline{\mathfrak  T}_{i,k}') \cup \mathrm{supp}(\mathfrak T_{i j E,k}) \cup \mathrm{supp}(\underline{\mathfrak  T}_{iE,k}') \cup  \mathrm{supp}(\mathfrak T_{i j L,k}) \subseteq \{ u_k=0\},
\end{equation}  
and
\begin{equation}\label{eq:constraint.data.assumption}
T.V.(\mathfrak T_{i j,k })+  T.V.(\underline{\mathfrak  T}_{i,k }') + T.V.(\mathfrak T_{i j E,k })+  T.V.(\underline{\mathfrak  T}_{i E,k }') + T.V.(\mathfrak T_{i j L,k }) \leq \ep^2.
\end{equation} 
\item For $\phi := \phi_{reg} + \sum_{k=1}^3 \tphi$, the non-degeneracy condition \ref{nondegeneracy} in Definition~\ref{roughdef} holds.
\end{itemize}

\pfstep{Step~2: Prescribing $\underline{\phi}_{reg}$ and using Lemma~\ref{lem:constraint}} 
By the non-degeneracy assumption and Lemma~\ref{lem:nondegeneracy.consequence} (with $\f R 2$ instead of $R$), we can now fix smooth functions $\psi_1$, $\psi_2$ compactly supported in $B(0,\f R 2)$ satisfying \eqref{eq:nondegeneracy}.

For $\bm \lambda = (\lambda_1, \lambda_2) \in [-\ep^{\f 54}, \ep^{\f 54}]\times [-\ep^{\f 54}, \ep^{\f 54}] =:\mathscr K$, define 
$$\phi:= \phi_{reg} + \sum_{k=1}^3 \tphi,\quad \underline{\phi}_{reg}^{\bm \lambda} := \sum_{j=1}^2 \lambda_j \psi_j,\quad \underline{\phi}^{\bm \lambda} := \underline{\phi}_{reg}^{\bm \lambda} + \sum_{k=1}^3 \underline{\widetilde{\phi}}_k.$$ 
We now check that $(\phi, \underline{\phi}^{\bm \lambda})$ obey conditions 1--3 preceding Lemma~\ref{lem:constraint}. The only non-trivial condition to check is condition~3, which translates to finding $\bm \lambda = (\lambda_1, \lambda_2) \in \mathscr K$ such that for $j=1,2$,
\begin{equation}\label{eq:rough.data.compatibility.rephrased}
\sum_{i=1}^2 \lambda_i \int_{\Sigma_0} e^{\underline{\gamma}} \psi_i \rd_j \phi \, dx^1\, dx^2 = - \int_{\Sigma_0} e^{\underline{\gamma}} (\sum_{k=1}^3 \widetilde{\uphi}_k) \rd_j \phi \, dx^1 \, dx^2.
\end{equation}
To see that \eqref{eq:rough.data.compatibility.rephrased} holds, note that given $\underline{\gamma}$ as in condition~3, we have $|e^{\underline{\gamma}} - 1| \leq \max \{e^{\underline{\gamma}}, e^{-\underline{\gamma}}\}|\underline{\gamma}| \ls \ep$ on $B(0,\f R2)$. Hence, using \eqref{eq:nondegeneracy}, \eqref{eq:constraint.data.assumption.1} and \eqref{eq:constraint.data.assumption.2}, we obtain that for $\ep_0$ sufficiently small,
\begin{equation}\label{eq:nondegeneracy.modulated.1}
\begin{split}
 &\:\left| \det  \begin{bmatrix}
 		\int_{\Sigma_0} e^{\underline{\gamma}}\psi_1 \rd_1 \phi \, dx^1\, dx^2 & 	\int_{\Sigma_0} e^{\underline{\gamma}}\psi_1 \rd_2 \phi \, dx^1\, dx^2   \\
 		\int_{\Sigma_0} e^{\underline{\gamma}}\psi_2 \rd_1 \phi \, dx^1\, dx^2 & 	\int_{\Sigma_0} e^{\underline{\gamma}} \psi_2 \rd_2 \phi \, dx^1\, dx^2  
 		\end{bmatrix} \right| \\
		\geq &\: \left| \det  \begin{bmatrix}
 		\int_{\Sigma_0} \psi_1 \rd_1 \phi \, dx^1\, dx^2 & 	\int_{\Sigma_0} \psi_1 \rd_2 \phi \, dx^1\, dx^2   \\
 		\int_{\Sigma_0} \psi_2 \rd_1 \phi \, dx^1\, dx^2 & 	\int_{\Sigma_0} \psi_2 \rd_2 \phi \, dx^1\, dx^2  
 		\end{bmatrix} \right| - C\ep^3 \gtrsim \ep^{\f52} - C\ep^3 \gtrsim \ep^{\f 52},
\end{split}
\end{equation}
which, after using again \eqref{eq:constraint.data.assumption.1}, \eqref{eq:constraint.data.assumption.2} and the upper bound $\|\psi_i \|_{H^3(\Sigma_t)} \leq 1$, imply the entry-wise bound
\begin{equation}\label{eq:rough.data.compatibility.rephrased.1}
\begin{bmatrix}
 		\int_{\Sigma_0} e^{\underline{\gamma}}\psi_1 \rd_1 \phi \, dx^1\, dx^2 & 	\int_{\Sigma_0} e^{\underline{\gamma}}\psi_1 \rd_2 \phi \, dx^1\, dx^2   \\
 		\int_{\Sigma_0} e^{\underline{\gamma}}\psi_2 \rd_1 \phi \, dx^1\, dx^2 & 	\int_{\Sigma_0} e^{\underline{\gamma}}\psi_2 \rd_2 \phi \, dx^1\, dx^2  
 		\end{bmatrix}^{-1} = O(\ep^{-\f 32}).
\end{equation}
On the other hand, using \eqref{eq:constraint.data.assumption.1}, \eqref{eq:constraint.data.assumption.2}, we see that the RHS of \eqref{eq:rough.data.compatibility.rephrased} obeys
\begin{equation}\label{eq:rough.data.compatibility.rephrased.2}
\left| \int_{\Sigma_0} e^{\underline{\gamma}} (\sum_{k=1}^3 \widetilde{\uphi}_k) \rd_j \phi \, dx^1 \, dx^2 \right| \ls \max_k \|\widetilde{\uphi}_k \|_{L^2(\Sigma_t)} \|\rd_j \phi\|_{L^2(\Sigma_t)} \ls \ep^3.
\end{equation}
Combining \eqref{eq:rough.data.compatibility.rephrased.1} and \eqref{eq:rough.data.compatibility.rephrased.2}, we see that \eqref{eq:rough.data.compatibility.rephrased} can be solved with $\bm \lambda$ such that $|\lambda_1|,\,|\lambda_2| \ls \ep^{\f 32}$; in particular, $\bm \lambda \in \mathscr K$.

Therefore, Lemma~\ref{lem:constraint} shows that there exist $\bm \lambda_0 \in \mathscr K$ and an admissible initial data set $(\phi, \phi':= e^{-\gamma}\underline{\phi}^{\bm \lambda_0}, \gamma, K)$ with the prescribed data. Returning to the assumptions in Step~1 above, it is easy to check that the data set satisfies all the conditions in Definition~\ref{roughdef}. \qedhere
\end{proof}

\subsection{Construction of approximate data (Proof of Lemma~\ref{lem:data.approx})}\label{sec:data.approx}

\begin{proof}[Proof of Lemma~\ref{lem:data.approx}]
Let $(\phi, \phi', \gamma, K)$ be as in Definition~\ref{roughdef}. In particular, we have decompositions $\phi = \rphi + \sum_{k=1}^3 \tphi$ and $\phi' = \rphi' + \sum_{k=1}^3 \tphi'$.

\pfstep{Step~1: Definitions of $\tphi^{(\de),*}$ and $\widetilde{\underline{\phi}}_k^{(\de),*}$: one-dimensional mollifications} As a first step towards prescribing $\tphi^{(\de)}$ and $\widetilde{\underline{\phi}}_k^{(\de)}$ for the initial data of the $\de$-impulsive waves, we first define approximations of them, denoted respectively by $\tphi^{(\de),*}$ and $\widetilde{\underline{\phi}}_k^{(\de),*}$, which are non-smooth but already satisfy the desired estimates.

Let $\varkappa:\mathbb R\to [0,1]$ be smooth and such that $\varkappa(\tau) = 1$ for $\tau\leq 1$, $\varkappa(\tau) = 0$ for $\tau\geq 2$, and $\int_{\mathbb R} \varkappa = 1$.

Define $\tphi^{(\de),*}$ and $\widetilde{\underline{\phi}}_k^{(\de),*}$ by performing $1$-dimensional mollifications and translating by $\f \de 2$ in the direction parallel to $\rd_{u_k}$:
\begin{align}
\label{eq:the.mollification}
\tphi^{(\de),*}(x) := &\: \f 8\de \int_{\mathbb R}  \varkappa( \f{8s}{\de}) \tphi(x^1-c_{k1}s + \f 12 c_{k1} \de, x^2-c_{k2}s + \f 12 c_{k2} \de) \, ds, \\
\label{eq:the.other.mollification}
\widetilde{\underline{\phi}}_k^{(\de),*}(x) := &\: \f 8{\de} \int_{\mathbb R}  \varkappa(\f{8 s}{\de}) (e^{\gamma} \widetilde{\phi}'_k)(x^1-c_{k1}s + \f 12 c_{k1} \de, x^2-c_{k2}s + \f 12 c_{k2} \de) \, ds.
\end{align}

Using the constraint equation \eqref{eq:gamma.constraint}, we can prove a bound (see e.g.,~\eqref{eq:constraint.bounds.ep2}) $|e^{\gamma}|,\, |e^{-\gamma}| \leq \ep^{\f 32} \la x\ra^{\ep}$ (for $\ep_0$ sufficiently small). This allows us to pass between bounds for $\tilde{\phi}'_k$ and $e^{\gamma} \tilde{\phi}'_k$ with only a small error.

In particular, the following are easy to check.
\begin{enumerate}
\item $\tphi^{(\de),*}$ and $\widetilde{\underline{\phi}}_k^{(\de),*}$ are supported in $\{x\in \Sigma_0: u_k(0,x) \geq - \f {3\de} 4\}\cap B(0,\f {2R}3)$.

\begin{itemize}
\item To see this, note that the support being in $B(0,\f {2R}3)$ is obvious by \eqref{eq:the.mollification}, \eqref{eq:the.other.mollification} since  $\tphi$, $\widetilde{\phi}'_k$ are supported in $B(0,\f R2)$, and $\de\leq 1$, $R\geq 10$. Now since
\begin{itemize}
\item $\tphi(y)$ is supported in $u_k = a_k + c_{ki} y^i \geq 0$, 
\item $\varkappa(\f{8s}{\de})$ is supported in $s\in [-\f \de 4, \f \de 4]$, 
\end{itemize}
we deduce that $\tphi^{(\de),*}$ is non-vanishing only when $a_k + c_{ki} (x^i + \f 14 c_{ki}\de + \f 12 c_{ki} \de) \geq 0$. Since $\sum_{i=1}^2 c_{ki}^2 = 1$, this means that $\mathrm{supp}(\tphi^{(\de),*}) \subseteq \{ x\in  \Sigma_0: u_k \geq -\f {3\de}4\}$. Similarly for $\widetilde{\underline{\phi}}_k^{(\de),*}$.
\end{itemize}
\item There exists a decreasing function $b:[0,1] \to [0,1]$ with $\lim_{\de\to 0} b(\de) = 0$ such that 
\begin{equation}\label{eq:de.approx.H1.nonquant}
 \| \tphi^{(\de),*} - \tphi \|_{H^{1+s'}(\Sigma_0)} + \| \widetilde{\underline{\phi}}_k^{(\de),*} -  e^\gamma \widetilde{\phi}'_k\|_{H^{s'}(\Sigma_0)} \leq  b(\de). 
 \end{equation}
\item The following more quantitative convergences hold for lower norms:
\begin{equation}\label{eq:de.approx.H1.quant}
\| \tphi^{(\de),*} - \tphi \|_{H^1(\Sigma_0)} + \| \widetilde{\underline{\phi}}_k^{(\de),*} - e^\gamma \widetilde{\phi}'_k \|_{L^2(\Sigma_0)} \ls \ep \de^{s'}.
\end{equation}
\item $\tphi^{(\de),*}$ and $\widetilde{\underline{\phi}}_k^{(\de),*}$ obey the following estimates (which follow from the given estimates \eqref{eq:assumption.rough.energy} and \eqref{eq:assumption.rough.energy.commuted}):
\begin{align}
	\label{eq:constraint.data.smooth.assumption.2}
	\|\tphi^{(\de),*} \|_{W^{1,\infty}(\Sigma_0)} + \|\tphi^{(\de),*} \|_{H^{1+s'}(\Sigma_0)} + \|\underline{\widetilde{\phi}}_k^{(\de),*} \|_{L^\infty(\Sigma_0)} + \|\underline{\widetilde{\phi}}_k^{(\de),*} \|_{H^{s'}(\Sigma_0)} \leq &\: \f{3\ep}2, \\
	\label{eq:constraint.data.smooth.assumption.3}
	\| \de^{iq} c^{\perp}_{kq} \rd_i  \tphi^{(\de),*} \|_{H^{1+s''}(\Sigma_0)} + \|\de^{iq}  c^{\perp}_{kq} \rd_i \underline{\widetilde{\phi}}_k^{(\de),*} \|_{H^{s''}(\Sigma_0)} + \|  \underline{\widetilde{\phi}}_k^{(\de),*} - \de^{iq}  c_{kq}  \rd_i  \tphi^{(\de),*} \|_{H^{1+s''}(\Sigma_0)} \leq &\: \f{3\ep}2.
\end{align}
\item The following holds (which follow from the estimates and support properties in \eqref{eq:T.assumptions.1}--\eqref{eq:T.assumptions.3}):
\begin{align}
	\label{eq:delta.waves.proof.1}
		\notag \|  \tphi^{(\de),*} \|_{H^2(\Sigma_0)}+ \| \underline{\widetilde{\phi}}_k^{(\de),*} \|_{H^1(\Sigma_0)}+    \| E_k \tphi^{(\de),*}\|_{H^2(\Sigma_0)} \qquad \qquad \qquad & \\
		+ \| E_k  \underline{\widetilde{\phi}}_k^{(\de),*} \|_{H^1(\Sigma_0)} + \| \underline{\widetilde{\phi}}_k^{(\de),*} - \de^{iq}  c_{kq}  \rd_i  \tphi^{(\de),*} \|_{H^1(\Sigma_0)}  \leq &\:  \f{9\ep}4 \cdot \delta^{-\frac{1}{2}}, \\
	\label{eq:delta.waves.proof.2}
		 \| \tphi^{(\de),*} \|_{H^2(\Sigma_0 \setminus S^k(-\f{3\de}4,-\f \de 4))} + \| \underline{\widetilde{\phi}}_k^{(\de),*} \|_{H^1(\Sigma_0 \setminus S^k(-\f{3\de}4,-\f \de 4))} \leq &\: \f{3\ep}2.
\end{align}
\begin{itemize}
\item To see \eqref{eq:delta.waves.proof.2}, we in particular use a support argument like in point 1~above to show that the singular part does not contribute outside $S^k(-\f{3\de}4,-\f \de 4)$.
\end{itemize}

\end{enumerate}

\pfstep{Step~2: Defining $\phi_{reg}^{(\de)}$, $\tphi^{(\de)}$, $\uphi_{reg}^{\bm \lambda, (\de)}$ and $\widetilde{\uphi}_k^{(\de)}$} We now define $\phi^{(\de)}$ and $\uphi^{\bm\lambda, (\de)}$, in anticipation of using Lemma~\ref{lem:constraint} to obtain an admissible initial data set. We define
\begin{equation}\label{eq:de.approx.def}
\phi^{(\de)} := \phi_{reg}^{(\de)} + \sum_{k=1}^3 \tphi^{(\de)},\quad \uphi^{\bm\lambda, (\de)} := \uphi_{reg}^{\bm \lambda, (\de)} + \sum_{k=1}^3 \widetilde{\uphi}_k^{(\de)},
\end{equation}
where $\phi_{reg}^{(\de)}$, $\tphi^{(\de)}$, $ \uphi_{reg}^{\bm \lambda, (\de)}$ and $\widetilde{\uphi}_k^{(\de)}$ are defined as follows:
\begin{enumerate}
\item For $\phi_{reg}$ as in Definition~\ref{roughdef}, define $\phi_{reg}^{(\de)}$ to be a smooth approximation of $\rphi$, supported in $B(0,R)$, and such that
\begin{equation}\label{eq:de.approx.H1.quant.0}
\| \rphi^{(\de)} - \rphi \|_{H^{2+s'}(\Sigma_0)} \leq \ep \de^{s'}.
\end{equation}
\item For $\tphi^{(\de),*}$, $\widetilde{\underline{\phi}}_k^{(\de),*}$ as in Step~1, define $\tphi^{(\de)}$ and $\widetilde{\underline{\phi}}_k^{(\de)}$ to respectively be smooth approximations of $\tphi^{(\de),*}$ and $\widetilde{\underline{\phi}}_k^{(\de),*}$, with support in $B(0, R)$, such that the following holds:
\begin{enumerate}
\item $\tphi^{(\de)}$ and $\widetilde{\underline{\phi}}_k^{(\de)}$ are supported in $\{x\in \Sigma_0: u_k(0,x) \geq - \de \}\cap B(0,R)$.
\item $(\tphi^{(\de)}, \widetilde{\underline{\phi}}_k^{(\de)})$ is close to $(\tphi^{(\de),*}, \widetilde{\underline{\phi}}_k^{(\de),*})$ in the following sense:
\begin{equation}\label{eq:de.approx.H1.quant.1}
\| \tphi^{(\de)} - \tphi^{(\de),*} \|_{H^{1+s'}(\Sigma_0)} + \| \widetilde{\underline{\phi}}_k^{(\de)} - \widetilde{\underline{\phi}}_k^{(\de),*}\|_{H^{s'}(\Sigma_0)} \leq \ep \de^{s'}.
\end{equation}
\item The estimates \eqref{eq:constraint.data.smooth.assumption.2}--\eqref{eq:delta.waves.proof.2} hold, with $(\tphi^{(\de),*}, \widetilde{\underline{\phi}}_k^{(\de),*})$ replaced by $(\tphi^{(\de)}, \widetilde{\underline{\phi}}_k^{(\de)})$, with $\ep$ replaced by $\f{4\ep}3$, and (in the case of \eqref{eq:delta.waves.proof.2}) with $S^k(-\f{3\de}4,-\f \de 4)$ replaced by $S^k(-\de,0)$.
\end{enumerate}
\item For $\uphi_{reg}^{\bm \lambda, (\de)}$, first define $\phi_{reg}'^{(\de),*}$ be a smooth approximation of $\phi_{reg}'$, supported in $B(0,R)$, and such that
\begin{equation}\label{eq:de.approx.H1.quant.2}
\| \phi_{reg}'^{(\de),*} - \rphi' \|_{H^{2+s'}(\Sigma_0)} \leq \ep \de^{s'}.
\end{equation}
Next, using \eqref{eq:nondegeneracy.in.thm}, we can fix $\psi_1$, $\psi_2$ as in the conclusion of Lemma~\ref{lem:nondegeneracy.consequence}. Define then 
\begin{equation}
\uphi_{reg}^{\bm \lambda, (\de)} := e^{\gamma} \phi_{reg}'^{(\de),*} + \lambda_1 \psi_1 + \lambda_2 \psi_2.
\end{equation}
\end{enumerate}

\pfstep{Step~3: Verifying \eqref{eq:integrability.in.appendix}} In order to apply Lemma~\ref{lem:constraint}, we first need to verify \eqref{eq:integrability.in.appendix}. 

Suppose $\underline{\gamma} = -\underline{\gamma}_{asymp} \om(|x|) \log|x| + \widetilde{\underline{\gamma}}$ is given such that $|\underline{\gamma}_{asymp}|\leq \ep$ and $\|\widetilde{\underline{\gamma}}\|_{W^{1,4}_{\f 14}(\Sigma_0)}\leq \ep$. Let $\mathscr K = [-\ep,\ep]\times [-\ep,\ep]$ and our goal is to find $\bm \lambda = (\lambda_1, \lambda_2) \in \mathscr K$ such that \eqref{eq:integrability.in.appendix} is satisfied. Given the definitions in Step~2, this means that need to solve for $\lambda_j$ (with $j=1,2$) which satisfies
\begin{equation}\label{eq:solving.for.modulation.delta}
\begin{split}
&\: \sum_{i=1}^2 \lambda_i \int_{\Sigma_0} e^{\underline \gamma} \psi_i \rd_j \phi^{(\de)} \, dx^1 \, dx^2 
= - \int_{\Sigma_0} (e^{2 \underline\gamma} \widetilde{\phi}_{reg}'^{(\de),*} + e^{\underline{\gamma}} \sum_{k=1}^3 \widetilde{\uphi}^{(\de)}_k) \rd_j \phi^{(\de)} \, dx^1 \, dx^2.
\end{split}
\end{equation}

\pfstep{Step~3(a): Controlling the RHS of \eqref{eq:solving.for.modulation.delta}}
We compute
\begin{equation}\label{eq:diff.in.compatibility}
\begin{split}
&\: (e^{2 \underline\gamma} \phi_{reg}'^{(\de),*} + e^{\underline{\gamma}} \sum_{k=1}^3 \widetilde{\uphi}^{(\de)}_k) \rd_j \phi^{(\de)} \\
= &\: \underbrace{(e^{2 \gamma} \phi_{reg}' + e^{2\gamma} \sum_{k=1}^3 \tphi') \rd_j \phi}_{=:I} + \underbrace{[ (e^{2 \underline\gamma} - e^{2\gamma}) \phi_{reg}'^{(\de),*} + ( e^{\underline{\gamma}} - e^\gamma ) \sum_{k=1}^3 \widetilde{\uphi}^{(\de)}_k ]\rd_j \phi^{(\de)}}_{=:II}\\
&\: + \underbrace{ e^{2 \gamma} (\phi_{reg}'^{(\de),*}\rd_j \phi^{(\de)} - \widetilde{\phi}_{reg}'\rd_j \phi)  + e^{\gamma} (\rd_j \phi^{(\de)} \sum_{k=1}^3 \widetilde{\uphi}^{(\de)}_k - \rd_j \phi \sum_{k=1}^3 e^\gamma \widetilde{\phi}'_k)}_{=:III} .
\end{split}
\end{equation}
Term I in \eqref{eq:diff.in.compatibility} integrates to $0$ since $(\gamma, K, \phi, \phi')$ is a given admissible initial data set and thus obeys \eqref{compatibility}. For the term $II$ in \eqref{eq:diff.in.compatibility}, note that $\gamma$ and $\underline{\gamma}$ can be bounded respectively by \eqref{eq:constraint.bounds.ep2} and the assumptions on $\underline{\gamma}$ (see beginning of Step~3). Hence, using $|e^x - 1| \leq \max\{e^x, e^{-x}\} |x|$, we obtain the following bound on $B(0,\f R2)$,
\begin{equation}\label{eq:integrability.de.termII.prep}
|e^{2 \underline\gamma} - e^{2\gamma}| \leq |e^{2 \underline\gamma} - 1| +  | e^{2\gamma} - 1 | \leq \max\{ e^{2 \underline\gamma}, e^{2\gamma}, \, e^{-2 \underline\gamma}, e^{-2\gamma}\} |(\underline\gamma| +  |\gamma|) \ls \ep.
\end{equation}
Using also the bounds for $\phi_{reg}'^{(\de),*}$, $\widetilde{\uphi}^{(\de)}_k$ and $\phi^{(\de)}$ from Step~2, we thus have
\begin{equation}\label{eq:integrability.de.termII}
\left| \int_{\Sigma_0} \mbox{(Term $II$ in \eqref{eq:diff.in.compatibility})} \, dx^1 \, dx^2 \right| \ls \ep^3.
\end{equation}
Finally, for term $III$, we use \eqref{eq:de.approx.H1.quant}, \eqref{eq:de.approx.def}--\eqref{eq:de.approx.H1.quant.2} together with \eqref{eq:assumption.rough.energy} and the bounds for $\gamma$ as for term $II$ to obtain 
\begin{equation}\label{eq:integrability.de.termIII}
\left| \int_{\Sigma_0} \mbox{(Term $III$ in \eqref{eq:diff.in.compatibility})} \, dx^1 \, dx^2 \right| \ls \ep^2 \de^{s'}.
\end{equation}

\pfstep{Step~3(b): Solving \eqref{eq:solving.for.modulation.delta}} Using \eqref{eq:datareg}, \eqref{eq:assumption.rough.energy}, \eqref{eq:de.approx.def}, \eqref{eq:de.approx.H1.quant}, \eqref{eq:de.approx.H1.quant.0}, \eqref{eq:de.approx.H1.quant.1}, and the fact that $|e^{\underline{\gamma}}-1| \ls \ep$ on the supports of $\psi_j$ and $\phi$, it follows that
\begin{equation}\label{eq:nondegeneracy.modulated}
\begin{split}
 &\:\left| \det  \begin{bmatrix}
 		\int_{\Sigma_0} e^{\underline{\gamma}}\psi_1 \rd_1 \phi^{(\de)} \, dx^1\, dx^2 & 	\int_{\Sigma_0} e^{\underline{\gamma}}\psi_1 \rd_2 \phi^{(\de)} \, dx^1\, dx^2   \\
 		\int_{\Sigma_0} e^{\underline{\gamma}}\psi_2 \rd_1 \phi^{(\de)} \, dx^1\, dx^2 & 	\int_{\Sigma_0} e^{\underline{\gamma}} \psi_2 \rd_2 \phi^{(\de)} \, dx^1\, dx^2  
 		\end{bmatrix} \right| \\
		\geq &\: \left| \det  \begin{bmatrix}
 		\int_{\Sigma_0} \psi_1 \rd_1 \phi \, dx^1\, dx^2 & 	\int_{\Sigma_0} \psi_1 \rd_2 \phi \, dx^1\, dx^2   \\
 		\int_{\Sigma_0} \psi_2 \rd_1 \phi \, dx^1\, dx^2 & 	\int_{\Sigma_0} \psi_2 \rd_2 \phi \, dx^1\, dx^2  
 		\end{bmatrix} \right| - C (\ep \de^{s'} + \ep^3) \gtrsim \ep^{\f52} - C (\ep \de^{s'} + \ep^3) \gtrsim \ep^{\f 52},
\end{split}
\end{equation}
for $\ep_0$ (and hence $\ep$) sufficiently small.

Using the lower bound on the determinant in \eqref{eq:nondegeneracy.modulated}, as well as the upper bounds in \eqref{eq:assumption.rough.energy} and \eqref{eq:nondegeneracy}, it follows that we have the entry-wise bound
\begin{equation}\label{eq:modulated.inverse.est}
\begin{bmatrix}
 		\int_{\Sigma_0} e^{\underline{\gamma}}\psi_1 \rd_1 \phi \, dx^1\, dx^2 & 	\int_{\Sigma_0} e^{\underline{\gamma}}\psi_1 \rd_2 \phi \, dx^1\, dx^2   \\
 		\int_{\Sigma_0} e^{\underline{\gamma}}\psi_2 \rd_1 \phi \, dx^1\, dx^2 & 	\int_{\Sigma_0} e^{\underline{\gamma}}\psi_2 \rd_2 \phi \, dx^1\, dx^2  
 		\end{bmatrix}^{-1} = O(\ep^{-\f 32}).
\end{equation}

Using \eqref{eq:modulated.inverse.est} and the estimates in Step~3(a), and recalling also that $\de^{s'} \leq \ep$, we then invert the linear matrix to solve \eqref{eq:solving.for.modulation.delta} with some $(\lambda_1, \lambda_2)$ satisfying
\begin{equation}\label{eq:easy.bounds.on.lambda}
 |\lambda_1| + |\lambda_2| \ls \ep^{-\f 32} (\ep^3 + \ep^{2} \de^{s'}) \ls \ep^{\f 32}.
 \end{equation}
In particular, $\bm \lambda = (\lambda_1, \lambda_2) \in \mathscr K$. We have thus verified \eqref{eq:integrability.in.appendix}.

\pfstep{Step~4: Application of Lemma~\ref{lem:constraint}} By Step~3 and Lemma~\ref{lem:constraint}, we know that there exist $\bm \lambda_0$ and functions $(\gamma^{(\de)},K^{(\de)})$ such that $(\phi^{(\de)}, (\phi')^{(\de)} = e^{-\gamma^{(\de)}} \uphi^{\bm\lambda_0,(\de)}, \gamma^{(\de)}, K^{(\de)})$ is an admissible initial data set.

\pfstep{Step~5: Checking the conclusions of Lemma~\ref{lem:data.approx}} First, we prove point 1~of Lemma~\ref{lem:data.approx}, i.e.~we check that $(\phi^{(\de)}, (\phi')^{(\de)} = e^{-\gamma^{(\de)}} \uphi^{\bm\lambda_0,(\de)}, \gamma^{(\de)}, K^{(\de)})$ (given by Step~4) forms an admissible initial data set for three $\de$-impulsive waves with parameters $(3\ep,s',s'',2R,\upkappa_0)$ in Definition~\ref{smoothdef}.
\begin{itemize}
\item The transversality condition holds trivially since $c_{ki}$ is defined as for the given $(\phi, \phi', \gamma, K)$.
\item The required support properties follow from points 1, 2(a), and 3 in Step~2.
\item The estimates in points \ref{datareg}--\ref{data5} of Definition~\ref{roughdef} and in \eqref{eq:delta.waves.1}--\eqref{eq:delta.waves.2} follow easily from \eqref{eq:de.approx.def}, the conditions on $\phi_{reg}^{(\de)}$, $\tphi^{(\de)}$, $ \uphi_{reg}^{\bm \lambda, (\de)}$ and $\widetilde{\uphi}_k^{(\de)}$ in Step~2, together with $\|\psi_j \|_{H^3(\Sigma_0)} \leq 1$ and the bound \eqref{eq:easy.bounds.on.lambda}.
\end{itemize}

We finally need to check the desired convergence (point 2~of Lemma~\ref{lem:data.approx}). By the definition \eqref{eq:de.approx.def}, and the estimates \eqref{eq:de.approx.H1.nonquant}, \eqref{eq:de.approx.H1.quant.0}, \eqref{eq:de.approx.H1.quant.1} and \eqref{eq:de.approx.H1.quant.2}, we have
\begin{equation}\label{eq:check.convergence.1}
\begin{split}
&\: \|\rphi^{(\de)} - \rphi \|_{H^{1+s'}(\Sigma_0)} + \|(\rphi')^{(\de)} - \rphi' \|_{H^{s'}(\Sigma_0)}  \\
&\: \qquad + \max_k (\|\tphi^{(\de)} - \tphi \|_{H^{1+s'}(\Sigma_0)} + \|(\tphi')^{(\de)} - \tphi' \|_{H^{s'}(\Sigma_0)})
 \ls b(\de) + |\lambda_1| + |\lambda_2|.
\end{split}
\end{equation}

To proceed, we need to bound $|\lambda_1|$, $|\lambda_2|$ with an estimate better than \eqref{eq:easy.bounds.on.lambda}. For this, we need a better bound compared to \eqref{eq:integrability.de.termII}. Instead of \eqref{eq:integrability.de.termII.prep}, we use Sobolev embedding (part 1 of Proposition~\ref{prop:Sobolev.weighted}) to obtain $|e^{2\gamma^{(\de)}} - e^{2\gamma}| \ls |\gamma_{asymp} - \gamma^{(\de)}_{asymp}| + \|\gamma - \gamma^{(\de)}\|_{H^2_{-\f 18}(\Sigma_0)}$ on the support of $\phi$. This in turn implies that 
$$\left| \int_{\Sigma_0} \mbox{(Term $II$ in \eqref{eq:diff.in.compatibility})} \, dx^1 \, dx^2 \right| \ls \ep^2(|\gamma_{asymp} - \gamma^{(\de)}_{asymp}| + \|\gamma - \gamma^{(\de)}\|_{H^2_{-\f 18}(\Sigma_0)}).$$
Hence, combining this with \eqref{eq:modulated.inverse.est}, \eqref{eq:integrability.de.termIII} and the fact that $I$ in \eqref{eq:diff.in.compatibility} integrates to $0$, we obtain
\begin{equation}\label{eq:check.convergence.2}
|\lambda_1| + |\lambda_2| \ls \ep^{\f 12} \de^{s'} + \ep^{\f 12} (|\gamma_{asymp} - \gamma^{(\de)}_{asymp}| + \|\gamma - \gamma^{(\de)}\|_{H^2_{-\f 18}(\Sigma_0)}).
\end{equation}
On the other hand, taking the difference of the elliptic equations for $(\gamma, K)$ (which hold because of Definition~\ref{roughdef}) and those for $(\gamma^{(\de)}, K^{(\de)})$ (which hold because of Step~4), we have
\begin{equation}\label{eq:check.convergence.3}
\begin{split}
&\: |\gamma_{asymp} - \gamma^{(\de)}_{asymp}| + \|\gamma - \gamma^{(\de)}\|_{H^2_{-\f 18}(\Sigma_0)} + \| K - K^{(\de)} \|_{H^1_{\f 78}(\Sigma_0)} \\
\ls &\: \ep (\de^{s'} + |\lambda_1| + |\lambda_2| + |\gamma_{asymp} - \gamma^{(\de)}_{asymp}| + \|\gamma - \gamma^{(\de)}\|_{H^2_{-\f 18}(\Sigma_0)} + \| K - K^{(\de)} \|_{H^1_{\f 78}(\Sigma_0)}).
\end{split}
\end{equation}

By \eqref{eq:check.convergence.1}--\eqref{eq:check.convergence.3}, we thus obtain
\begin{align*}
\mbox{LHS of \eqref{eq:check.convergence.1}} \ls &\: b(\de) + \ep \de^{s'}, \\
|\lambda_1| + |\lambda_2|  +|\gamma_{asymp} - \gamma^{(\de)}_{asymp}| + \|\gamma - \gamma^{(\de)}\|_{H^2_{-\f 18}(\Sigma_0)} + \| K - K^{(\de)} \|_{H^1_{\f 78}(\Sigma_0)} \ls &\: \ep \de^{s'}.
\end{align*}
Since $\lim_{\de\to 0} b(\de) = 0$ by point 2~in Step~1, this implies the desired convergence statement. \qedhere
\end{proof}

\bibliographystyle{DLplain}
\bibliography{Threewaves}

\end{document}